\documentclass[aps,prd,superscriptaddress,eqsecnum,floatfix,nofootinbib,preprint,tightenlines]{revtex4-2}

\usepackage[english]{babel}
\usepackage{amsmath,amssymb,bm,slashed}
\usepackage{amsthm}
\usepackage{graphicx}
\usepackage[sort&compress]{natbib}
\usepackage[usenames,dvipsnames,svgnames]{xcolor}
\usepackage{centernot}
\usepackage[hidelinks]{hyperref}
\usepackage[capitalize]{cleveref}
\usepackage[counterclockwise]{rotating} 

\allowdisplaybreaks

\crefname{section}{Sec.}{Secs.}
\Crefname{section}{Section}{Sections}
\crefrangeformat{equation}{Eqs.~(#3#1#4)--(#5#2#6)}
\Crefrangeformat{equation}{Equations~(#3#1#4)--(#5#2#6)}
\crefrangeformat{figure}{Figs.~#3#1#4--#5#2#6}
\Crefrangeformat{figure}{Figures~#3#1#4--#5#2#6}

\definecolor{red}{rgb}{0.8,0.0,0.0}
\definecolor{green}{rgb}{0.0,0.6,0.0}
\definecolor{darkblue}{rgb}{0.0,0.1,0.7}
\definecolor{brown}{rgb}{0.6,0.1,0.0}
\definecolor{gray}{rgb}{0.6,0.6,0.6}
\definecolor{darkgreen}{rgb}{0.0, 0.545098, 0.0}
\definecolor{verydarkgreen}{rgb}{0.0, 0.4, 0.0}
\definecolor{veryverydarkgreen}{rgb}{0.0, 0.3, 0.0}
\definecolor{purple}{rgb}{0.5,0.0,0.5}
\definecolor{applegreen}{rgb}{0.55, 0.71, 0.0}
\definecolor{babypink} {rgb}{0.64, 0.44, 0.44}
\definecolor{orange}{rgb}{0.9,0.4,0.0}
\definecolor{darkorange}{rgb}{0.8,0.35,0.0}

\definecolor{dkgreen}{rgb}{0.0,0.4,0.0}

\newcommand{\str}[1]{}    

\newtheorem{theorem}{Theorem}[section]
\newtheorem{lemma}[theorem]{Lemma}

\DeclareMathOperator{\Tr}{Tr}

\newcommand{\pr}{\ensuremath{\text{pr}}} 

\newcommand{\Vud}{\ensuremath{|V_{ud}|}}
\newcommand{\Vus}{\ensuremath{|V_{us}|}}
\newcommand{\Vub}{\ensuremath{|V_{ub}|}}
\newcommand{\Vcd}{\ensuremath{|V_{cd}|}}
\newcommand{\Vcs}{\ensuremath{|V_{cs}|}}
\newcommand{\Vcb}{\ensuremath{|V_{cb}|}}
\newcommand{\Vcx}{\ensuremath{|V_{cx}|}}  
\newcommand{\bra}[1]{\langle #1 |}
\newcommand{\ket}[1]{| #1 \rangle}
\newcommand{\avg}[1]{\left\langle #1 \right\rangle}

\newcommand{\vacuum}{\emptyset}

\newcommand{\abs}[1]{\left| #1 \right|} 
\newcommand{\order}[1]{\mathcal{O}(#1)} 
\newcommand{\matrixel}[3]{\left< #1 \vphantom{#2#3} \right| #2 \left| #3 \vphantom{#1#2} \right>} 

\DeclareMathOperator{\arccosh}{arcCosh}

\DeclareMathOperator{\PV}{PV}
\DeclareMathOperator{\PAV}{PAV}
\DeclareMathOperator{\sgn}{sgn}
\DeclareMathOperator{\diag}{diag}

\newcommand{\GeV}{~\ensuremath{\text{GeV}}}   
\newcommand{\MeV}{~\ensuremath{\text{MeV}}}  
\newcommand{\fm}{~\ensuremath{\text{fm}}} 

\newcommand{\Dpi}{\ensuremath{D\to\pi}}
\newcommand{\DK}{\ensuremath{D\to K}}
\newcommand{\DsK}{\ensuremath{D_s\to K}}

\newcommand{\RDpi}{\ensuremath{R_{\mu/e}^{\Dpi}}}
\newcommand{\RDK}{\ensuremath{R_{\mu/e}^{\DK}}}
\newcommand{\RDsK}{\ensuremath{R_{\mu/e}^{\DsK}}}

\newcommand{\VcdCombined}{\ensuremath{0.2238(11)^{\rm Expt}(15)^{\rm QCD}(04)^{\rm EW}(02)^{\rm SIB}[22]^{\rm QED}}}
\newcommand{\VcdFinal}{\VcdCombined}

\newcommand{\VcsCombined}{\ensuremath{0.9589(23)^{\rm Expt}(40)^{\rm QCD}(15)^{\rm EW}(05)^{\rm SIB}[95]^{\rm QED}}}
\newcommand{\VcsFinal}{\VcsCombined}
\newcommand{\VcdElectronicDsK}{\ensuremath{0.258(15)^{\rm Expt}(01)^{\rm QCD}[03]^{\rm QED}}}
\newcommand{\VcdFinalDsK}{\VcdElectronicDsK}

\newcommand{\RDpiFinal}{\ensuremath{0.98671(17)^{\rm QCD}[500]^{\rm QED}}}
\newcommand{\RDKFinal}{\ensuremath{0.97606(16)^{\rm QCD}[500]^{\rm QED}}}
\newcommand{\RDsKFinal}{\ensuremath{0.98099(10)^{\rm QCD}[500]^{\rm QED}}}

\newcommand{\azaf}{Department of Physics, University of Arizona, Tucson, Arizona 85721, USA}
\newcommand{\coloaf}{Department of Physics, University of Colorado, Boulder, Colorado 80309, USA}
\newcommand{\fnalaf}{Theory Division, Fermi National Accelerator Laboratory, Batavia, Illinois, 60510, USA}
\newcommand{\iuaf}{Department of Physics, Indiana University, Bloomington, Indiana 47405, USA}
\newcommand{\mitaf}{Center for Theoretical Physics, Massachusetts Institute of Technology, \\ Cambridge, MA 02139, USA}
\newcommand{\msuaf}{Department of Computational Mathematics, Science and Engineering, and Department of Physics and Astronomy, Michigan State University, East Lansing, Michigan 48824, USA}
\newcommand{\ucsbaf}{Department of Physics, University of California, Santa Barbara, California 93106, USA}
\newcommand{\ugraf}{CAFPE and Departamento de Física Teórica y del Cosmos, \\ Universidad de Granada, E-18071 Granada, Spain}
\newcommand{\uiucaf}{Department of Physics, University of Illinois, Urbana, Illinois, 61801, USA}
\newcommand{\icasuuiaf}{Illinois Center for Advanced Studies of the Universe, University of Illinois, \\ Urbana, Illinois, 61801, USA}
\newcommand{\unizar}{Departmento de Física Teórica, Universidad de Zaragoza, \\ 50009 Zaragoza, Spain}
\newcommand{\utahaf}{Department of Physics and Astronomy, University of Utah, \\ Salt Lake City, Utah 84112, USA}

\begin{document}
\preprint{MIT-CTP/5513, FERMILAB-PUB-22-943-T}
\title{\texorpdfstring{\boldmath$D$}{D}-meson semileptonic decays to pseudoscalars from \\ four-flavor lattice QCD}

\author{Alexei~Bazavov}\affiliation{\msuaf}
\author{Carleton~DeTar}\affiliation{\utahaf}
\author{Aida~X.~El-Khadra}\affiliation{\uiucaf}\affiliation{\icasuuiaf}
\author{Elvira~Gámiz}\affiliation{\ugraf}
\author{Zechariah~Gelzer}\affiliation{\uiucaf}
\author{Steven~Gottlieb}\affiliation{\iuaf}
\author{William~I.~Jay}\email{willjay@mit.edu}\affiliation{\mitaf}
\author{Hwancheol~Jeong}\affiliation{\iuaf}
\author{Andreas~S.~Kronfeld}\affiliation{\fnalaf}
\author{Ruizi~Li}\affiliation{\iuaf}
\author{Andrew~T.~Lytle}\affiliation{\uiucaf}\affiliation{\icasuuiaf}
\author{Paul~B.~Mackenzie}\affiliation{\fnalaf}
\author{Ethan~T.~Neil}\affiliation{\coloaf}
\author{Thomas~Primer}\affiliation{\azaf}
\author{James~N.~Simone}\affiliation{\fnalaf}
\author{Robert~L.~Sugar}\affiliation{\ucsbaf}
\author{Doug~Toussaint}\affiliation{\azaf}
\author{Ruth~S.~\surname{Van~de~Water}}\affiliation{\fnalaf}
\author{Alejandro~Vaquero}\affiliation{\utahaf}\affiliation{\unizar}


\collaboration{Fermilab Lattice and MILC Collaborations}
\noaffiliation

\date{\today}

\begin{abstract}
We present lattice-QCD calculations of the hadronic form factors for the semileptonic decays $D\to\pi\ell\nu$, $D\to K\ell\nu$, and $D_s\to K\ell\nu$.
Our calculation uses the highly improved staggered quark (HISQ) action for all valence and sea quarks and includes $N_f=2+1+1$ MILC ensembles with lattice spacings ranging from $a\approx0.12\fm$ down to $0.042\fm$. 
At most lattice spacings, an ensemble with physical-mass light quarks is included. 
The HISQ action allows all the quarks to be treated with the same relativistic light-quark action, allowing for nonperturbative renormalization using partial conservation of the vector current.
We combine our results with experimental measurements of the differential decay rates to determine $\Vcd^{\Dpi}=\VcdFinal$~and $\Vcs^{\DK}=\VcsFinal$. 
This result for $\Vcd$ is the most precise to date, with a lattice-QCD error that is, for the first time for the semileptonic extraction, at the same level as the experimental error.
Using recent measurements from BES~III, we also give the first-ever determination of $\Vcd^{\DsK}=\VcdFinalDsK$ from $\DsK l\nu$.
Our results also furnish new Standard Model calculations of the lepton flavor universality ratios
$\RDpi=\RDpiFinal$,
$\RDK=\RDKFinal$, and
$\RDsK=\RDsKFinal$,
which are consistent within $2\sigma$ with experimental measurements.
Our extractions of $\Vcd$ and $\Vcs$, when combined with a value for $\Vcb$, provide the most precise test of second-row CKM unitarity, finding agreement with unitarity at the level of one standard deviation.
\end{abstract}

\maketitle

\tableofcontents

\section{Introduction}

Historically, measurements in quark-flavor physics have a strong precedent of anticipating the direct discovery of new particles.
To name one instance, consider the charm quark, decays of which are the subject of this paper.
Its existence was conjectured on the basis of symmetry~\cite{Bjorken:1964gz, Glashow:1970gm}, and its mass was predicted to explain the rates of strangeness-changing neutral-current processes~\cite{Glashow:1970gm, Gaillard:1974mw}.
The discovery of the $J/\psi$~\cite{E598:1974sol, SLAC-SP-017:1974ind} was then immediately interpreted as charmonium~\cite{Appelquist:1974zd, DeRujula:1974rkb, Appelquist:1974yr, Eichten:1974af}.
Another example is the measurement in 1987 of large mixing in neutral $B$~mesons by the ARGUS Collaboration~\cite{ARGUS:1987xtv}, which suggested the unusually large mass for the top quark (see, e.g., Ref.~\cite{Marciano:1989xd}), eight years before its direct observation at the Tevatron in 1995~\cite{CDF:1994vkk, CDF:1995wbb, D0:1995jca}.
In light of several anomalies in measurements of $B$-meson decays and tension in several tests of the Standard Model (SM) flavor
structure~\cite{Artuso:2022ijh, Blanke:2022deg}, one can speculate that this area of particle physics is again pointing toward something new.
To illuminate the situation, it is timely to improve the theoretical ingredients in confronting experiment with the Standard Model for other quark-flavor processes.
In this paper, we report on lattice-QCD calculations relevant to the second row of the Cabibbo-Kobayashi-Maskawa (CKM) matrix, enabling stringent tests of second-row CKM unitarity.

Within the Standard Model (SM), charged-current flavor-changing processes are described by the CKM matrix
\begin{equation}
    V_{\rm CKM} = 
    \begin{pmatrix}
    V_{ud}  & V_{us}    & V_{ub}\\
    V_{cd}  & V_{cs}    & V_{cb}\\
    V_{td}  & V_{ts}    & V_{tb}
    \end{pmatrix},
\end{equation}
which describes the mismatch between the propagating mass eigenstates and the flavor eigenstates which participate in the weak interaction.
By construction, the CKM matrix is unitary, so each row and column should have unit norm.
Deviations from this expectation can arise if $V_{\rm CKM}$ is a $3\times3$ submatrix in an extended flavor sector or if non-SM processes contribute to measured decay and mixing rates.
It is important to test the CKM paradigm using independent determinations from multiple processes, for example, comparing leptonic and semileptonic decays with the same flavor charge.
Improved precision for the individual matrix elements leads directly to more stringent tests of the Standard Model.
Any statistically significant deviation from the predictions of CKM-unitarity would constitute evidence for new physics beyond the Standard Model.

The strongest test of unitarity comes from the first row, where the matrix elements are determined most precisely, with the exception of \Vub, which plays a negligible role in the first row unitarity relation at the current level of precision. 
Either taking the most precise value of $\Vud$ that comes from superallowed $\beta$ decays~\cite{Hardy:2020qwl}~\footnote{Recent calculations of the universal electroweak radiative corrections relevant for superallowed $\beta$ decays in Refs.~\cite{Seng:2018yzq, Seng:2018qru, Czarnecki:2019mwq, Seng:2020wjq, Shiells:2020fqp} found larger values than those estimated before, shifting the central value of $|V_{ud}|$ and increasing the tension with unitarity. In addition, further, previously unaccounted, nuclear-structure uncertainties in the inner radiative correction have considerably increased the error for earlier determinations~\cite{Seng:2018qru, Gorchtein:2018fxl}.}
and $\Vus$ as extracted from semileptonic $K_{\ell 3}\equiv K\to\pi\ell\nu$ decays, or using only inputs from kaon and pion decays ({\it i.e.}, $\Vus$ from semileptonic decays and $\Vus/\Vud$ from the ratio of leptonic decays, $K_{\ell 2}\equiv K\to\ell\nu$ over $\pi_{\ell 2}\equiv\pi\to\ell\nu$~\cite{Marciano:2004uf}), 
the combination $\Vud^2 + \Vus^2 + \Vub^2$ is in tension with unitarity at the $3\sigma$ level~\cite{Workman:2022ynf}. There is also a $\sim 3\sigma$ tension between the semileptonic and the leptonic determinations of $\Vus$~\cite{Workman:2022ynf}, where the leptonic determination uses $\Vud$ from superallowed decays as an external input.  
In those tests, the relevant QCD nonperturbative inputs for semileptonic and leptonic decays, the form factor $f_+^{K\pi}(0)$~\cite{Carrasco:2016kpy, FermilabLattice:2018zqv, RBCUKQCD:2015joy, Ishikawa:2022otj} and the ratio of decay constants $f_K/f_\pi$~\cite{Dowdall:2013rya, Carrasco:2014poa, Bazavov:2017lyh, Miller:2020xhy, Miller:2020xhy, Durr:2016ulb, QCDSF-UKQCD:2016rau}, respectively, are calculated using lattice QCD with uncertainties that have reached the $\sim0.18$\% level~\cite{Aoki:2021kgd}. Experimental data for the decay widths of $K_{\ell 3}$ and $K_{\ell 2}/\pi_{\ell 2}$ decays are similarly precise~\cite{Cirigliano:2022yyo,Moulson:2017ive}, leaving electromagnetic corrections as an important source of uncertainty in the extraction of the corresponding CKM matrix elements. Pioneering work addressing the calculation of structure-dependent QED corrections both for pion and kaon leptonic decays using lattice techniques~\cite{Giusti:2017dwk, DiCarlo:2019thl} and kaon semileptonic decays~\cite{Seng:2021boy, Seng:2021wcf, Seng:2022wcw} including lattice calculations of the $\gamma W$-box contribution, have been recently performed, opening the door to an important reduction of the electromagnetic uncertainty.

Similarly precise tests for the CKM matrix elements in the second row have been limited both by theory and experimental uncertainties.
On the theory side, the error for the decay constants $f_{D}$ and $f_{D_s}$ (roughly $0.35$--$0.2\%$~\cite{Aoki:2021kgd}) are now subleading in the extraction of $\Vcd$ and $\Vcs$, respectively, from leptonic decays thanks to the progress made by lattice calculations in the last years~\cite{Bazavov:2017lyh,Carrasco:2014poa}.
However, the situation is very different for semileptonic extractions of those CKM matrix elements.
Since the decay rates are not suppressed by the lepton mass, experimental measurements are more precise.
For leptonic decays, the HFLAV world averages for 
$f_{D_s}\Vcs$ and $f_D\Vcd$ have fractional errors of roughly $1\%$ and $2\%$, respectively~\cite{HFLAV:2022pwe}.
The corresponding semileptonic decay-rate measurements are roughly a factor of two more precise in each case, with the fractional errors in $f_+^{D\to K}(0)\Vcs$ and
$f_+^{D\to \pi}(0)\Vcd$ around $0.5\%$ and $1\%$, respectively~\cite{HFLAV:2022pwe}. 
Lattice-QCD calculations of semileptonic form factors (including both normalization and shape), while more complex than for decay constants for leptonic decays, have a long history in lattice QCD~\cite{FermilabLattice:2004ncd,Becirevic:2007cr,DiVita:2010mlb,Na:2010uf, Na:2011mc,Koponen:2011ev,Bailey:2012sa,Koponen:2012di,Koponen:2013tua,LATTICE-FERMILAB:2015wnj,FermilabLattice:2017eea,Kaneko:2017xgg,Lubicz:2017syv,Lubicz:2018rfs}.
Now, however, the current experimental errors and the forthcoming improvements by BES~III motivate further reducing the lattice-QCD errors to the level of experimental precision.

In this work, we leverage the same theoretical tools that were successfully employed in the calculation of decay constants and the $K_{\ell 3}$ form factor \cite{FermilabLattice:2014tsy,Bazavov:2017lyh,FermilabLattice:2018zqv}: the same highly improved relativistic lattice actions and gauge-field ensembles with physical quark masses and small lattice spacings.
In particular, we compute the hadronic form factors for the semileptonic decays $D\to\pi\ell\nu$, $D\to K\ell\nu$, and $D_s\to K\ell\nu$ in lattice QCD, with the goal of obtaining percent-level determinations of $\Vcd$ and $\Vcs$ when combined with experimental data. Our values for $\Vcd$ and $\Vcs$ provide a stringent test of unitarity and their precision allows a commensurate comparison with leptonic determinations.
As a key aspect of our analysis, we report the correlations between the hadronic form factors in the different decay channels as well as between the final values for $\Vcd$ and $\Vcs$ (see \cref{sec:phenomenology}).
Preliminary results for the present calculation of the form factors have been presented in Refs.~\cite{FermilabLattice:2019ycs,FermilabLattice:2021bxu}.
We note that the HPQCD collaboration has recently presented a precise lattice-QCD calculation of the form factors for $\DK$ decay~\cite{Chakraborty:2021qav,Parrott:2022rgu}, with a quoted lattice-QCD uncertainty close to the experimental one in the extraction of $\Vcs$ ~\cite{Chakraborty:2021qav}.
On the other hand, this paper yields the first percent-level determination of $\Vcd$ and enables the first stringent test of second-row CKM unitarity from semileptonic $D$-meson decays.

With the hadronic form factors for a given decay in hand, it is straightforward to construct the lepton flavor universality (LFU) ratios $R_{\mu/e}$, which are defined as the ratio of the branching fractions to muon versus electron final states; see \cref{sec:lfu}.
These ratios are expected to be close but not identically equal to unity in the SM, with differences coming from lepton-mass, isospin-breaking, and QED effects. 
Lattice QCD calculations offer a theoretically clean method for determining the SM prediction to high precision (up to QED corrections), contributing to stringent LFU tests in those channels.

The rest of this article is organized as follows.
\Cref{sec:definitions} reviews the definitions and formalism for relating experimentally measured decay rates to the hadronic form factors we calculate.
\Cref{sec:simulation} gives details related to the lattice-QCD simulation.
\Cref{sec:correlator_analysis} reports the statistical analysis of Euclidean correlation functions which yields renormalized form factors.
\Cref{sec:chiral_ctm} describes the final chiral-continuum fit, which interpolates the form factors to the physical hadron masses and extrapolates to the continuum limit.
\Cref{sec:syst_errors} analyzes the uncertainties in our calculation and summarizes the complete statistical and systematic error budget for the form factors.
\Cref{sec:phenomenology} discusses applications to phenomenology, including determinations of the CKM matrix elements and the LFU ratios in each channel.
Finally, \cref{sec:conclusions} gives some concluding remarks.
Four appendices provide additional technical information.
\Cref{sec:stagg} contains useful formulae appearing in the statistical analysis of staggered correlation functions.
\Cref{ssec:HQETErrors} presents useful information about staggered fermions and heavy quark effective theory when the bare lattice quark mass is large.
\Cref{sec:shrinkage} describes linear and nonlinear shrinkage techniques for correlation and covariance matrix, the latter of which is a novel aspect of the correlator analysis presented in this work. 
\Cref{app:extras} provides supporting details and figures regarding various fits, which exceed the scope of the main text but illustrate the robustness of our analysis.

\section{Definitions}
\label{sec:definitions}

The differential decay rate for the semileptonic decay $H\to L\ell \nu$ of a heavy pseudoscalar meson $H \in \{D, D_s\}$ to a light pseudoscalar meson $L \in \{K,\pi\}$ is given by 
\begin{align}
\begin{split}
    \frac{d\Gamma}{dq^2}
    =&\frac{G_F^2}{24\pi^3}
    \eta_{\rm EW}^2 \Vcx^2
    (1 - \epsilon)^2
    (1 + \delta_{\rm EM}) \times\\
    &\left[
    \abs{\bm{p}}^3 \left(1 + \frac{\epsilon}{2}\right)\abs{f_+(q^2)}^2
    + \abs{\bm{p}} M_H^2 \left( 1 - \frac{M_L^2}{M_H^2}\right)^2
    \frac{3\epsilon}{8} \abs{f_0(q^2)}^2
    \right],
    \label{eq:dGammadq2}
\end{split}
\end{align}
where $\epsilon = m_\ell^2/q^2$ (with $m_\ell$ the lepton mass),\footnote{
In our notation, $m_\ell$ with a cursive subscript always refers to the lepton mass in the decay $H\to L\ell \nu$.
The light-quark mass is denoted $m_l$.
}
$q$ is the momentum transfer, $M_H$ and $M_L$ are the masses of the heavy initial and light final mesons, and $\bm{p}$ is the three-momentum of the final-state meson in the rest frame of the initial hadron.
Short-distance electroweak corrections to $G_F$ are contained in
$\eta_{\rm EW} = 1 + (\alpha_\text{QED}/\pi) \ln (M_Z/\mu)|_{\mu=M_D} = 1.009(2)$~\cite{Sirlin:1981ie},
where the error is an estimate of the scale uncertainty from a factor-of-two variation around $\mu = M_D$.\footnote{%
Physically, the scale dependence of $\eta_{\rm EW}$ should cancel against that of the structure-dependent electromagnetic corrections which, though calculable in principle, have never been computed for these decays.
Computing these corrections exceeds the scope of the present work.
}
Long-distance and structure-dependent electromagnetic corrections are described by $\delta_{\rm EM}$.\footnote{Systematic uncertainties from neglected electromagnetic corrections and strong isospin breaking are discussed in \cref{ssec:sib_qed}.}

The form factors $f_+(q^2)$ and $f_0(q^2)$ encapsulate the nonperturbative hadronic structure of the decay.
They arise in the usual way from the Lorentz-covariant decomposition of the relevant transition matrix elements,
\begin{align}
\bra{L} \mathcal{V}^\mu \ket{H}
	&\equiv \sqrt{2 M_H} \left[ v^\mu f_\parallel(q^2) + p_\perp^\mu f_\perp(q^2) \right], \label{eq:f||f} \\
	&\equiv f_+(q^2) \left( k^\mu + p^\mu - \frac{M_H^2 - M_L^2}{q^2} q^\mu \right) + f_0(q^2) \frac{M_H^2 - M_L^2}{q^2}q^\mu ,
    \label{eq:vector_decomposition}\\
\bra{L} \mathcal{S} \ket{H}
	&= \frac{M_L^2 - M_H^2}{m_h - m_x} f_0(q^2). \label{eq:scalar_decomposition}
\end{align}
In these expressions, $k^\mu$, and $p^\mu$ refer to the four-momentum of the heavy initial and light final mesons; $m_h$ and $m_x$ refer to the masses of the heavy and light quarks in the current; $v^\mu = k^\mu / M_H$ is the four-velocity of the heavy meson; $p_\perp^\mu = p^\mu - (p \cdot v) v^\mu$ is the component of the final-state hadron's momentum orthogonal to $v$; and $q^\mu = k^\mu - p^\mu$ is the momentum transfer.
The same form factor $f_0$ appears in Eqs.~\eqref{eq:vector_decomposition} owing to partial conservation of the vector current (PCVC), namely $\partial_\mu \mathcal{V}^\mu = (m_h - m_x)\mathcal{S}$ as an operator identity.

In lattice gauge theory, we introduce bilinears of lattice fermion fields---$J=V^0$, $V^i$, and $S$---and associated matching factors $Z_J$, such that $Z_JJ$ and the corresponding $\mathcal{J}$ have the same matrix elements (up to controlled uncertainties).
In this notation, and in the rest frame of the decaying meson, the relations between form factors and matrix elements \pagebreak take
the following forms:
\begin{align}
f_\parallel(q^2)	&= Z_{V^0} \frac{\matrixel{L}{V^0}{H}}{\sqrt{2 M_H}}, \label{eq:f_parallel}\\
f_\perp(q^2)		&= Z_{V^i} \frac{1}{p^i} \frac{\matrixel{L}{V^i}{H}}{\sqrt{2 M_H}}, \label{eq:f_perp}\\
f_0(q^2)			&= Z_m Z_S \frac{m_h-m_l}{M_H^2 - M_L^2} \matrixel{L}{S}{H}. \label{eq:f_0}
\end{align}
[No sum is implied in \cref{eq:f_perp}.]
Using the preceding equations, the vector form factor is given by a linear combination of $f_\perp(q^2)$ and $f_0(q^2)$,
\begin{align}
    f_+(q^2) &= \left( \frac{M_H - E_L}{\sqrt{2 M_H}} \right)
        \left( 1 - \frac{E_L^2 - M_L^2}{(M_H - E_L)^2}\right) f_\perp(q^2)
        + \left( \frac{M_H^2 - M_L^2}{M_H - E_L}\right) \frac{f_0(q^2)}{2 M_H},
    \label{eq:fplus_perp_0}
\end{align}
which will be useful below.

In momentum space, PCVC implies the following condition for the lattice currents:\footnote{%
In Minkowski space, the basic momentum-space operator relation reads
$i q_\mu \avg{\mathcal{V}^\mu(q)} = (m_h-m_l)\avg{\mathcal{S}}$.
\Cref{eq:PCVC}, in which all terms come with a positive sign and without factors of $i$, amounts to a definition of the sign convention for Wick rotation and the phase convention for the lattice currents.}
\begin{align}
    Z_{V^0}(M_H - E_L) \matrixel{L}{V^0}{H} + Z_{V^i} \bm{q}\cdot \matrixel{L}{\bm{V}}{H} = Z_m Z_S (m_h - m_x) \matrixel{L}{S}{H}.
    \label{eq:PCVC}
\end{align}
which can be used to extract the renormalization factors for the temporal and spatial components of the vector current, $Z_{V^0}$ and $Z_{V^i}$~\cite{Na:2010uf}, as explained in detail in \cref{sec:renormalization}. 
With the present treatment of all valence quarks with the highly improved staggered quark (HISQ) action~\cite{Follana:2006rc}, the local scalar density enjoys absolute normalization, $Z_m Z_S=1$~\cite{Karsten:1980wd,Smit:1987zh}.
Furthermore, PCVC allows one to express any single matrix element in terms of the other two involved in the relation in \cref{eq:PCVC}, for example,
\begin{align}
    f_+^{\rm alt}(q^2) &= \frac{1}{\sqrt{2 M_H}} \left[ f_\parallel(q^2) + (M_H - E_L) f_\perp(q^2) \right],
    \label{eq:fplus_parallel_perp}\\
    f_0^{\rm alt}(q^2) &= \frac{\sqrt{2 M_H}}{M_H^2 - M_L^2} \left[ (M_H - E_L) f_\parallel(q^2) + (E_L^2 - M_L^2) f_\perp(q^2) \right]
    \label{eq:f0_parallel_perp},
\end{align}
with $f_\parallel$ and $f_\perp$ computed using \cref{eq:f_parallel,eq:f_perp}.\footnote{
Another expression for $f_+$ in terms of $f_0$ and $f_\parallel$ exists but involves a delicate numerical cancellation near $q^2_{\rm max}$. For this reason it is excluded from the subsequent discussion.}
These alternative constructions will be used to check for systematic errors in our analysis; see \cref{ssec:fplus_f0}.

\section{Simulation details \label{sec:simulation}}

\begin{table}[tp]
\centering
\caption{
	A summary of the lattice spacings, lattice spatial ($N_s$) and temporal ($N_t$) sizes, valence quark masses, intermediate scale-setting parameters, number of source times and configurations 
	$N_{\rm src} \times N_{\rm configs}$, source-sink separations $T/a$,
	and approximate Goldstone (pseudoscalar taste) pion masses used in our calculation.
	The text describes the ensembles' sea and valence masses in more detail.
	The gauge ensembles were generated by the MILC collaboration
	~\cite{MILC:2010pul,MILC:2012znn,Bazavov:2017lyh}.
	The values for the gradient-flow scale $w_0/a$ have been calculated previously~\cite{MILC:2015tqx,Brown:2018jtv}. 
	The simulation program is described in detail in Ref.~\cite{MILC:2012znn} and was later extended to smaller lattice spacings ($a\approx 0.042\fm$), as described in Ref.~\cite{Bazavov:2017lyh}.
	The number of source times $N_{\rm src}$ refers to the number of loose-solve source times employed in the truncated solver method; on each configuration one corresponding fine solve is used.
	Values for the sea- and valence-quark masses are given in \cref{table:bare_quark_masses}.
\label{table:ensembles}
}
\begin{tabular}{l  c  c  l  c  c  c c}
\hline \hline
$\approx a$ &  $N_s^3 \times N_t$ &   $m_l$    &   $m_h / m_c$ & $w_0/a$ & $N_{\rm src} \times N_{\rm configs}$ & $T/a$ & $\approx M_{\pi, P}$ \\
   ~[fm] & & & & & & & [MeV] \\
\hline
0.12    & $48^3 \times 64$   & physical    & $0.9, 1.0, 1.4$           & 1.4168(10) & $32 \times 1352$ & \{12, 13, 14, 16, 17\} & 135\\
0.088   & $64^3 \times 96$   & physical    & $0.9, 1.0, 1.5, 2.0$      & 1.9470(13) & $24 \times 980$ & \{16, 17, 19, 22, 25\} & 130\\
0.088   & $48^3 \times 96$   & $m_s/10$ & $0.9, 1.0, 1.5, 2.0$ & 1.9299(12) & $24 \times 697$ & \{16, 19, 22, 25\} & 224\\
0.057   & $96^3 \times 192$  & physical    & $0.9, 1.0, 1.1, 2.2$      & 3.0119(19) & $32 \times 877$ & \{25, 28, 30, 34, 37\} & 134\\
0.057   & $64^3 \times 144$  & $m_s/10$ & $0.9, 1.0, 2.0$      & 2.9478(31) & $36 \times 916$ & \{23, 30, 34, 37\} & 231\\
0.057   & $48^3 \times 144$  & $m_s/5$ & $0.9, 1.0, 2.0$      & 2.8956(33) & $36 \times 823$ & \{23, 30, 34, 37\} & 325\\
0.042   & $64^3 \times 192$  & $m_s/5$ & $0.9, 1.0, 2.0$      & 3.9222(29) & $24 \times 1008$ & \{34, 39, 45, 52\} & 308\\
\hline \hline
\end{tabular}
\end{table}

\begin{table}
    \centering
    \caption{Sea- and valence-quark masses in lattice units for the ensembles used in this calculation.
    The first two columns specify the ensemble by the approximate lattice spacing and the ratio of light- and strange-quark masses.
    The next three columns give the sea-quark masses.
    The final three columns contain the valence-quark masses.}
    \label{table:bare_quark_masses}
    \begin{tabular}{cc@{\hspace{0.75em}}lll@{\hspace{0.75em}}lll}
\hline\hline
$\approx a$ & $m_l/m_s$ & $(am_l)^{\rm sea}$ & $(am_s)^{\rm sea}$ & $(am_c)^{\rm sea}$ & $(am_l)^{\rm valence}$ & $(am_s)^{\rm valence}$ & $(am_h)^{\rm valence}$ \\
\hline
0.12        & 1/27      & 0.001907           & 0.05252            & 0.6382             & $(am_l)^{\rm sea}$ & $(am_s)^{\rm sea}$ &\{0.5744, 0.6382, 0.8935\} \\[2pt]
0.088       & 1/27      & 0.0012             & 0.0363             & 0.432              & $(am_l)^{\rm sea}$ & $(am_s)^{\rm sea}$ &\{0.389, 0.432, 0.648, 0.864\} \\
0.088       & 1/10      & 0.00363            & 0.0363             & 0.43               & $(am_l)^{\rm sea}$ & $(am_s)^{\rm sea}$ &\{0.389, 0.432, 0.648, 0.864\} \\[2pt]
0.057       & 1/27      & 0.0008             & 0.022              & 0.26               & $(am_l)^{\rm sea}$ & $(am_s)^{\rm sea}$ &\{0.257, 0.286, 0.572\} \\
0.057       & 1/10      & 0.0024             & 0.024              & 0.286              & $(am_l)^{\rm sea}$ & $(am_s)^{\rm sea}$ &\{0.257, 0.286, 0.572\} \\
0.057       & 1/5       & 0.0048             & 0.024              & 0.286              & $(am_l)^{\rm sea}$ & $(am_s)^{\rm sea}$ &\{0.257, 0.286, 0.572\} \\[2pt]
0.042       & 1/5       & 0.00316            & 0.0158             & 0.188              & 0.00311            & 0.01555            &\{0.164, 0.1827, 0.365\} \\
\hline\hline
\end{tabular}
\end{table}

\begin{figure}[bp]
    \centering
    \includegraphics[width=0.49\textwidth]{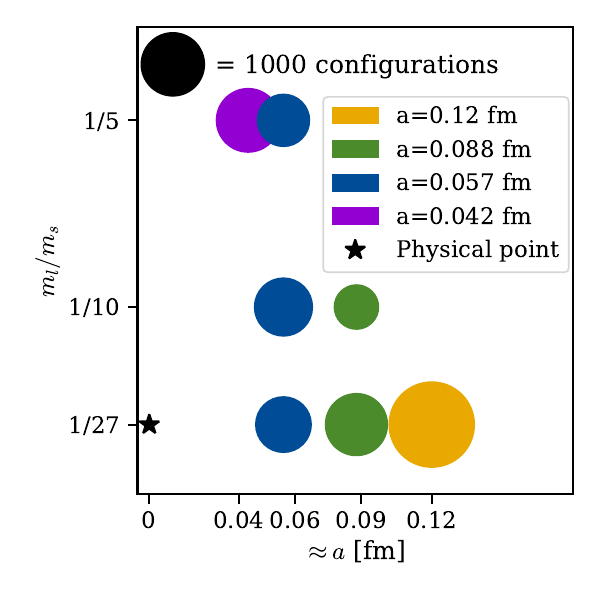}
    \caption{
        Summary of the lattice spacings and light-quark valence masses used in the present calculation.
        The sizes of the colored circles are proportional to the number of configurations in each ensemble. 
        Quantitative details are given in \cref{table:ensembles}.
    }
    \label{fig:ensembles}
\end{figure}

Our calculation uses ensembles generated by the MILC Collaboration using 
a one-loop Symanzik improved gauge action and $N_f=2+1+1$ flavors of dynamical sea quarks with the HISQ action~\cite{MILC:2010pul,MILC:2012znn,Bazavov:2017lyh}\footnote{
We have adopted a policy for sharing collaboration-generated gauge configuration files with highly-improved staggered sea quarks.
The policy, along with a list of lattices that are shared without restriction as well as bibliographic guidance for citations, can be found on our GitHub page linked \href{https://github.com/milc-qcd/sharing/wiki/LatticeSharing}{\textbf{here}}.
}.
\Cref{table:ensembles} and \cref{fig:ensembles} summarize the ensembles used in this work.
Lattice spacings range from $a\approx 0.12$~fm down to $a\approx 0.042$~fm.
An ensemble with physical-mass light quarks appears for all lattice spacings but $a\approx 0.042\fm$.
For the finer lattice spacings, we also include ensembles with heavier-than-physical light quarks with
$m_l \approx m_ s/10$ and
$m_l \approx m_s/5$.

The masses of the valence light and strange quarks generally match those in the sea.
In all ensembles the charm and strange quarks in the sea have (close to) physical masses.
The heavy valence quarks used in this study range from around nine-tenths to around twice the physical charm mass.
The precise values for the sea- and valence-quark masses are given in \cref{table:bare_quark_masses}.

Although the primary targets of this work are the dimensionless form factors $f_0$ and $f_+$, many intermediate quantities (e.g., $f_\parallel$ and $f_\perp$ and masses) are dimensionful.
Throughout this work, the scale is set on each ensemble using previously calculated values for the gradient-flow scale $w_0/a$~\cite{MILC:2015tqx,Brown:2018jtv}, also listed in \cref{table:ensembles}.
Details of the intermediate scale-setting scheme in the chiral-continuum analysis are discussed below in \cref{ssec:chiral_ctm_fits}.

\section{Correlator Analysis \label{sec:correlator_analysis}}

\subsection{Definitions}

\begin{figure}[b]
    \centering
    \includegraphics[width=0.49\textwidth]{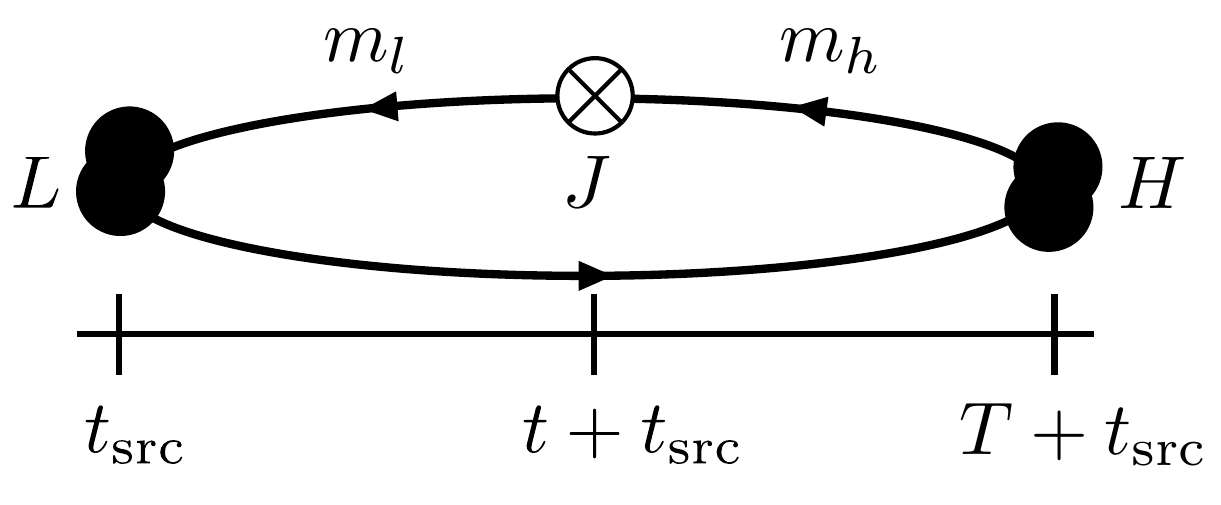}
    \caption{Schematic picture of the three-point functions defined in \cref{eq:c3_S,eq:c3_V0,eq:c3_Vi}.
    The final-state hadron $L\in\{\pi, K\}$ is created with momentum $\bm{p}$ at the time $t_{\rm src}$.
    An external current $J$ is inserted at time $t+t_{\rm src}$.
    The initial-state hadron $H\in\{D,D_s\}$ is destroyed at rest at time $T+t_{\rm src}$.
    \label{fig:schematic_3pt}
    }
\end{figure}

To access the matrix elements $\matrixel{L}{S}{H}$, $\matrixel{L}{V^0}{H}$, and $\matrixel{L}{V^i}{H}$,
we compute the following two- and three-point correlation functions:
\begin{align}
C_{H}^{P}(t)
    &= \sum_{\bm{x},\bm{y}}
    \left\langle
        P_{H}  (t_{\rm src},   \bm{x})
        P_{H}  (t+t_{\rm src}, \bm{y})
    \right\rangle, \label{eq:c2_heavy_P}\\
C_{H}^{A^0}(t)
    &= \sum_{\bm{x},\bm{y}}
    \left\langle
        A^0_{H}(t_{\rm src},   \bm{x})
        A^0_{H}(t+t_{\rm src}, \bm{y})
    \right\rangle, \label{eq:c2_heavy_A0}\\
C_{L}^{P}(t, \bm{p})
    &= \sum_{\bm{x},\bm{y}}
    e^{i\bm{p}\cdot (\bm{x}-\bm{y})}
    \left\langle
        P_{L}  (t_{\rm src},   \bm{x})
        P_{L}  (t+t_{\rm src}, \bm{y})
    \right\rangle, \label{eq:c2_light_P}\\
C_{H \to L }^{S}(t,T,\bm{p})
    &= \sum_{\bm{x},\bm{y},\bm{z}}
    e^{i\bm{p}\cdot(\bm{x}-\bm{y})}
    \left\langle
        P_{L}  (t_{\rm src},   \bm{x})
           S   (t+t_{\rm src}, \bm{y})
        P_{H}  (T+t_{\rm src}, \bm{z})
    \right\rangle, \label{eq:c3_S} \\
C_{H \to L}^{\bm{V}}(t,T,\bm{p})
    &= \sum_{\bm{x},\bm{y},\bm{z}}
    e^{i\bm{p}\cdot(\bm{x}-\bm{y})}
    \left\langle
        P_{L}  (t_{\rm src},   \bm{x})
        \bm{V} (t+t_{\rm src}, \bm{y})
        P_{H}  (T+t_{\rm src}, \bm{z})
    \right\rangle, \label{eq:c3_Vi} \\
C_{H \to L}^{V^0}(t,T,\bm{p})
    &= \sum_{\bm{x},\bm{y},\bm{z}}
    e^{i\bm{p}\cdot(\bm{x}-\bm{y})}
    \left\langle
        P_{L}  (t_{\rm src},   \bm{x})
        V^0    (t+t_{\rm src}, \bm{y})
        A^0_{H}(T+t_{\rm src}, \bm{z})
    \right\rangle, \label{eq:c3_V0}
\end{align}
where the labels denote the initial heavy hadron $H\in\{D, D_s\}$ and the final light hadron $L\in \{\pi,K\}$.
The schematic structure of the three-point correlators in \cref{eq:c3_S,eq:c3_Vi,eq:c3_V0} is depicted in \cref{fig:schematic_3pt} and the spin-taste structure of the operators in our simulations specified in \cref{table:spin_taste}.
The operators used for the scalar current $S$ and temporal vector current $V^0$ are local, but the spatial vector current $\bm{V}$ is the taste-singlet one-link operator.
The tastes of the meson creation and annihilation operators are chosen so that the correlation functions are overall taste singlets.
For three-point functions involving $S$ and $\bm{V}$, \cref{eq:c3_S,eq:c3_Vi}, it therefore suffices to use local pseudoscalar operators $P$, corresponding to Goldstone pseudoscalar mesons, at the source and sink.
For three-point functions involving $V^0$, \cref{eq:c3_V0}, we use the local axial vector operator $A^0$, corresponding to a non-Goldstone pseudoscalar meson, at either the source or the sink.
Our choice in \cref{eq:c3_V0} is to use $A^0$ for the initial-state hadrons ($D$ and $D_s$) and $P$ for the final-state hadrons ($\pi$ and $K$).
To reduce statistical noise, APE smearing~\cite{APE:1987ehd} is applied to the gauge field appearing in the one-link vector current, with 20 iterations and staple weight $0.05$.

\begin{table}
\caption{The spin-taste structure of the staggered operators used in this work.
    Pseudoscalar mesons are created and annihilated using $P$ ($\pi$, $K$, $D$, and $D_s$) and $A^0$ ($D$ and $D_s$).
    Transitions between these states are induced by the currents $S$, $V^0$, and $V^i$.
    The operator $A^0$ is necessary to conserve taste in \cref{eq:c3_V0}.}
    \label{table:spin_taste}
\begin{tabular}{cccccc}
\hline \hline
    Operator    &   Spin $\otimes$ Taste        &   Locality\\
    \hline
    $P$         &   $\gamma^5 \otimes \xi_5$    &   Local\\
    $A^0$       &   $\gamma^0 \gamma^5 \otimes \xi_0 \xi_5$  & Local\\
    $S$         &   $1 \otimes 1$   & Local\\
    $V^0$       &   $\gamma^0 \otimes \xi_0$    &   Local\\
    $V^i$       &   $\gamma^i \otimes 1$    &   One-link \\
\hline \hline
\end{tabular}
\end{table}

In \cref{eq:c2_heavy_P,eq:c2_heavy_A0,eq:c2_light_P,eq:c3_S,eq:c3_V0,eq:c3_Vi}, we work in the rest frame of the
heavy initial hadron $H$ and compute the recoiling light hadron $L$ with eight different lattice momenta $\bm{p} = 2\pi \bm{n}/N_s a$, where $N_s$ is the spatial extent of the lattice,
and $\bm{n}$ is $(0,0,0)$, $(1,0,0)$, $(1,1,0)$, $(2,0,0)$, $(2,1,0)$, $(3,0,0)$, $(2,2,2)$, or $(4,0,0)$.
For each choice of heavy-quark mass in \cref{table:ensembles} and momentum above, we compute the three-point function for several different source-sink separations $T$, given in \cref{table:ensembles}.
The final-state light-quark and spectator-quark propagators are computed using random corner-wall sources~\cite{MILC:2004qnl}.
The heavy-quark propagators are computed sequentially from the end of the spectator-quark propagator at time $T+t_{\rm src}$ as shown in \cref{fig:schematic_3pt}.
For the light- and strange-quark propagators, we employ the truncated solver method~\cite{Bali:2009hu,Alexandrou:2012zz}, using a single fine solve together with 24 to 36 loose solves on each configuration (see \cref{table:ensembles} for details).
To reduce autocorrelation in Monte Carlo time, the source locations for the fine and loose solves are precessed in Euclidean time from one configuration to the next.
    
As usual, states with both positive and negative parities contribute to the staggered correlation functions.
For the operators considered here, the negative-parity states decay smoothly with Euclidean time, while the positive-parity states oscillate while decaying in Euclidean time.
The spectral decompositions of \cref{eq:c2_heavy_P,eq:c2_heavy_A0,eq:c2_light_P,eq:c3_S,eq:c3_V0,eq:c3_Vi} take the following forms:
\begin{align}
    C_{L}^\mathcal{O}(t, \bm{p})
    &= \sum_{n=0} (-1)^{n(t+1)} \frac{\left|\matrixel{\vacuum}{\mathcal{O}_{L}}{n}\right|^2}{2E_L^{(n)}(\bm{p})}
    \left( e^{-E_L^{(n)}(\bm{p})t} + e^{-E_L^{(n)}(\bm{p})(N_t-t)} \right), \label{eq:2pt_spectral_decomp_final}\\
    C_{H}^\mathcal{O}(t)
    &= \sum_{m=0} (-1)^{m(t+1)} \frac{\left|\matrixel{\vacuum}{\mathcal{O}_{H}}{m}\right|^2}{2M_H^{(m)}}
    \left( e^{-M_H^{(m)}t} + e^{-M_H^{(m)}(N_t-t)} \right), \label{eq:2pt_spectral_decomp_initial}\\    
    \begin{split}
    C_{H \to L}^J(t,T,\bm{p})
    &= \sum_{m,n}
        (-1)^{n(t+1)} (-1)^{m(T-t-1)}
        \frac{
            \matrixel{\vacuum}{\mathcal{O}_L}{n}
            \matrixel{n}{J}{m}            
            \matrixel{m}{\mathcal{O}_H}{\vacuum}}
            {4 E_L^{(n)}(\bm{p}) M_H^{(m)}}\\
        & \phantom{000}\times
        \left(e^{-E_L^{(n)}(\bm{p})t} + e^{-E_L^{(n)}(\bm{p})(N_t - t)} \right)
        \left(e^{-M_H^{(m)}(T-t)} + e^{-M_H^{(m)}(N_t-T+t)}\right)
        \label{eq:3pt_spectral_decomp},
    \end{split}
\end{align}
where $\mathcal{O} \in \{P, A^0\}$ is the appropriate interpolating operator and $\ket{\vacuum}$ denotes the QCD vacuum state.
In the final line, the ground-state term contains the transition matrix elements, $\left.\matrixel{n}{J}{m}\right|_{n=m=0} \equiv \matrixel{L}{J}{H}$, from which one can extract the desired form factors via \cref{eq:f_parallel,eq:f_perp,eq:f_0}.

For the sake of visualization, certain ratios of correlation functions prove useful:
\begin{align}
R_{\parallel}(t, T, \bm{p}) &=
     \frac{ \bar{C}_{H \to L}^{V^0}(t,T, \bm{p}) \sqrt{2 E_L}}{ \sqrt{\bar{C}_{L}^{A^0}(t, \bm{p}) \bar{C}_{H}^P(T-t) e^{-E_L t} e^{-M_H (T-t)}}},
    \label{eq:ratio_v4}\\
R_{\perp}(t, T, \bm{p})     &= \frac{\sqrt{2 E_L}}{p^i}
    \frac{ \bar{C}_{H \to L}^{V^i}(t,T, \bm{p})}{ \sqrt{\bar{C}_{L}^{P}(t, \bm{p}) \bar{C}_{H}^P(T-t) e^{-E_L t} e^{-M_H (T-t)}}},
    \label{eq:ratio_vi}\\
R_0(t, T, \bm{p}) &= 2 \sqrt{M_H E_L} \left(\frac{m_h - m_l}{M_H^2 - M_L^2}\right)
    \frac{ \bar{C}_{H \to L}^S(t,T, \bm{p})}{ \sqrt{\bar{C}_{L}^P(t, \bm{p}) \bar{C}_{H}^P(T-t) e^{-E_L t} e^{-M_H (T-t)}}},
    \label{eq:ratio_s}
\end{align}
where the bars (e.g., $\bar{C}^P_{L}$) denote the time-slice-averaged correlators defined in \cref{eq:c2_avg,eq:c3_avg}.
Up to discretization effects (and renormalization), these ratios asymptotically approach the form factors at large Euclidean times:
\begin{align}
    R_\parallel(t, T, \bm{p}) &\stackrel{0 \ll t \ll T}{\longrightarrow} Z_{V^0}^{-1} f_\parallel(\bm{p}),\\
    R_\perp(t, T, \bm{p})     &\stackrel{0 \ll t \ll T}{\longrightarrow} Z_{V^i}^{-1} f_\perp(\bm{p}),\\
    R_0(t, T, \bm{p})         &\stackrel{0 \ll t \ll T}{\longrightarrow} f_0(\bm{p}).
\end{align}

The subsequent analysis of statistical and systematic uncertainties was conducted in a blinded fashion.
More precisely, all of our three-point correlation functions
were multiplied by a blinding factor $X \in (0.95, 1.05)$, which was chosen randomly and held fixed across all ensembles, momenta, currents, and heavy-quark masses in the three-point functions.
The blinded results were carried all the way through the phenomenological applications described in \cref{sec:phenomenology}.
The blinding factor was removed only after the estimate of systematic errors was complete and the analysis was frozen.

\subsection{Statistical analysis}

\begin{figure}
    \centering
    \includegraphics[width=1.0\textwidth]{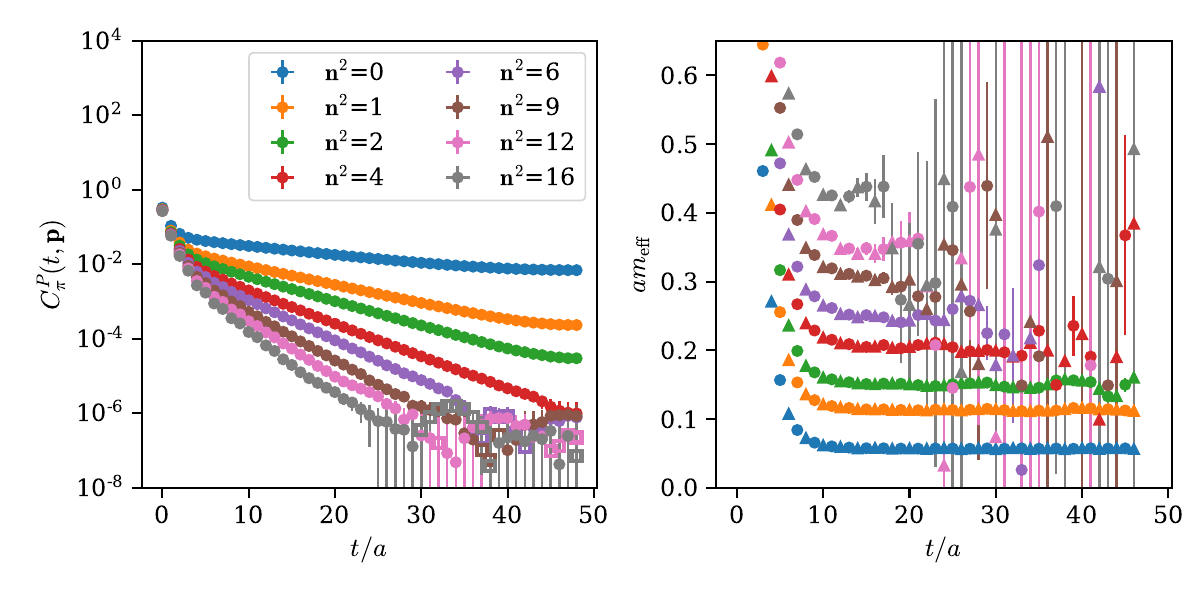}
    \caption{Pion two-point correlation functions $C_{\pi}^P(t,\bm{p}=2\pi\bm{n}/N_sa)$ and effective masses on the physical-mass $a\approx 0.088\fm$ ensemble.
    To reduce the visual impact of opposite parity states, the effective mass is computed separately for even and odd times using \cref{eq:effective_mass} and plotted using circles and triangles, respectively.
    After folding the data around the midpoint of the lattice, the correlator is defined for $t/a \in [0, N_t/2] = [0, 48]$.
    Because of the form of \cref{eq:effective_mass} involves offsets by 2, the effective mass is defined on times
    $t/a \in [2, N_t/2-2] = [2, 46]$.
    }
    \label{fig:pion_2pt_functions}
\end{figure}

\begin{figure}
    \centering
    \includegraphics[width=1.0\textwidth]{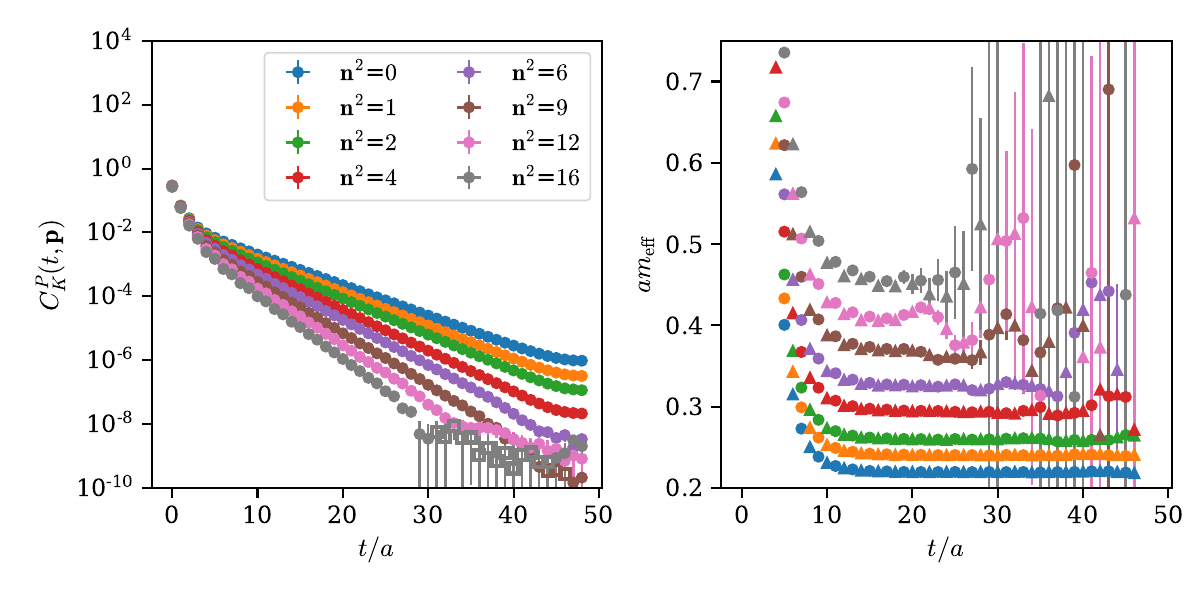}
    \caption{Kaon two-point correlation functions $C_{K}^P(t,\bm{p}=2\pi\bm{n}/N_sa)$ and effective masses on the physical-mass $a\approx 0.088\fm$ ensemble. 
    To reduce the visual impact of opposite parity states, the effective mass is computed separately for even and odd times using \cref{eq:effective_mass} and plotted using circles and triangles, respectively.
    After the data is folded around the midpoint of the lattice, the correlator is defined for $t/a \in [0, N_t/2] = [0, 48]$.
    Because of the form of \cref{eq:effective_mass} involves offsets by 2, the effective mass is defined on times
    $t/a \in [2, N_t/2-2] = [2, 46]$.}
    \label{fig:kaon_2pt_functions}
\end{figure}

\begin{figure}
    \centering
    \includegraphics[width=1.0\textwidth]{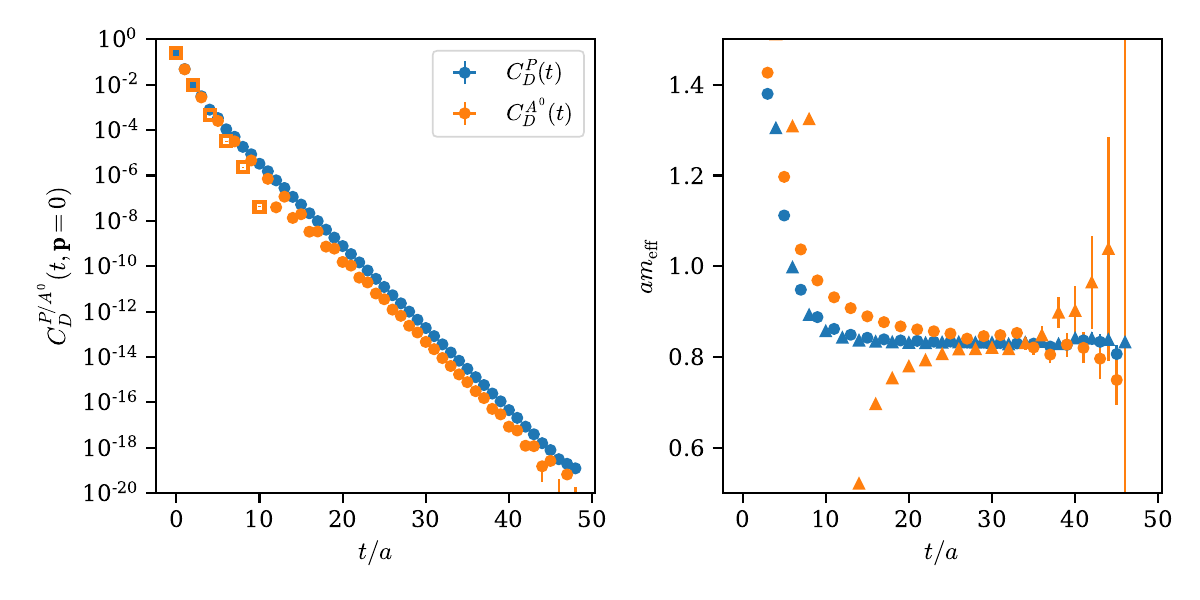}
    \caption{$D$-meson two-point correlation functions $C_D^{P/A^0}(t,\bm{p})$ at $\bm{p=0}$ and effective masses with the heavy-quark mass near its physical value on the physical-mass $a\approx 0.088\fm$ ensemble.
    The corresponding results for the $D_s$ meson are very similar. To reduce the visual impact of opposite parity states, the effective mass is computed separately for even and odd times using \cref{eq:effective_mass} and plotted using circles and triangles, respectively.
    Open symbols indicate values of the correlation function that are negative and, for ease of visualization, have been multiplied by $-1$.
    These negative points are responsible for the behavior of the odd-site effective mass for $C_D^{A^0}$ at early times.
    After folding the data around the midpoint of the lattice, the correlator is defined for $t/a \in [0, N_t/2] = [0, 48]$.
    Because the form of \cref{eq:effective_mass} involves offsets by 2, the effective mass is defined on times
    $t/a \in [2, N_t/2-2] = [2, 46]$.}
    \label{fig:d_2pt_functions}
\end{figure}

The statistical analysis consists of two stages.
First, two-point functions are analyzed in isolation.
Second, the two- and three-point functions are analyzed together to extract the form factors.
Several features are common to the fits in both stages.
To avoid possible contamination from autocorrelation in Monte Carlo time, the data are binned by 10 configurations prior to fitting. 
The amount of binning was chosen by looking for stability and saturation of errors in the fit results for the masses and form factors.

Our analysis employs standard Bayesian fits, which can be described generally as least-squares regression to a model function $f(\bm{a})$ with parameters $\bm{a}$ for some data set $D$.
The likelihood function $\pr(D|\mathbf{a}) \propto \exp \left[-\frac{1}{2} \chi^2\right]$ is written in terms of the augmented chi-squared function $\chi^2=\chi^2_{\rm data} + \chi^2_{\rm prior}$, with
\begin{align}
\chi^2_{\rm data} &= (\bar{y} - f(\bm{a}))^T \Sigma^{-1} (\bar{y} - f(\bm{a})), \label{eq:chi2}\\
\chi_{\rm prior}^2 &= (\bm{a} - \tilde{\bm{a}})^T \tilde{\Sigma}^{-1} (\bm{a} - \tilde{\bm{a}}) \label{eq:chi2_prior},
\end{align}
where $\bar{y}$ is a vector with the data means, $\Sigma$ is the covariance matrix, $\tilde{\bm{a}}$ is a vector with the prior values, and $\tilde{\Sigma}$ is the prior covariance matrix.
These expressions are standard \cite{Lepage:2001ym,Morningstar:2001je,Jay:2020jkz}.
In the present analysis, $\bar{y}$ and $\Sigma$ correspond to the measured means and covariance matrices of the correlation functions.
The function $f(\bm{a})$ corresponds to the spectral decomposition of  \cref{eq:2pt_spectral_decomp_final,eq:2pt_spectral_decomp_initial,eq:3pt_spectral_decomp}, with the energies and matrix elements serving as the parameters $\bm{a}$.
The choices for the priors $\tilde{\bm{a}}$ are discussed below.
Instead of using the usual binned-sample covariance matrix in \cref{eq:chi2}, we used an improved estimator $\hat{S}_n$ employing nonlinear shrinkage, which corrects for finite-sample-size effects~\cite{Ledoit:2018}; for technical details, see \cref{sec:shrinkage}.\footnote{%
To avoid possible confusion, we emphasize our correlator fits use \emph{nonlinear} shrinkage.
The chiral-continuum fits described below use linear shrinkage, since it combines data from different ensembles, each with a different statistical size;
see the discussion in \cref{ssec:chiral_ctm_fits}.}
The general procedure is as follows.
First, we compute the binned variances $\sigma$.
Second, we compute the correlation matrix $C_n$ using the full (unbinned) data.
Third, we compute the shrinkage estimator $\hat{C}_n$, taking the effective sample size to be the ratio of the total configurations to the bin size.
Finally, we construct $\hat{S}_n = \diag(\sigma)\hat{C}_n\diag(\sigma)$.
Apart from the usage of shrinkage, a similar procedure has been employed in the past, e.g., in Ref.~\cite{Bazavov:2017lyh}. 
In all cases, statistical uncertainties in the fit parameters are determined via bootstrap resampling with $500$ draws.
For fits on bootstrap-resampled pseudoensembles, the covariance matrix is held fixed to the binned-sample covariance matrix with shrinkage ($\hat{S}_n$) for the full ensemble~\cite{Toussaint:2008ke}.

As mentioned above, the statistical analysis begins with two-point functions.
\Cref{fig:pion_2pt_functions,fig:kaon_2pt_functions,fig:d_2pt_functions} display
representative two-point functions and effective masses for the pion, the kaon and the $D$~meson, respectively, on the physical-mass $a\approx 0.12\fm$ ensemble with a heavy-quark mass near its physical value for the $D$~meson. 
For the correlation functions themselves, dramatic oscillations from opposite-parity states are present only for the heavy mesons (see \cref{fig:d_2pt_functions}).
When plotted in the usual way, oscillations are visible in all the effective masses aside from the zero-momentum pion. 
To reduce the distraction of opposite-parity states and bring out the approach to the ground state, the effective mass is constructed separately on even and odd time slices using 
\begin{equation}
    am_{\rm eff}(t) \equiv \frac{1}{2}\arccosh \left[ (C(t+2)+ C(t-2))/2C(t) \right].
    \label{eq:effective_mass}
\end{equation}
In the effective mass plots, the triangle and circle markers correspond to the even and odd time slices, respectively.
As expected, the statistical noise grows exponentially for correlators with nonzero momentum.
High-momentum correlators therefore become noisy at large times, especially those considered here with $\bm{n}=(2,2,2)$ or $(4,0,0)$.
Even so, clear plateaus spanning several time slices are typically present in the effective mass at each momentum. 
For $D$~mesons, contributions from excited states are visibly larger when the interpolating operator $A_0$ is used. 
This observation informs certain analysis choices below. 
The behavior of two-point functions is similar to the ones shown in \cref{fig:pion_2pt_functions,fig:kaon_2pt_functions,fig:d_2pt_functions} for other masses and lattice spacings.

\begin{table}
\caption{Preferred analysis settings for fits of two-point functions to \cref{eq:2pt_spectral_decomp_final} and \cref{eq:2pt_spectral_decomp_initial}.
    The same settings are applied uniformly across all ensembles.
    The larger $t_{\rm min}$ cut for $C_H^{A_0}(t)$ is taken to avoid the excited-state contributions visible in \cref{fig:d_2pt_functions}.
    }
\label{table:analysis_choices}
\begin{tabular}{ l c c l l}
    \hline\hline
    Correlator      &   $N_{\rm decay} + N_{\rm osc}$ &   $ t_{\rm min}$ [fm]   & $t_{\rm max}$ cut\\
    \hline
    $C_{\pi}^P(t, \bm{p}=\bm{0})$      &   3   +   0   & $\approx 0.5$  & Noise $\leq 30\%$\\
    $C_{K}^P(t, \bm{p}=\bm{0})$        &   3   +   1   & $\approx 0.5$  & Noise $\leq 30\%$\\
    $C_{L}^P(t, \bm{p}\neq\bm{0})$ &   3   +   1   & $\approx 0.5$  & Noise $\leq 30\%$\\
    $C_H^P(t)$                   &   3   +   2   & $\approx 0.5$  & Noise $\leq 30\%$\\
    $C_H^{A_0}(t)$               &   3   +   2   & $\approx 0.75$--$0.85$  & Noise $\leq 30\%$\\
    \hline\hline
\end{tabular}
\end{table}

\begin{table}
\caption{
Priors used for the energies in fitting two-point functions to the spectral decomposition. All energy values are in $\MeV$.
At each lattice spacing, the values are converted to lattice units. Internally, the actual fit parameters are (the logarithm of) energy differences~\cite{Lepage:2001ym}.
For instance, the splitting between the first and second excited states for the pion is $1700(400) - 1300(400) = 400(566) \MeV$.
Priors for the amplitudes are discussed in the main text.
For ensembles with heavier-than-physical pions, the central values for the priors for $E_\pi$ and $E_K$ are increased using the tree-level expectation from chiral perturbation theory for the 
quark-mass dependence of the hadron mass (see \cref{ssec:chiral_ctm_function} below).
}
\label{table:priors_2pt_energies}
\begin{tabular}{l @{\hspace{12pt}} cc @{\hspace{12pt}} cc @{\hspace{12pt}} cc @{\hspace{12pt}} cc}
\hline \hline
$n$ &   $E_\pi$   &   $E_\pi^{\rm osc}$   &   $E_K$   &   $E_K^{\rm osc}$   &   $E_D$   &   $E_D^{\rm osc}$   &   $E_{D_s}$   &   $E_{D_s}^{\rm osc}$   \\
\hline
0   &   $135(50)$   &   $500(300)$  &   $498(100)$  &   $800(300)$  &   $1865(200)$ &   $2300(700)$ & $1968(200)$   & $2317(200)$ \\
1   &   $1300(400)$ &   ---         &   $1460(400)$ &   ---         &   $2565(700)$ &   $3000(700)$ & $2300(400)$   & $2713(400)$ \\
2   &   $1700(400)$ &   ---         &   $1860(400)$ &   ---         &   $3200(700)$ &   ---         & $2700(400)$   & ---         \\
\hline \hline
\end{tabular}
\end{table}

Each two-point correlator is fit to the corresponding spectral decomposition, \cref{eq:2pt_spectral_decomp_final} or \cref{eq:2pt_spectral_decomp_initial}, using the choices in \cref{table:analysis_choices}.
We have verified that our results are stable under reasonable variations of these choices, such as including more states or changing the value of $t_{\rm min}$.
The preferred number of states is roughly the minimal number required to achieve statistically significant fits
(with, say, $\chi^2/{\rm DOF} \lesssim 1$ or $p \gtrsim 0.1$ for goodness of fit\footnote{%
In this work, we use the augmented $\chi^2$ when quoting reduced $\chi^2/{\rm DOF}$.
Throughout the analysis, judgements about fit quality are insensitive to the precise definition used, and indistinguishable results are obtained for other reasonable definitions, e.g., the alternative quality-of-fit metrics defined in Appendix B of Ref.~\cite{FermilabLattice:2016ipl}.}),
while the cuts on $t_{\rm min}$ and $t_{\rm max}$ are designed to retain as much of the data as possible without undue contamination from excited states at early times or statistical noise at late times. 
The choices for the number of states and $t_{\rm min}$ are broadly similar to Fermilab-MILC work on decay constants~\cite{Bazavov:2017lyh}.
The main difference from Ref.~\cite{Bazavov:2017lyh} is that the present analysis includes more states for the pion and kaon,
e.g., $3+1$ versus $1+1$, in order to include data from shorter Euclidean times, which is advantageous for the subsequent analysis with
three-point functions.
\Cref{table:priors_2pt_energies} summarizes the priors used for the energies in fitting the two-point functions.
For the amplitudes, we choose broad priors in lattice units: $0.50(20)$ for decaying states and $0.1(1.0)$ for the oscillating states.

As an example, \cref{fig:pion_2pt_stability} shows the stability of the ground-state mass  extracted from $C_\pi^P(t,\bm{p}=\bm{0})$ on the physical-mass $a\approx 0.12\fm$ ensemble using fits with different choices for the number of states and $t_{\rm min}/a$.
Consistency with expectations from the effective mass is demonstrated in \cref{fig:pion_2pt_E_meff}.
Similar studies inform the other choices in \cref{table:analysis_choices}.

\begin{figure}
    \centering
    \includegraphics[width=1.0\textwidth]{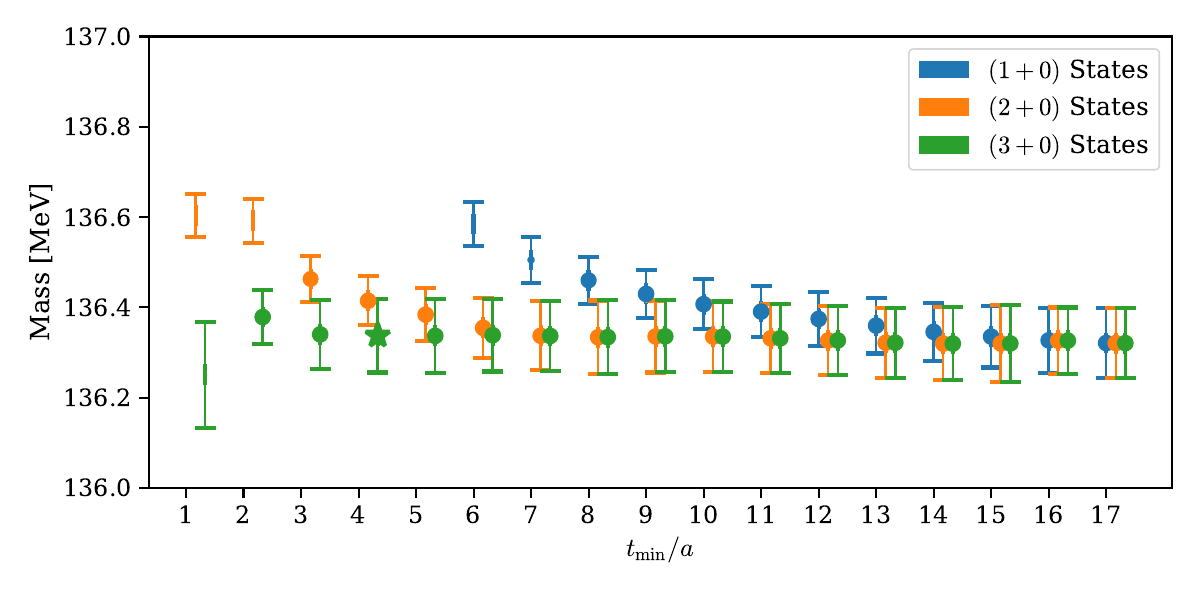}
    \caption{Stability of the ground-state mass in multi-exponential fits to $C_\pi^P(t,\bm{p}=\bm{0})$ on the physical-mass $a\approx 0.12\fm$ ensemble.
    The different colors show the posterior values for the ground-state mass using different numbers of states.
    The blue points indicating 1-state fits are aligned with the value of $t_{\rm min}/a$ on the horizontal axis.
    The corresponding results for 2- and 3-state fits at the same $t_{\rm min}/a$ are offset slightly to the right.
    The size of the markers is proportional to the $p$-value of the fit.
    As described in \cref{table:analysis_choices}, the preferred fit uses 3 states and $t_{\rm min}/a =4$ and is indicated by the green star.
    \label{fig:pion_2pt_stability}
    }
\end{figure}
\begin{figure}
    \centering
    \includegraphics[width=0.5\textwidth]{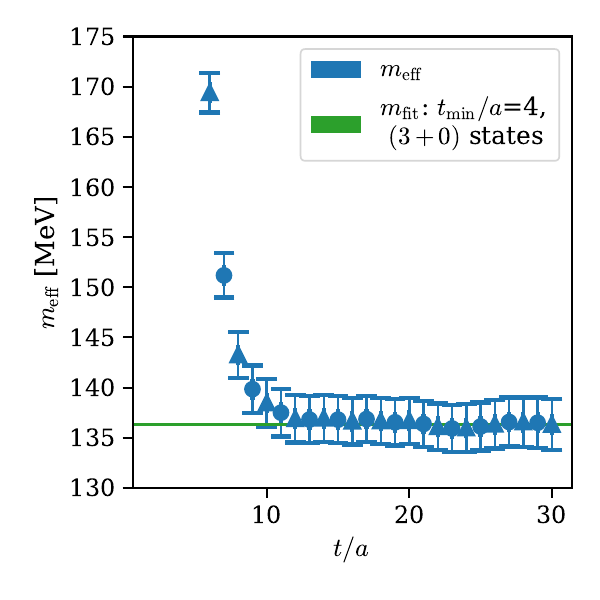}
    \caption{
    Comparison of the ground-state mass from the preferred fit and the effective mass for $C_\pi^P(t,\bm{p}=\bm{0})$ on the physical-mass $a\approx 0.12\fm$ ensemble.
    To reduce the visual impact of opposite parity states, the effective mass is computed separately for even and odd times using \cref{eq:effective_mass} and plotted using circles and triangles, respectively.}
    \label{fig:pion_2pt_E_meff}
\end{figure}

\begin{figure}
    \centering
    \includegraphics[width=0.49\textwidth]{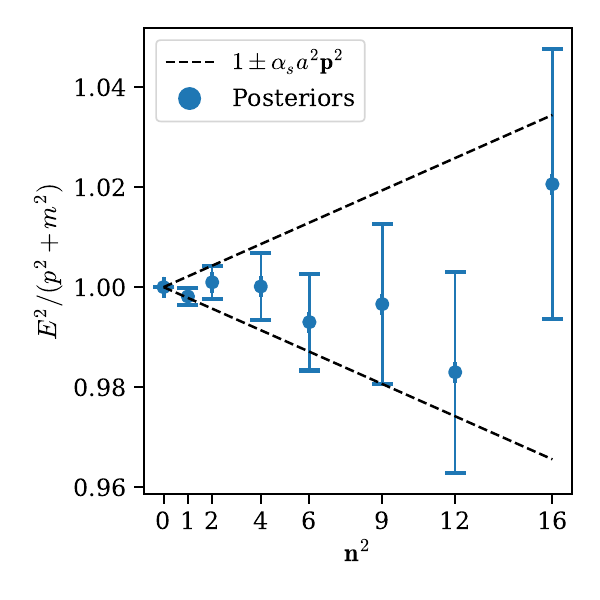}
    \includegraphics[width=0.49\textwidth]{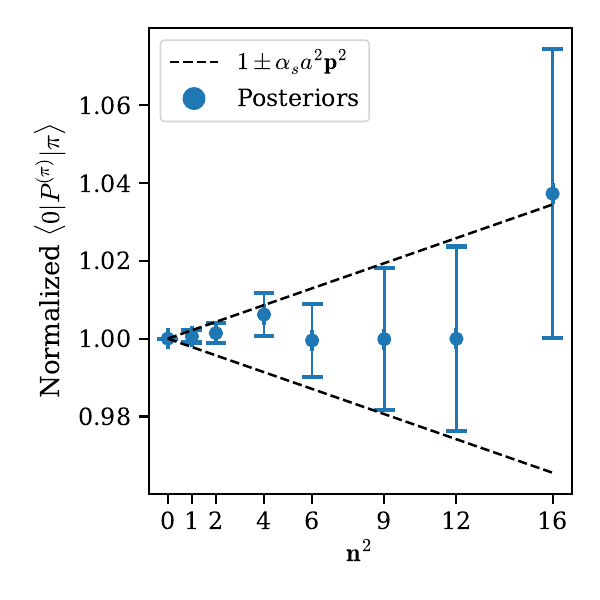}
    \caption{
    Left: The pion dispersion relation.
    Within statistical uncertainties, $E^2 = m^2+\bm{p}^2$.
    Right: The behavior of the overlap factor $\matrixel{\vacuum}{P^{(\pi)}}{\pi}$, normalized by the value at zero momentum.
    Within statistical uncertainties, the overlap factor is constant.
    }
    \label{fig:pion_dispersion_tests}
\end{figure}

\begin{figure}
    \centering
    \includegraphics[width=0.49\textwidth]{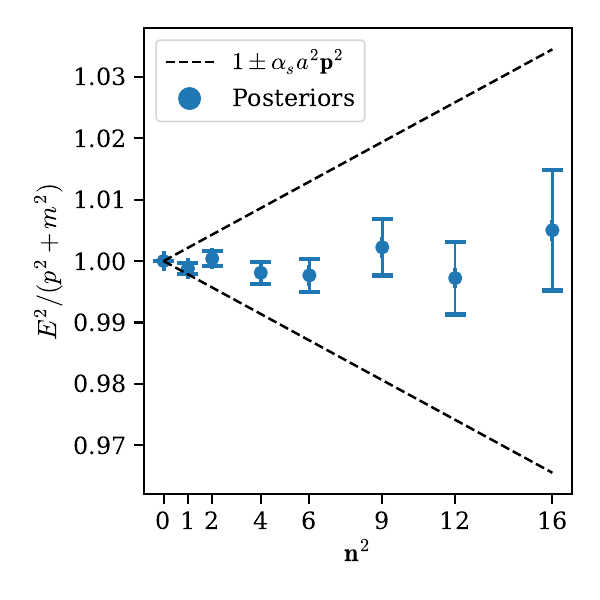}
    \includegraphics[width=0.49\textwidth]{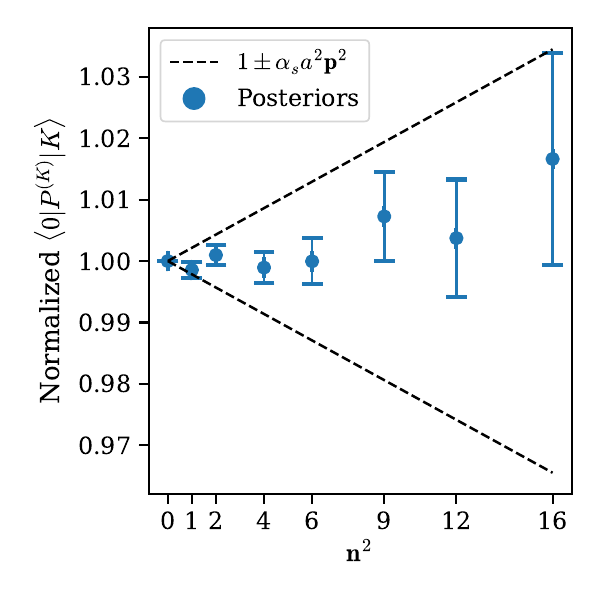}
    \caption{
    Left: The kaon dispersion relation.
    Within statistical uncertainties, $E^2 = m^2+\bm{p}^2$.
    Right: The behavior of the overlap factor $\matrixel{\vacuum}{P^{(\pi)}}{\pi}$, normalized by the value at zero momentum.
    Within statistical uncertainties, the overlap factor is constant.
    \label{fig:kaon_dispersion_tests}
    }
\end{figure}

The second stage of the analysis combines data from two-point and three-point functions to extract $f_0$, $f_\parallel$, and $f_\perp$.
The basic procedure consists of simultaneous correlated fits to 
the spectral decompositions, \cref{eq:2pt_spectral_decomp_final,eq:2pt_spectral_decomp_initial,eq:3pt_spectral_decomp}, for a particular value of the heavy-quark mass and the current $J$ using the choices for the numbers of states and the fit ranges in \cref{table:analysis_choices}.
For instance, a simultaneous correlated fit to 
$C_{D}^{P}(t)$,
$C_{\pi}^{P}(t, \bm{p})$, 
and $C_{\Dpi}^{S}(t,T,\bm{p})$ furnishes  $\matrixel{\pi}{S}{D}$.
For consistency between the two-point and three-point functions, the fit window for the three-point functions is taken to be $t\in [t_{\rm min}^{\rm src}, T-t_{\rm min}^{\rm snk}]$, where $t_{\rm min}^{\rm src}$ and $t_{\rm min}^{\rm snk}$ are the values of $t_{\rm min}$ associated with the source and sink operators, which in general differ.
For example, when the $A_0$ sink operator is used $t_{\rm min}^{\rm src} < t_{\rm min}^{\rm snk}$; see \cref{table:analysis_choices}.
The Bayesian priors used in these fits incorporate knowledge about the ground-state energies and overlap factors coming from the two-point fits.
Let $M_{\rm 2pt}\pm \delta M_{\rm 2pt}$ denote the posterior value of the ground-state energy emerging from a fit to \cref{eq:2pt_spectral_decomp_final} or \cref{eq:2pt_spectral_decomp_initial}, and let
$E_{\rm 2pt}(\bm{p}^2)\equiv \sqrt{M_{\rm 2pt}^2 + \bm{p}^2}$ denote the value of the energy obtained by boosting the central value.
Similarly, let $A_{\rm 2pt} \pm \delta A_{\rm 2pt}$ denote the posterior value of the ground-state amplitude from the same fit.\footnote{
At large times, a generic two-point function is $C(t) = A_{\rm 2pt}^2 \left(e^{M_{\rm 2pt}t} + e^{M_{\rm 2pt}(T-t)}\right) + \cdots$, so the amplitude $A_{\rm 2pt}$ contains the momentum-dependent relativistic normalization of states in the denominator.}
For the joint fits to the two- and three-point functions at zero momentum, the central values for the amplitude and energy priors are taken to match the two-point posterior central values ($M_{\rm 2pt}$ and $A_{\rm 2pt}$), while the prior widths are taken to be ten times the posterior widths
($10\times \delta M_{\rm 2pt}$ and $10 \times \delta A_{\rm 2pt}$).
For nonzero momentum, the prior central values are obtained by boosting the corresponding ground-state results assuming the continuum relativistic dispersion relation; the fractional prior widths are taken to match the expected size of discretization effects, e.g.,
$\delta E_{\rm 2pt} / E_{\rm 2pt} = \order{\alpha_s a^2\bm{p}^2}$.
\Cref{table:priors_2pt3pt} summarizes the choices of these priors.
For the excited states, the priors for the energy differences are the same as those in \cref{table:priors_2pt_energies}, and the priors for the amplitudes are as above.

For generic transition matrix elements, a prior of $V_{nm} = 0.1(10)$ in lattice units is used, where
\begin{equation}
V_{nm}\equiv
    \frac{\matrixel{\vacuum}{\mathcal{O}_L}{n}
    \matrixel{n}{J}{m}            
    \matrixel{m}{\mathcal{O}_H}{\vacuum}}
    {4 E_L^{(n)}(\bm{p}) M_H^{(m)}}.
\end{equation}
For the special case of the ground state, the central value of $V_{00}$ is estimated from the plateau in ratios following \cref{eq:plateau} below,
and the width is taken to be 50\%.

\begin{table}
\caption{
Summary of how priors for simultaneous fits to two- and three-point functions incorporate information from the two-point
fits.
Values for $\alpha_s$ are given in \cref{table:alpha}.
}
\label{table:priors_2pt3pt}
\begin{tabular}{cc c }
\hline\hline
Momentum    & Quantity & Prior value \\
\hline
$\bm{p}^2=0$ & Energy & $M_{\rm 2pt} \pm 10\times \delta M_{\rm 2pt}$\\
     & Amplitude & $A_{\rm 2pt} \pm 10\times \delta A_{\rm 2pt}$\\
\hline
$\bm{p}^2>0$ & Energy & $E_{\rm 2pt}\times\left(1 \pm \alpha_s a^2 \bm{p}^2\right)$ \\
    & Amplitude & $A_{\rm 2pt} \sqrt{M_{\rm 2pt}/E_{\rm 2pt}} \times \left( 1 \pm \alpha_s a^2 \bm{p}^2 \right)$ \\
\hline\hline
\end{tabular}
\end{table}


\begin{figure}
    \centering
    \includegraphics[width=1.0\textwidth]{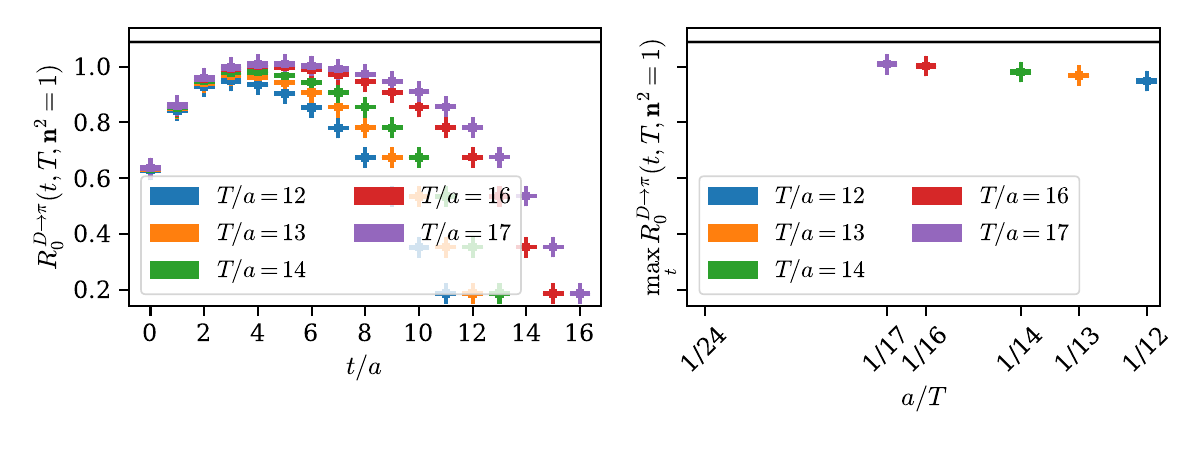}
    \includegraphics[width=1.0\textwidth]{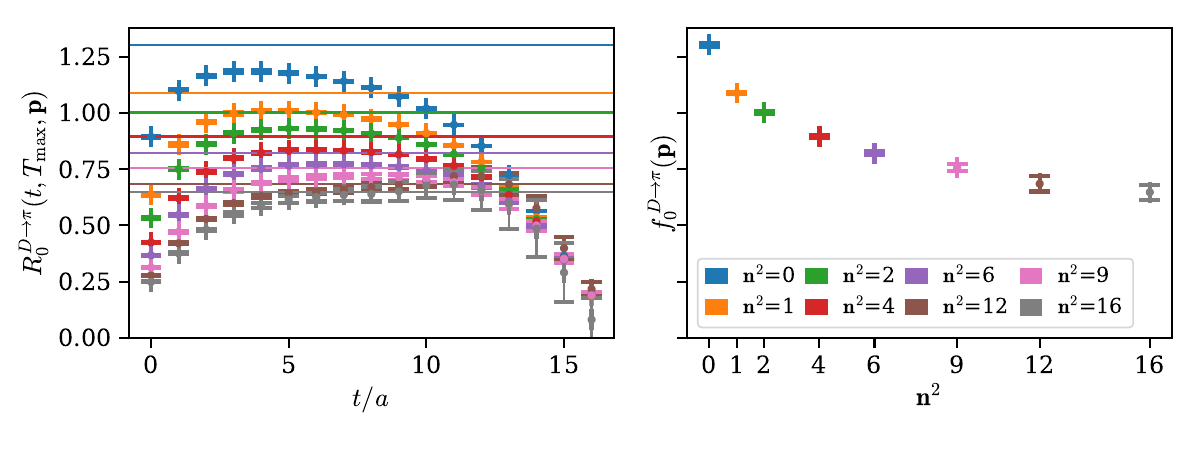}
    \caption{
    Visual tests involving the ratio $R_{0}^{D\to\pi}$ in \cref{eq:ratio_s} to form factors coming from the spectral decomposition.
    Data and results are taken from the physical-mass $0.12\fm$ ensemble with the charm-quark mass near its physical value.
    \textbf{(Top)}
    The approach to the asymptotic plateau region for the ratio $R_0^{\Dpi}$ at fixed momentum as a function of the source-sink separation.
    The right panel shows the maximum point from each color set of points on the left, $\max_t R_0^{\Dpi}(t, T)$.
    The horizontal black line in the top panels shows the form factor's posterior value, taken from the joint fit to the spectral decomposition.
    \textbf{(Bottom)}
    The form factor's momentum dependence.
    The left panel shows the ratio $R_0^{\Dpi}(t, T_{\rm max}, \bm{p})$ at fixed source-sink separation, with each color corresponding to a different momentum.
    Horizontal lines show the central value for the form factors coming from the fits (including all source-sink separations $T$).
    The right panel shows the smooth momentum dependence of $f_0^{\Dpi}(q^2)$.
    }
    \label{fig:d2pi_rbar}
\end{figure}

Once statistically acceptable fits (e.g., $\chi^2/{\rm DOF} \lesssim 1$ or $p \gtrsim 0.1$) are obtained, a variety of visualizations give confidence that the bare form factors have been extracted reliably.
For instance, the fits must reproduce the data visually with reasonable uncertainties and give results for the ground-state masses and overlap factors that agree with the initial analysis of two-point functions in isolation.
As the priors in \cref{table:priors_2pt3pt} suggest, energies are expected to satisfy the continuum dispersion relation, $E^2 = (M^2 + \bm{p}^2) (1 + \order{\alpha_s a^2 \bm{p}^2})$, and overlap factors are expected to be constant, since only point-like interpolators were used for the source and sink operators. 
\Cref{fig:pion_dispersion_tests,fig:kaon_dispersion_tests} demonstrate that both conditions are well satisfied.
The blue points correspond to the posterior (``best-fit'') results, while the dashed lines show the size of the priors for $\bm{p}^2 > 0$, as defined in \cref{table:priors_2pt3pt}.
As the figure shows, the posteriors typically are much narrower than the priors.
We have verified that statistically consistent results, with similar statistical precision, are obtained if the priors are relaxed by inflating the width by a factor of ten.
\Cref{fig:d2pi_rbar} shows representative results for joint fits for $\Dpi$.
The top rows show the approach to the asymptotic ($T/a\to\infty$) plateau region.
In the top left panel, data are plotted at fixed momentum $\bm{p}=2\pi(1,0,0)/N_sa$, with each color corresponding to a different source-sink separation $T$.
The top right panel shows the approach to the asymptotic plateau versus $T/a$, with each point corresponding to the maximum point in the curves on the top left: $\max_t R_0^{\Dpi}(t, T, \bm{p}=2\pi(1,0,0)/N_sa)$.
The horizontal black line in the top panels shows the form factor's posterior value, taken from the joint fit to the spectral decomposition.
The bottom panels shows the form factor's momentum dependence.
In the bottom left panel, the data correspond to the ratio $R_0^{\Dpi}(t, T_{\rm max}, \bm{p})$, with each color corresponding to a different momentum.
In each case, only the largest source-sink separation $T_{\rm max}$ is displayed.
Horizontal lines denote the posterior central values for the form factors, coming from fits including all source-sink separations $T$.
The bottom right panel shows the smooth momentum dependence of $f_0^{\Dpi}(q^2)$.
Additional details, along with similar figures for the decays $\DK$ and $\DsK$ are given in \cref{app:correlator_fits}.

\subsection{Nonperturbative renormalization}
\label{sec:renormalization}

Bare matrix elements are renormalized nonperturbatively by imposing the PCVC relation, \cref{eq:PCVC}.
\Cref{fig:bare_ward_identity} shows the matrix elements entering this expression, before and after renormalization, for the physical-mass $a\approx 0.12\fm$ ensemble with the charm-quark mass near is physical value.
The black points show the quantity
$(m_c-m_q)\matrixel{L}{S}{H} - q^\mu \matrixel{L}{V^\mu}{H}$ with $L \in \{\pi,K\}$, $H\in\{D,D_s\}$, and $m_q \in \{m_l, m_s\}$.
The fact that the open black circles differ slightly from zero gives a visual indication that the renormalization factors $Z_{V^0}$ and $Z_{V^i}$ are necessary to satisfy PCVC.
The closed black squares, statistically consistent with zero, show the precision with which the PCVC relation is satisfied after renormalization.

\begin{figure}
    \centering
    \includegraphics[width=\textwidth]{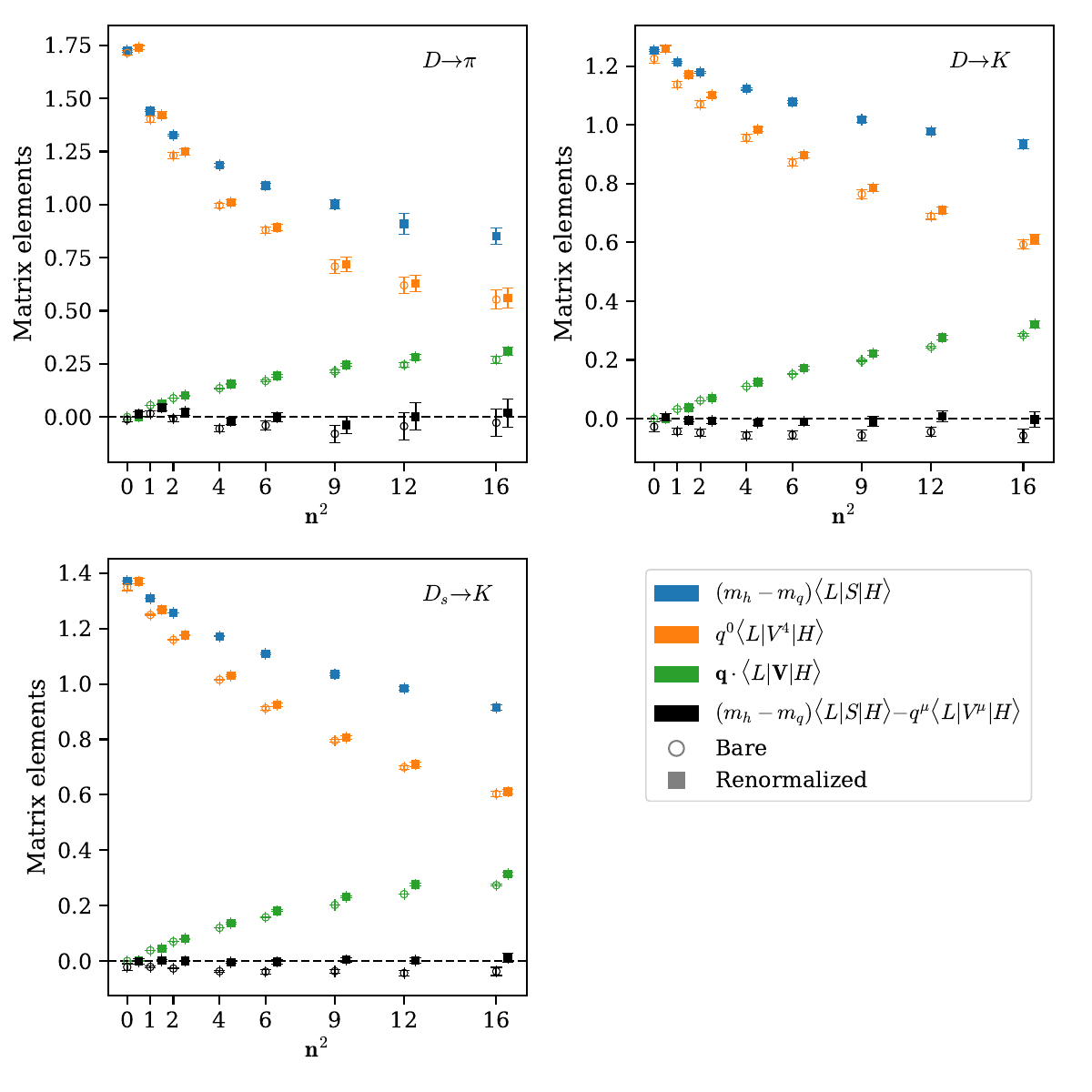}
    \caption{Matrix elements (in arbitrary units) entering the PCVC relation, \cref{eq:PCVC}, before and after renormalization for the physical-mass $a\approx0.12\fm$ ensemble with the charm-quark mass near its physical value, for the decays $D\to\pi$ (top left), $D\to K$ (top right), and $\DsK$ (bottom)
    The open circles denote bare quantities, while the filled squares are renormalized.}
    \label{fig:bare_ward_identity}
\end{figure}

In principle, much freedom exists for extracting the vector-current renormalization factors.
The present analysis fits the bare matrix elements as a function of momentum to \cref{eq:PCVC} for each ensemble and choice of $m_h$, treating $Z_{V^0}$ and $Z_{V^i}$ as free parameters.
Recall $Z_m Z_S=1$ for the local staggered scalar current.
When constructing the renormalized matrix elements, correlations between the bare matrix elements and $Z_{V^0}$ and $Z_{V^i}$ are incorporated via the bootstrap resampling discussed above.

\begin{figure}
    \centering
    \includegraphics[width=1.0\textwidth]{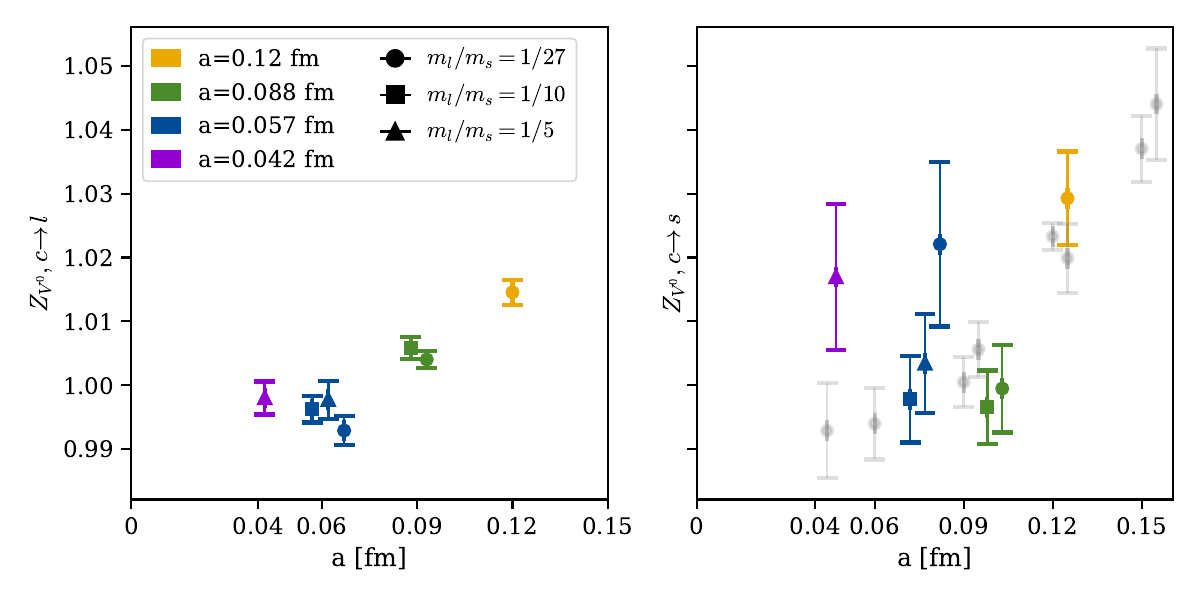}
    \caption{Vector-current renormalization factors $Z_{V^0}$, in all cases with the charm-quark mass near its physical value.
    \textbf{(Left)}
    Results for the quark-level transition $c\to l$ appearing in the decays $D\to \pi$ and $D_s \to K$.
    As discussed in the text, the data are taken from fits to the $D_s\to K$ data.
    \textbf{(Right)}
    Results for the quark-level transition $c\to s$ appearing in the decay $D\to K$.
    The light gray points show the renormalization factors computed by HPQCD on the same ensembles, using the same local $V^0$ current but slightly different valence masses~\cite{Chakraborty:2021qav}.
    The two sets of values agree at $1$--$2\sigma$.}
    \label{fig:renormalization_V0}
\end{figure}

\begin{figure}
    \centering
    \includegraphics[width=1.0\textwidth]{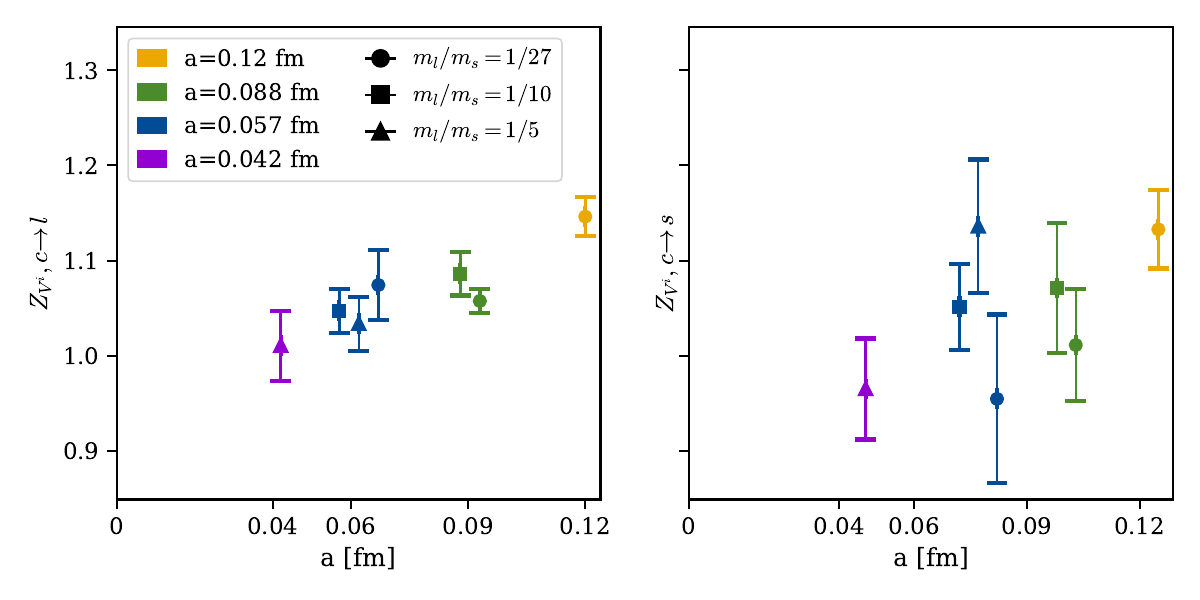}
    \caption{Vector-current renormalization factors $Z_{V^i}$, in all cases with the charm-quark mass near its physical value.
    \textbf{(Left)}
    Results for the quark-level transition $c\to l$ appearing in the decays $D\to \pi$ and $D_s \to K$.
    As discussed in the text, the data are taken from fits to the $D_s\to K$ data.
    \textbf{(Right)}
    Results for the quark-level transition $c\to s$ appearing in the decay $D\to K$.}
    \label{fig:renormalization_Vi}
\end{figure}

\Cref{fig:renormalization_V0,fig:renormalization_Vi} show the results for renormalization factors for the temporal and spatial components of the vector current, respectively, in all cases for the data with the charm-quark mass near its physical value.
The transition $c\to l$ appears in both $D\to \pi$ and $D_s\to K$ decays, differing only by the spectator quark.
The data for the latter decay are statistically more precise, which in turns yields more precise values for the renormalization factors of the $\bar l c$ currents. 
We thus use the renormalization factors extracted from the $D_s\to K$ data to renormalize both $D_s\to K$ and $D\to\pi$ matrix elements.

At a given lattice spacing, uncertainties both in the bare matrix elements (coming from the correlator fits) and in the renormalization factors contribute to the total error budget for the form factors.
The relative importance of the renormalization error depends both on the form factor ($f_\parallel$ or $f_\perp$) and the momentum.
For instance, for the $\Dpi$ decay on the physical-mass $a\approx 0.12\fm$ ensemble, the renormalization error in $f_\parallel$ from $Z_{V^0}$ is $\lesssim 0.1\%$.
For the same decay and lattice spacing, the renormalization error in $f_\perp$ from $Z_{V^i}$ is around $1\%$. For comparison, the individual statistical errors in both $f_\parallel$ and $f_\perp$ (neglecting the renormalization error) range from around $1\%$ at low momentum to around $8\%$ at large momentum. 
These observations are consistent with the expectation that renormalization with PCVC should enable sub-percent determinations of form factors.

\section{Chiral-continuum analysis}
\label{sec:chiral_ctm}

This section describes our chiral-continuum analysis, yielding results for $f_+(q^2)$ and $f_0(q^2)$ at physical quark mass and in the continuum limit.
\Cref{ssec:chiral_ctm_function} describes the fit function used in the analysis and its connection to effective field theory (EFT).
\Cref{ssec:chiral_ctm_fits} presents the results of the fits and describes our definition of the physical point in isospin-symmetric QCD.
\Cref{ssec:fplus_f0} presents a cross check on our results by constructing $f_+$ and $f_0$ in different ways.
\Cref{ssec:z-expansion} re-expresses our results in a compact form using the model-independent $z$~expansion.
\Cref{ssec:spectator_dependence} considers the spectator dependence of the form factors by comparing our results for $\Dpi$ and $\DsK$.
Finally, \cref{ssec:us-vs-them} compares our form factors with published results in the literature.

\subsection{Description of the chiral-continuum fit function \label{ssec:chiral_ctm_function}}

Together, the bare matrix elements and renormalization factors calculated in \cref{sec:correlator_analysis} furnish
the form factors $f_\parallel$, $f_\perp$, and $f_0$ at four different lattice spacings, three different pion masses, and several values of the heavy-quark mass.
These results are extrapolated to the continuum limit and interpolated to the physical point using guidance from effective field theory.

We treat the light-quark mass dependence of the form factors $f_\parallel$ and $f_\perp$ with SU(2) heavy-meson rooted staggered chiral-perturbation theory~\citep{Aubin:2005aq,Aubin:2007mc}.
Following earlier work~\cite{Bailey:2015dka,FermilabLattice:2019ikx}, we use the version for a hard final-state hadron~\citep{Flynn:2008tg,Bijnens:2010ws,Bijnens:2010jg}, hereafter referred to as ``hard SU(2) $\chi$PT.''
We include the complete set of chiral logarithms and analytic corrections through next-to-leading order (NLO) in the chiral expansion.
To account for truncation errors, we also include all analytic terms consistent with the power-counting scheme of Ref.~\citep{Aubin:2005aq,Aubin:2007mc} through next-to-next-leading order (NNLO).
These choices amount to considering the following functional form for $P \in \{\parallel, \perp, 0, +\}$:
\begin{equation}
\begin{split}
    w_0^{d_P} f_{P}(E)
    = \frac{c_0}{w_0\left(E + \Delta_{xy, P}\right)}
    \times \Big[
    1 &+ \delta f_{\rm logs} + c_l \chi_l + c_s \chi_s + c_H \chi_H + c_E \chi_E \\
    &+ c_{l^2} \chi_l^2 + c_{ls} \chi_l \chi_s + c_{s^2} \chi_s^2\\
    &+ c_{lH} \chi_l \chi_H + c_{lE} \chi_l \chi_E + c_{sH} \chi_s \chi_H + c_{sE} \chi_s \chi_E\\
    & + c_{H^2} \chi_H^2 + c_{HE} \chi_H \chi_E + c_{E^2} \chi_E^2\\
    &+ \delta f_{\rm artifacts}^{(a^2 + h^2)}
    \Big],
\end{split}\label{eq:chipt_hard_su2}
\end{equation}
where the exponent $d_P\in\{1/2$, $-1/2$, $0$, $0\}$ for $P \in\{\parallel, \perp, 0, +\}$ (respectively) and $c_0$ is a dimensionless constant. 
Although, in principle, the function describing the chiral logarithms, $\delta f_{\rm logs}$, depends on the form factor, in hard SU(2) $\chi$PT it is the same for all $P \in \{\parallel, \perp, 0, +\}$, as discussed below.

The leading pole factor (in terms of the final-state hadron's energy $E$) arises from the exchange of a virtual $W$ boson, which couples to an excited meson $D_x^*$ composed of $c$ and the final-state quark $x$, contributing a factor proportional to
\begin{equation}
    \frac{1}{E + \Delta_{xy, P}} = \frac{2M_{D_y}}{M_{D^*_x}^2-q^2}.
    \label{eq:pole}
\end{equation}
The intrinsic angular momentum and parity of the $D_x^*$ are those of the virtual $W$ boson, which is $J^P=1^-$ for $f_+$ and  $J^P=0^+$ for $f_0$.
According to the leading-order expectations of the heavy-quark expansion~\cite{Burdman:1993es},
the same pole arises pairwise for $f_\perp$ as $f_+$, and similarly for the pair $f_\parallel$ and~$f_0$ (cf.~\cref{table:pole_locations}). 
\Cref{eq:pole} implies that location of the pole in the energy can be written as
\begin{equation}
    \Delta_{xy, P} = \frac{M_{D^*_x(J^P)}^2 - M_{D_y}^2 - M_{L}^2}{2 M_{D_y}},
\end{equation}
where $y$ is the spectator quark and $L \in \{\pi, K\}$ is the final-state hadron.
Values for the $\Delta_{xy, P}$ are collected in \cref{table:pole_locations} for the decays of interest.

\begin{table}
    \centering
    \caption{Approximate pole locations $\Delta_{xy,P}$ appearing in the decays $\Dpi$, $\DsK$, and $\DK$.}
    \label{table:pole_locations}
    \begin{tabular}{cc cccccc}
    \hline\hline
    Decay     & $c\to x$ & $J^P$ & $D_x^*(J^P)$ & $D_y$ & $L$ & $\Delta_{xy,P}$           & 
    $ (\Delta_{xy,P})^{\rm PDG}$ \\
    \hline
    $\Dpi$ & $c\to l$ & $1^-$ & $D^*$           & $D$   & $\pi$    & $\Delta_{ll,+/\perp}$     & $140\MeV$ \\
           &          & $0^+$ & $D^*_0(2300)$   & $D$   & $\pi$    & $\Delta_{ll,0/\parallel}$ & $480\MeV$ \\
    \hline
    $\DsK$& $c\to l$ & $1^-$ & $D^*$            & $D_s$ & $\pi$    & $\Delta_{ls,+/\perp}$     & $-25\MeV$ \\
          &          & $0^+$ & $D^*_0(2300)$    & $D_s$ & $\pi$    & $\Delta_{ls,0/\parallel}$ & $300\MeV$ \\
    \hline
    $\DK$ & $c\to s$ & $1^-$ & $D_s^*$          & $D$   & $K$      & $\Delta_{sl,+/\perp}$     & $200\MeV$ \\
          &          & $0^+$ & $D_{s0}^*(2317)$ & $D$   & $K$      & $\Delta_{sl,0/\parallel}$ & $440\MeV$ \\
    \hline\hline
\end{tabular}
\end{table}

The $\chi_n$ are dimensionless expansion parameters defined according to
\begin{align}
    \chi_l &= \frac{(M_\pi^{\rm sim})^2}{8 \pi^2 f_{\pi}^2}, \label{eq:chil}\\
    \chi_s &= \frac{(M_K^{\rm sim})^2- (M_K^{\rm PDG})^2}{8 \pi^2 f_{\pi}^2},
    \label{eq:chis}\\
    \chi_E &= \frac{\sqrt{2} E}{4 \pi f_{\pi}}, \label{eq:chiE}\\
    \chi_H &= \frac{\Lambda_{\rm HQET}}{M_{H_{(s)}}^{\rm sim}} - \frac{\Lambda_{\rm HQET}}{M_{D_{(s)}}^{\rm PDG}}
    \label{eq:chiH},
\end{align}
where $f_{\pi}$ is the physical pion decay constant and $\Lambda_{\rm HQET}$
is the scale of heavy-quark effective theory.
As in Ref.~\cite{Bazavov:2017lyh}, we take $\Lambda_{\rm HQET}=800\MeV$.
The parameters $\chi_l$ and $\chi_E$ describe the analytic dependence on the light-quark mass $m_l$ (via the leading-order expression $M_\pi^2 = 2 \mu m_l$) and the final-state hadron energy $E$, respectively. 
Their normalization is such that, according to the typical $\chi$PT power counting, the corresponding coefficients in the fit function $c_l$ and $c_E$ are expected to be of order~1.
The parameter $\chi_H$ describes the heavy quark mass mistuning through the difference between the simulated heavy meson mass $M_{H_{(s)}}^{\rm sim}$ and the physical $M_{D^0}$ or $M_{D_s}$ from Ref.~\cite{Workman:2022ynf}. This term allows for a simultaneous description of results across several different heavy quark masses. 
Finally, $\chi_s$ describes the strange-quark mass mistuning.\footnote{Since the expansion parameters $\chi_l$ and $\chi_s$ are written in terms of the simulated hadron masses, they implicitly accommodate mistuning between the masses of the sea and valence quarks.
As shown in \cref{table:bare_quark_masses}, this feature is only relevant for the finest ensemble, where values for $m_l$ and $m_s$ differ by a small amount ($\approx 1\%$) between the sea and valence quarks.
}

In hard SU(2) $\chi$PT, the chiral logarithms for $f_\parallel$ and $f_\perp$ in \cref{eq:chipt_hard_su2} have the common form~\cite{Bailey:2015dka,FermilabLattice:2019ikx},
\begin{equation}
\begin{split}
    \delta f_{\rm logs}^\text{SU(2)}
    &=\frac{1}{(4\pi f_\pi)^2} \left(
    - \frac{1}{16}\sum_\xi \mathcal{I}_1(M_{\pi,\xi})
    + \frac{1}{4} \mathcal{I}_1(M_{\pi,I}) 
    + \mathcal{I}_1(M_{\pi,V}) - \mathcal{I}_1(M_{\eta, V}) + 
    [V \to A]
    \right)\\
    &\phantom{=}\times
    \begin{cases}
    (1+3g^2), \quad \Dpi \\
    3g^2,   \phantom{(1+{})} \quad \DK  \\
    1,       \phantom{({}+3g^2)} \quad \DsK
    \end{cases}\hspace{-1em},
\end{split}
\label{eq:chiral_logs}
\end{equation}
where $\mathcal{I}_1(M) \equiv M^2 \ln(M^2/\Lambda^2) + 4I_1^{\rm FV}(M, ML)$, with $I_1^{\rm FV}(M, ML)$ being a calculable finite-volume correction to the chiral logarithm which vanishes exponentially for large volumes; see \cref{ssec:finite_volume} below.
Hard SU(2) $\chi$PT enjoys the further simplification that nonanalytic self-energy corrections vanish for all three decays considered here.
These expressions [\cref{eq:chiral_logs} and the self energies] were originally derived for a non-staggered heavy quark, but because the heavy-quark taste is conserved in all-staggered $\chi$PT, they hold in the present case too~\cite{Bernard:2013qwa}.
In heavy-meson $\chi$PT, compact expressions are available for $f_\perp$ and $f_\parallel$~\cite{Becirevic:2002sc,Becirevic:2003ad}, while the corresponding results for $f_+$ and $f_0$ follow as linear combinations.
Because of their simple connection to heavy-meson $\chi$PT, previous lattice calculations have historically worked primarily in terms of $f_\parallel$ and $f_\perp$.
However, since the chiral logarithms have the same functional form for $f_\parallel$ and $f_\perp$ in hard SU(2) $\chi$PT [see \cref{eq:chiral_logs}], the same functional form also describes the chiral logarithms for $f_0$ and $f_+$.
In other words, \cref{eq:chipt_hard_su2} may be used directly for all four form factors, with a $1^-$ pole for $f_{+,\perp}$ or a $0^+$ pole for $f_{0,\parallel}$. 

Following Ref.~\citep{Aubin:2007mc}, the arguments of the chiral logarithms involve the masses of mesons with different tastes $\xi\in\{I,P,V,A,T\}$, that can be expressed as
\begin{align}
    M_{\pi,\xi}^2 &= M_{uu,\xi}^2 = M_{dd,\xi}^2, \\
    M_{ij,\xi}^2 &= \mu (m_i + m_j) + \Delta_\xi, \label{eq:staggered_GMOR}\\
    M_{\eta,V(A)}^2 &= M_{uu,V(A)}^2 + \frac{1}{2} \delta'_{V(A)}, \\
    \bar{\Delta} &= \frac{1}{16} \sum_\xi \Delta_\xi.
\end{align}
The low-energy constant $\mu$ and the taste splittings $\Delta_\xi$ have been tabulated for these ensembles in Ref.~\cite{FermilabLattice:2018zqv}. 
At NLO in the chiral expansion, the taste splittings $\Delta_\xi$ and the hairpin parameters $\delta'_{V,A}$ both scale like $\alpha_s^2 a^2$, so their ratio remains approximately constant as the lattice spacing changes.
We follow Ref.~\citep{Bazavov:2017lyh} and take
$\delta'_A / \bar{\Delta} = -0.88(09)$ and
$\delta'_V / \bar{\Delta} = +0.46(23)$.

Chiral logs described above include the dominant discretization effects coming from the taste-symmetry breaking of staggered fermions at NLO in the chiral expansion.
We also remove the leading-order (tree-level) heavy-quark discretization effects in the form factors prior to fitting by applying a multiplicative normalization factor $Z^{\rm HQET, LO}_{hx}$, described in \cref{ssec:HQETErrors}.
Because of the tree-level improvement of the HISQ action, the remaining discretization effects are expected to arise at order $\alpha_s (a\Lambda)^2$ or $\alpha_s(a m_h)^2$, where $\Lambda$ is the scale of generic discretization effects.
They are thus expected to be well described by an expansion in terms of the parameters $x_{a^2}$ and $x_h$:
\begin{align}
    x_{a^2} &= \frac{a^2 \bar{\Delta}}{8 \pi^2 w_0^2f_\pi^2}, \label{eq:xa2}\\
    x_h &= \frac{2}{\pi}am_h.\label{eq:xh}
\end{align}
The quantity $x_{a^2}$ gives a dimensionless measure of order $\alpha_s (a\Lambda)^2$ discretization corrections, while $x_h$ is the natural expansion parameter for heavy-quark discretization effects.
The HISQ action was designed specifically to control lattice artifacts for charm physics, and the leading heavy-quark corrections are suppressed both by $\alpha_s$ and the velocity $v\approx \sqrt{1/10}$ of the charm quark within the heavy hadron.
Our preferred model thus takes the following simple \emph{Ansatz} for the discretization effects
\begin{align}
    \delta f_{\rm artifacts}^{(a^2 + h^2)}    &= c_{a^2} x_{a^2} + \alpha_s v c_{h^2} x_{h}^2. \label{eq:artifacts_a2h2}
\end{align}
Values for $\alpha_s$ are given in \cref{table:alpha}.

To check for truncation effects with high-order discretization effects, we also consider variations:
\begin{align}
\delta f_{\rm artifacts}^{(a^2 + h^2 + a^4)}    &= c_{a^2} x_{a^2} + \alpha_s v c_{h^2} x_{h}^2 + c_{a^4} x_{a^2}^2, \label{eq:artifacts_a2h2a4}\\
\delta f_{\rm artifacts}^{(a^2 + h^2 + h^4)}    &= c_{a^2} x_{a^2} + \alpha_s v c_{h^2} x_{h}^2 + v c_{h^4} x_{h}^4. 
\label{eq:xh4}
\end{align}
With the HISQ action, discretization effects of order $x_h^4$ are suppressed by $v$ (at the tree level) or by~$\alpha_s$.
These suppression factors are numerically similar enough that the last term in \cref{eq:xh4} tests both.

\begin{table}
\caption{Values for the strong coupling constant, which are based on the continuum value of $\alpha_s(5\GeV, N_f=4)$ from Ref.~\cite{Chakraborty:2014aca}.
Continuum perturbation theory is used to convert to the $\alpha_V$ scheme and to run to the scale $2/a$~\cite{Komijani:2018}.
}
\label{table:alpha}
\begin{tabular}{lc}
\hline \hline
$\approx a$ [fm]    & $\alpha_s(2/a)$\\
\hline
0.15    & 0.3509\\
0.12    & 0.3091\\
0.088   & 0.2646\\
0.06    & 0.2236\\
0.042   & 0.2036\\
\hline \hline
\end{tabular}
\end{table}

\subsection{Chiral-continuum fits \label{ssec:chiral_ctm_fits}}

We perform correlated fits, using the methodology described around \cref{eq:chi2,eq:chi2_prior}, to \cref{eq:chipt_hard_su2} with $\delta f_{\rm artifacts}^{(a^2 + h^2)}$  in \cref{eq:artifacts_a2h2}
for each of the above form factors for $\Dpi$, $\DK$, and $\DsK$ including all of the ensembles and heavy-quark masses described in \cref{table:ensembles}.

In our preferred fits, the input data for the form factors $f_{\perp / \parallel / 0}$ are defined using a single matrix element each via 
\cref{eq:f_parallel,eq:f_perp,eq:f_0}. 
The chiral-continuum fit results using the alternative constructions $f^{\rm alt}_{\rm P}$ of \cref{eq:fplus_parallel_perp,eq:f0_parallel_perp} are considered below in the analysis of systematic effects.  

\begin{table}
    \caption{Summary of the priors used in the chiral-continuum fits to \cref{eq:chipt_hard_su2}.
    Values for $(\Delta_{xy,P})^{\rm PDG}$ in the different decays are given in \cref{table:pole_locations}.
    }
    \label{table:priors_chip}
    \begin{tabular}{cc}
    \hline\hline
    Parameter   & Value \\
    \hline
    $c_0$     & $1 \pm10$ \\
    $c_n$     & $0 \pm 1$ \\
    $g$       & $0.5 \pm 0.2$ \\
    $\Delta_{xy,P}$ & $(\Delta_{xy,P})^{\rm PDG} \pm 200 \MeV$ \\
    \hline\hline
    \end{tabular}
\end{table}

The free parameters varied in the fits are the coefficients $c_n$, the coupling $g$, and the mass splittings $\Delta_{xy,\rm P}$.
The Bayesian priors for these parameters are given in \cref{table:priors_chip}.
The leading coefficient $c_0$ is well determined by the data, so the preferred analysis uses a broad prior (the results are insensitive to the central value, and any reasonable variation gives indistinguishable results). 
The chiral-continuum fit function is based on power-counting arguments from effective field theory, according to which the coefficients $c_n$ are expected to be of order unity.
The preferred analysis therefore uses priors of $0 \pm 1$ for the parameters $c_n$.
The dimensionless (reduced) ``$DD^*\pi$'' coupling appearing as a coefficient of the chiral logarithms is expected to be $g \approx 0.5$, both from experimental measurement~\cite{CLEO:2001sxb,BaBar:2013zgp,BaBar:2013thi} and previous lattice-QCD calculations~\cite{Detmold:2011bp,Detmold:2012ge,Can:2012tx,Becirevic:2012pf,Flynn:2015xna,Bernardoni:2014kla,FermilabLattice:2015mwy}.
For compatibility with these results, our fits take a prior of $0.5\pm 0.2$.
Because a broad width is used for $\Delta_{xy,\rm P}^*$, and since the fits are insensitive to the precise value, the priors do not distinguish between the $J^P=0^+$ and $1^-$ states.
The other inputs to the fits are the measured initial- and final-state hadron masses on each ensemble, the staggered parameters described in the previous section, and the pion decay constant (which is held fixed to its physical value in \cref{table:physical_point_inputs}).

A few words are in order regarding our choice of intermediate scale setting using $w_0/a$ and its role in the chiral-continuum fit function.
The dimensionless expansion parameters $\chi_l$ and $\chi_E$ contain factors of the pion decay constant in the denominator.
Using the continuum values for $f_\pi$ and $w_0$ in \cref{table:physical_point_inputs}, we express the denominator as a dimensionless number.
For the numerator, we use the measured values of $a M_\pi$ and $w_0/a$ to construct the dimensionless product.
In other words, on each ensemble $\chi_l$ is constructed as
\begin{align}
    \chi_l
    \rightarrow
    \left(\frac{(w_0/a)^{\rm sim}}{w_0}\right)^2\frac{(a M_\pi^{\rm sim})^2}{8 \pi^2 (f_\pi^{\rm PDG})^2},
\end{align}
and similarly for $\chi_E$ and $\chi_H$.
In the continuum limit and at the physical point, the 
$w_0$ dependence cancels in all the analytic terms.
Since $f_+$ and $f_0$ are dimensionless, the only residual scale-setting dependence enters through the pole term, where the energy $E$ (in physical units) must be converted to $w_0$ units.

On each ensemble, the input data for the form factors and meson masses and energies are correlated using the results of the bootstrap fits from \cref{sec:correlator_analysis}.
The resulting correlation matrices tend to be near-singular, with small, poorly determined eigenvalues posing a difficult challenge for the fits.
For a given form factor, the dominant source of these eigenvalues is highly correlated data at nearby heavy valence masses (e.g., $0.9\,m_c$, $1.0\, m_c$, and $1.5\,m_c$).
A common solution to this problem is SVD cuts, which have recently been used in another lattice-QCD analysis of $\DK$ form factors~\cite{Chakraborty:2021qav} and which are summarized lucidly in Ref.~\cite{Dowdall:2019bea}.
Another solution is shrinkage of the eigenvalue spectrum, as described in \cref{sec:shrinkage}.

In the analysis of correlation functions in \cref{sec:correlator_analysis}, we could use nonlinear shrinkage, which has the desirable feature of not involving any tunable parameters.
As described in \cref{sec:shrinkage}, however, the amount of shrinkage applied to the eigenvalue spectrum is controlled by the concentration ratio (the ratio of the number of random variables to the number of independent statistical samples).
Since the chiral-continuum extrapolation combines data from different ensembles, there is no clear-cut concentration ratio.
For the chiral-continuum fits, we therefore employ linear shrinkage, which entails a parameter $\lambda$.
We find that $\lambda=0.1$ is large enough to regulate the small eigenvalues (thus giving good fits) without discarding correlations unnecessarily.
As with SVD cuts~\cite{Dowdall:2019bea}, linear shrinkage improves the quality of fits and tends to increase the uncertainty in the posterior values.
The quantitative effect of linear shrinkage is discussed alongside other systematic effects in \cref{ssec:fit_stability}.
A qualitative comparison of nonlinear shrinkage, linear shrinkage, and SVD cuts is given in \cref{sec:shrinkage}.

These fits deliver the form factors in the continuum limit and at the physical point.
The continuum limit of \cref{eq:chipt_hard_su2} is defined by setting $\delta f_{\rm artifacts}^{(a^2+h^2)}$ equal to zero and setting the taste splittings to zero in $\delta f_{\rm P, logs}$. 
The physical point is defined by setting the input meson masses equal to their physical values, given in \cref{table:physical_point_inputs}.
By construction, all quantities involving $\chi_H$ also vanish identically at the physical mass of the decaying heavy meson.
Our simulations and chiral analysis are both done in the isospin limit ({\it i.e.}, with a pair of degenerate quarks with mass $m_l = (m_u + m_d)/2$), so the final results for the form factors correspond to QCD in the isospin limit.
The physical meson masses in \cref{table:physical_point_inputs} are chosen accordingly, following the prescription in Ref.~\cite{Aoki:2021kgd}.  The systematic uncertainty with the isospin-symmetric approximation is discussed below in \cref{ssec:sib_qed}.


\begin{figure}
    \centering
    \includegraphics[width=1.0\textwidth]{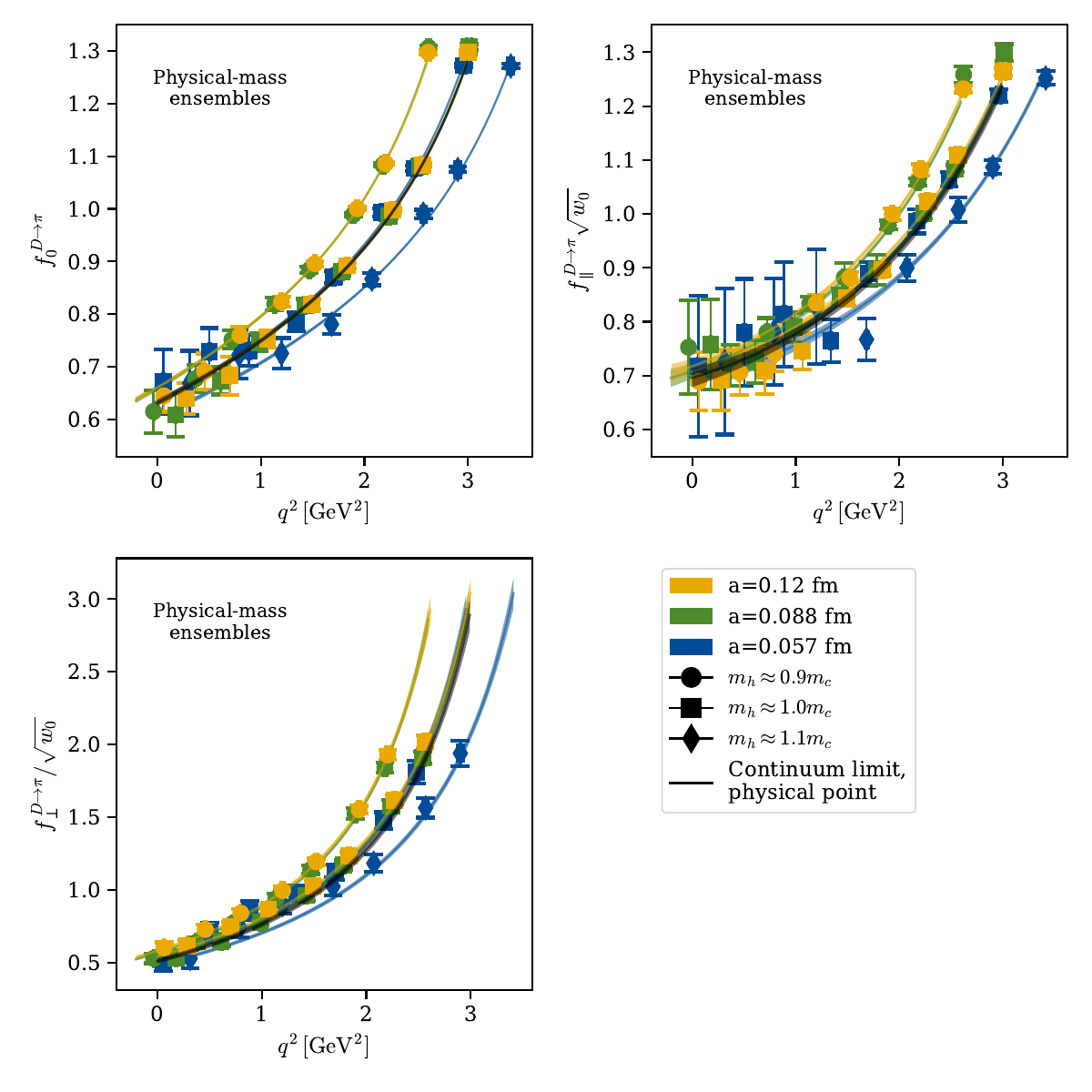}
    \caption{
    The result of the chiral-continuum fit for the $\Dpi$ form factors
    constructed using \cref{eq:f_perp,eq:f_parallel,eq:f_0}
    in units of the gradient-flow scale $w_0$.
    For visual clarity, only the physical-mass ensembles with heavy valence masses $m_h/m_c\in\{0.9, 1.0, 1.1\}$ are shown, although all ensembles and heavy valence masses in \cref{table:ensembles} are included in the fit. Different colors label different lattice spacings and different shapes correspond to the different values of the heavy-quark mass.
    Points with $m_h/m_c\approx 1.1$ were only simulated on the $a\approx 0.06\fm$ ensemble.
    }
    \label{fig:d2pi_data_with_fit} 
\end{figure}

The results for the $\Dpi$ form factors are shown in \cref{fig:d2pi_data_with_fit}.
To avoid plotting many overlapping data and curves, the figures restrict to the three ensembles with physical-mass pions and heavy valence masses with $m_h/m_c\in\{0.9, 1.0, 1.1\}$.
In all cases, the nearly coincident data around the physical charm mass ($m_h \approx m_c$) suggest a mild dependence on the lattice spacing.
The black band denotes the result in the continuum limit and at the physical point.
The results for $\DK$ and $\DsK$ are quite similar and given in \cref{app:chipt_fits}.
\Cref{table:chiral_fit_quality} summarizes the fit quality for the preferred fits.

\begin{table}
\caption{
External inputs used to define the physical point in isospin-symmetric QCD using \cref{eq:chipt_hard_su2}.
As described in the text, the experimentally measured values of the heavy mesons are also used implicitly as inputs in \cref{eq:chipt_hard_su2}.}
\label{table:physical_point_inputs}
\begin{tabular}{l  l  l}
\hline \hline
Quantity    & Value                     & Reference \\
\hline
$f_\pi$     & $130.2(8) \MeV$           & Ref.~\cite{Aoki:2021kgd}\\
$M_{\pi^0}$ & $134.9768(5) \MeV$        & Ref.~\cite{ParticleDataGroup:2020ssz}\\
$M_{K^0}$   & $497.611(13) \MeV$        & Ref.~\cite{ParticleDataGroup:2020ssz}\\
$w_0$       & $0.17177(67) \fm$          & Ref.~\cite{Aoki:2021kgd}\\
\hline
$M_{D^0}$   & $1864.83(05) \MeV$        & Ref.~\cite{ParticleDataGroup:2020ssz}\\
$M_{D_s}$   & $1968.34(07) \MeV$        & Ref.~\cite{ParticleDataGroup:2020ssz}\\
\hline \hline
\end{tabular}

\end{table}

\begin{table}
\caption{
Summary of the reduced $\chi^2$ values and associated degrees of freedom (in brackets) for the preferred fits to \cref{eq:chipt_hard_su2} for all decays and form factors.}
\label{table:chiral_fit_quality}
\begin{tabular}{l  r  r  r}
\hline \hline
                & $\Dpi$    & $\DK$ & $\DsK$ \\
\hline
$f_0$         &  0.91 [126] &  0.48 [128] &  1.31 [134] \\
$f_\parallel$ &  0.59 [112] &  0.41 [123] &  0.88 [128] \\
$f_\perp$     &  0.64 [110] &  0.32 [111] &  0.66 [113] \\
$f_+$         &  0.59 [106] &  0.29 [109] &  0.60 [111] \\
\hline \hline
\end{tabular}
\end{table}

\subsection{Alternative constructions of \texorpdfstring{$f_+$}{f+} and \texorpdfstring{$f_0$}{f0}}
\label{ssec:fplus_f0}

Our default construction for $f_0$ is given by \cref{eq:f_0} and obtained in the preceding section. In an analogous way we construct $f_+$ from the continuum-limit results for $f_0$ and $f_\perp$ in the preceding section following \cref{eq:fplus_perp_0}.
As discussed in \cref{sec:definitions}, the PCVC relation in \cref{eq:PCVC} provides the alternative constructions given in \cref{eq:fplus_parallel_perp,eq:f0_parallel_perp}.
Additional freedom exists in whether the linear combinations in \cref{eq:fplus_parallel_perp,eq:fplus_perp_0,eq:f0_parallel_perp} are taken before or after the chiral-continuum limit. 
A comparison of the different constructions is given in \cref{fig:different_constructions} for $\Dpi$ and $\DsK$ ($\DK$ is similar), where excellent stability is observed throughout the kinematic range.
In the legend, the notation $CL$ specifies whether the continuum limit is taken before or after computing the linear combination [$CL(f_\perp) + CL(f_0)$ versus $CL(f_\perp + f_0)$, respectively].
In all cases, $f_\parallel$ and $f_\perp$ are directly related to vector matrix elements via \cref{eq:f_parallel,eq:f_perp}.
Because the results are statistically consistent, our preferred analysis takes the results with the best statistical precision (our default analysis).
We take $f_0$ from \cref{eq:f_0}. 
We construct $f_+$ via \cref{eq:fplus_perp_0}, using results for $f_\perp$ and $f_0$ given by \cref{eq:f_perp,eq:f_0}, each separately extrapolated to the continuum limit.


\begin{figure}
    \centering
    \includegraphics[width=0.49\textwidth]{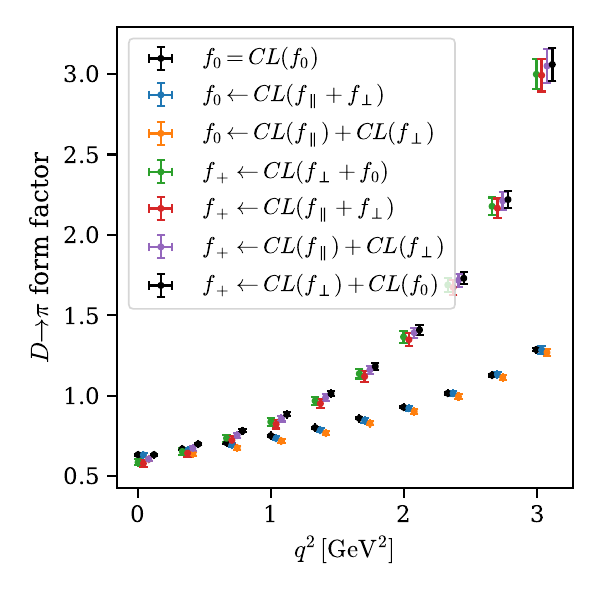}
    \includegraphics[width=0.49\textwidth]{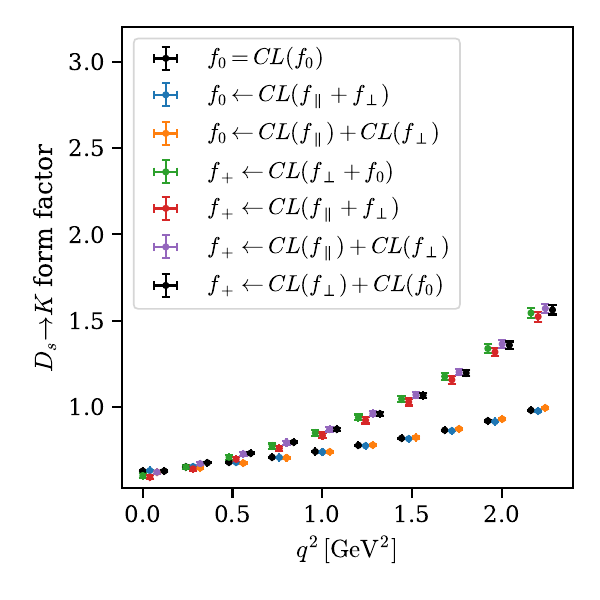}
\caption{
    Comparison of the form factors coming from different continuum-limit constructions for the decays $\Dpi$, and $\DsK$.
    Points are offset horizontally for readability; the same value of $q^2$ is used in each grouping.
    The notation used in the legend is explained in the main text.
    The black points denote the preferred results with the best statistical precision.
    Similar agreement was also found for $\DK$.
    \label{fig:different_constructions}
    }
\end{figure}

\subsection{Model-independent \texorpdfstring{$z$}{z}~expansion}
\label{ssec:z-expansion}

The previous section gave results for $f_+(q^2)$ and $f_0(q^2)$ in the continuum limit and at the physical point.
To facilitate comparison with experimental measurements and other theoretical calculations, it is convenient to re-express our results using the $z$~expansion.
To start, consider the decays $\Dpi$ and $\DK$.
The $z$~expansion leverages the known analytic structure of the form factors in the complex $q^2$-plane to write the form factors as a rapidly convergent expansion in the variable $z(q^2, t_0)$,
\begin{align}
    z(q^2, t_0) = \frac{\sqrt{t_+ - q^2} - \sqrt{t_+ - t_0}}{\sqrt{t_+ - q^2} + \sqrt{t_+ - t_0}},
    \label{eq:z}
\end{align}
where $t_+ = (M_D + M_L)^2$ denotes the start of the multiparticle cut, $L\in\{\pi, K\}$, and $t_0\in [0, t_+]$ can be chosen for convenience.
This map sends the branch cut onto the unit circle, $|z(q^2, t_0)|=1$,
and the rest of the first Riemann sheet onto the open unit disk, $|z(q^2, t_0)|<1$.
Note that
\begin{align}
    z(t_+, t_0) &= -1, \\
    z(t_0, t_0) &= 0, \\
    z(-\infty, t_0) &= +1.
\end{align}
Further, \cref{eq:z} maps the physical region for semileptonic decay onto an interval on the real axis.
Similar considerations apply for the decay $\DsK$, except that the multiparticle cut begins at $t_+=(M_D+M_\pi)^2$ [and \emph{not} at $(M_{D_s}+M_K)^2$].
Below, we take $t_0=0$, so $q^2\in[0,q^2_{\rm max}]$ is mapped to $z\in[0,-z_{\rm max}]$. Because $-z_{\rm max}\approx0.332$, $0.190$, and $0.192$, for \Dpi, \DK, and \DsK, respectively, one expects a series expansion in $z$ to converge within our precision in roughly four or fewer terms. 

The form factors can be expressed in $z$ in various ways~\cite{Boyd:1994tt, Bourrely:2008za}.
We follow Bourrely, Caprini, and Lellouch~\citep{Bourrely:2008za} as,
\begin{align}
    f_0(z) &=
        \frac{1}{1-q^2(z)/M^2_{0^+}}
        \sum_{n=0}^{M-1} b_m z^m,
    \label{eq:z_f0}\\
    f_+(z) &=
        \frac{1}{1-q^2(z)/M^2_{1^-}}
        \sum_{n=0}^{N-1} a_n \left( z^n - \frac{n}{N} (-1)^{n-N} z^N\right).
    \label{eq:z_fplus}
\end{align}
In these expressions, $M_{J^P}$ refers to a possible sub-threshold ($M_{J^P}^2 < t_+$) pole, which requires explicit removal.
For the scalar or vector form factors, the pole corresponds to any sub-threshold particle with quantum numbers $J^P = 0^+$ or $1^-$, respectively, corresponding to the helicity of the virtual $W$ boson.
Such poles are present for the decays $\DK$ and $\DsK$ with $J^P=1^-$.
No sub-threshold poles are present for $\Dpi$, but the fits are more stable if the nearby poles are nevertheless included, as shown previously for $\Dpi$~\cite{Lubicz:2017syv}.

\begin{table}
\caption{
Pole masses and cut positions used in \cref{eq:z_f0,eq:z_fplus}.
The closest pole and the start of the cut are the same for both $\Dpi$ and $\DsK$,
since they both involve the same $c\to d$ quark-level transition.
\label{table:zexp_poles}
}
\begin{tabular}{ccclcl}
\hline \hline
Decay   & $\sqrt{t_+}$ &  pole      & $J^P=1^-$  & pole & $J^P=0^+$ \\
\hline
$\DK$   & $M_D+M_K$ & $D_s^*$   & $2112.2(4) \MeV$  & $D_{s0}^*$ & $2317.8(5) \MeV$ \\
$\Dpi$, $\DsK$  & $M_D+M_\pi$ & $D^*$     & $2006.85(05) \MeV$ & $D_0^*$ & $2300(15) \MeV$ \\
\hline \hline
\end{tabular}
\end{table}

Since the input data from the continuum results in \cref{ssec:fplus_f0} spans the full kinematic range of the decay, the $z$~expansion amounts to a convenient change of variables. 
To carry out this procedure, we evaluate each form factor at four evenly spaced points spread throughout the physical $q^2$-region:  $[0.1, 0.37, 0.63, 0.9] \times q^2_{\rm max}$.
We then perform a joint correlated fit of these synthetic data to \cref{eq:z_f0,eq:z_fplus}, imposing the kinematic constraint $f_+(0)=f_0(0)$ by taking a common coefficient for $n=0$: $a_0 \equiv b_0$.
The pole masses entering \cref{eq:z_f0,eq:z_fplus} are given in \cref{table:zexp_poles}.
\Cref{table:z_results} reports the correlated posterior values for $a_n$ and $b_m$ emerging from the preferred fits for the three decays analyzed.
The preferred fits have $N=M=4$ terms for all three decays.
As shown in \cref{fig:zexpansion_stability}, the posteriors for the coefficients stabilize with these choices.
In all cases, statistical uncertainties in the fit parameters are determined via bootstrap resampling with $500$ draws. 
These bootstrap fits also furnish estimates of the $21\times 21$ correlation matrix associated with the full set of form factors ($f_+$ and $f_0$ for all three decays).
The block-diagonal correlations for each decay are also given in \cref{table:z_results}, while the full correlation matrix is given in the supplementary material.
The results for $f_+(q^2)$ and $f_0(q^2)$ coming directly from the chiral-continuum fits (before applying the $z$~expansion) are compared with those from the $z$~expansion in \cref{fig:zexpansion_final}.

An alternative, common form of the $z$~expansion uses~\cite{Boyd:1994tt}
\begin{align}
f_+(q^2) &= \frac{1}{P(q^2) \phi(q^2)}\sum_{n=0}^{N-1} a_n z^n \label{eq:zexpansion_expt},
\end{align}
with $P(q^2)=1$ for $\Dpi$ or $z(q^2, M_{D_s^*}^2)$ for $\DK$ and outer function $\phi(q^2)$ given by
\begin{align}
\phi(q^2) &= \sqrt{\frac{\pi}{3}}m_c
    \left( \frac{z(q^2,0)}{-q^2} \right)^{5/2}
    \left( \frac{z(q^2,t_0)}{t_0-q^2} \right)^{-1/2}
    \left( \frac{z(q^2,t_-)}{t_--q^2} \right)^{-3/4}
    \left( \frac{t_+ - q^2}{(t_+ - -t_0)^{1/4}} \right),
\end{align}
where $t_0 = t_+(1-\sqrt{1-t_-/t_+})$ and $m_c = 1.25 \GeV$.
For comparison with the experimental determination of the shapes in \cref{ssec:CKM} below, we use \cref{eq:zexpansion_expt} together with the refitting procedure described in Ref.~\cite{Chakraborty:2021qav}.

\begin{figure}
    \centering
    \includegraphics[width=0.49\textwidth]{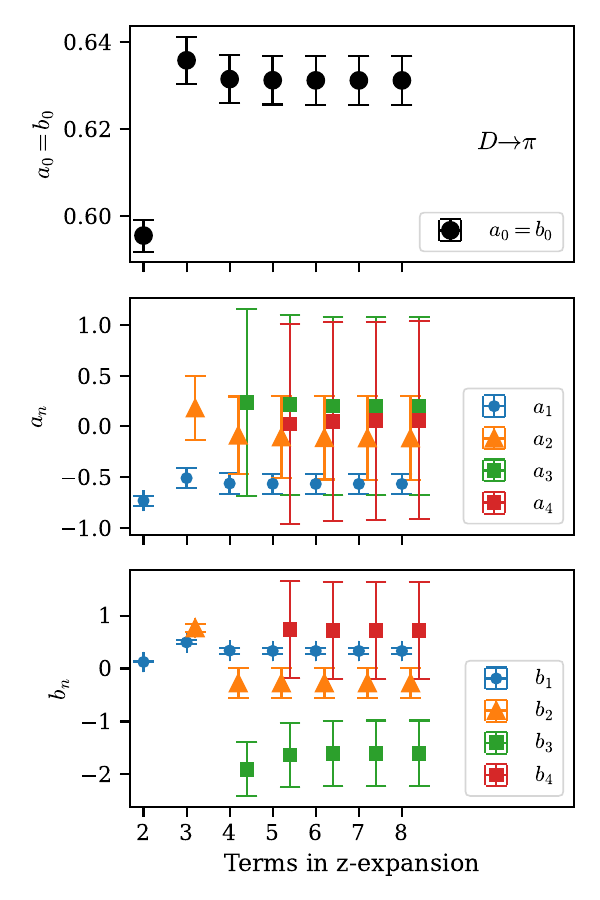}
    \includegraphics[width=0.49\textwidth]{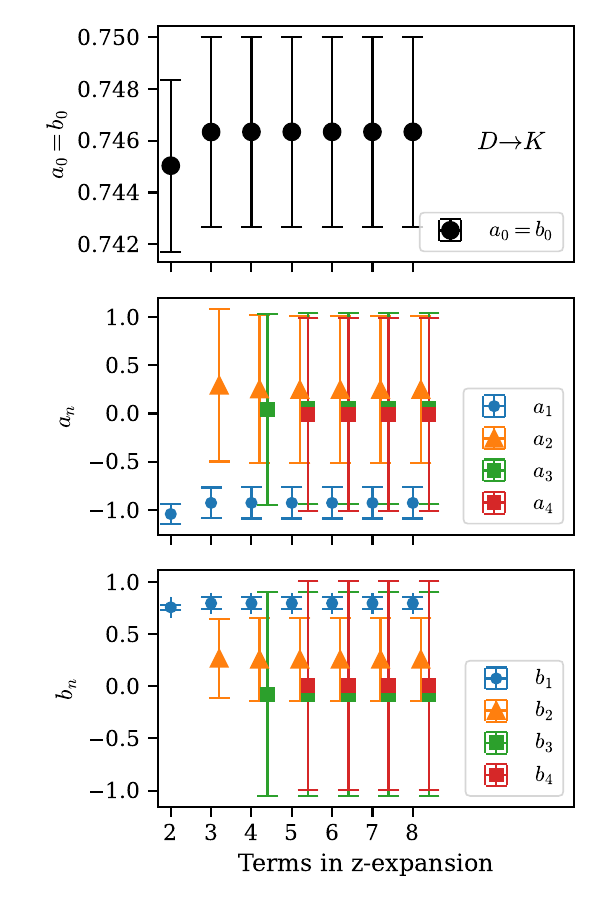}
    \caption{
    Stability analysis for the fit parameters appearing in the $z$~ expansion for the decays $\Dpi$ and $\DK$.
    The preferred fit uses $N=M=4$ terms (i.e., up to and including $a_3$ and $b_3$), at which point the coefficients' central values and errors have stabilized and higher-order terms are expected to contribute negligibly at the current level of precision.
    Results for $\DsK$ are qualitatively similar.
    }
    \label{fig:zexpansion_stability}
\end{figure}


\begin{table}
    \caption{
    Correlated posterior values for $a_n$ and $b_m$ for the coefficients of the $z$~expansion for the decays $\Dpi$, $\DK$, and $\DsK$.
    The simultaneous fit to \cref{eq:z_fplus,eq:z_f0} constrains $a_0\equiv b_0$. 
    The pole masses used in the fits are given in \cref{table:zexp_poles}.
    The full correlation matrix is given in the supplementary material.
    The supplementary material also contains a script, \texttt{reconstuct.py}, which shows an example of how to read the $z$-expansion coefficients and recreate our final results for the form factors as function of the momentum transfer, correctly including the full correlation matrix. 
    \label{table:z_results}
    }
    \begin{tabular}{cccccccc}
$\Dpi$ & $a_0 \equiv b_0$ & $a_1$ & $a_2$ & $a_3$ & $b_1$ & $b_2$ & $b_3$\\
 & 0.6300(51) & -0.610(99) & -0.20(30) & 0.30(19) & 0.330(51) & -0.31(25) & -1.90(39)\\
\hline
 & 1.0000 & 0.5670 & 0.5189 & -0.2018 & 0.7547 & 0.3473 & 0.0861\\
 &  & 1.0000 & 0.8912 & -0.2826 & 0.5148 & 0.2529 & 0.0747\\
 &  &  & 1.0000 & -0.1482 & 0.5082 & 0.2782 & 0.1162\\
 &  &  &  & 1.0000 & -0.1728 & -0.0496 & 0.0354\\
 &  &  &  &  & 1.0000 & 0.8277 & 0.6066\\
 &  &  &  &  &  & 1.0000 & 0.9442\\
 &  &  &  &  &  &  & 1.0000\\
\end{tabular}

    \begin{tabular}{cccccccc}
$\DK$ & $a_0 \equiv b_0$ & $a_1$ & $a_2$ & $a_3$ & $b_1$ & $b_2$ & $b_3$\\
 & 0.7452(31) & -0.948(97) & 0.14(40) & 0.07(12) & 0.776(62) & 0.14(34) & 0.03(13)\\
\hline
 & 1.0000 & -0.0332 & 0.0747 & -0.0201 & 0.7753 & 0.4920 & -0.0189\\
 &  & 1.0000 & 0.3272 & -0.1586 & -0.0909 & -0.1090 & 0.0420\\
 &  &  & 1.0000 & -0.7543 & 0.2071 & 0.2565 & 0.1457\\
 &  &  &  & 1.0000 & -0.0594 & -0.1119 & -0.2259\\
 &  &  &  &  & 1.0000 & 0.9087 & 0.1012\\
 &  &  &  &  &  & 1.0000 & 0.2126\\
 &  &  &  &  &  &  & 1.0000\\
\end{tabular}

    \begin{tabular}{cccccccc}
$\DsK$ & $a_0 \equiv b_0$ & $a_1$ & $a_2$ & $a_3$ & $b_1$ & $b_2$ & $b_3$\\
 & 0.6307(20) & -0.562(65) & -0.19(20) & 0.33(29) & 0.347(27) & 0.44(18) & -0.21(43)\\
\hline
 & 1.0000 & 0.1825 & 0.2612 & -0.0266 & 0.8467 & 0.5197 & 0.0973\\
 &  & 1.0000 & 0.9274 & -0.2432 & 0.1899 & 0.1915 & 0.1180\\
 &  &  & 1.0000 & -0.0514 & 0.3243 & 0.3065 & 0.1764\\
 &  &  &  & 1.0000 & -0.0551 & -0.1580 & -0.2098\\
 &  &  &  &  & 1.0000 & 0.8344 & 0.4260\\
 &  &  &  &  &  & 1.0000 & 0.8442\\
 &  &  &  &  &  &  & 1.0000\\
\end{tabular}

\end{table}

\begin{figure}
    \centering
    \includegraphics[width=0.45\textwidth]{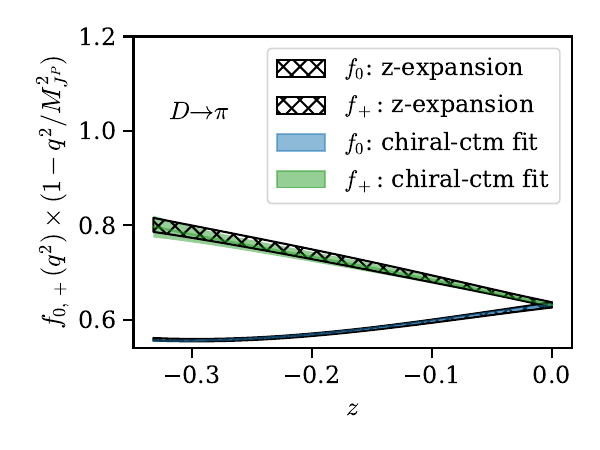}
    \includegraphics[width=0.45\textwidth]{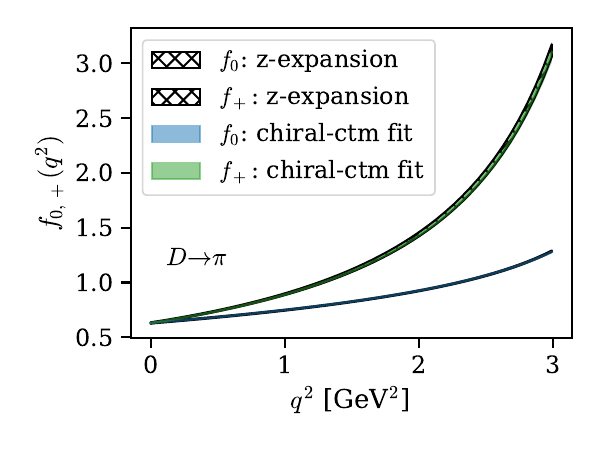}
    \includegraphics[width=0.45\textwidth]{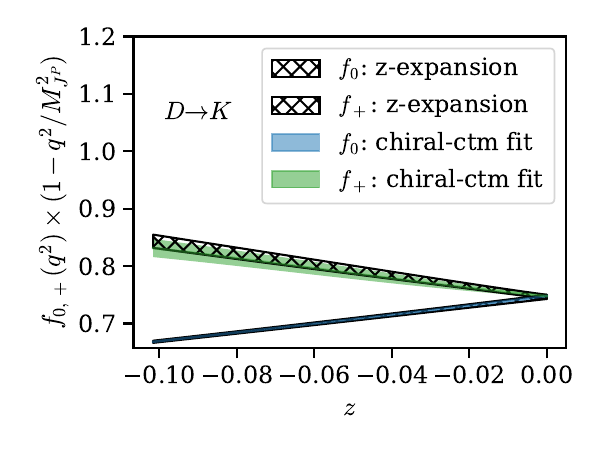}
    \includegraphics[width=0.45\textwidth]{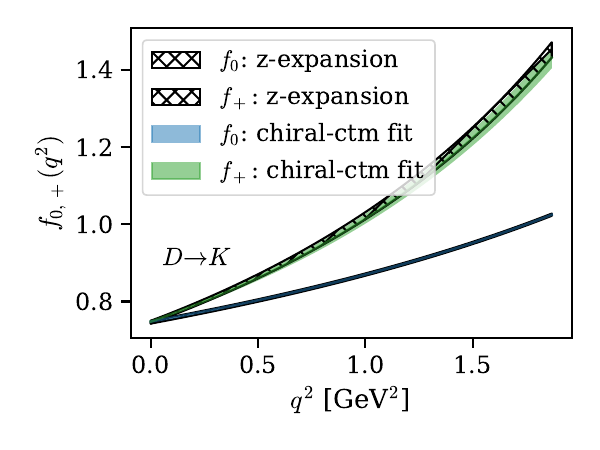}
    \includegraphics[width=0.45\textwidth]{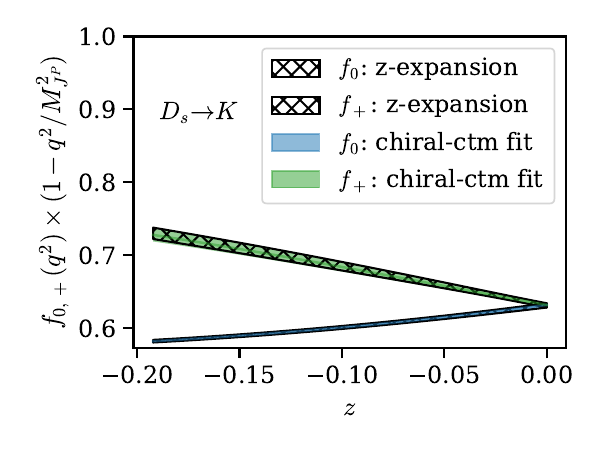}
    \includegraphics[width=0.45\textwidth]{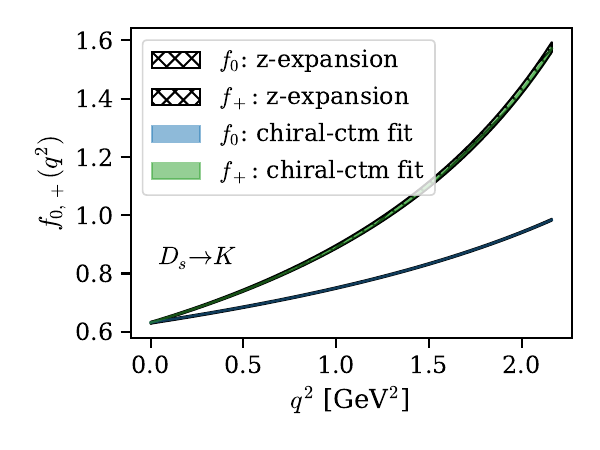}
    \caption{
    Final results for $f_+$ and $f_0$ for the decays $\Dpi$, $\DK$, and $\DsK$
    in the continuum limit and at the physical point before and after fitting to $z$~expansion.
    The solid curves show the results after the chiral-continuum fit, while the hatched curves show the result of the $z$~expansion.
    The left column shows the product $(1-q^2/M_{J^P}^2)f_{0,+}$ as a function of $z$, while the right column shows the form factors versus~$q^2$.}
    \label{fig:zexpansion_final}
\end{figure}

\begin{table}
    \centering
    \caption{
    Final results for $f_+(0)=f_0(0)$, $f_+(q^2_{\rm max})$, and $f_0(q^2_{\rm max})$ for the decays $\Dpi$, $\DK$, and $\DsK$, together with comparisons with existing $N_f=2+1+1$ results in the literature from HPQCD~\cite{Chakraborty:2021qav,Parrott:2022rgu} and ETMC~\cite{Lubicz:2017syv}.
    The results of the present work, denoted ``Fermilab-MILC'', are all given at the physical point and in the continuum limit in isospin-symmetric QCD.
    Included in these results are all systematic errors discussed in \cref{sec:syst_errors} and summarized in \cref{table:final_error_budget}.
    Not included are additional systematic uncertainties associated with QED, isospin, and electroweak corrections (these effects are estimated in \cref{ssec:sib_qed}).
    The different groups use slightly different conventions to define the isospin-symmetric point.
    Shifts from these differences are expected to be small.
    \Cref{fig:sib} suggests that the largest differences, perhaps amounting to a few percent, will be present near $q^2_{\rm max}$.}
    \label{table:ff_final_endpoint_results}
    \begin{tabular}{lllll}
\hline
process & collaboration &   $f_0(0)$ & $f_0(q^2_{\rm max})$ & $f_+(q^2_{\rm max})$ \\
\hline
   $\Dpi$ &     FNAL/MILC & 0.6300(51) &           1.2783(61) &            3.119(57) \\
   $\Dpi$ &       ETMC 17 &  0.612(35) &            1.134(49) &            2.130(96) \\
    \hline
    $\DK$ &     FNAL/MILC & 0.7452(31) &           1.0240(21) &            1.451(17) \\
    $\DK$ &      HPQCD 22 & 0.7441(40) &           1.0136(36) &            1.462(16) \\
    $\DK$ &      HPQCD 21 & 0.7380(40) &           1.0158(41) &            1.465(20) \\
    $\DK$ &       ETMC 17 &  0.765(31) &            0.979(19) &            1.336(54) \\
   \hline
    $\DsK$ &     FNAL/MILC & 0.6307(20) &           0.9843(18) &            1.576(13) \\
\hline
\end{tabular}

\end{table}

\subsection{Spectator dependence\label{ssec:spectator_dependence}}

From the hadronic perspective, the decay channels $\Dpi$ and $\DsK$ are quite similar, differing only by the mass of the valence spectator quark.
As illustrated in \cref{fig:spectator_comparison}, we find that the vector and scalar form factors for these two transitions agree with with each other at the level of $\lesssim 2\%$ throughout the full kinematic range of the $\DsK$ decay.
The first experimental measurement of the decay $\DsK$ by BES~III~\cite{BESIII:2018xre} confirms this picture within experimental uncertainties while old, unpublished results by the HPQCD collaboration~\cite{Koponen:2012di,Koponen:2013tua} are also consistent with our findings.

\begin{figure}
    \centering
    \includegraphics[width=0.5\textwidth]{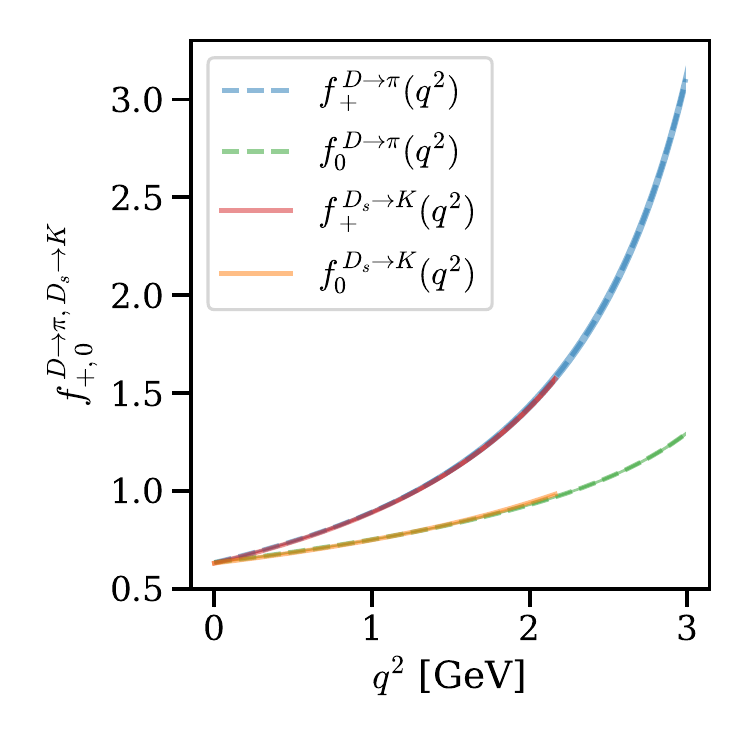}
    \caption{
    Comparison of the vector and scalar form factors between the decays $\Dpi$ and $\DsK$, which differ only by the mass of the valence spectator quark.
    The form factors agree at the level of $\lesssim 2\%$ throughout the full kinematic range of the $\DsK$ decay.
    The long dashed lines extending to $q^2\approx 3 \GeV^2$ correspond to $\Dpi$, while the shorter solid lines correspond to $\DsK$.}
    \label{fig:spectator_comparison} 
\end{figure}

\subsection{Comparison with existing results in the literature}
\label{ssec:us-vs-them}

The form factors under consideration have been computed previously using lattice QCD with $N_f=2+1+1$ flavors of dynamical fermions by ETMC~\cite{Lubicz:2017syv,Lubicz:2018rfs} (for both $\Dpi$ and $\DK$) and by HPQCD (for $\DK$)~\cite{Chakraborty:2021qav,Parrott:2022rgu}. 
The more recent HPQCD calculation~\cite{Parrott:2022rgu} includes the same set of $D\to K$ correlators as the earlier one~\cite{Chakraborty:2021qav}, but they are analyzed together with tensor-current three-point functions, data for heavier-than-charm quark masses, and $D_s\to\eta_s\ell\nu$ form factor data~\cite{Parrott:2020vbe}.
Both the correlator fits and the description of the heavy-quark-mass dependence and discretization effects are thus different.
Our $\Dpi$ results for the form factors and the semimuonic differential decay rate are compared with those of ETMC in \cref{fig:d2pi_lit_compare}.
At large $q^2$, our results for the form factors are significantly larger than those in Ref.~\cite{Lubicz:2017syv}.
Due to phase-space suppression, the difference is less visibly pronounced in the differential decay rate $d\Gamma/dq^2$.
In the low $q^2$ region, which is most relevant for extractions of $\Vcd$, good agreement is observed at the level of $\approx 1\sigma$.
Similarly, our $\DK$ results are compared with those of ETMC and HPQCD in \cref{fig:d2k_lit_compare}.
Mild tension, at the level of $\approx 2\sigma$, is observed between our results and ETMC.
Good agreement with HPQCD is observed throughout the kinematic range.
Our results for $f_+(0)=f_0(0)$, $f_0(q^2_{\rm max})$, and $f_+(q^2_{\rm max})$ are summarized in \cref{table:ff_final_endpoint_results} alongside the published results of Refs.~\cite{Lubicz:2017syv,Lubicz:2018rfs,Chakraborty:2021qav,Parrott:2022rgu}.

\begin{figure}
    \centering
    \includegraphics[width=0.49\textwidth]{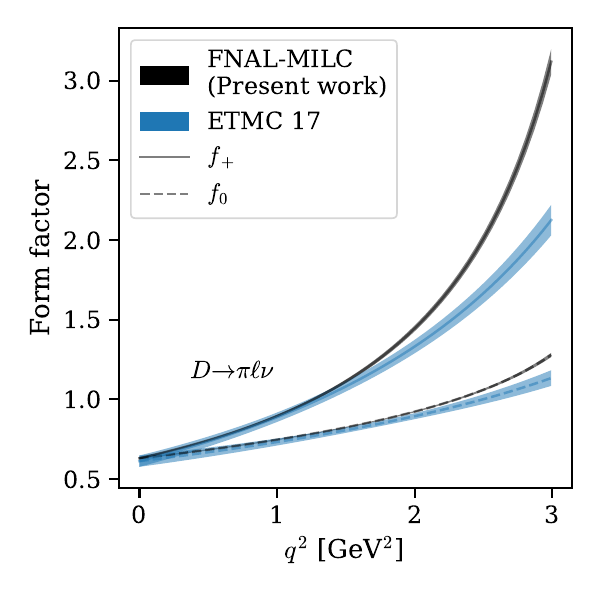}
    \includegraphics[width=0.49\textwidth]{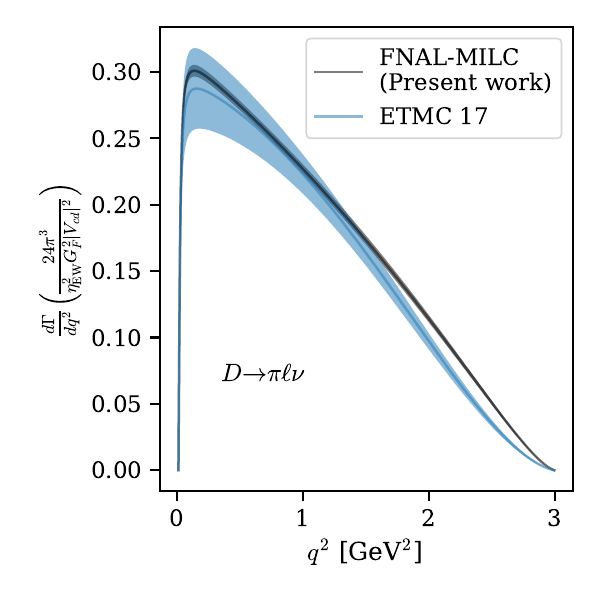}
\caption{
    Comparison of our results for the $\Dpi$ form factors and semimuonic differential decay rate $(d\Gamma/dq^2)(24\pi^3/G_F^2)/\Vcd^2$ with published results from ETMC~\cite{Lubicz:2017syv}.
    No QED or electroweak corrections [cf.\ $\eta_{\rm EW}$ in \cref{eq:dGammadq2}] or errors have been included.
    To account for differences in defining the physical isospin-symmetric point, the errors in our curves have been inflated with an estimate of SIB effects; see \cref{ssec:sib_qed} below.}
    \label{fig:d2pi_lit_compare}
\end{figure}

\begin{figure}
    \centering
    \includegraphics[width=0.49\textwidth]{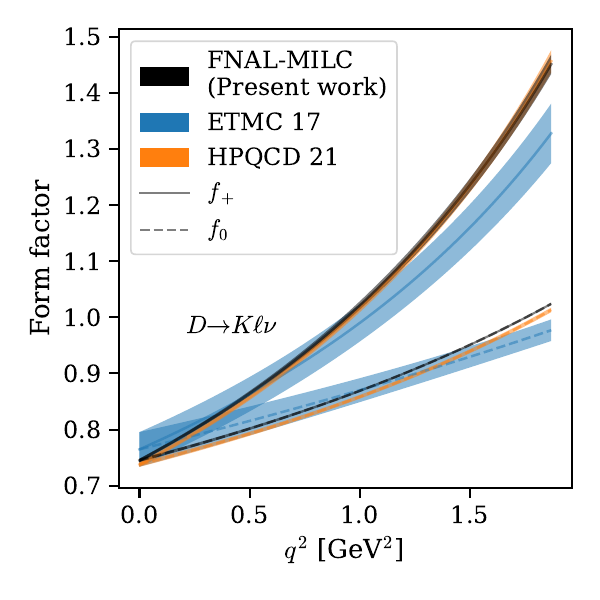}
    \includegraphics[width=0.49\textwidth]{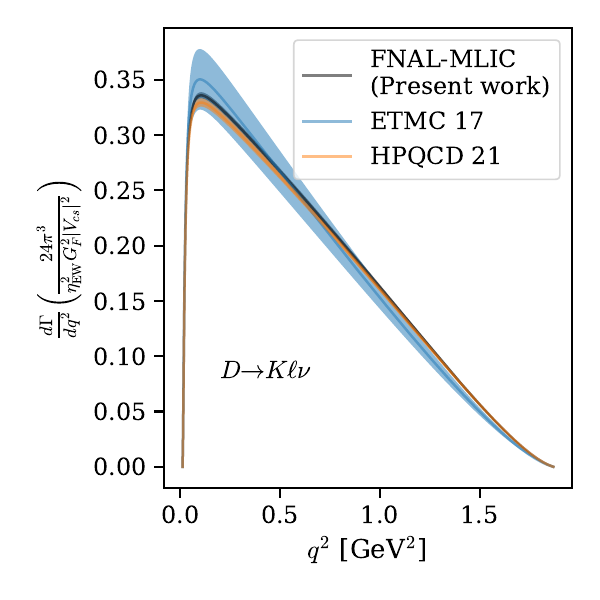}
\caption{
    Comparison of our results for the $\DK$ form factors and semimunoic differential decay rate $(d\Gamma/dq^2)(24\pi^3/G_F^2)/\Vcs^2$ with published results from ETMC~\cite{Lubicz:2017syv} and HPQCD~\cite{Chakraborty:2021qav}.
    No QED or electroweak corrections (cf.\ $\eta_{\rm EW}$ in \cref{eq:dGammadq2}) or errors have been included.
    To account for differences in defining the physical isospin-symmetric point, the errors in our curves have been inflated with an estimate of SIB effects; see \cref{ssec:sib_qed} below.}
    \label{fig:d2k_lit_compare}
\end{figure}

\section{Systematic Error Analysis\label{sec:syst_errors}}

The fits to the $z$~expansion described in \cref{ssec:z-expansion} and given in \cref{table:z_results} provide our final results for the pure-QCD form factors at the physical point in isospin-symmetric QCD.
In this section, we examine and quantify the  various statistical and systematic uncertainties contributing to the calculations.
The complete final error budget is summarized in \cref{table:final_error_budget} for $f_+(0)=f_0(0)$, $f_+(q^2_{\rm max})$, and $f_0(q^2_{\rm max})$ for all decay modes.
As discussed in \cref{ssec:finite_volume,ssec:topology}, the very small corrections for the leading finite-volume shifts ($\lesssim 0.01\%$) and the effect of nonequilibrated topological charge (relevant for $a\approx 0.042\fm$ only) have been applied to the form factors prior to fitting and thus are not included as separate errors.
Systematic errors associated with isospin breaking effects and QED corrections, which are external to our calculation in isospin-symmetric QCD but necessary for comparison with experimental results, are discussed in \cref{ssec:sib_qed}.

\begin{table}
\caption{
Complete statistical and systematic error budget for the vector and scalar form factors at $q^2=0$ and $q^2_{\rm max}$ for the decays $\Dpi$, $\DK$, and $\DsK$.
All values are given in percent.
The breakdown of the chiral-continuum fit errors is discussed in \cref{ssec:stat_error_budget}.
Corrections for finite-volume and topological-charge effects, discussed in \cref{ssec:finite_volume,ssec:topology}, are applied prior to the chiral-continuum fit and are negligibly small ($<0.01\%$).
Experimental uncertainties on the meson masses are also negligible at our current level of precision.}
\label{table:final_error_budget}

\begin{tabular}{l | ccc | ccc | ccc}
\hline \hline
Decay   & $\Dpi$ & & & $\DK$ & & & $\DsK$ & & \\
\hline
Source  & $f_+(0)$ & $f_+(q^2_{\rm max})$ & $f_0(q^2_{\rm max})$   
        & $f_+(0)$ & $f_+(q^2_{\rm max})$ & $f_0(q^2_{\rm max})$   
        & $f_+(0)$ & $f_+(q^2_{\rm max})$ & $f_0(q^2_{\rm max})$\\
\hline
Statistics $f_\perp$ & 0.21 & 1.46 & 0.01 & 0.07 & 0.95 & 0.01 & 0.07 & 0.73 & 0.02\\
Statistics $f_0$ & 0.70 & 0.39 & 0.40 & 0.39 & 0.36 & 0.22 & 0.29 & 0.18 & 0.12\\
Ctm. $w_0$ & 0.31 & 0.09 & 0.15 & 0.24 & 0.14 & 0.17 & 0.32 & 0.27 & 0.27\\
$\chi$EFT $f_\perp$ & 0.12 & 0.46 & 0.01 & 0.03 & 0.50 & 0.00 & 0.03 & 0.39 & 0.01\\
$\chi$EFT $f_0$ & 0.24 & 0.13 & 0.13 & 0.11 & 0.14 & 0.05 & 0.05 & 0.03 & 0.02\\
Discr. $f_\perp$ & 0.08 & 0.73 & 0.00 & 0.06 & 0.41 & 0.01 & 0.03 & 0.41 & 0.01\\
Discr. $f_0$ & 0.05 & 0.02 & 0.07 & 0.02 & 0.02 & 0.02 & 0.00 & 0.00 & 0.01\\
$f_\pi^{\rm PDG}$ & 0.16 & 0.10 & 0.13 & 0.06 & 0.04 & 0.03 & 0.12 & 0.04 & 0.09\\
\hline
Total error & 0.87 & 1.84 & 0.48 & 0.49 & 1.29 & 0.28 & 0.46 & 0.99 & 0.30\\
\hline \hline
\end{tabular}
\end{table}

\subsection{Chiral-continuum fits: stability analysis \label{ssec:fit_stability}}

The results in \cref{sec:chiral_ctm} are the product of several choices.
In this section, we examine the stability of the results under reasonable variations to these choices for the fiducial point $q^2=0$.
First, the model for the EFT is varied.
The staggered chiral logarithms are replaced with their continuum counterparts, setting the known taste splittings to zero by hand.
Another alternative is simply dropping the chiral logarithms $\delta f_{\rm P, logs}$ in \cref{eq:chipt_hard_su2}.
This variation is reasonable, since the ensembles with physical-mass pions reduce the approach to the physical point from an extrapolation to an interpolation.
The final EFT variation consists of augmenting the analytic terms in \cref{eq:chipt_hard_su2} to include all the N${}^3$LO terms ({\it i.e.}, terms cubic in the $\chi_\ell$, $\chi_H$, and $\chi_E$).
Second, we consider variations to the model for discretization effects as given in \cref{eq:artifacts_a2h2a4,eq:xh4}.
Third, the widths of our Bayesian priors are increased, and the fits are rerun.
In one variation, the widths of the priors for the coefficients of the leading-order analytic terms ($c_l$, $c_E$, and $c_H$) are increased by a factor of ten.
In another variation, the widths of all the priors are increased by a factor of two.
Fourth, the choice of the linear shrinkage parameter is tested by fits varying it by a factor of 2 from its fiducial value ($\lambda=0.1$).
Finally, the choice of data used in the fits is varied, rerunning after dropping the coarsest ensemble ($a\approx 0.12\fm$) and after dropping the finest ensemble ($a\approx 0.042\fm$).

As \cref{fig:d2pi_stability} shows, for $\Dpi$, that all variations are statistically consistent with the preferred fit at the level of one standard deviation. 
Stability plots for $\DK$ and $\DsK$ are similar and given in \cref{fig:d2k_stability,fig:ds2k_stability} in \cref{app:chipt_fits}.


\begin{figure}
    \centering
    \includegraphics[width=0.75\textwidth]{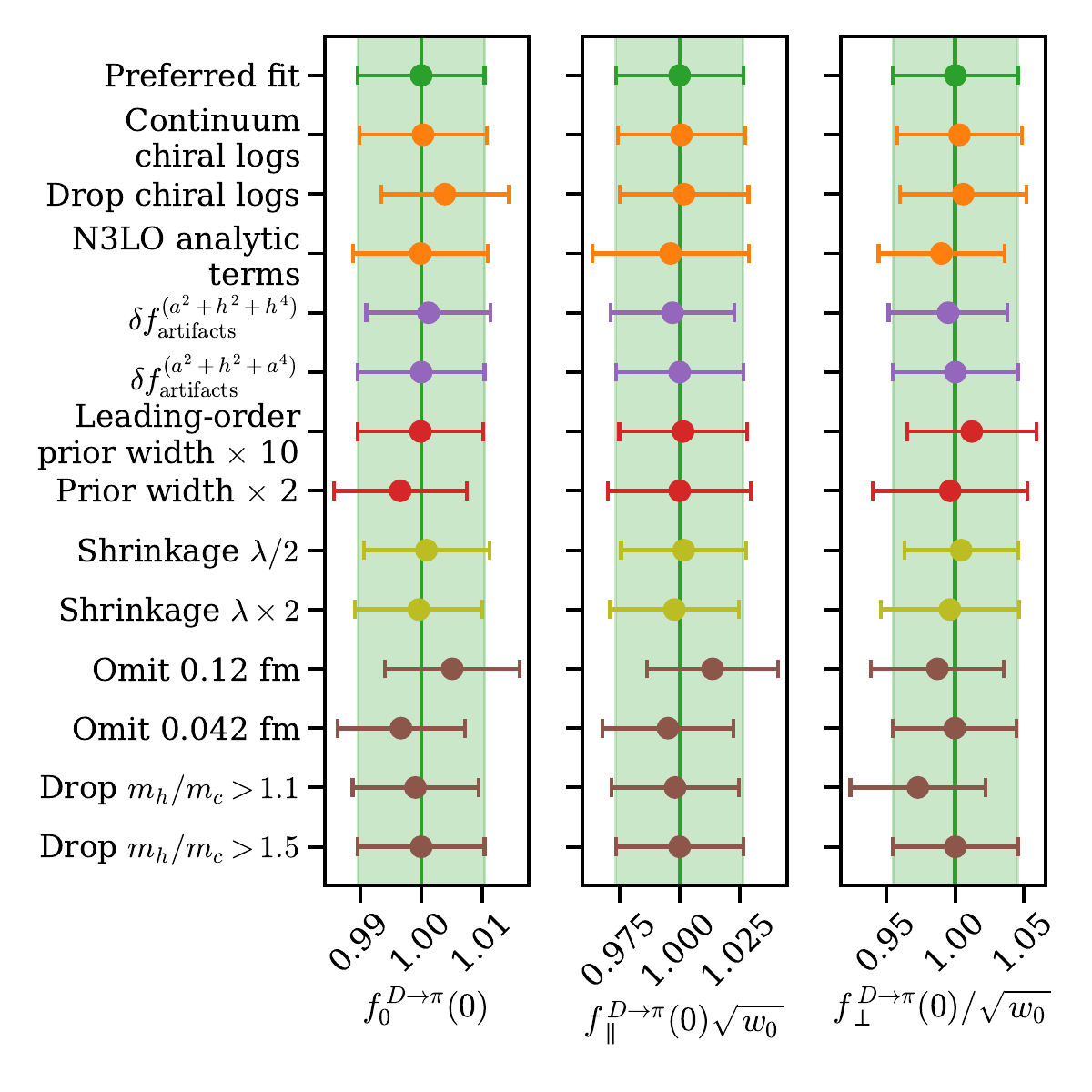}
    \caption{Stability of the $\Dpi$ form factors $f_{\perp/\parallel/0}$ at $q^2=0$
    under variations to the EFT model, the model for discretization effects, 
    to the choice of data included in the fit, and other analysis choices as described in the main body.
    The central values have been normalized by the central value of preferred fit in green.
    All variations are statistically consistent with the preferred fit, highlighted by the green band in each panel.
    The statistical significance of the fits is indicated by the marker size, with larger points denoting better fits.}
    \label{fig:d2pi_stability}
\end{figure}

The discussion in \cref{ssec:fplus_f0} demonstrates good agreement for the physical form factors constructed in different ways, while the discussion above shows that alternative discretization models, as well as continuum-$\chi$PT fit functions (without taste splittings in the chiral logarithms), give consistent results.

\subsection{Chiral-continuum fits: error breakdown}
\label{ssec:stat_error_budget}

The form-factor results coming out of the chiral-continuum fits contain several sources of uncertainty:
statistical errors in the form factor on each ensemble (the correlated uncertainty from the bare form factors and renormalization constants),
scale-setting errors coming from the continuum value of $w_0$, choices in the fit function and chiral interpolation, discretization effects, and errors in the input parameters (physical meson masses and $f_\pi$ in \cref{table:physical_point_inputs}).
The different sources of error are entangled in the total fit uncertainty; in particular, the fit function, chiral interpolation and discretization errors are rather difficult to separate unambiguously.
Nevertheless, an estimate of each error can be obtained using the package \texttt{gvar}~\cite{gvar:2022} following the methodology described in Ref.~\cite{Bouchard:2014ypa}.
The discretization error is defined to be the error coming from the parametric uncertainty in $\delta f_{\rm artifacts}^{(a^2+h^2)}$ from $c_{a^2}$ and $c_{h^2}$.
The combined uncertainty from all other fit parameters in \cref{eq:chipt_hard_su2} is defined to be the error in the fit function and chiral interpolation.
This error includes the uncertainty from the  $D D^*\pi$ coupling, $g$, which turns out to have a small influence on the final results.
The experimentally measured values of the meson masses also contribute negligibly to the total error.

Numerical results for the error breakdown are shown in \cref{table:final_error_budget} for $q^2=0$ and $q^2_{\rm max}$, and
\cref{fig:d2pi_final_error_budget,fig:d2k_final_error_budget,fig:ds2k_final_error_budget} show the error budgets through the full kinematic range for $\Dpi$, $\DK$, and $\DsK$, respectively, after fits to the $z$~expansion.
The colored curves sum in quadrature to give the total error in black.
Not shown are contributions from uncertainties less than $0.01\%$; this includes  the experimental values for the input meson masses. 
Since the lattice data span the full kinematic range in $q^2$, errors from the $z$~expansion are also negligible.

Several important qualitative features are evident in the error budgets.
For all three decays, $f_+$ has the largest errors near $q^2_{\rm max}$, since this kinematic region involves an extrapolation $(\bm{p}\to \bm{0})$.
Second, because the $z$-expansion analysis uses a correlated joint fit to $f_0$ and $f_+$, the final errors in each case include contributions from statistical uncertainties in both $f_0$ and $f_+$.
Third, because the $f_\perp$ term vanishes in \cref{eq:fplus_perp_0} at $q^2=0$, the contributions from statistical errors in $f_\perp$ decrease for small $q^2$.
Fourth, although the form factors are dimensionless, the scale-setting uncertainty is significant and tends to decrease for large $q^2$.
At the physical point, the scale-setting uncertainties vanish identically for the chiral logarithms and analytic terms.
The full uncertainty comes from the leading-order term in \cref{eq:chipt_hard_su2}: since the posterior values for $c_0$ and $\Delta_{xy,\rm P}$ are both implicitly in intermediate units of $w_0$, so must the energy be.
The associated scale-setting uncertainty thus decreases when the energy is small.

The error budgets for $\DK$ and $\DsK$ are qualitatively similar as shown in \cref{fig:d2k_final_error_budget,fig:ds2k_final_error_budget}. 
Over the whole kinematic range, statistics is the dominant source of error for all three channels, except for $f_0^{\DsK}$ near $q^2_{\rm max}$, where the scale-setting uncertainty dominates.

\begin{figure}
    \centering
    \includegraphics[width=1.0\textwidth]{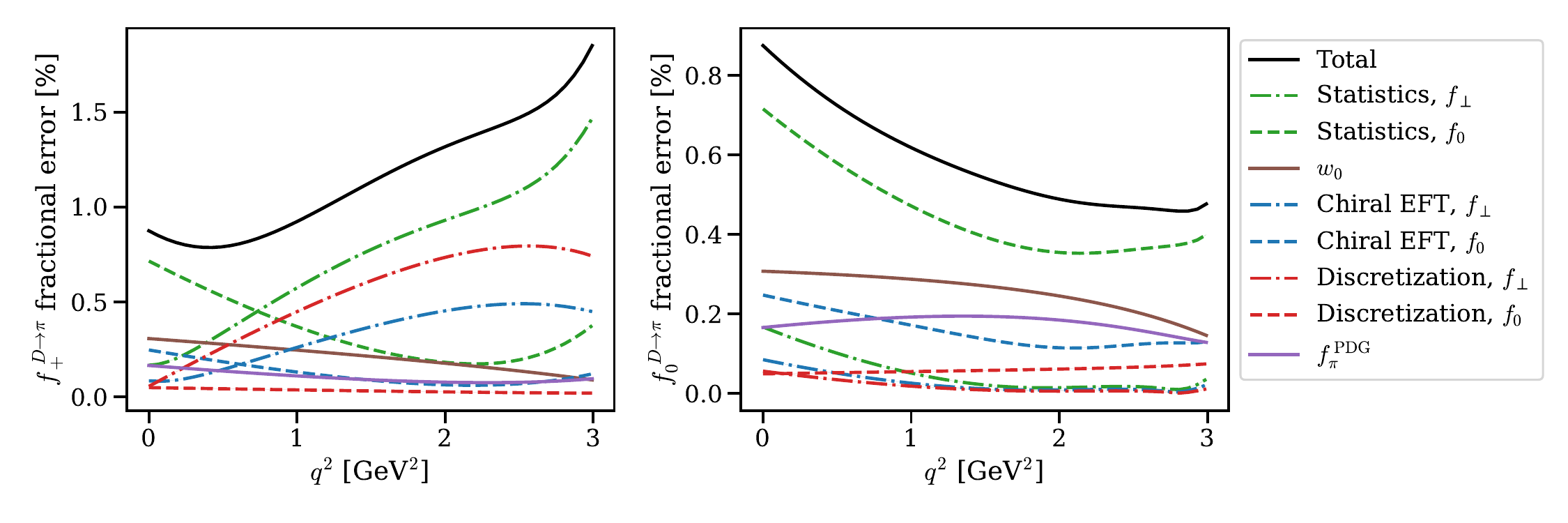}
\caption{
    Final error budget for the form factors $f^{\Dpi}_+$ and $f^{\Dpi}_0$ after the fit to the $z$~expansion.
    Contributions less than $0.01\%$ are not shown.
    \label{fig:d2pi_final_error_budget}
    }
\end{figure}

\subsection{Finite-volume corrections}
\label{ssec:finite_volume}

In principle, the finite volume of our simulations is a systematic effect influencing the results for the form factors.
Within chiral perturbation theory, the leading corrections amount to replacing loop integrals by discrete sums~\cite{Arndt:2004bg,Laiho:2005ue}.
The basic infinite-volume loop integral appearing in the present analysis is
\begin{equation}
    i \mu^\epsilon \int \frac{d^{4-\epsilon}q}{(2\pi)^{4-\epsilon}} \frac{1}{q^2-M^2} = \frac{1}{16\pi^2} I_1(M),
\end{equation}
with $I_1(M) = M^2 \ln(M^2/\Lambda^2)$ as in \cref{eq:chiral_logs}.
In a finite volume, this integral becomes the discrete sum
\begin{equation}
    \mathcal{I}_1(M) =
    \frac{1}{L^3} \sum_{\bm{q}} \int \frac{dq^0}{2\pi}\frac{1}{q^2 - m^2 + i \epsilon}
    \equiv I_1(m) + I_1^{\rm FV}(m),
\end{equation}
where $I_1^{\rm FV}(M)$ is the finite-volume correction that vanishes exponentially for large volumes.
The correction has the explicit form
\begin{equation}
    I_1^{\rm FV}(M) =
    \frac{1}{4\pi^2}M^2 \sum_{|\bm{n}|\neq 0} \frac{K_1(nML)}{nML},
\end{equation}
with the sum running over all nonzero lattice vectors $\bm{n}\in \mathbb{Z}^3$ in the finite volume, and where $K_1$ is a modified Bessel function of the second kind. 
As described in \cref{sec:chiral_ctm}, the effect of this correction has already been included explicitly in our fits to \cref{eq:chipt_hard_su2}.
To quantify the overall size of the finite-volume effect, it is useful to compute the dimensionless ratio:
\begin{equation}
    \frac{I_1^{\rm FV}(m)}{I_1(m)} = 
     \frac{4}{\ln\left(M^2/\Lambda^2\right)}
    \sum_{|\bm{n}|\neq 0} \frac{K_1(nML)}{nML}.
\end{equation}
As shown in \cref{table:finite_volume}, the finite-volume corrections amount to $\lesssim 2\%$ shifts in $I_1(m)$.
In the chiral-continuum fits to \cref{eq:chipt_hard_su2}, the overall contribution from the chiral logarithms enter at the level of a few percent.
The total size of finite-volume corrections to the form factors may be estimated to be at the few permyriad level, $\order{0.01}\%$.
Since the leading correction to the chiral logarithm has already been included in our fits to \cref{eq:chipt_hard_su2}, and since the effect is so small, we do not include any additional error for residual finite-volume effects in our final systematic error budget. 

\begin{table}
\caption{Finite-volume corrections to the chiral logarithm $I_1(M_\pi)$ for the ensembles given in \cref{table:ensembles}.}
\label{table:finite_volume}
\begin{tabular}{c c c c c}
\hline\hline
 $\approx a$ [fm] & $m_l/m_s$ & $L/a$ &     $M_\pi L$ &  $I_{\rm FV}(M_\pi)/I_1(M_\pi)$ [\%] \\
\hline
0.120 &        1/27 & 48 & 3.9 & 1.32 \\
0.088 &        1/10 & 48 & 4.7 & 0.65 \\
0.088 &        1/27 & 64 & 3.7 & 2.06 \\
0.057 &         1/5 & 48 & 4.5 & 1.31 \\
0.057 &        1/10 & 64 & 4.3 & 1.25 \\
0.057 &        1/27 & 96 & 3.7 & 1.91 \\
0.042 &         1/5 & 64 & 4.3 & 1.70 \\
\hline\hline
\end{tabular}
\end{table}

\subsection{Nonequilibrated topological charge}
\label{ssec:topology}

Efficiently sampling regions with different topological charges $Q$ in lattice-QCD simulations becomes slow in standard algorithms, which use a continuous updating procedure for the gauge fields.
Brower et al.~\cite{Brower:2003yx} realized that chiral perturbation theory can be used to study the $Q$-dependence of observables, and they showed how to extract physical results from numerical data at fixed topology.
Their calculations confirmed the theoretical expectation that, due to  locality and cluster decomposition, the effects from fixed topology should be suppressed for large volumes.
Subsequent calculations by Bernard and Toussaint~\cite{Bernard:2017npd} extended these ideas to heavy-light decay constants and meson masses in the context of heavy-meson chiral-perturbation theory.
The analysis was extended to light form factors in Ref.~\cite{FermilabLattice:2018zqv}. 

Following those works, we account for the effect of the difference between the correct $\avg{Q^2}$ and the simulation $\avg{Q^2}_{\rm sample}$ in the extraction of heavy-light form factors by applying a correction factor $\Delta_Q f_{\rm P}$, independent on $q^2$, valid for all form factors considered in this work, and given by
\begin{align}
    \Delta_Q f_{\rm P}
    &\equiv f_{\rm P, corrected} - f_{\rm P, sample},\\
    &= -\frac{1}{2 \chi_T V}
    \left.\frac{\partial^2f_{\rm P}}{\partial \theta^2}\right|_{\theta=0}
    \left( 1 - \frac{\avg{Q^2}_{\rm sample}}{\chi_T V} \right),
\end{align}
with $f_{\rm P,sample}$ the simulation value of a given form factor at any value of $q^2$, $\theta$ the vacuum angle, $\chi_T$ the topological susceptibility, and $V$ the four-dimensional lattice volume. The second derivative of the form factors with respect to the vacuum angle is obtained using LO heavy-light $\chi$PT with $\theta\ne 0$
 \begin{align}   
    \left.\frac{\partial^2f_{\rm P}}{\partial \theta^2}\right|_{\theta=0}
    &=
    -\frac{1}{4}\left(
    \frac{m_l m_s}{m_x (m_l + 2 m_s)}
    \right)^2
    \left.f_{\rm P}(\theta)\right|_{\theta=0}.
\end{align}
where $m_{l,s}$ are the light and strange quark masses respectively, and $m_x$ is the mass of the spectator quark in the transition, {\it i.e.}, $m_l$ for $D\to\pi(K)\ell\nu$, $m_s$ for $D_s\to K\ell\nu$.
The value of $\avg{Q^2}_{\rm sample}$ is understood to be the measured value from the simulation.
For the chiral susceptibility, we take the prediction from leading-order staggered $\chi$PT~\cite{Billeter:2004wx},
\begin{equation}
    \chi_T = \frac{1}{4}f_\pi^2 M^2
\end{equation}
where $1/M^2 = 2/M_{ll,I}^2 + 1/M_{ss,I}^2$ involves the taste-singlet non-Goldstone states.
At leading order, the only change from the familiar result~\cite{Leutwyler:1992yt} is the replacement of $M_\pi^2$ by $M^2$.
The masses of the taste-singlet mesons are calculated using \cref{eq:staggered_GMOR}.
Of the ensembles considered in this work (cf.\ \cref{table:ensembles}), 
the effects of nonequilibrated topological charge are relevant only for the finest ensemble ($a\approx 0.042 \fm$), for which $\avg{Q^2}_{\rm sample}=27.59$~\cite{Bernard:2017npd}.
The resulting corrections, $(\Delta_Q f_{\rm P})/f_{\rm P}\lesssim 0.0003$, are applied to the form factor data on the $a\approx 0.042\fm$ ensemble prior to the chiral-continuum fit in \cref{sec:chiral_ctm}.
Having accounted for the effect explicitly, and given the smallness of the correction, no further systematic error is assigned for nonequilibrated topological charge. 

\section{Phenomenology}
\label{sec:phenomenology}

The analysis of the preceding sections yields the semileptonic form factors for \Dpi, \DK, and \DsK\ in the idealized case of isospin-symmetric QCD.
For phenomenological applications, we have to consider the effects of strong isospin and QED, and then combine corrected results  with experimental data.
In this section, we first (\cref{ssec:expt}) explore several options for combining the experimental results with the lattice-QCD form factors and next (\cref{ssec:sib_qed}) estimate QED and strong isopin-breaking effects.
We are then in a position to determine via \cref{eq:dGammadq2} the CKM matrix elements \Vcd\ and \Vcs\ with a full error budget (\cref{ssec:CKM}) and to carry out tests of CKM unitarity (\cref{ssec:unitarity}).
We also compute the Standard Model predictions for the LFU ratios $R_{\mu/e}$.

\subsection{Experimental measurements}
\label{ssec:expt}

The differential decay rates $d\Gamma/dq^2$ for semileptonic decays of $D_{(s)}$~mesons to pseudoscalar light mesons have been measured by FOCUS (shape only)~\cite{FOCUS:2004meh}, Belle~\cite{Belle:2006idb}, BaBar~\cite{BaBar:2007zgf, BaBar:2014xzf}, CLEO~\cite{CLEO:2009svp}, and BES~III~\cite{BESIII:2015tql, BESIII:2016gbw, BESIII:2017ylw, BESIII:2018xre, BESIII:2018nzb}.
\Cref{table:expt_summary} summarizes the published measurements according to decay channel. 
\begin{table}
\caption{Summary of published measurements of semileptonic decays of $D$~mesons to pseudoscalar light mesons.
FOCUS 2005~\cite{FOCUS:2004meh} obtained shape information only and is omitted.
\label{table:expt_summary}}
\begin{tabular}{c  c  c}
\hline \hline
Decay                           & Measurements  &   Notes\\
\hline
$D^0 \to \pi^- e^+ \nu$       & BaBar 2015~\cite{BaBar:2014xzf}&\\
                                & Belle 2006~\cite{Belle:2006idb}& $e^+$ and $\mu^+$ averaged\\
                                & BES~III 2015~\cite{BESIII:2015tql}&\\
                                & CLEO 2009~\cite{CLEO:2009svp}&\\
$D^+ \to \pi^0 e^+ \nu$       & BES~III 2017~\cite{BESIII:2017ylw}&\\
                                & CLEO 2009~\cite{CLEO:2009svp}&\\
\hline
$D^0 \to \pi^- \mu^+ \nu$   & Belle 2006~\cite{Belle:2006idb}& $e^+$ and $\mu^+$ averaged\\
                                & BES~III 2018~\cite{BESIII:2018nzb}&\\
$D^+ \to \pi^0 \mu^+ \nu$   & BES~III 2018~\cite{BESIII:2018nzb}&\\
\hline
$D^0 \to K^-e^+\nu$           & BaBar 2007~\cite{BaBar:2007zgf}&\\
                                & Belle 2006~\cite{Belle:2006idb}&$e^+$ and $\mu^+$ averaged\\
                                & BES~III 2015~\cite{BESIII:2015tql}&\\
                                & CLEO 2009~\cite{CLEO:2009svp}&\\
$D^+ \to \bar{K}^0 e^+ \nu$   & BES~III 2017~\cite{BESIII:2017ylw}& \\
                                & CLEO 2009~\cite{CLEO:2009svp}&\\
\hline                                
$D^0 \to K^-\mu^+\nu$         & Belle 2006~\cite{Belle:2006idb}&$e^+$ and $\mu^+$ averaged\\
                                & BES~III 2019~\cite{BESIII:2018ccy}&\\
$D^+ \to \bar{K}^0 \mu^+ \nu$ & BES~III 2016~\cite{BESIII:2016gbw}& total rate only\\
\hline
$D_s^+ \to K^0 e^+ \nu$       & BES~III 2019~\cite{BESIII:2018xre}& \\
\hline \hline
\end{tabular}
\end{table}
Due to the experimental challenge of reconstructing muons in the final state, more measurements exist for the electron channels.
The only published data available for the semimuonic final states are from BES~III, which measured the rates for $\Dpi\mu\nu$~\cite{BESIII:2018nzb} and $\DK\mu\nu$~\cite{BESIII:2016gbw,BESIII:2018ccy}. 
Although Belle measured both the semielectronic and semimuonic final states~\cite{Belle:2006idb}, numerical values for the rate were not reported; instead, only values for the product $\Vcx f_+^{D\to\pi/K}(q^2)$ averaged over the lepton final state are available~\cite{Rong:2014hea, Fang:2014sqa}, without any correlation information.
Since both experimental data and lattice-QCD form factors have now reached a level of precision where the effects of the scalar form factor (which are proportional to $m_\ell^2$) are no longer negligible, as discussed below, we exclude Belle data from our subsequent analysis.

Besides the experimental difficulties associated with semimuonic final states, the extraction of the CKM matrix elements \Vcd\ and \Vcs\ poses an additional complication.
Contributions to the differential decay rate from the scalar form factor enter \cref{eq:dGammadq2} with a factor of $m_\ell^2$.
As \cref{fig:scalar_vector_breakdown} shows, the scalar form factor is negligible for semielectronic final states everywhere except the lowest $q^2$ bin, where its contribution is roughly $1\%$.
The situation for semimuonic final states is entirely different, where contributions from $f_0$ are roughly $(m_\mu/m_e)^2 \approx 10^5$ times larger and, thus, contribute at the few-percent level throughout the full kinematic range. 
Many extractions of $\Vcd$ and $\Vcs$ have neglected the contributions of the scalar form factor. But, with errors of $\lesssim 1\%$ both from experiment and from the results of this paper, determinations of $\Vcd$ and $\Vcs$ require the inclusion of both terms in \cref{eq:dGammadq2}.

\begin{figure}
    \centering
    \includegraphics[width=1.0\textwidth]{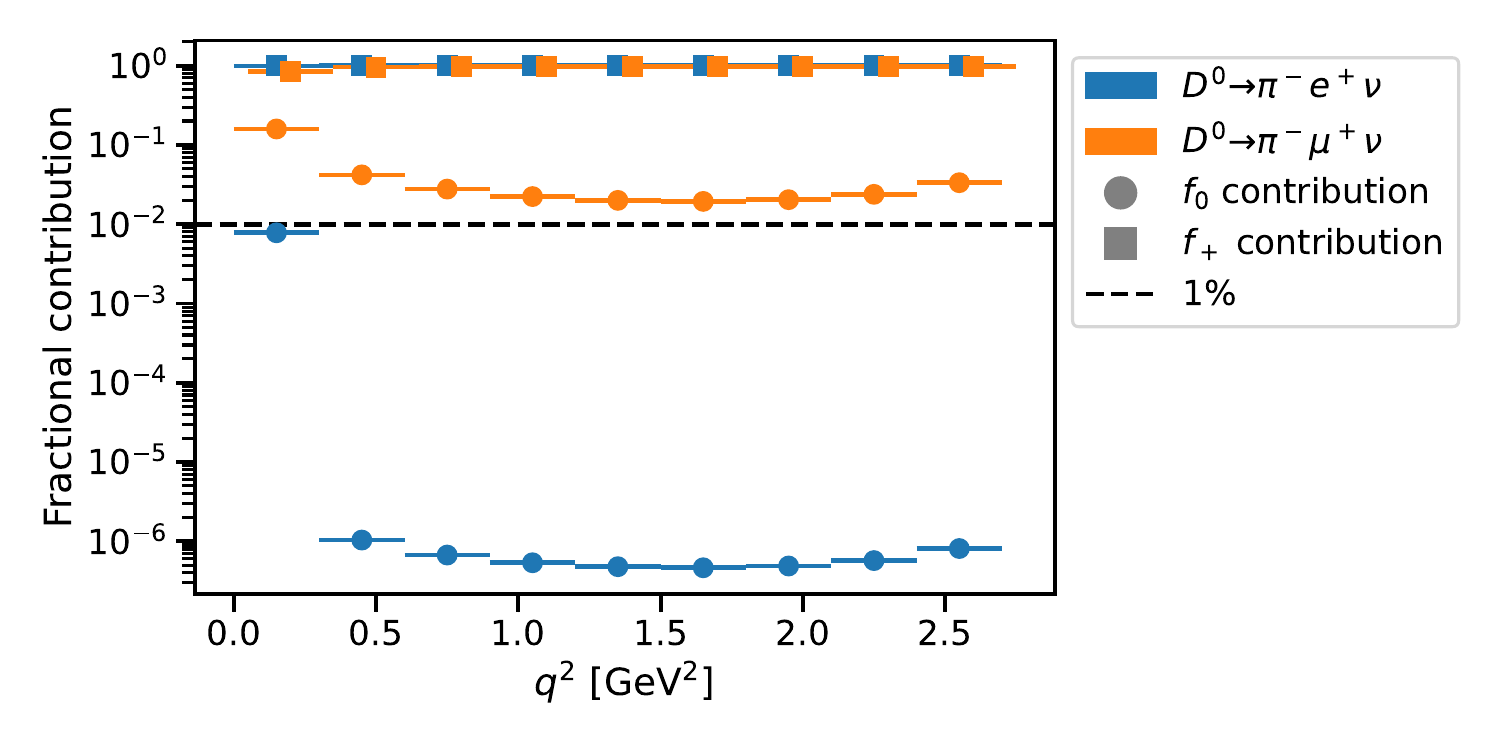}
    \caption{
    Fractional contributions from the scalar and vector form factors to the differential decay rate $d\Gamma/dq^2$ for $D^0\to\pi^-\ell\nu$ for electronic and muonic final states.
    Scalar contributions enter \cref{eq:dGammadq2} with a factor of $m_\ell^2$.
    In muonic decays, in general, scalar contributions are thus a factor of $(m_\mu/m_e)^2 \approx 10^5$ larger than in the corresponding electronic decays. 
    \label{fig:scalar_vector_breakdown}
    }
\end{figure}

\subsection{Systematic uncertainty from strong isospin effects and QED}
\label{ssec:sib_qed}

The form-factor results reported in \cref{table:z_results} are computed in isospin-symmetric QCD, {\it i.e.}, in simulations with degenerate light quarks of mass $m_l = (m_u + m_d)/2$ in the sea and valence sectors.
This theory is slightly different from nature, which includes corrections from electromagnetic effects and strong isospin breaking (SIB).
An estimate of these neglected effects is necessary before combination with experimental data.

Consider first SIB.
Isospin violation in the sea may be ignored at the current level of precision.
Because the matrix elements yielding the form factors are symmetric under exchange of the up and down sea quarks ($m_u \leftrightarrow m_d$), the leading contributions from SIB in the sea are of order $(m_d-m_u)^2$.
This behavior appears in the $\chi$PT~\cite{Aubin:2007mc} expressions, showing that sea SIB is smaller than the NNLO terms in the chiral expansion~\cite{Bazavov:2017lyh}.
To estimate the valence correction, we evaluate the form factors with a different definition of the physical point, replacing the masses of the neutral initial and final hadrons that define the physical point, see \cref{table:physical_point_inputs}, with their charged counterparts, and then computing the fractional shift 
$1 - (f_{+,0}^{\rm neutral}/f_{+,0}^{\rm charged})$
as a function of $q^2$.
To account for this systematic effect, we increase our errors on the form factors by 
$\pm(1 - f_{+,0}^{\rm neutral}/f_{+,0}^{\rm charged})$,
leaving the central value unchanged.
The systematic error profiles are shown as functions of $q^2$ in \cref{fig:sib}. 
Although this treatment of SIB does not distinguish between SIB in the sea and valence sectors, it is conservative insofar as both sea and valence effects contribute the variation with the hadron masses.
Guidance from EFT calculations or dedicated simulations with $m_u \neq m_d$ would be useful to help quantify this effect more precisely. 
Due to phase-space suppression at large $q^2$, isospin effects will turn out to be a small (and sometimes neglible) contribution to the systematic error budgets for quantities of phenomenological interest.

\begin{figure}
    \centering
    \includegraphics[width=1.0\textwidth]{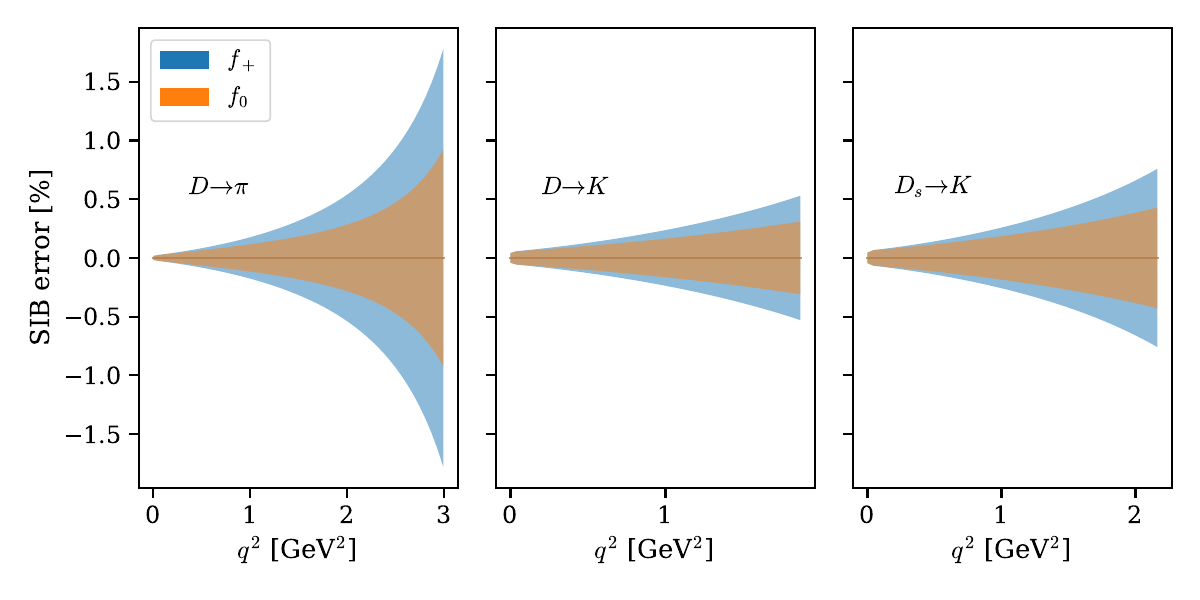}
    \caption{
    Systematic error envelopes $1 - f_{+,0}^{\rm neutral}/f_{+,0}^{\rm charged}$ estimating the effects of isospin breaking from the variation in form factors using an alternative definition of the physical point with the charged initial and final hadron masses (cf.\ \cref{table:physical_point_inputs}). 
    The total errors on the theoretical prediction for the form factor are increased, leaving the central value unchanged.
    Although the systematic uncertainty from SIB increases with $q^2$, its effect on $\Vcx$ and $R_{\mu/e}$ ends up being small due to phase-space suppression.
    \label{fig:sib}
    }
\end{figure}

Some effects of QED are taken into account in the experimental measurements.
For instance, final-state radiation tends to degrade the momentum resolution, which can lead to mis-measurement of the positron momentum if background radiative events (e.g., $D^0\to\pi^-e^+\nu \gamma$) are not handled correctly.
Experimental groups correct for this effect using the Monte Carlo tool PHOTOS~\cite{Golonka:2005pn,Barberio:1993qi}. See Refs.~\cite{BaBar:2007zgf,CLEO:2009svp} for a discussion.

The long-distance electromagnetic corrections to the semileptonic decays themselves ($\delta_{\rm EM}$ in \cref{eq:dGammadq2}) have not been calculated for the decays $D_{(s)}\to K/\pi\ell\nu$.
However, the analogous corrections to the decay amplitudes for $K\to\pi\ell\nu$ have been computed in the framework of $\chi$PT~\cite{Cirigliano:2008wn,Cirigliano:2011ny} and more recently in a hybrid framework combining $\chi$PT and Sirlin's representation of SM radiative corrections~\cite{Seng:2021boy,Seng:2021wcf,Seng:2022wcw}.
The more recent calculations confirm the older results but with smaller final uncertainties.
The overall picture, substantiated by \cref{table:delta_EM}, is that final states with a charged hadron (e.g., $\pi^-e^+$) tend to have shifts of $\delta_{\rm EM} \approx 1$--$1.5\%$, while the shifts for final states with a neutral hadron (e.g., $\pi^0 e^+$) are roughly a factor of 3--4 smaller. 
Differences between the decays with an $e^+$ or a $\mu^+$ in the final state are around an order of magnitude smaller. Since, as mentioned above, no similar calculations exist for the decays at hand, we are unable to apply a concrete correction $\delta_{\rm EM}$.
Instead, using the results for $K\to\pi\ell\nu$ as a rough guide, we include an additional systematic uncertainty.
For extractions of the CKM matrix elements, we add a conservative error of $\pm 1\%$ to the final value $\Vcd$ or $\Vcs$.
In all cases, the uncertainty is inflated without shifting the central values.

\begin{table}
\caption{
Long-distance electromagnetic corrections for the $K_{\ell 3}$ decay amplitude, taken from Ref.~\cite{Cirigliano:2008wn,Cirigliano:2011ny,Seng:2021boy,Seng:2021wcf,Seng:2022wcw}.
Since the shifts are computed for the amplitude, the factor of two is necessary for use with the decay rate.
Entries correspond to $\frac{1}{2}\delta_{\rm EM}$ in~\%.
\label{table:delta_EM}
}
\begin{tabular}{l c c}
\hline\hline
Decay & Cirigliano et al.~\cite{Cirigliano:2008wn,Cirigliano:2011ny} & Seng et al.~\cite{Seng:2021boy,Seng:2021wcf,Seng:2022wcw} \\
\hline
$K^0 \to \pi^- e^+ \nu$       &   $0.50 \pm 0.11$     &   $0.580 \pm 0.016$\\
$K^0 \to \pi^- \mu^+ \nu$   &   $0.70 \pm 0.11$     &   $0.77 \pm 0.04$\\
\hline
$K^+ \to \pi^0 e^+ \nu$       &   $0.05 \pm 0.13$     &   $0.105 \pm 0.024$\\
$K^+ \to \pi^0 \mu^+ \nu$   &   $0.08 \pm 0.13$     &   $0.25 \pm 0.05$\\
\hline\hline
\end{tabular}
\end{table}

In the analysis below, we also report values for the correlated ratio $\Vcd/\Vcs$ as well as LFU ratios.
A few additional remarks are necessary concerning the QED uncertainty for these quantities.

Consider first the ratio $\Vcd/\Vcs$.
As the results in \cref{table:delta_EM} show, QED corrections for $K^+$ decays are $\lesssim 0.25\%$, which suggests similarly small corrections for $D^+$ decays.
For the decays of $D^0$, the QED corrections will be dominated by the Coulomb interaction between the charged final-state particles.
The Coulomb shift in the rate is approximately given by $1 + \pi \alpha/\beta$, where $\beta = \sqrt{1 - M_L^2 m_\ell^2 / (p_L \cdot p_\ell)^2}$ is relative velocity between the charged final-state particles~\cite{Ginsberg:1968pz,Atwood:1989em,Cirigliano:2008wn,Cali:2019nwp,deBoer:2018ipi}.
For the decays considered here, the kinematics are such that $\beta\approx 1$.
Therefore, within the uncertainties of our calculation, the Coulomb corrections for $D^0$ decays are essentially constant over the kinematic range of the decays and would cancel in the ratio. 
Overall, we take a conservative $0.5\%$ QED systematic uncertainty for the ratio $\Vcd/\Vcs$.

Similar considerations apply for the LFU ratios.
Again using $K\to\pi\ell\nu$ and \cref{table:delta_EM} for guidance, the correlated~\cite{Seng:2022wcw} differences
$\delta_{\rm EM}(\mu)-\delta_{\rm EM}(e)$ are about 0.3--0.4\%.
As for the ratio of CKM matrix elements, the Coulomb corrections are expected to introduce a factor, $(1+\alpha\pi)$, which cancels, within the precision of our calculation, in the ratio.
For the same reasons as above, we thus take a conservative $0.5\%$ QED systematic uncertainty for the LFU ratios.

\subsection{CKM matrix elements}
\label{ssec:CKM}

Our analysis extracts the CKM matrix elements using two different methods: the joint $z$-expansion method and the binned method, discussed below.


First, the \emph{joint $z$-expansion method} fits experimental data for $d \Gamma/dq^2$ together with synthetic data for our lattice-QCD form factors $f_+(q^2)$ and $f_0(q^2)$.
More precisely, the expected model for the decay rate is given by \cref{eq:dGammadq2} using the four-parameter $z$~expansions for both $f_+(q^2)$ and $f_0(q^2)$ via \cref{eq:z_fplus} and \cref{eq:z_f0}.
The CKM matrix element $\Vcx_{\rm joint}$ is treated as a free parameter in the fit which serves as a floating relative normalization factor between the experimental data for the rate and synthetic data for the form factors.
The synthetic data are computed using our results for $f_+(q^2)$ and $f_0(q^2)$ given in \cref{table:z_results}, evaluated for $q^2$ at $[0.1, 0.3, 0.63, 0.9]\times q^2_{\rm max}$. 
The locations of these points are the same as the synthetic points used in \cref{ssec:z-expansion}.
Since this method works directly with the full expression for the differential decay rate, \cref{eq:dGammadq2}, it makes no assumptions about the relative size of the vector and scalar contributions.

\begin{figure}
    \centering
    \includegraphics[width=0.49\textwidth]{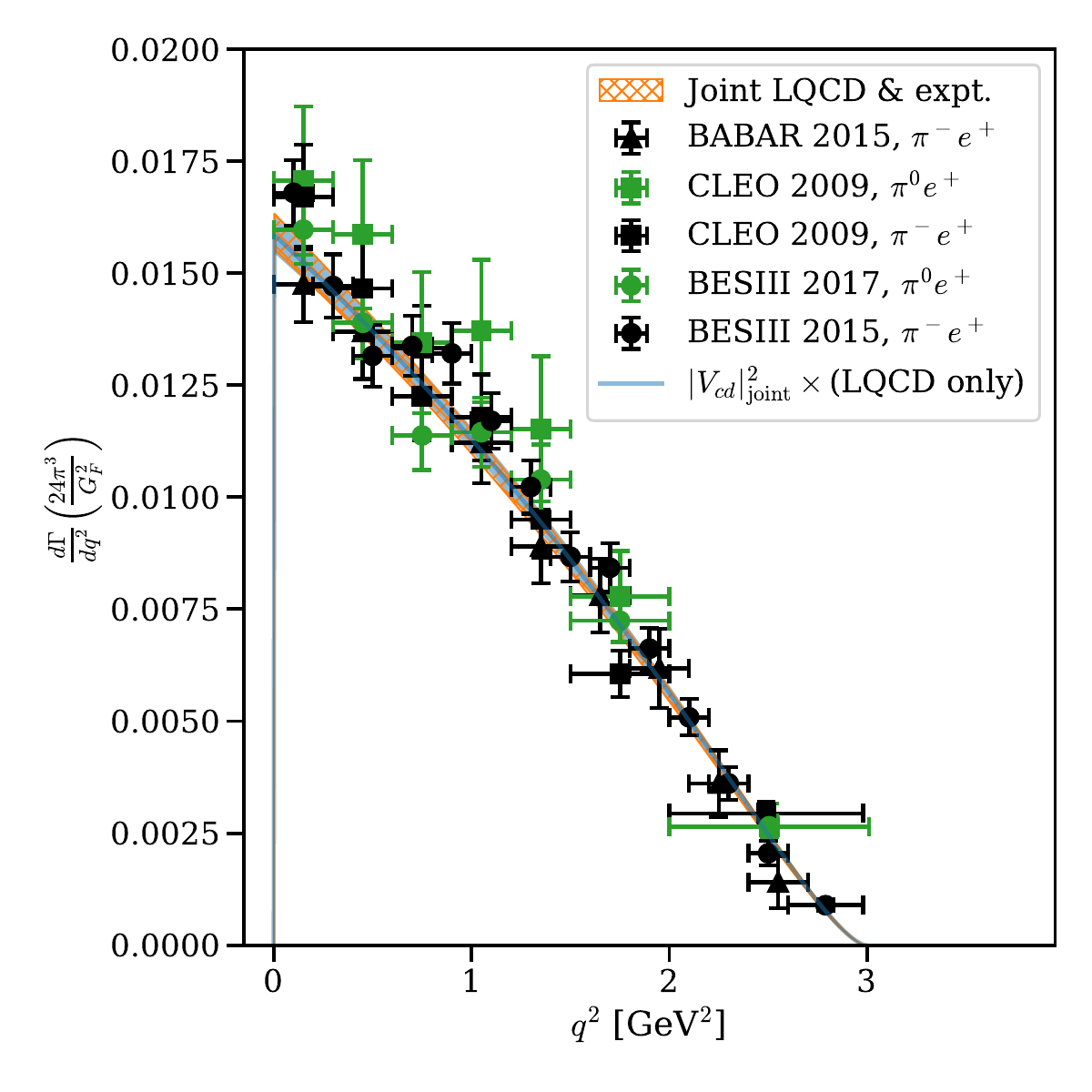}
    \includegraphics[width=0.49\textwidth]{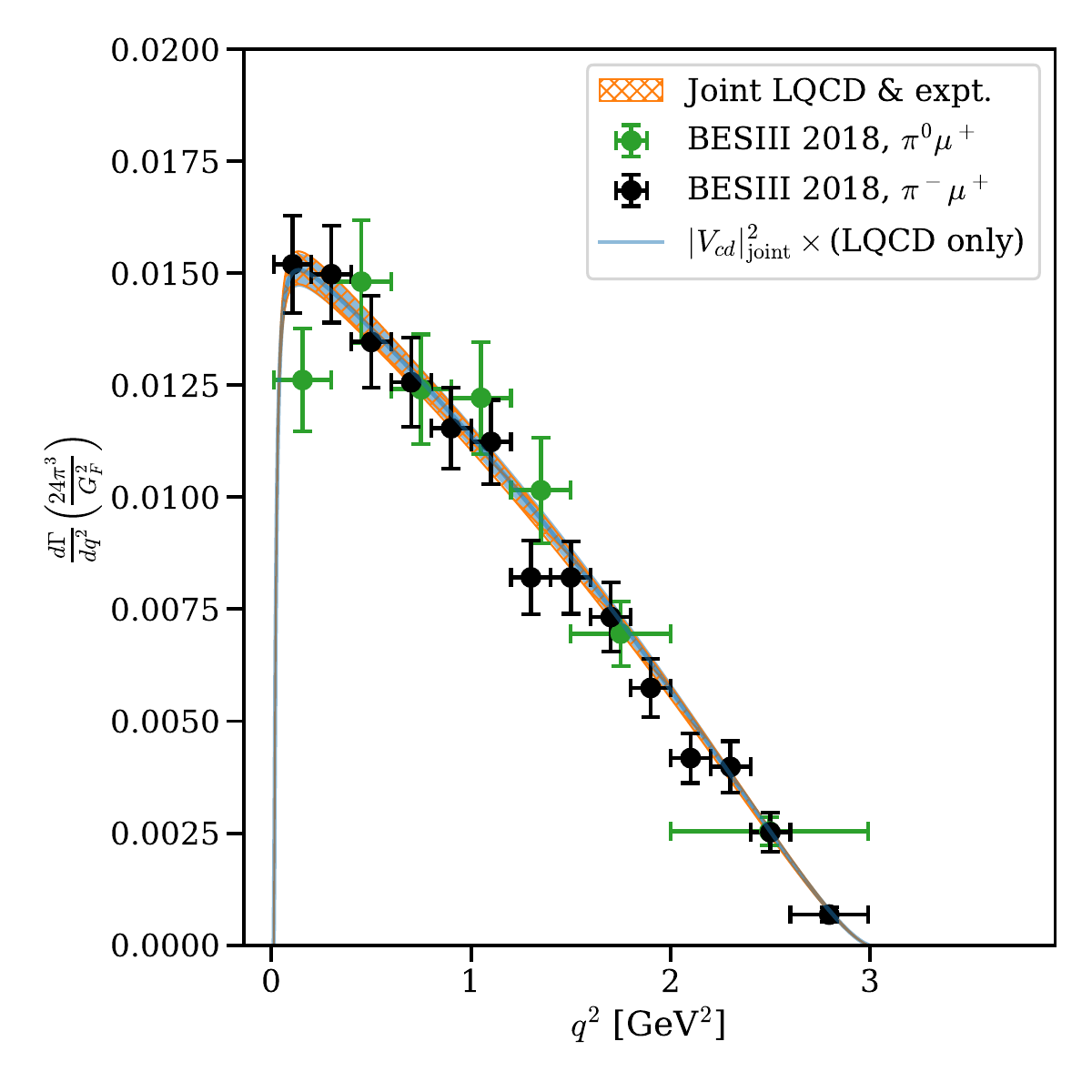}
    \includegraphics[width=0.49\textwidth]{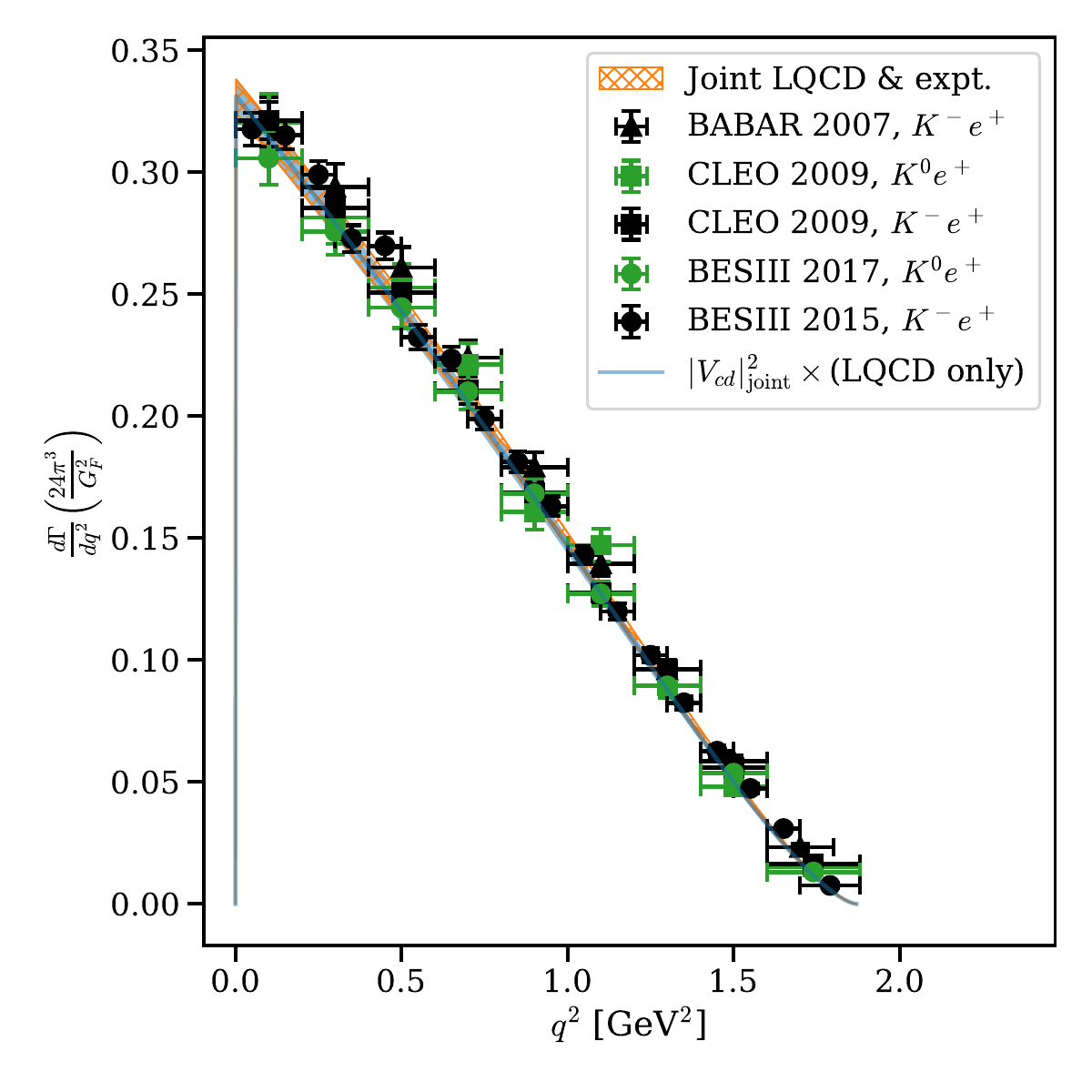}
    \includegraphics[width=0.49\textwidth]{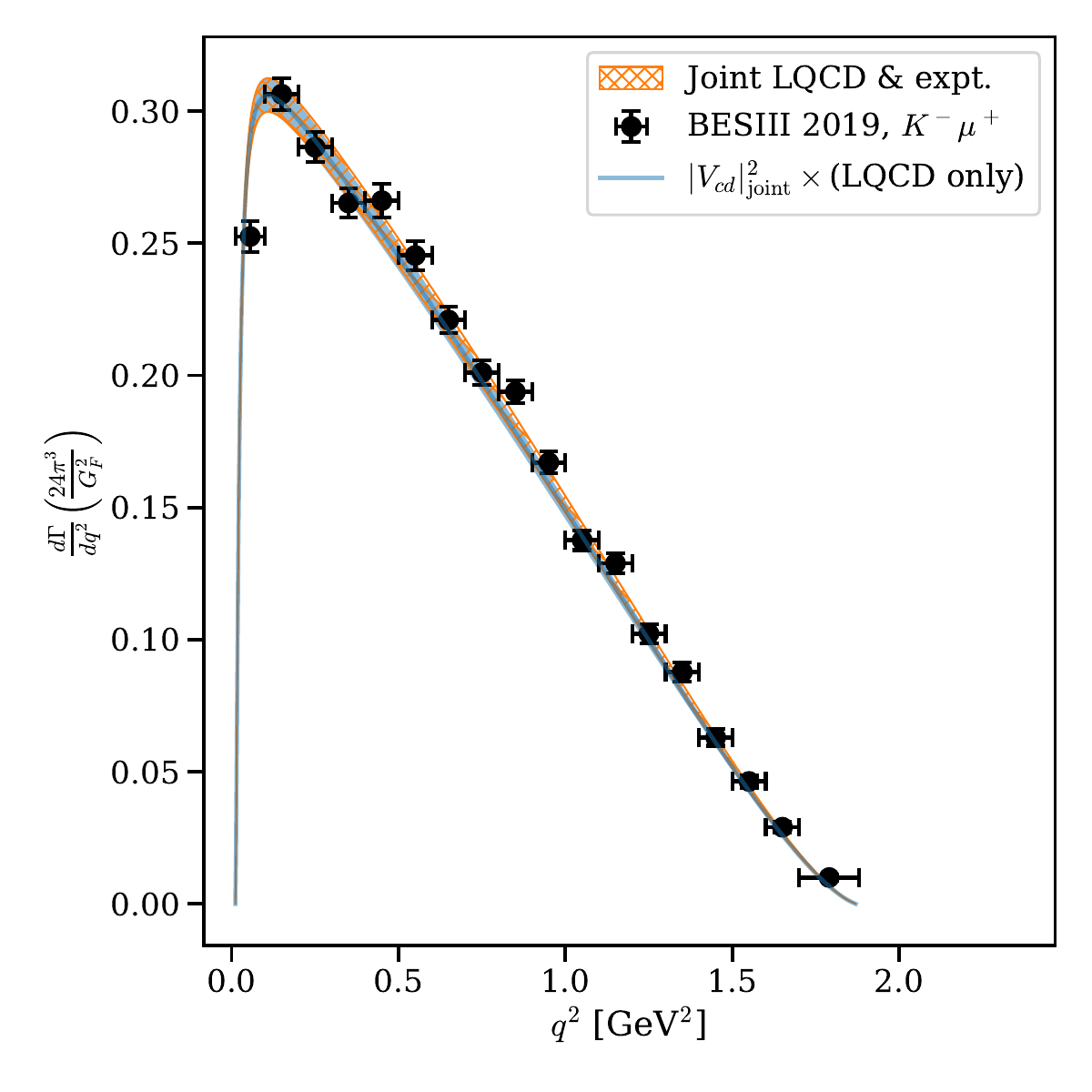}
    \caption{
    The differential decay rates for $\Dpi$ (top row) and $\DK$ (bottom row) in the semielectronic (left) and semimuonic (right) channels.
    The blue curves shows the result of evaluating \cref{eq:dGammadq2} using our lattice-QCD form factors, normalized by $\Vcx^2_{\rm joint}$.
    The hatched orange curves show the result of the joint fit of experimental data and synthetic lattice-QCD data to the $z$~expansion.
    The black data points indicate charged-hadron ($K^-/\pi^- \ell^+$) final states, while the green points indicate  experimental measurements for neutral-hadron ($K^0/\pi^0 \ell^+$) final states.
    In the top row, $\Dpi$ results come from BaBar~\cite{BaBar:2014xzf}, CLEO~\cite{CLEO:2009svp}, and BES~III~\cite{BESIII:2015tql,BESIII:2017ylw,BESIII:2018nzb}.
    In the bottom row, results come from BaBar~\cite{BaBar:2007zgf}, CLEO~\cite{CLEO:2009svp}, and BES~III~\cite{BESIII:2015tql,BESIII:2017ylw,BESIII:2018ccy}.
    Results from different experiments have are distinguished by different markers.
    The points have been slightly offset horizontally for readability.
    \label{fig:d2x_rate_with_data}
    }
\end{figure}

\begin{figure}
    \centering
    \includegraphics[width=0.49\textwidth]{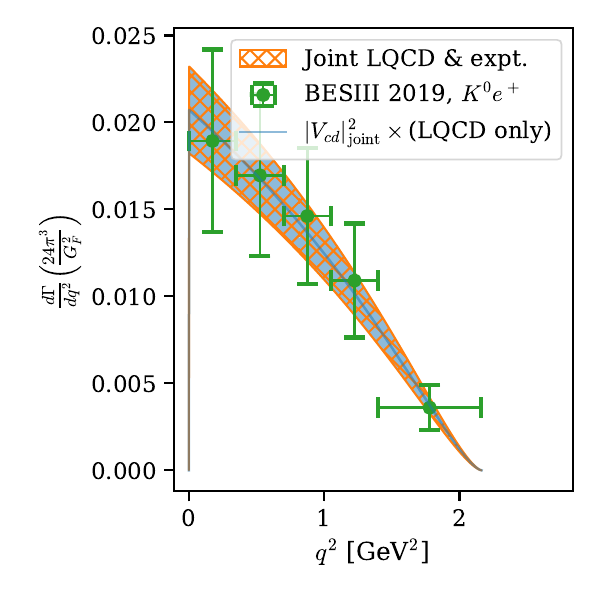}
    \includegraphics[width=0.49\textwidth]{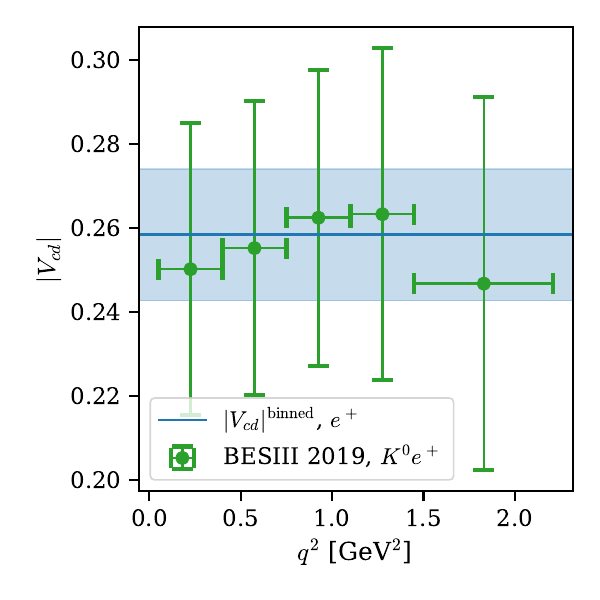}
    \caption{
    \textbf{(Left)} The differential decay rate for $D_s\to K^0 e^+\nu$.
    The blue curve shows the result of evaluating \cref{eq:dGammadq2} using our lattice-QCD form factors, normalized by $\Vcd^2_{\rm joint}$.
    The hatched orange curve show the result of the joint fit of experimental data and synthetic lattice-QCD data to the $z$~expansion.
    The data points indicate experimental measurements from BES~III~\cite{BESIII:2018xre}.
    \textbf{(Right)} The binwise estimate of the CKM matrix element
    $\left[|V_{cd}(q_i^2)|\right]_{\rm Binned}$ from the decay $\DsK$.
    \label{fig:ds2k_rate_with_data}
    }
\end{figure}

Joint $z$-expansion fits have been carried out including all experimental data for each decay process. 
The corresponding differential decay rates (orange curves) are shown, together with the experimental data, in \cref{fig:d2x_rate_with_data} ($\Dpi$ and $\DK$) and \cref{fig:ds2k_rate_with_data} ($\DsK$). 
For completeness, the fit posteriors for the $z$-expansion coefficients are given in \cref{sec:joint_fit_results}.
Measurements from CLEO were reported including correlations between different decay channels~\cite{CLEO:2009svp};
these correlations are included in our analysis.
The published results from BES~III do not include correlations between different decays
(e.g., $D^0\to\pi^- e^+\nu$, $D^+\to\pi^0 e^+\nu$, $D^0\to\pi^- \mu^+\nu$, and $D^+\to\pi^0 \mu^+\nu$).
We have experimented with different models for the missing off-diagonal blocks of the full correlation matrix, ranging from zero correlation to 100\% correlation.
Our results for $\Vcd$ and $\Vcs$ are extremely insensitive to the precise treatment of these off-diagonal correlations and give statistically indistinguishable results.
We therefore report values from our preferred analysis, which uses a simple model for the correlations in which the off-diagonal blocks are taken to be constant, with correlation coefficient equal to the mean of the corresponding diagonal blocks.\footnote{%
We thank the BES~III collaboration for providing us with the correlations for the differential rate $d\Gamma/dq^2$ for the decays $D_s^+\to K^0 e^+\nu$, $D^0\to\pi^-\mu^+\nu$, and $D^+\to\pi^0\mu^+\nu$ as well as for information and guidance regarding the treatment of off-diagonal correlations between different decays (Lei Li, private communication, 22 July 2022; Hailong Ma, private communication, 11 Dec 2022)}.
Regarding the measurements of $D^0\to K^-e^+\nu$ coming from BaBar~\cite{BaBar:2007zgf}, our fits drop the largest $q^2$ bin, since it is constructed by a normalization constraint (one minus the sum of the other bins).
Especially when fitting $d\Gamma/dq^2$ for semimuonic channels, including the scalar form factor is essential to achieving a good description of the data at low $q^2$. 
However, the higher parameters $b_1$ and $b_2$ associated with the scalar form factor are  constrained entirely by our precise synthetic data for $f_0$. 
The influence of neglecting $f_0$ is considered below.

The plots in \cref{fig:d2x_rate_with_data,fig:ds2k_rate_with_data} also include comparisons with the shape obtained from our form factors, given by the parameters in \cref{table:z_results} (lattice-QCD-only results), normalized by $\Vcx^2_{\rm joint}$. Those correspond to the blue curves in the plots. 
In addition, fits are also conducted to the experimental data alone, although we do not use these results to extract the CKM matrix elements. All those $z$-expansion fits, as well as the joint fit, enforce the kinematic identity $f_+(0)=f_0(0)$ by imposing $a_0=b_0$ [cf.\ \cref{eq:z_f0} and \cref{eq:z_fplus}].
The best-fist posterior values for those $z$-expansion fits are also given in \cref{sec:joint_fit_results}.

Another visualization of the form factors' shapes, which is independent of the overall normalization, comes from comparing the ratios $r_1\equiv a_1/a_0$ and $r_2\equiv a_2/a_0$ of the $z$-expansion coefficients from \cref{eq:zexpansion_expt} after applying the refitting procedure of Ref.~\cite{Chakraborty:2021qav}. 
These ratios are displayed for $\Dpi$ and $\DK$ in \cref{fig:shape_comparison}.
Our results, given by the black ellipses, show good agreement with the experimental shapes.
For \DK, we also find good agreement with the lattice QCD calculation from HPQCD~\cite{Chakraborty:2021qav}.
For \Dpi, we find $r_1=-2.009(55)$ and $r_2=0.14(36)$, with a correlation of $\rho_{12}=-0.58$.
For \DK, we find $r_1=-2.05(11)$ and $-0.59(52)$, with a correlation of $\rho_{12}=-0.29$.

\begin{figure}
    \centering
    \includegraphics[width=0.49\textwidth]{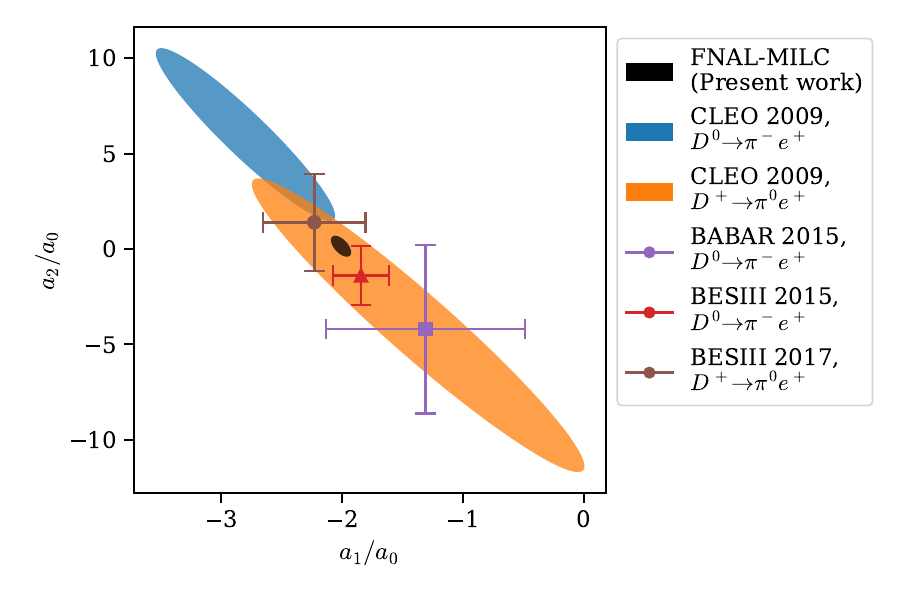}
    \includegraphics[width=0.49\textwidth]{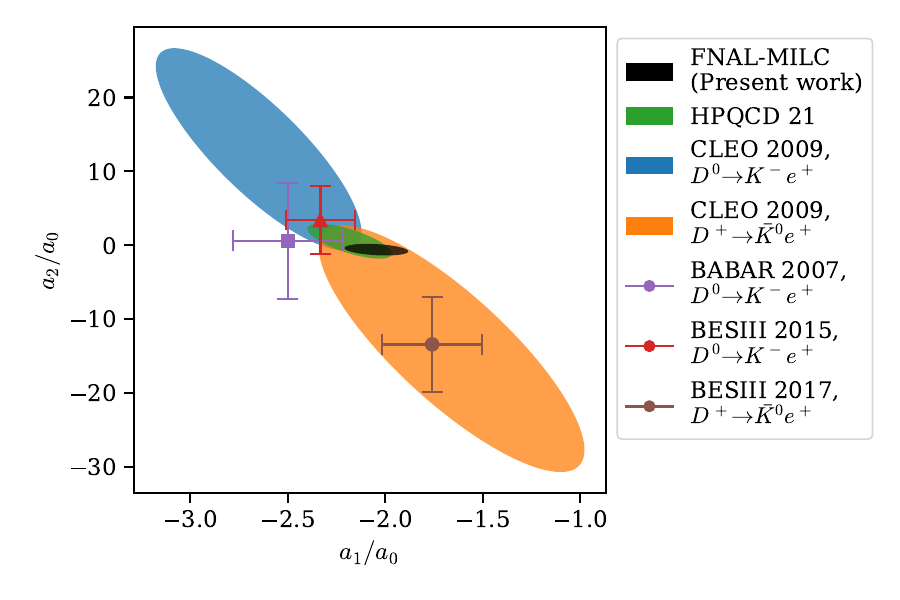}
    \caption{
    Comparing the shapes of the vector form factor $f_+$ between lattice QCD and experiment for the decays $\Dpi$ (left) and $\DK$ (right) via ratios of $z$-expansion coefficients from \cref{eq:zexpansion_expt}.
    Where published correlations are available, the ellipses show the $68\%$ confidence intervals.
    Systematic errors from QED and isospin breaking are not included in the lattice QCD results. 
    \label{fig:shape_comparison}
    }
\end{figure}


The second method we use to extract CKM matrix elements, the \emph{binned method}, combines lattice-QCD results with experimental data for the rate $d\Gamma/dq^2$ to give a binwise estimate of the CKM matrix element:
\begin{align}
    \left[|V_{cx}(q_i^2)|\right]_{\rm Binned}
    \equiv 
    \left[
    \left( \frac{d\Gamma}{dq^2} \right)_{\rm Expt}
    \frac{24 \pi^3}{G_F^2 \eta_{\rm EW}^2}
    \frac{1}{\avg{
    (\cdots)_{\rm LQCD}
    }}_{q_i^2}
    \right]^{1/2},
    \label{eq:binwise_Vcx}
\end{align}
where the quantity in the denominator is understood to be the binwise average (i.e., integrated over the bin) of the lattice-QCD form factors together with the appropriate kinematic factors appearing in \cref{eq:dGammadq2},
\begin{align}
(\cdots)_{\rm LQCD} \equiv    
(1 - \epsilon)^2 \left(
    \abs{\bm{p}}^3 (1 + \epsilon/2) f_+(q^2)
    +
    \abs{\bm{p}} M_H^2 \left(1 - \frac{M_L^2}{M_H^2}\right)^2 \frac{3}{8}\epsilon f_0(q^2)
\right).
\end{align}
This expression depends on the lepton mass via $\epsilon\equiv m_\ell^2/q^2$ as well as the experimentally measured hadron masses $M_H$ and $M_L$ for each mode (e.g., $D^0$ and $\pi^-$ or $D^+$ and $\pi^0$ for $\Dpi$).
A weighted, correlated average ({\it i.e.}, a fit to a constant) then gives $|V_{cx}|_{\rm Binwise}$. 
The binned method is entirely general and makes no assumptions about the relative size of the vector and scalar contributions.
Results for $\Vcd$ from $\Dpi$ and $\Vcs$ from $\DK$ for each $q^2$ bin and experiment, as well as the correlated average over bins including only semielectronic (blue lines) or only semimuonic data (red lines), are shown in \cref{fig:binwise_Vcx}.
The semimuonic results lie roughly $1\sigma$ below the semielectronic results for both $\Vcd$ and $\Vcs$, so below we report the values in each channel as well as the combined results. Those combined extractions, including all leptonic channels, lie between the two bands in \cref{fig:binwise_Vcx}, and are statistically consistent with the individual determinations, as shown in \cref{fig:VcdVcs}.
For $\Vcd$ from $\DsK$, results are shown in \cref{fig:ds2k_rate_with_data}. 
As argued above, with present statistical precision, the presence of the scalar form factor is quantitatively important for the differential rate $d\Gamma/dq^2$, especially for semimuonic channels.
\Cref{fig:Vcs_drop_f0} shows the effect of dropping the contribution from $f_0$ for $D\to K\mu\nu$.
Values for $\Vcs^{\rm binned}$ are observed to shift by a few percent and, when considered as a function of $q^2$, become statistically inconsistent with a constant.
Similar few-percent shifts occur for $D\to\pi\mu\nu$.

\begin{figure}
    \centering
    \includegraphics[width=0.49\textwidth]{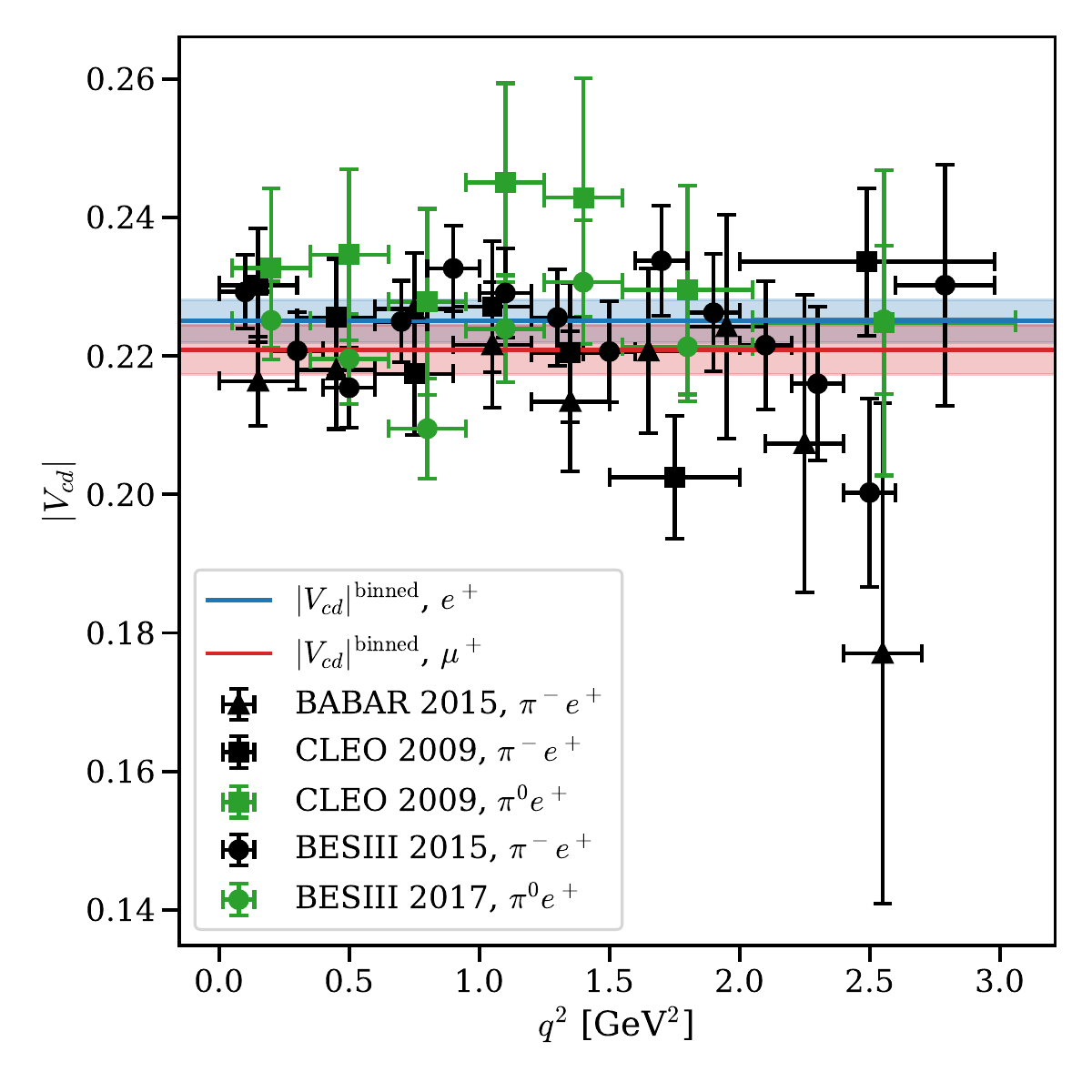}
    \includegraphics[width=0.49\textwidth]{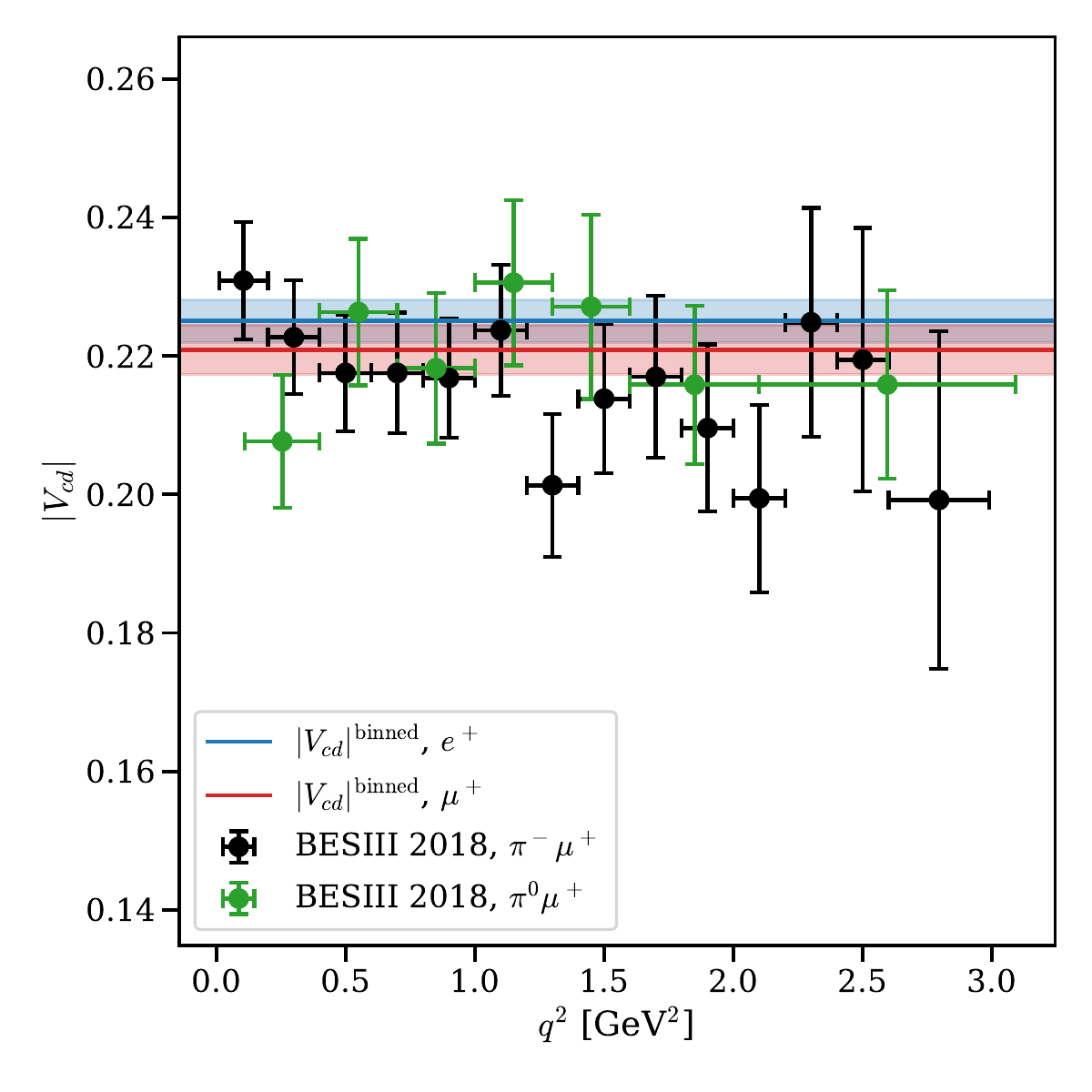}
    \includegraphics[width=0.49\textwidth]{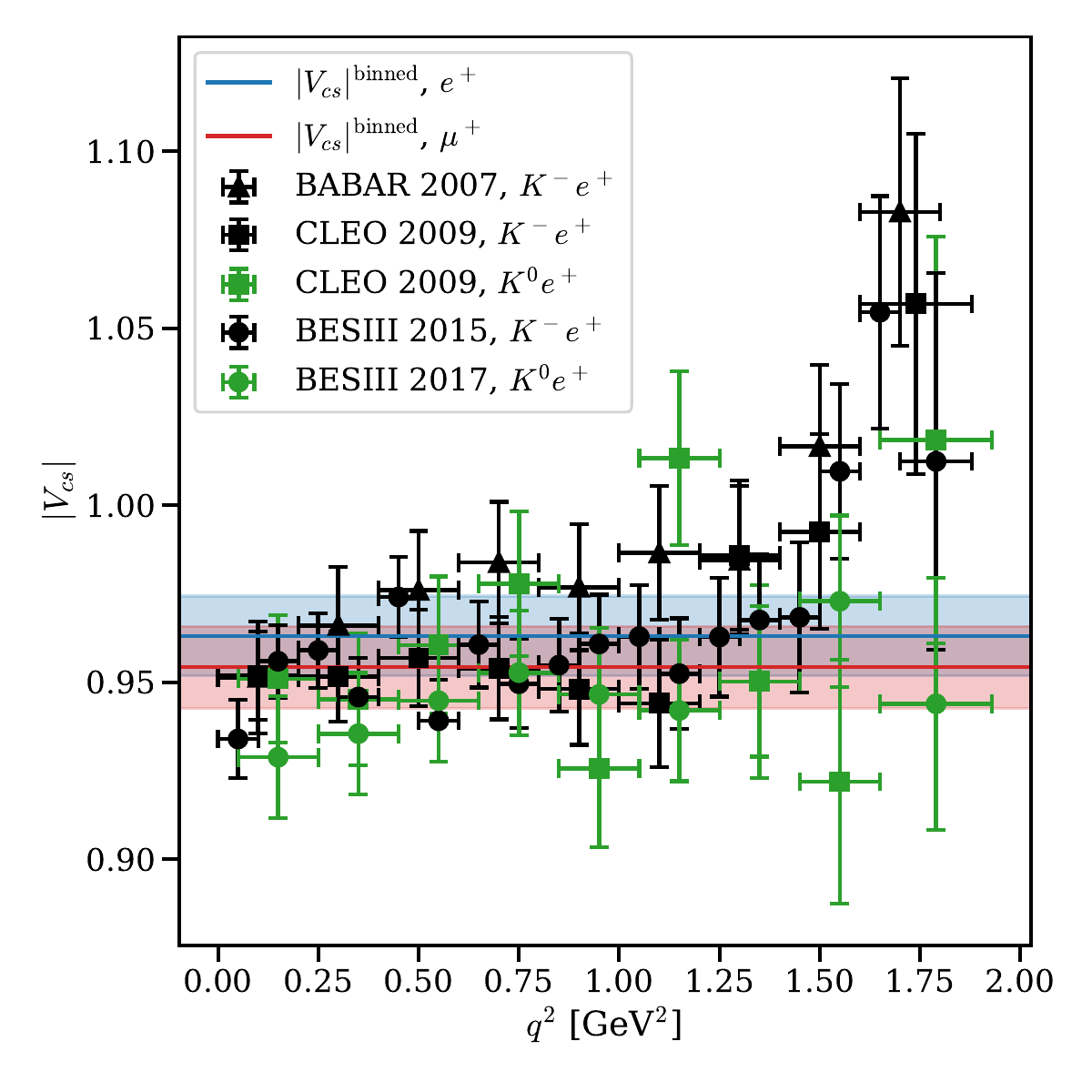}
    \includegraphics[width=0.49\textwidth]{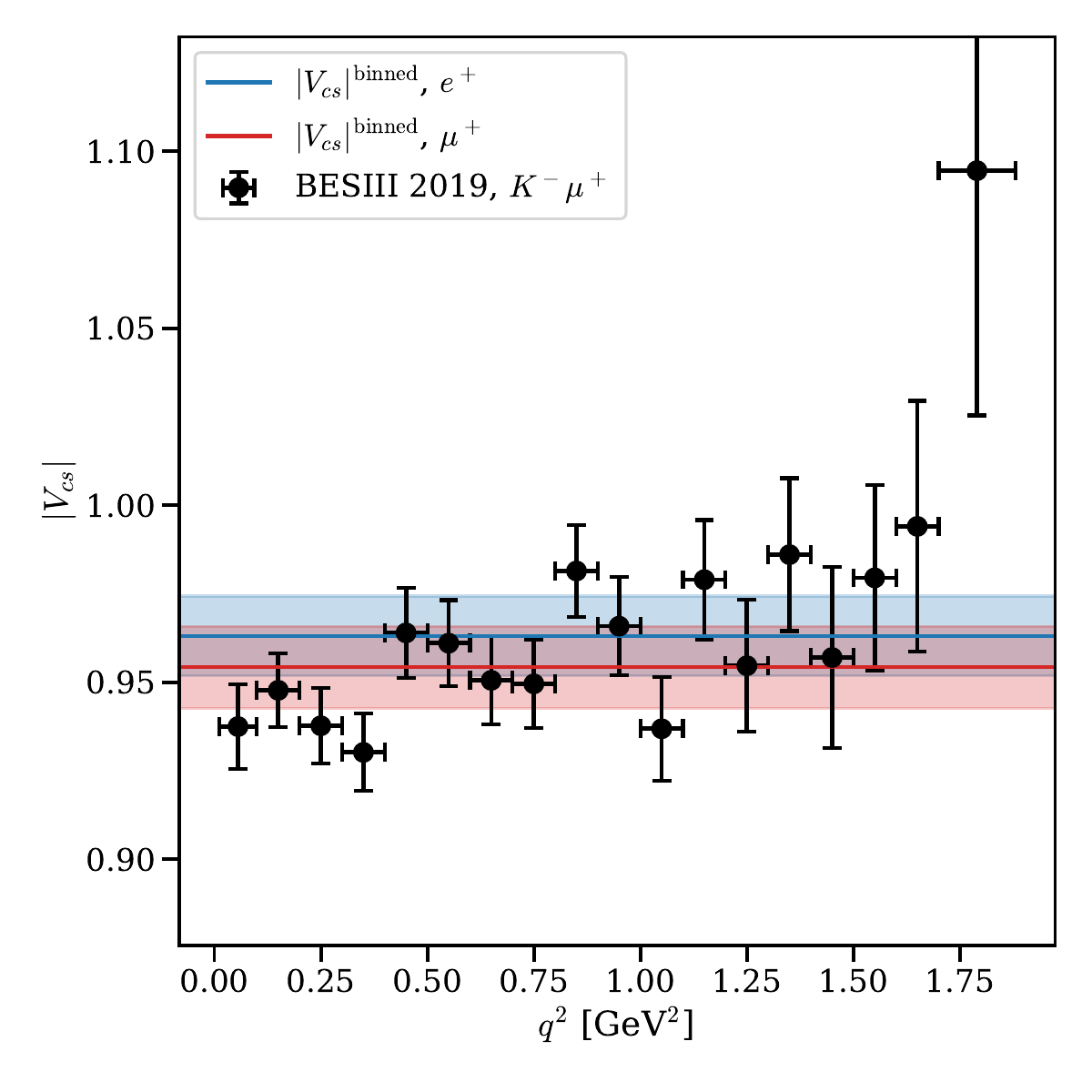}
    \caption{
    The binwise estimate of the CKM matrix element $\left[|V_{cx}(q_i^2)|\right]_{\rm Binned}$ from the decays $\Dpi$ (top rows) and $\DK$ (bottom row) in the semielectronic (left) and semimuonic (right) channels.
    The horizontal bands show the resulting values for $\Vcx^{\rm Binned}$ from correlated fits to a constant in each channel.
    The result for $\Vcd_{\rm Binned}^{\Dpi\mu^+\nu}$ (red) lies slightly below $\Vcd_{\rm Binned}^{\Dpi e^+\nu}$ (blue).
    The combined extraction using both channels lies between the two bands and is statistically consistent with each.
    A comparison of the different extractions of $\Vcx$ is given in \cref{fig:VcdVcs}.
    For $\Dpi$, experimental data are taken from BaBar~\cite{BaBar:2014xzf}, CLEO~\cite{CLEO:2009svp}, and BES~III~\cite{BESIII:2015tql,BESIII:2017ylw,BESIII:2018nzb}.
    For $\DK$, experimental data are taken from BaBar~\cite{BaBar:2007zgf}, CLEO~\cite{CLEO:2009svp}, and BES~III~\cite{BESIII:2015tql,BESIII:2017ylw,BESIII:2018ccy}.
    Although all the correlated fits have good quality ($\chi^2/{\rm DOF}\approx1$, $p\gtrsim 0.05$), the residuals for $\DK$ are visually larger near $q^2_{\rm max}$.
    \label{fig:binwise_Vcx}
    }
\end{figure}

Because the joint-fit and binned methods explicitly account for (potentially) percent-level contributions from the scalar form factor, they constitute our main extractions for $\Vcx$.
For continuity with previous studies, we also consider the \emph{endpoint method}, in which $\Vcx$ is defined according to
\begin{align}
[\Vcx]_{\rm Endpoint} \equiv \frac{[\Vcx \eta_{\rm EW} f_+(0)]_{\rm Expt}}{\eta_{\rm EW}[f_+(0)]_{\rm LQCD}}.
\end{align}
The experimental values are taken from the HFLAV world averages: $\Vcd \eta_{\rm EW} f_+^{\Dpi}(0) = 0.1426(18)$,
$\Vcs \eta_{\rm EW} f_+^{\DK}(0) = 0.7180(33)$~\cite{HFLAV:2022pwe}. 
The resulting values for $[\Vcx]_{\rm Endpoint}$ are shown in \cref{fig:VcdVcs} and given in \cref{table:vcx_summary}.
Although these endpoint results give a statistical precision comparable to our preferred extractions, it's worth emphasizing that our precise values for $f_+(0)$ were made possible by leveraging information about the form factor across the full kinematic range of the decays.
The final errors can potentially be much larger in a simulation that works directly at the endpoint ($q^2=0$).
For example, preliminary work by our collaboration has focused on $q^2\approx 0$ on many of the same ensembles and with comparable statistics~\cite{FermilabLattice:2019ycs}.
Using the preliminary values of $f_+(0)$ from these proceedings gives values for $\Vcd$ and $\Vcs$ with errors that are roughly 2.5 to 3.5 larger than the final errors in the present work. 

\begin{figure}
    \centering
    \includegraphics[width=0.49\textwidth]{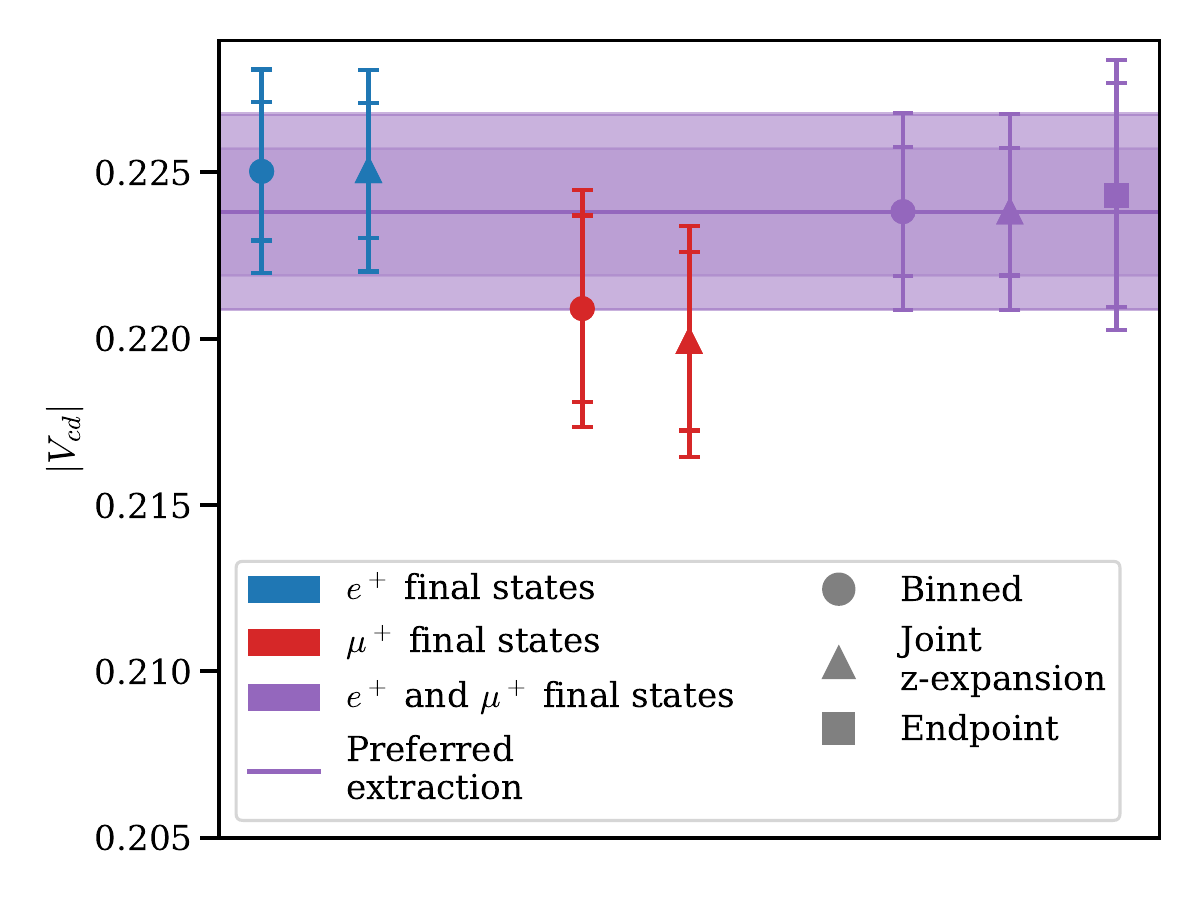}
    \includegraphics[width=0.49\textwidth]{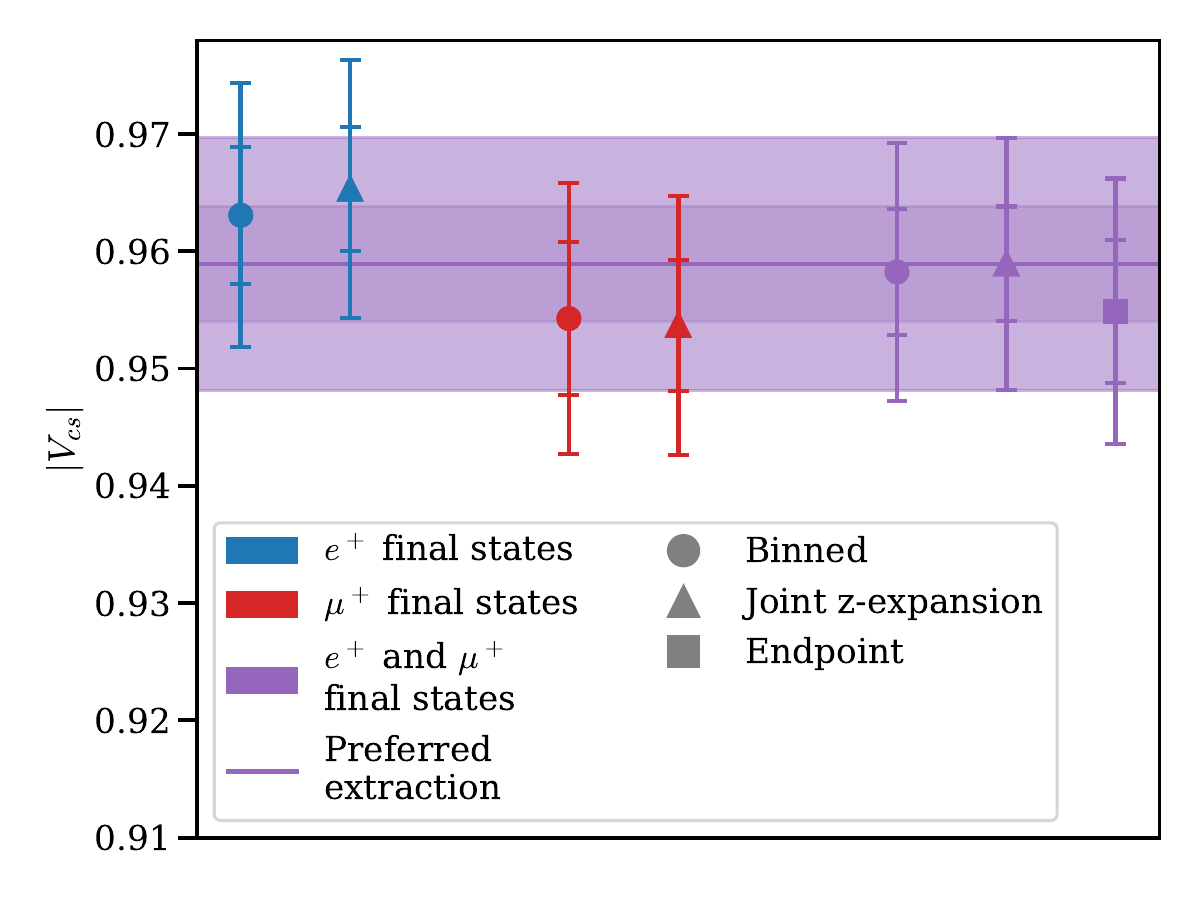}
    \caption{
    Determinations of $\Vcd$ and $\Vcs$ using experimental measurements of the decays $\Dpi$ and $\DK$.
    The outer and inner error bars and bands show the results with and without QED uncertainties, respectively.
    \label{fig:VcdVcs}
    }
\end{figure}

\begin{figure}
    \centering
    \includegraphics[width=0.49\textwidth]{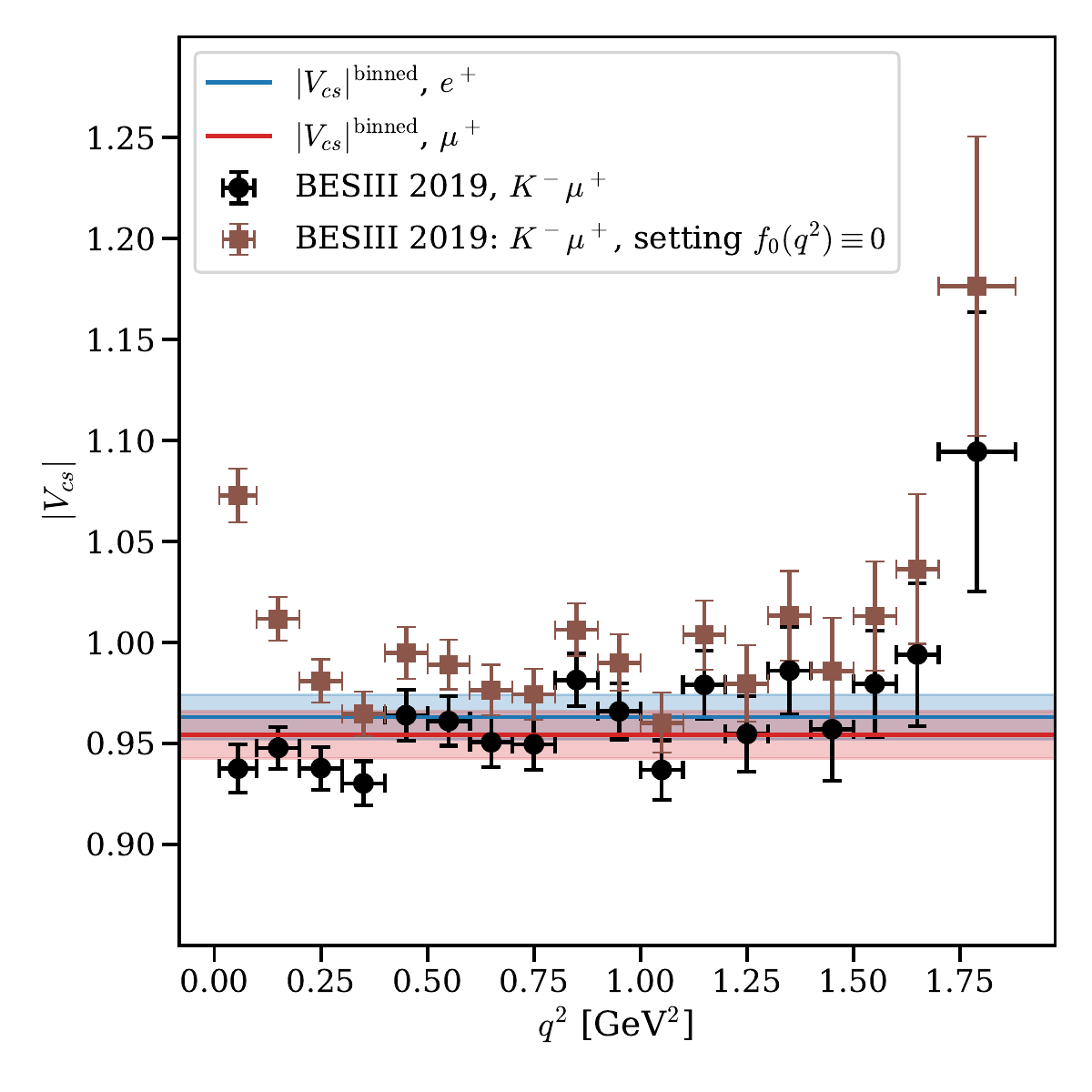}
    \caption{
    The effect of neglecting the scalar form factor (setting $f_0(q^2)\equiv 0$ in \cref{eq:dGammadq2}) when computing $\Vcs^{\rm binned}$ for $D\to K\mu\nu$.
    The red and blue horizontal lines and the black data are reproduced from the bottom-right panel of \cref{fig:binwise_Vcx}.
    Similar few-percent shifts occur for $D\to\pi\mu\nu$.
    }
    \label{fig:Vcs_drop_f0}
\end{figure}

The results for $\Vcd$ and $\Vcs$ from the different methods described above and for different leptons in the final states are summarized in \cref{fig:VcdVcs} and \cref{table:vcx_summary}.
Since, as shown in the plot, the binned and joint-fit extractions give statistically consistent values well within 1$\sigma$, we take the joint-fit extractions to define our preferred results:
\begin{align}
\Vcd^{\Dpi \ell^+ \nu} &= \VcdCombined, \label{eq:Vcdpreferred}\\
\Vcd^{\DsK e^+ \nu} &= \VcdElectronicDsK, \\
\Vcs^{\DK \ell^+\nu} &= \VcsCombined, \label{eq:Vcspreferred}
\end{align}
where the first error comes from the experimental differential decay rate uncertainty,
the second error comes from our form factor calculation (see \cref{table:final_error_budget}), the third error shows the uncertainty in $\eta_{EW}$, and the fourth and fifth from our estimate of SIB and long-distance QED corrections described in \cref{ssec:sib_qed}.
The errors in these expressions combine in quadrature to give the total errors in \cref{table:vcx_summary}. 
Since our preferred extraction of $\Vcs$ includes both $e^+$ and $\mu^+$ final states, the experimental contribution to the error is smaller by roughly a factor of two than in Ref.~\cite{Chakraborty:2021qav}.

We also repeated our analysis separating charged-hadron and neutral-final states (e.g., $\pi^- e^+$ versus $\pi^0 e^+$).
No statistically significant difference was observed within the uncertainties, consistent with what was observed in Ref.~\cite{Chakraborty:2021qav}.

\begin{table}
    \centering
    \caption{Summary of results for $\Vcd$ and $\Vcs$ from different decays and different extraction methods. 
    The final column gives the result when errors from QED are neglected.
    \label{table:vcx_summary}
    }
    \begin{tabular}{lllll}
\hline\hline
       &              Process &      Method &      $\Vcx$ & $\Vcx$ (no QED) \\
\hline
$\Vcd$ &    $D\to\pi e^+ \nu$ & $z$-expansion &  0.2251(30) &      0.2251(20) \\
$\Vcd$ &    $D\to\pi e^+ \nu$ &      binned &  0.2250(31) &      0.2250(21) \\
$\Vcd$ &  $D\to\pi \mu^+ \nu$ & $z$-expansion &  0.2199(35) &      0.2199(27) \\
$\Vcd$ &  $D\to\pi \mu^+ \nu$ &      binned &  0.2209(36) &      0.2209(28) \\
$\Vcd$ & $D\to\pi \ell^+ \nu$ & $z$-expansion &  0.2238(29) &      0.2238(19) \\
$\Vcd$ & $D\to\pi \ell^+ \nu$ &      binned &  0.2238(30) &      0.2238(19) \\
$\Vcd$ & $D\to\pi \ell^+ \nu$ &    endpoint &  0.2243(41) &      0.2243(34) \\
\hline
$\Vcs$ &     $D\to K e^+ \nu$ & $z$-expansion & 0.9653(110) &      0.9653(53) \\
$\Vcs$ &     $D\to K e^+ \nu$ &      binned & 0.9631(113) &      0.9631(58) \\
$\Vcs$ &   $D\to K \mu^+ \nu$ & $z$-expansion & 0.9537(111) &      0.9537(56) \\
$\Vcs$ &   $D\to K \mu^+ \nu$ &      binned & 0.9543(116) &      0.9543(65) \\
$\Vcs$ &  $D\to K \ell^+ \nu$ & $z$-expansion & 0.9589(108) &      0.9589(49) \\
$\Vcs$ &  $D\to K \ell^+ \nu$ &      binned & 0.9582(110) &      0.9582(54) \\
$\Vcs$ &  $D\to K \ell^+ \nu$ &    endpoint & 0.9549(110) &      0.9549(61) \\
\hline
$\Vcd$ &   $D_s\to K e^+ \nu$ & $z$-expansion & 0.2582(155) &     0.2582(153) \\
$\Vcd$ &   $D_s\to K e^+ \nu$ &      binned & 0.2583(157) &     0.2583(155) \\
\hline\hline
\end{tabular}

\end{table}

For the first time, our calculation provides a value of $\Vcd$ from $\Dpi$ for which lattice QCD errors are at the same level as the experimental errors, $\sim0.5\%$ each. This represents an improvement by roughly a factor of six from the existing state of the art~\cite{Lubicz:2017syv,Riggio:2017zwh}.
For $\Vcd^{\DsK}$, experimental errors dominate and are substantially larger than for $\Dpi$.
Since the theoretical uncertainty is actually the smallest for $\DsK$, additional experimental measurements of this channel would be particularly welcome.
On the other hand, theoretical error exceeds the experimental error by roughly a factor of two in the extraction of $\Vcs$ from $\DK$, leaving room for improvements in the theory side.
Experimental errors also dominate the CKM extractions from the semimuonic channels, where we have only included recent results from BES-III.
Another key ingredient for improved semileptonic extractions of $\Vcd$ and $\Vcs$ would be the calculation of long-distance structure-dependent EM corrections or a more robust estimate of their effect on these decays, since our lack of knowledge of these corrections currently dominates the uncertainty of the most precise determinations.

A comparison of our final results for $\Vcd$ and $\Vcs$ with existing results in the literature appears in \cref{fig:VcdVcs_comparison}, including leptonic decays, global fits assuming CKM unitarity fits, and scattering. 
Our determinations of $\Vcd$ and $\Vcs$ agree well, at the level of 1--2 standard deviations, with previous leptonic~\cite{Davies:2010ip,FermilabLattice:2011njy,Na:2012iu,Carrasco:2014poa,Boyle:2017jwu,Bazavov:2017lyh} and semileptonic~\cite{Na:2010uf,Na:2011mc,Lubicz:2017syv,Riggio:2017zwh,Chakraborty:2021qav} determinations reported in FLAG~\cite{Aoki:2021kgd}.

Our correlated results for $\Vcd$ and $\Vcs$ also yield the ratio,
\begin{equation}
    \Vcd/\Vcs =
    0.2334(13)^{\rm Expt}(16)^{\rm QCD}(02)^{\rm SIB}[11]^{\rm QED}
    \label{eq:VcdoverVcs}
\end{equation}
where the correlation coefficient between $\Vcd$ and $\Vcs$, neglecting QED, is 0.18.
As described in \cref{ssec:sib_qed}, we have taken a conservative $0.5\%$ systematic uncertainty for QED effects in the ratio.

Using the latest measurements $f_{D^+} \Vcd$ and $f_{D_s} \Vcs$ reported by HFLAV~\cite{HFLAV:2022pwe} and the ratio of decay constants $f_{D_s}/f_{D^+}$ computed by our collaboration in a similar set of ensembles and with the same action in Ref.~\cite{Bazavov:2017lyh}, one finds $[\Vcd/\Vcs]^{\rm leptonic} = 0.2212(58)$, where the error is dominated by the experimental uncertainty.
Both values are plotted in \cref{fig:vcd_by_vcs} together with previous
leptonic~\cite{Davies:2010ip,FermilabLattice:2011njy,Na:2012iu,Carrasco:2014poa,Boyle:2017jwu,Bazavov:2017lyh}
semileptonic~\cite{Na:2010uf,Na:2011mc,Lubicz:2017syv,Riggio:2017zwh,Chakraborty:2021qav} determinations combined in averages by FLAG~\cite{Aoki:2021kgd}
and the result from the PDG global unitarity fit~\cite{Workman:2022ynf}
(the global-fit methodologies of CKMfitter~\cite{Charles:2004jd} and UTfit~\cite{UTfit:2022hsi} give very similar results).
The leptonic extraction above agrees with our semileptonic result within roughly $2\sigma$, although, as plotted in \cref{fig:vcd_by_vcs}, leptonic determinations tend to give smaller values of the ratio.
The error in our result is more than a factor of two smaller than the leptonic one, with similar uncertainties from lattice QCD and experiment.
Results for $\Vcd/\Vcs$ from the PDG global fit assuming unitarity and from the ratio  $\Vus/\Vud$ (see \cref{ssec:unitarity} below for more details) are also shown in \cref{fig:vcd_by_vcs}.
Our result agrees well with both of them.

\begin{figure}
    \centering
    \includegraphics[width=0.49\textwidth]{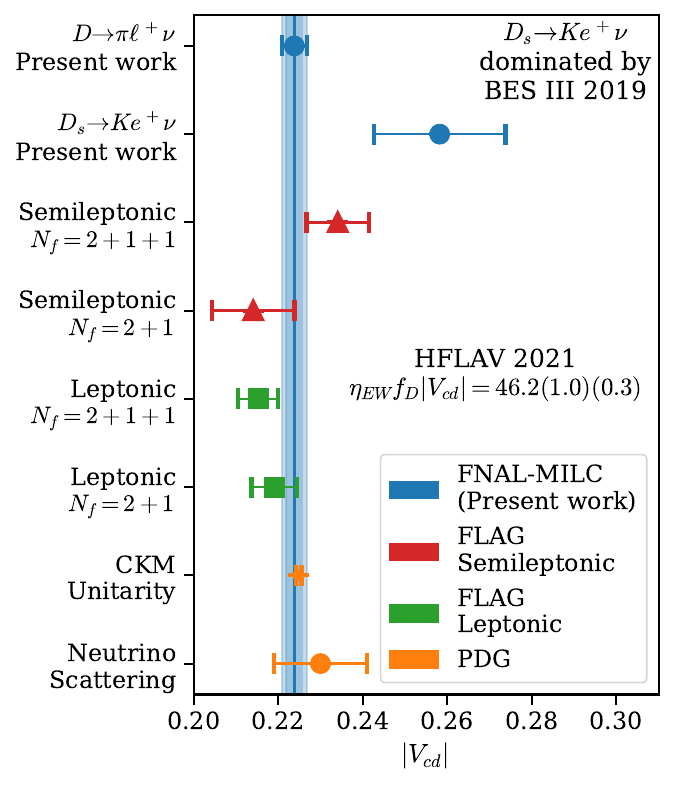}
    \includegraphics[width=0.49\textwidth]{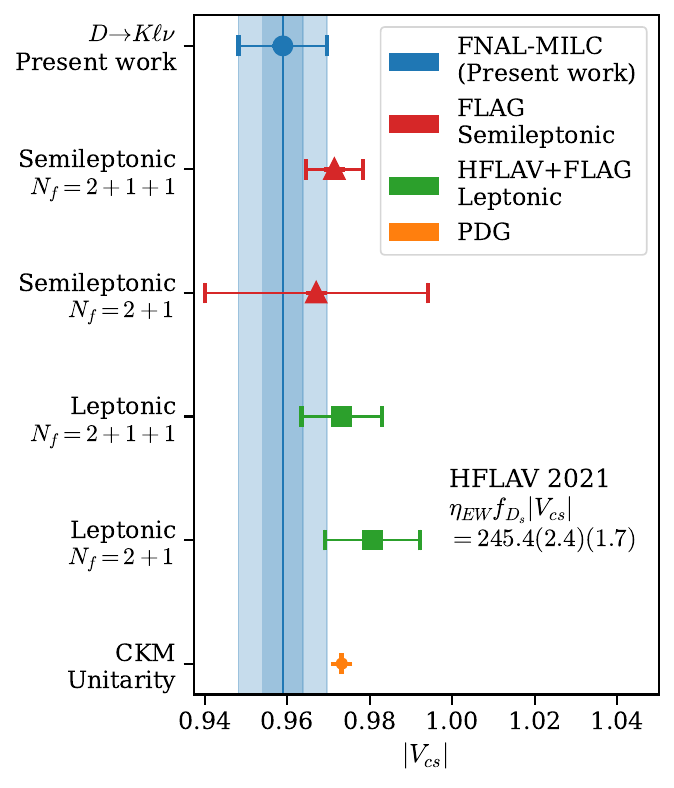}
    \caption{
    Comparison of our preferred determinations of $\Vcd^{\Dpi}$ and $\Vcs^{\DK}$ (blue bands) with existing results in the literature. The outer and inner error bands show our preferred result with and without QED uncertainties, respectively.
    The world's first determination $\Vcd^{\DsK}$ is also given.
    Results from FLAG are taken from Ref~\cite{Aoki:2021kgd}.
    Results from the PDG appear in Ref.~\cite{Workman:2022ynf}.
    We emphasize that FLAG uses slightly different conventions for the semileptonic extraction of $|V_{cd(cs)}|$ as we used here; for instance they do not include short-distance electroweak corrections to $G_F$ or an error from QED.
    For the leptonic results, we combine the latest experimental averages reported in HFLAV~\cite{HFLAV:2022pwe} with the FLAG averages for $f_D$ and $f_{D_s}$~\cite{Aoki:2021kgd}. 
    ``CKM unitarity'' denotes the global fit result reported by the PDG, which includes all available measurements (for all nine matrix elements) imposing three-generation unitarity.     
    \label{fig:VcdVcs_comparison}
    }
\end{figure}

\begin{figure}
    \centering
    \includegraphics[width=0.75\textwidth]{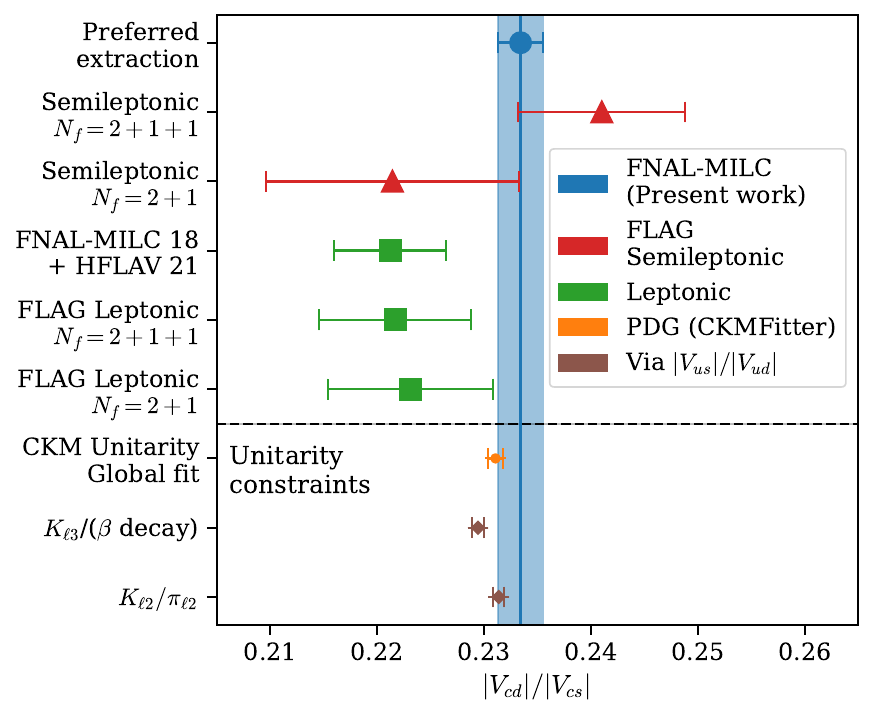}
    \caption{Comparison of different extractions of the ratio $\Vcd/\Vcs$.  
    The blue point and band show the value from the preferred extractions of the present work. 
    Error bands with and without QED error are indistinguishable.
    The red and green points denote semileptonic and leptonic extractions given by FLAG~\cite{Aoki:2021kgd}.
    The points below the dashed line are constraints from unitarity. 
    The orange point is computed using values from CKMFitter's global fit (as reported in the PDG), assuming CKM unitarity.
    The brown points comes from two different extractions of $\Vus/\Vud$ which, as explained in the text, are related to $\Vcd/\Vcs$ by CKM unitarity. 
    \label{fig:vcd_by_vcs} }
\end{figure}

\subsection{Tests of CKM unitarity}
\label{ssec:unitarity}

Our results for $\Vcd$ and $\Vcs$ enable a test of unitary in the second row of the CKM matrix, including theoretical correlations between $\Vcd$ and $\Vcs$.
Using our preferred extractions in \cref{eq:Vcdpreferred} and \cref{eq:Vcspreferred}, and $\Vcb^{\rm incl + excl} = (40.8\pm 1.4)\times10^{-3}$ from a combined average of inclusive and exclusive semileptonic $B$-decays~\cite{Workman:2022ynf}\footnote{
In particular, see the review ``Semileptonic $b$-Hadron Decays, Determination of $V_{cb}$, $V_{ub}$"} yields the following result for the deviation from unitarity in the second row:
\begin{align}
    \Vcd^2 +& \Vcs^2 + \Vcb^2 -1
        =-0.0286(44)^{\rm Expt}(78)^{\rm QCD}[194]^{\rm QED}(28)^{\rm EW}
        = -0.029(22).
\end{align}
Because $\Vcb$ is so small compared to $\Vcd$ and $\Vcs$, numerically indistinguishable results are obtained (within current precision) if inclusive or exclusive values are taken for $\Vcb$.
This result is compatible with three-generation CKM unitary within approximately one standard deviation. 
The precision of this test is roughly $2\%$ and is limited by the systematic uncertainty from QED in our extractions of $\Vcd$ and $\Vcs$. 
We show the constraints on $\Vcd$ and $\Vcs$ from our calculation in \cref{fig:ckm_2nd_row_unitarity}, together with constraints coming from leptonic decays~\cite{Bazavov:2017lyh,HFLAV:2022pwe} and second-row unitarity.
The leptonic inputs used for the green ellipse are summarized in \cref{table:leptonic_inputs}. 
As the figure shows, semileptonic tests of second-row CKM unitarity are now slighty more precise than leptonic tests.
The leptonic and semileptonic results are consistent at the level of roughly 1-2 standard deviations. 

\begin{figure}
    \centering
    \includegraphics[width=0.67\textwidth]{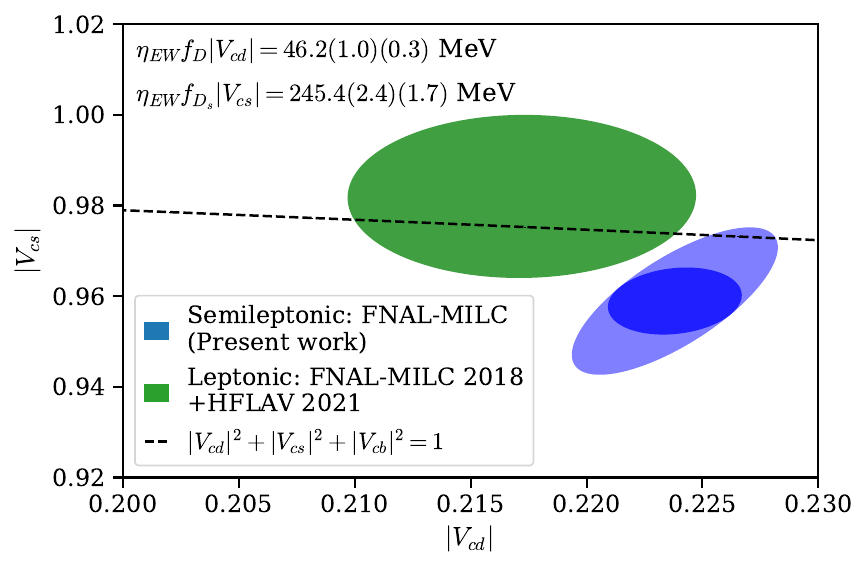}
    \caption{Constraints on $\Vcd$ and $\Vcs$ from our results, $D-$meson leptonic decays, and unitarity. The blue ellipse shows the preferred values of the present work from semileptonic decays in \cref{eq:Vcdpreferred} and \cref{eq:Vcspreferred}.
    The green ellipse is the result of combining the latest results for the products $\eta_{\rm EW}\Vcd f_{D^+}$ and $\eta_{\rm EW} \Vcs f_{D_s}$ with leptonic decay constants from lattice-QCD calculations; the inputs values are summarized in \cref{table:leptonic_inputs}.
    The dotted line comes from assuming unitarity of the second row, taking $\Vcb^{\rm incl + excl} = (40.8\pm 1.4)\times10^{-3}$~\cite{Workman:2022ynf}. 
    In all cases, the ellipses shows the correlated $1\sigma$ (68\%) confidence intervals.
    The inner blue ellipse shows our result without the QED uncertainty.
    }
    \label{fig:ckm_2nd_row_unitarity}
\end{figure}

One can perform further tests of the unitarity of the CKM matrix using the fact that in the Standard Model, $\Vcd=\Vus +\order{A^2\lambda^5}$ and $\Vcs=\Vud + \order{A^2\lambda^4}$. 
Including the dominant corrections~\cite{Buras:1994ec} with the Wolfenstein parameters taken from global unitarity fits by CKMFitter~\cite{Charles:2004jd} (using values from the January 2022 update) gives
\begin{align}
    \Vcs &= 0.97282(32) \text{ from } \Vud^{0^+\to0^+}, \\
    \Vcd &= 0.22317(53) \text{ from } \Vus^{K\ell3},    \\
    \Vcd/\Vcs &= 0.22941(55) \text{ from } \Vus^{K\ell3}/\Vud^{0^+\to0^+},
\end{align}
using $\Vud=0.97367(32)$ from superallowed $0^+\to 0^+$ nuclear $\beta$ decays~\cite{Hardy:2020qwl,Cirigliano:2022yyo} and $\Vus=0.22330(53)$ from $K_{\ell 3}$ decays~\cite{Aoki:2021kgd,Cirigliano:2022yyo}.
Alternatively, the ratio of $K_{\ell2}$ to $\pi_{\ell2}$ decays yields~\cite{Aoki:2021kgd,Cirigliano:2022yyo}
\begin{equation}
    \Vcd/\Vcs = 0.23135(51) \text{ from } |V_{us}/ V_{ud}|^{K_{\ell 2}/\pi_{\ell 2}}.
\end{equation}
As shown in \cref{fig:vcd_by_vcs}, our preferred value in \cref{eq:VcdoverVcs} lies roughly $1\sigma$ above the result coming from  $|V_{us}/V_{ud}|^{K_{\ell2}/\pi_{\ell2}}$ and roughly $2\sigma$ above that from $\Vus^{K_{\ell3}}/\Vud^{0^+\to0^+}$.
Our preferred value for $\Vcd$ in \cref{eq:Vcdpreferred} shows excellent agreement with $\Vcd$ from $\Vus^{K\ell3}$.
Our preferred value for $\Vcs$ in \cref{eq:Vcspreferred} lies somewhat below $\Vcs$ from $\Vud^{0^+\to0^+}$ but is consistent at 1-2 standard deviations.

\begin{table}
\caption{Leptonic inputs used for comparison in \cref{fig:ckm_2nd_row_unitarity}.
HFLAV reports the product $\eta_{\rm EW}\Vcx f_{D_{(s)}}$~\cite{HFLAV:2022pwe}.
Following the prescription of the PDG~\cite{Workman:2022ynf}, we include an EW+QED error of $0.7\%$ for the product $\Vcx f_{D_{(s)}}$.
}
    \centering
    \begin{tabular}{l l }
        \hline\hline
        Value                           & Source \\
        \hline
        $\eta_{\rm EW}\Vcd f_{D^+} = 46.2(1.0)(0.3)^{\rm EW+QED} \MeV$ & HFLAV~\cite{HFLAV:2022pwe}\\
        $\eta_{\rm EW}\Vcs f_{D_s}= 245.4(2.4)(1.7)^{\rm EW+QED} \MeV$  & HFLAV~\cite{HFLAV:2022pwe}\\
        $f_{D^+} = 212.7(0.6) \MeV$     & Fermilab-MILC 2018~\cite{Bazavov:2017lyh}\\
        $f_{D_s} = 249.9(0.4) \MeV$     & Fermilab-MILC 2018~\cite{Bazavov:2017lyh}\\
        $f_{D_s}/f_{D^+} = 1.1749(16)$  & Fermilab-MILC 2018~\cite{Bazavov:2017lyh}\\
        \hline\hline
    \end{tabular}
    \label{table:leptonic_inputs}
\end{table}

\subsection{Lepton flavor universality}
\label{sec:lfu}

For a given semileptonic decay $H\to L\ell\nu$, the lepton flavor universality (LFU) ratio $R_{\mu/e}$ is defined as the ratio of branching fractions into muon versus electron final states 
\begin{align}
R_{\mu/e}^{H\to L}
    \equiv \frac{\mathcal{B}(H\to L\mu\nu)}{\mathcal{B}(H\to Le\nu)}
    = \frac{\Gamma_\mu}{\Gamma_e},
\end{align}
where the total rates to each final state are defined in the usual way,
\begin{align}
    \Gamma_\ell
    \equiv \int_{m_\ell^2}^{q^2_{\rm max}} dq^2\, \left(\frac{d\Gamma}{dq^2}\right).
    \label{eq:total_rate}
\end{align}
In the SM, the LFU ratios are close but not identically equal to unity.
This difference from unity arises from at least three effects.
First, the lower boundary of the integration region in \cref{eq:total_rate} depends on the lepton mass.
Second, the differential decay rate in \cref{eq:dGammadq2} itself depends on the lepton mass, with the scalar form factor contributing more for larger masses. 
Finally, QED corrections depend in principle on both the charges of the final state and the lepton mass. 
The coefficients $\frac{G_F^2}{24\pi^3} \left( \eta_{\rm EW} |V_{cx}| \right)^2$ are independent of $q^2$ and cancel in the ratio, meaning that predictions for $R_{\mu/e}$ are entirely calculable using our lattice-QCD form factors, up to corrections from QED and SIB.
The rates $d\Gamma/dq^2$ for the decay $\Dpi$, using as inputs our form factors $f_0(q^2)$ and $f_+(q^2)$ together with the estimates of systematic uncertainties from QED and SIB (see \cref{ssec:sib_qed}), are shown in \cref{fig:rates} for both semielectronic and semimuonic final states. 
When computing the rates, the meson masses were taken to be the average of the experimentally measured masses for the charged and neutral states (e.g., $D^0$ and $D^+$ or $\pi^0$ and $\pi^+$).
The final results for the SM predictions of the ratios $R_{\mu/e}$ are
\begin{align}
    \RDpi &= \RDpiFinal, \\
    \RDK  &= \RDKFinal , \\
    \RDsK &= \RDsKFinal.
\end{align}
The dominant error is the systematic uncertainty from QED corrections, which we conservatively take to be $0.5\%$, as described in \cref{ssec:sib_qed}.
Our prediction for $\RDK$ is in good agreement with a recent calculation by HPQCD, which found $\RDK = 0.97594(19)^{\rm QCD}[500]^{\rm QED}$
and used the same estimate of the QED uncertainty~\cite{Chakraborty:2021qav}.\footnote{%
The central value we quote here differs slightly from the published value in Ref.~\cite{Chakraborty:2021qav}.
We thank HPQCD for providing the correct central value
(William Parrott, private communication, 16 Dec 2022).}
We also find good agreement with previous lattice QCD results by ETMC and experimental measurements of $\RDpi$ and $\RDK$, as shown in \cref{fig:comparisonLFU}.
The measurement of $\RDpi$ by BES~III for the channel $D^0\to\pi^-$ lies below our result but is consistent at the $2\sigma$ level.
Because the QED error is dominant for the lattice-QCD predictions of the LFU ratios, the insets in \cref{fig:comparisonLFU} compare the lattice-QCD results with the QED uncertainty removed.

\begin{figure}
    \centering
    \includegraphics[width=0.49\textwidth]{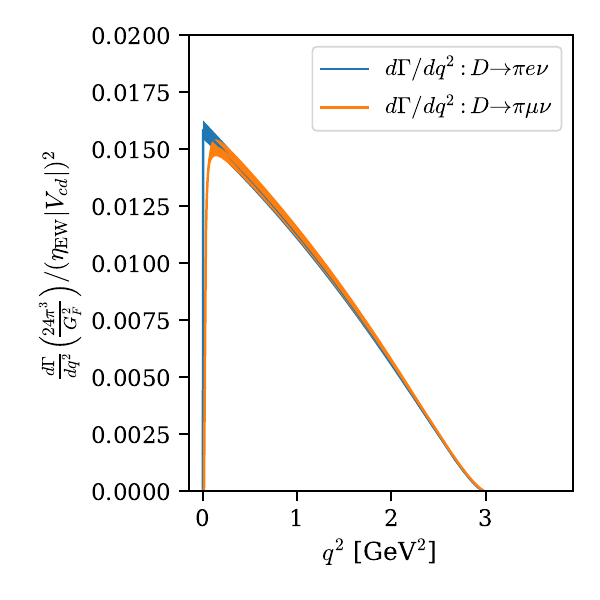}
    \caption{
    Our prediction for the decay rates $d\Gamma/dq^2$ for the decays $\Dpi$.
    The majority of the total rate comes from small $q^2$, where $(d\Gamma/dq^2)_\mu < (d\Gamma/dq^2)_e$.
    The Standard Model therefore predicts 
    $R_{\mu/e} < 1$.
    \label{fig:rates}
    }
\end{figure}

\begin{figure}
    \centering
    \includegraphics[width=0.49\textwidth]{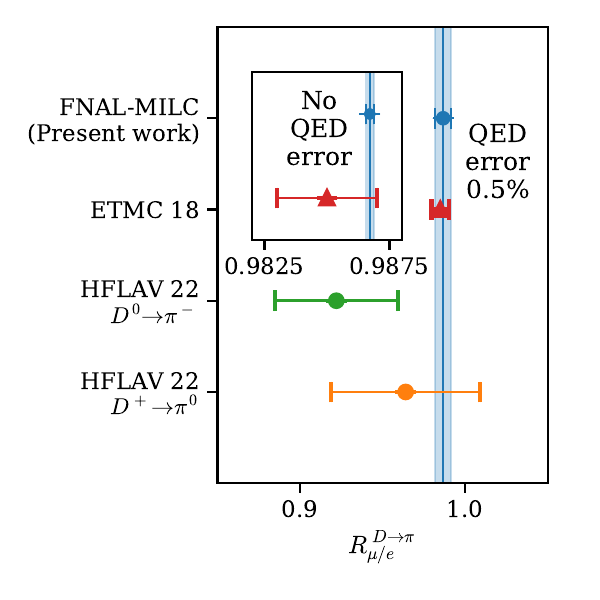}
    \includegraphics[width=0.49\textwidth]{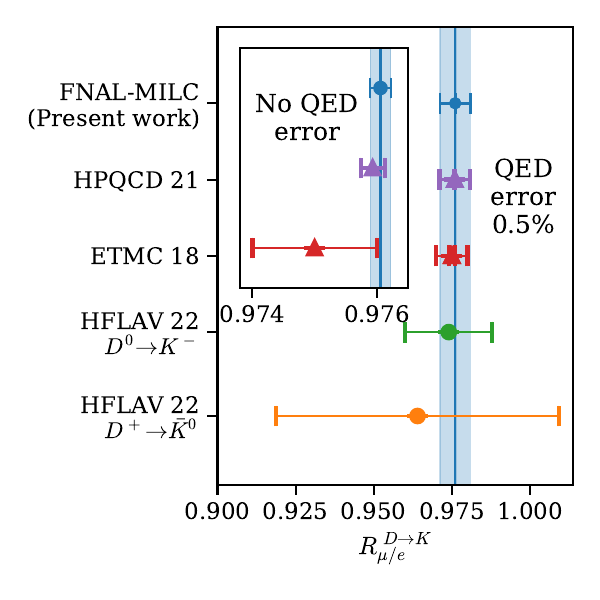}
    \caption{Comparison of $R_{\mu/e}^{\Dpi}$ and $R_{\mu/e}^{\DK}$ 
    with experimental HFLAV averages~\cite{HFLAV:2022pwe}, which are dominated by measurements from BES III~\cite{BESIII:2016gbw,BESIII:2018ccy,BESIII:2018nzb}, and other SM predictions from lattice QCD~\cite{Riggio:2017zwh,Chakraborty:2021qav}.
    In the main body of both figures, all lattice QCD results are presented with a QED uncertainty of $0.5\%$. 
    The results from ETMC 18 were reported in the isospin-symmetric limit of QCD, without including QED or SIB uncertainties~\cite{Riggio:2017zwh}, so we have added the QED uncertainty for a like-to-like comparison.
    The insets compare lattice QCD results when QED uncertainty is removed.
    \label{fig:comparisonLFU}}
\end{figure}

\clearpage

\section{Conclusions \label{sec:conclusions}}

We have calculated the hadronic form factors $f_+(q^2)$ and $f_0(q^2)$ relevant for the semileptonic decays $\Dpi\ell\nu$, $\DK\ell\nu$, and $\DsK\ell\nu$ using lattice QCD.
These decays occur at tree level in the SM and are important channels for determining the CKM matrix elements $\Vcd$ and $\Vcs$.
Our calculation uses $N_f = 2+1+1$ flavors of dynamical staggered quarks and includes several ensembles with all quarks near their physical masses.
The use of the HISQ action permits all the quarks to be treated with the same relativistic light-quark action and allows for nonperturbative renormalization using PCVC, \cref{eq:PCVC}.
Our results improve significantly on the previous precision for the form factors for $\Dpi$ and $\DsK$, and have precision comparable to that of recent $N_f=2+1+1$ calculations by HPQCD for $\DK$~\cite{Chakraborty:2021qav,Parrott:2022rgu}.
We agree well with HPQCD's $\DK$ form factors over the entire kinematic range, especially with their latest results in Ref.~\cite{Parrott:2022rgu}, while for both $\Dpi$ and $\DK$, our form factors are significantly larger near $q^2_{\rm max}$ than the $N_f=2+1+1$ results of ETMC~\cite{Lubicz:2017syv}. 
\cref{table:z_results} shows the $z$-expansion parameters from which our final results for the form factors, computed in isospin-symmetric QCD where $m_u=m_d$, can be reconstructed, while a complete error budget, including all statistical and systematic uncertainties, is given in \cref{table:final_error_budget} for the edges of the kinematic range.

Our results suggest a very mild spectator dependence for $\Dpi$ and $\DsK$, with close agreement at $\lesssim 2\%$ level throughout the kinematic range between the respective form factors (cf. \cref{fig:spectator_comparison}).
This picture was also recently confirmed, within experimental uncertainty, by the first measurement of the decay $\DsK$ by BES~III~\cite{BESIII:2018xre}.

When combined with the available experimental data for the corresponding decay rates, summarized in \cref{table:expt_summary}, our form factors enable the extraction of the CKM matrix elements $\Vcd$ and $\Vcs$ with percent-level uncertainties.
These extractions include correlations between all the lattice form factors and between the different experimental channels.\footnote{In the supplementary material, we provide correlated values for all the $z$-expansion coefficients needed to reproduce our final results for all three decays.}
The values obtained from our preferred extractions are
\begin{align*}
\Vcd^{\Dpi \ell^+ \nu} &= \VcdCombined,\\
\Vcd^{\DsK e^+ \nu} &= \VcdElectronicDsK,\\
\Vcs^{\DK \ell^+\nu} &= \VcsCombined. \\
\Vcd/\Vcs & = 0.2329(13)^{\rm Expt}(16)^{\rm QCD}(02)^{\rm SIB}[11]^{\rm QED}
\end{align*}
For $|V_{cd}|$ we obtain the most precise determination to date, with lattice-QCD form factors errors that, for the first time in a semileptonic extraction, are commensurate with experimental uncertainties. The improved determination of $\Dpi$ form factors, together with the fact that we account for theoretical correlations among channels, also allows us to provide the most precise determination of the ratio $\Vcd/\Vcs$, around a factor of two more precise than the leptonic determination.
The rate for $\DsK$ was only recently measured for the first time by BES~III~\cite{BESIII:2018xre}, and our calculation delivers the first extraction of $\Vcd^{\DsK}$. 
Although this determination is not yet competitive with the one from $\Dpi$, the error is dominated by the statistics-limited experimental uncertainty.
Our result for $\Vcd^{\DsK}$ lies roughly $2\sigma$ above $\Vcd^{\Dpi}$, albeit with large uncertainty.
Experimental improvements for this Cabibbo-suppressed decay would immediately give improved precision for $\Vcd^{\DsK}$ and help clarify the situation. 

Our determinations of $\Vcd$ and $\Vcs$, combined with the value of $\Vcb$ from Ref.~\cite{Workman:2022ynf}, give a precise test of second-row CKM unitarity.
We find consistency with unitarity at the level of roughly $2\%$ and one standard deviation, with an uncertainty dominated by the systematic effect of QED.
As shown in \cref{fig:ckm_2nd_row_unitarity}, the precision of the semileptonic constraint is now slightly better than the corresponding leptonic one.

After demonstrating consistency between the form factor shapes from our calculations and those measured in experiments, we computed the SM prediction for the lepton flavor universality ratios $R_{\mu/e}$ with sub-percent precision for all three decays:
\begin{align*}
    \RDpi &= \RDpiFinal,\\
    \RDK  &= \RDKFinal ,\\
    \RDsK &= \RDsKFinal.
\end{align*}
These results agree with previous $N_f=2+1+1$ lattice calculations, considerably improving the precision for $D\to\pi$, and with experimental measurements within $2\sigma$, for $\Dpi$ and $\DK$.

With the total precision for $\Vcd$ and $\Vcs$ approaching the subpercent level, the effects of the scalar form factor in the differential rate, \cref{eq:dGammadq2}, become quantitatively important.
For semielectronic decays, contributions from $f_0$ enter at roughly the $1\%$ level in the lowest $q^2$ bin.
For semimuonic decays, the effect is much larger, a roughly $10\%$ effect in the lowest $q^2$ bin and a few-percent effect throughout the rest of the kinematic range.
\Cref{fig:Vcs_drop_f0} showed that naively neglecting contributions from $f_0$ can shift values for $\Vcs$ by a few percent in the case of $D\to K \mu\nu$ (similar results hold for $D\to\pi\mu\nu$).

Future progress in the precision of $\Vcd$, $\Vcs$, and the LFU ratios will depend crucially on improved understanding of QED corrections to these decays, which are already the dominant source of uncertainty. The one exception is the decay $\Vcs^{\DsK \ell^+\nu}$, for which the experimental error is still large.
One avenue for improvement is through EFT calculations in the spirit of those for $K\to\pi\ell\nu$~\cite{Cirigliano:2008wn,Cirigliano:2011ny,Seng:2021boy,Seng:2021wcf,Seng:2022wcw}, which were used in \cref{ssec:sib_qed} to estimate our systematic uncertainties (cf. \cref{ssec:sib_qed}).
As usual, the intermediate mass of the charm quark (which is simultaneously too heavy for $\chi$PT to apply and too light for reliable application of HQET) may present a challenge for robust treatment with EFT.
Another possibility is carrying out lattice simulations to compute the structure-dependent QED corrections to the semileptonic decay amplitudes.
Such calculations have not yet reached a mature state, but the field is progressing rapidly, particularly for the QED corrections to leptonic decays~\cite{Carrasco:2015xwa, Giusti:2017dwk, DiCarlo:2019thl, Desiderio:2020oej, Frezzotti:2020bfa, Frezzotti:2021slr, Gagliardi:2022szw}.

Regarding the pure QCD calculation, it should be straightforward to improve the precision of our form factor results.
A leading contribution to the error budget is statistics (cf. \cref{table:final_error_budget} and \cref{fig:d2pi_final_error_budget,fig:d2k_final_error_budget,fig:ds2k_final_error_budget}), for which the physical mass ensembles at $a\approx 0.06\fm$ and $0.09\fm$ play the largest role.
As part of our ongoing work toward $B$-meson semileptonic decays, we are simulating on a finer physical-mass ensemble with $a\approx 0.04\fm$.
We expect that new data from this ensemble will reduce the uncertainties both from statistics and from the continuum extrapolation.
Future calculations will also benefit from ongoing work in the community to improve scale-setting measurements (e.g., $w_0$ or the $\Omega$-baryon mass) on the HISQ ensembles used in this work.

\section*{Acknowledgments}
We thank Claude Bernard, Urs Heller, Javad Komijani, and Jack Laiho for collaboration and essential contributions to previous projects, which paved the way for this work.  
We also thank Claude Bernard for helpful advice about scale setting and about the chiral expansion and Javad Komijani for useful correspondence regarding $\alpha_s$ on the HISQ ensembles.
We thank Jake Bennett and Alan Schwartz for answering questions about the Belle data.
We thank the BES~III collaboration, and especially Lei Li and Hailong Ma, for providing us with their data for $\DsK e \nu$ as well as correlation data for $D\to\pi\mu\nu$.
We thank Ryan Mitchell for useful comments about the CLEO and BES~III detectors.
We thank William Parrott for answering questions about HPQCD's evaluation of LFU ratios.

This material is based upon work supported in part 
by the U.S. Department of Energy, Office of Science under grant Contract Numbers
DE-SC0010120 (S.G.),
DE-SC0011090 (W.J.), DE-SC0021006 (W.J.),
DE-SC0015655 (A.X.K., Z.G., A.T.L.), and
DE-SC0010005 (E.T.N.); 
by the U.S. National Science Foundation under Grants No. PHY17-19626 and PHY20-13064 (C.D., A.V.);
by the Simons Foundation under their Simons Fellows in Theoretical Physics program (A.X.K.);
by SRA (Spain) under Grant No.\ PID2019-106087GB-C21 / 10.13039/501100011033 (E.G.);
by the Junta de Andalucía (Spain) under Grants No.\ FQM-101, A-FQM-467-UGR18 (FEDER), and P18-FR-4314 (E.G.);
by AEI (Spain) under Grant No.\ RYC2020-030244-I / AEI / 10.13039/501100011033 (A.V.).
This document was prepared by the Fermilab Lattice and MILC Collaborations using the resources of the Fermi National Accelerator Laboratory (Fermilab), a U.S. Department of Energy, Office of Science, HEP User Facility.
Fermilab is managed by Fermi Research Alliance, LLC (FRA), acting under Contract No.\ DE-AC02-07CH11359.

Computations for this work were carried out in part on facilities of the USQCD Collaboration, which are funded by the Office of Science of the U.S. Department of Energy.
An award of computer time was provided by the Innovative and Novel Computational Impact on Theory and Experiment (INCITE) program. This research used resources of the Argonne Leadership Computing Facility, which is a DOE Office of Science User Facility supported under contract DE-AC02-06CH11357. This research also used resources of the Oak Ridge Leadership Computing Facility, which is a DOE Office of Science User Facility supported under Contract No.\ DE-AC05-00OR22725.
This research used resources of the National Energy Research Scientific Computing Center (NERSC), a U.S. Department of Energy Office of Science User Facility located at Lawrence Berkeley National Laboratory, operated under Contract No.\ DE-AC02-05CH11231.
The authors acknowledge support from the ASCR Leadership Computing Challenge (ALCC) in the form of time on the computers Summit and Theta.
The authors acknowledge the \href{http://www.tacc.utexas.edu}{Texas Advanced Computing Center (TACC)} at The University of Texas at Austin for providing HPC resources that have contributed to the research results reported within this paper.
This research is part of the Frontera computing project at the Texas Advanced Computing Center. Frontera is made possible by National Science Foundation award OAC-1818253~\cite{Frontera}.
This work used the Extreme Science and Engineering Discovery Environment (XSEDE), which is supported by National Science Foundation Grant No.\ ACI-1548562.
This work used XSEDE Ranch through the allocation TG-MCA93S002~\cite{XSEDE}.

\clearpage
\appendix

\section{Analysis of staggered correlation functions}
\label{sec:stagg}

As demonstrated in Ref.~\cite{Bailey:2008wp}, averaging over adjacent time slices can dramatically suppress contributions from oscillating states.
Consider a two-point correlation function $C_2(t)$
Let $E$ denote the ground state energy.
Then the averaged two-point function is $\overline{C}_2(t)$:
\begin{align}
\overline{C}_2(t) &=
	\frac{e^{-E t}}{4}
	\left[
		\frac{C_2(t)}{e^{-E t}}
		+ \frac{2 C_2(t+1)}{e^{-E (t+1)}}
		+ \frac{C_2(t+2)}{e^{-E (t+2)}}
	\right], \label{eq:c2_avg} \\
	&= \frac{\abs{\matrixel{\vacuum}{\mathcal{O}}{E}}^2}{2 E} e^{-E t} + \mathcal{O}({\Delta E^2}), 
\end{align}
where $\mathcal{O}$ is an interpolating operator as given in \cref{table:spin_taste} and $\ket{\vacuum}$ is the QCD vacuum.
Similarly, consider a three-point correlation function $C_3(t,T)$ with ground states $E_L$ and $E_H$ at the source and sink, respectively, and connected by the current $J$.
The averaged three-point function is $\overline{C}_3(t,T)$:
\begin{align}
\overline{C}_3(t, T)
	= &\frac{ e^{-E_L t} e^{-E_H(T-t)} }{8} \times \nonumber \\
	    		& \left[ \frac{C_3(t,T)}{e^{-E_L t} e^{-E_H(T-t)}}
    			+ \frac{2 C_3(t+1,T)}{e^{-E_L (t+1)} e^{-E_H(T-t-1)}}
	    		+ \frac{C_3(t+2,T)}{e^{-E_L (t+2)} e^{-E_H(T-t-2)}} \right.  \label{eq:c3_avg} \\
			& + \left. \frac{C_3(t,T+1)}{e^{-E_L t} e^{-E_H(T+1-t)}}
	    		+ \frac{2 C_3(t+1,T+1)}{e^{-E_L (t+1)} e^{-E_H(T-t)}}
	    		+ \frac{C_3(t+2,T+1)}{e^{-E_L (t+2)} e^{-E_H(T-t-1)}} \right] \nonumber \\
	=&  \frac
	        {\matrixel{\vacuum}{\mathcal{O}_L}{E_L}
	        \matrixel{E_L}{J}{E_H}
	        \matrixel{E_H}{\mathcal{O}_H}{\vacuum}}
	        {4 E_L E_H}
       e^{-E_L t} e^{-E_H(T-t)} + \order{\Delta E_H^2,\Delta E_L^2}.
\end{align}
These averaged two- and three-point functions are used in \cref{eq:ratio_v4,eq:ratio_vi,eq:ratio_s}.

\section{Discretization errors for HISQ}
\label{ssec:HQETErrors}

Several of the results in this appendix were first derived in Ref.~\cite{Monahan:2012dq}.
Our discussion follows closely that of Ref.~\cite{Bazavov:2017lyh}.
Let $am_0$ and $am_1$ denote a quark's bare and rest masses, respectively.
The two quantities are related by the transcendental equation
\begin{equation}
    am_0 = a \widetilde{\mathcal{S}h}(am_1)
         = \sinh(am_1)\left(1 - \frac{1}{6}\mathbb{N}(am_1) \sinh^2(am_1) \right).
\end{equation}
In this expression, $\mathbb{N}(am_1)$ denotes the coefficient of the Naik improvement term appearing in the HISQ action
\begin{align}
    \mathbb{N}(am_1) &= \frac{4 - 2 \sqrt{1 + 3X(am_1)}}{\sinh^2(a m_1)}, \\
    X(am_1) &= \frac{2 am_1}{\sinh(2am_1)}.
\end{align}
When bare masses are not small, $am_0 \centernot\ll 1$, quark bilinears can lose their conventional normalization.
This phenomenon has been discussed in the literature for both Wilson~\cite{El-Khadra:1996wdx,El-Khadra:1997kgw} and staggered fermions.
Consider a quark bilinear containing a heavy quark $h$ and a generic (heavy or light) quark $x$.
Arguments from leading-order HQET~\citep{Bazavov:2017lyh} show that the conventional normalization can be restored, at leading order, by multiplying matrix elements containing the bilinear by the factor $Z^{\rm HQET,\,LO}_{hx}$
\begin{align}
    \widetilde{\mathcal{C}h}(am_1)
    &= \cosh(am_1) \left(1 - \frac{1}{2}\mathbb{N}(am_1) \sinh^2(a m_1) \right), \\
    Z^{\rm HQET,\,LO}_{hx}
    &=\begin{cases}
    \sqrt{\widetilde{\mathcal{C}h}(am_{1,h})\,\widetilde{\mathcal{C}h}(am_{1,x})}, \quad x \text{ nonrelativistic}\\
    \sqrt{\widetilde{\mathcal{C}h}(am_{1,h})},\phantom{\,\widetilde{\mathcal{C}h}(am_{1,x})} \quad x \text{ ultrarelativistic}
    \end{cases}\hspace{-1em}.
\end{align}
Residual discretization effects from next-to-leading HQET appear at order $x_h^4$ and $\alpha_s x_h^2$, where $x_h$ is the parameter linear in the heavy quark mass defined in \cref{eq:xh}.

\section{Shrinkage of covariance and correlation matrices \label{sec:shrinkage}}

Analysis of highly correlated Monte Carlo data encountered in lattice gauge theory presents a formidable statistical challenge.
Many problems are phrased in terms of least-squares minimization of a suitable $\chi^2$ function.
Examples in the present work include the correlator analysis in \cref{sec:correlator_analysis} to extract energies and matrix elements and the chiral-continuum fits of \cref{sec:chiral_ctm}.
The essential difficulty is that covariance matrix appearing in the $\chi^2$ functions can be nearly singular, and the small eigenvalues in the sample covariance matrix are poorly determined.
Shrinkage estimators, which we review here, are a class of tools for improving the sample covariance matrix by ``regulating'' the small eigenvalues.
For motivation, we follow closely the discussion and notation of Ref.~\cite{Ledoit:2004}, beginning with a technical result.

\begin{lemma}[Ledoit and Wolf]
Let $M$ be a real, symmetric matrix.
The eigenvalues are the most dispersed diagonal elements obtainable by rotation.
\end{lemma}
\begin{proof}
Consider a real, symmetric $p\times p$ matrix $M$.
Let $R \in \text{SO}(p)$ be a rotation, under which $M$ transforms into $R^T M R$.
The average of the eigenvalue spectrum $\avg{\lambda} \equiv (1/p) \Tr[M]$ is clearly invariant under rotations.
Let $\bm{v}_i$ denote the $i^\text{th}$ column of the rotation $R$.
The $i^\text{th}$ diagonal element of $R^T M R$ is $\bm{v}_i^T M \bm{v}_i$, and the dispersion of the diagonal elements around the average of the spectrum is defined via
\begin{equation}
    \frac{1}{p} \sum_{i}^{p} \left( \bm{v}_i^T M \bm{v}_i - \avg{\lambda} \right)^2.
\end{equation}
This expression is not invariant under rotations, but a closely related quantity is:
\begin{equation}
    \Tr[ (R^T M R - \avg{\lambda} I)^2 ]
    = \frac{1}{p} \sum_{i}^{p} \left( \bm{v}_i^T M \bm{v}_i - \avg{\lambda} \right)^2
    + \sum_{i=1}^{p} \sum_{\substack{j=1 \\ j\neq i}}^{p} \left(\bm{v}_i^T M \bm{v}_j\right)^2.
\end{equation}
The second term on the right-hand side is non-negative and vanishes precisely when the rotation $R$ diagonalizes $M$.
In other words, since the left-hand side is constant, the dispersion is maximized when the eigenvalues of $M$ appear on the diagonals of $R^T M R$, which was to be shown.
\end{proof}

This result has important consequence for the near-singular covariance matrices encountered in practical problems.
Let $\diag({\bm{\lambda}}) = U^T \Sigma U$ denote the spectral decomposition of the ``true'' population covariance matrix of a statistical distribution, where $U$ contains the eigenvectors and $\bm{\lambda}$ are the eigenvalues.
The corresponding sample covariance matrix has decomposition $\diag({\bm{\lambda}_n}) = U_n^T S_n U_n$.
As usual, $S_n$ is an unbiased estimator of $\Sigma$.
Therefore, $U^T S_n U$ is also an unbiased estimator of $\diag{\bm{\lambda}}$.
Unfortunately, one does not typically have access to the population eigenvectors of $U$ and is instead obliged to work with the sample estimates of $U_n$.
As the preceding lemma makes clear, the sample eigenvalues $\bm{\lambda}_n$ will be more widely dispersed than those of the population $\bm{\lambda}$.
Indeed, $\bm{\lambda}_n$ is \emph{not} an unbiased estimator of $U^T \Sigma U$ due to correlations between the eigenvectors in $U_n$ and the eigenvalues in $\bm{\lambda}_n$.
The general idea behind shrinkage estimators is to apply some function which decreases the dispersion of the sample eigenvalues $\bm{\lambda}_n$ to better approximate the population $\bm{\lambda}$.

The remainder of this appendix is organized as follows.
\Cref{ssec:linear_shrinakge} describes linear shrinkage, which is used in the chiral-continuum analysis (cf.~\cref{ssec:chiral_ctm_fits}).
\Cref{ssec:nonlinear_shrinkage} describes nonlinear shrinkage, which is used in the correlator fits (cf.~\cref{sec:correlator_analysis}).

\subsection{Linear shrinkage\label{ssec:linear_shrinakge}}

Linear shrinkage was introduced by Ledoit and Wolf in Ref.~\cite{Ledoit:2004}.
There seems to be some knowledge of this technique in the recent lattice-gauge-theory literature~\cite{Rinaldi:2019thf}.
Because lattice data often vary over many orders of magnitude, it is common to invert the correlation matrix instead of the covariance matrix, with shrinkage techniques being applied to them instead.

The linear shrinkage estimator $\hat{C}_n$ is defined as the convex sum of two matrices:
\begin{equation}
    \hat{C}_n = (1 - \lambda) C_n + \lambda C_{\rm target},
\end{equation}
with $\lambda \in [0,1]$.
As the parameter $\lambda$ is varied, the shrinkage estimator smoothly interpolates between the sample correlation matrix $C_n$ and the target matrix $C_{\rm target}$.
Many options are possible for $C_{\rm target}$.
Examples in the literature~\cite{Ledoit:2004,Rinaldi:2019thf} advocate using the identity matrix as the shrinkage target.
The idea is that suppressing the correlations by a small amount (say, $\lambda = 0.05$ or $0.1$) will correct the small eigenvalues while preserving the rest of the correlated structure to the data.\footnote{
As discussed in the main text, the preferred value of $\lambda=0.1$ was chosen to regulate the small eigenvalues (thus giving good fits) with unnecessarily discarding correlations, which can also cause fit quality to degrade.
Ultimately, our results are insensitive to the precise choice of $\lambda$, as shown in \cref{fig:d2pi_stability,fig:d2k_stability,fig:ds2k_stability}.
}
Besides using the identity matrix, our analysis also experimented with block-diagonal matrices (e.g., to retain the full correlations between different momenta at fixed valence mass).
The more complicated choices did not improve fit results compared with the simpler choice of the identity matrix.
The preferred chiral-continuum analysis of \cref{sec:chiral_ctm} therefore uses only the identity matrix.
Once a shrinkage estimator for the correlation matrix has been chosen, the corresponding covariance matrix follows in the usual way,
\begin{equation}
    \hat{S}_n = \diag(\bm{\sigma}) \hat{C}_n \diag(\bm{\sigma}),
\end{equation}
where $\bm{\sigma}$ is a vector containing the standard deviations.
The shrinkage estimator, which enjoys a smaller condition number and approximates the population covariance matrix better than the sample estimate,
is then inverted to give $\hat{S}_n^{-1}$, which is used in our fits.

\subsection{Nonlinear shrinkage\label{ssec:nonlinear_shrinkage}}

Nonlinear shrinkage has been described by Ledoit and Wolf~\cite{Ledoit:2018}, whose notation and presentation we follow closely.
A complete theoretical justification exceeds the scope of the work; the interested reader is invited to consult the original paper for proofs, additional references, and numerical evidence supporting the applicability in realistic finite data.
To keep the present work self-contained, we restrict ourselves to reproducing the required formulae with some discussion.

Suppose the sample covariance matrix $S_n$ is computed from $n$ observations of $p$ total random variables.
Consider the diagonalization of this matrix, $S_n = U_n^T \diag(\bm{\lambda}_n) U_n$.
Individual eigenvalues are denoted $\lambda_{n,i}$, $i\in\{1,\dots,p\}$ and, 
without loss of generality, may be supposed to be sorted in ascending order.
For large $p$ and $n$, suppose the eigenvalues follow some asymptotic cumulative distribution function $F(x)$ with associated spectral density $f(x) = F'(x)$.
Nonlinear shrinkage is a method for adjusting the empirical spectral density locally to improve the spread in eigenvalues for finite $n$.

Nonlinear shrinkage is based on the Hilbert transform, which maps continuous real functions $g(x)$ to $\mathcal{H}_g(x)$ via
\begin{equation}
    \mathcal{H}_g(x) \equiv \frac{1}{\pi} \PV \int_{-\infty}^\infty dx'\,\frac{g(x')}{x' - x}\,,
\end{equation}
where $\PV$ denotes the Cauchy principal value.
Conceptually, and as described at length in Ref.~\cite{Ledoit:2018}, the Hilbert transform acts like a local attractor, pulling eigenvalues towards regions of greater density.
Define the oracle function 
\begin{align}
    d(x) &\equiv \frac{x}{[\pi c x f(x)]^2 + [1 - c - \pi c x \mathcal{H}_f(x)]} \nonumber \\
        &= \frac{x}{1 + c^2 [ \varphi(x)^2 + \mathcal{H}_\varphi(x)^2] -2 c \mathcal{H}_\varphi(x)}\,,
\end{align}
where $c\equiv p/n$ is the concentration ratio, $\varphi(x) = \pi x f(x)$ and $\mathcal{H}_\varphi(x) = 1 + \pi x \mathcal{H}_f(x)$ is its Hilbert transform.
Given a set of sample eigenvalues $\bm{\lambda}_n$, $d(\bm{\lambda}_n)$ provides a shrinkage estimator.
To see this, first observe that as the number of samples becomes large ($c\to 0$), no shrinkage occurs ($d(x) \to x$), in agreement with intuition.
For small but finite concentration, the linear term in the denominator will dominate:
\begin{equation}
    d(x) \approx x \left[1 + 2 c \mathcal{H}_\varphi(x) + \order{c^2} \right].
\end{equation}
Since the Hilbert transform attracts eigenvalues, anomalously large or small eigenvalues will be pulled locally toward regions of higher density, shrinking the spectrum.
The same qualitative behavior is also present for generic $c$, as described in Ref.~\cite{Ledoit:2018}.
For some given finite data set, the underlying distributions $F(x)$ and $f(x)$ are typically unknown.
Moreover, since neither the empirical density nor the empirical CDF are continuous (the former is a sum of $\delta$ functions, one at each eigenvalue), the necessary Hilbert transform does not exist.
Instead, one works with a kernel estimator for $f(x)$, for which the necessary derivatives do exist:
\begin{align}
    \tilde{f}_n(x)
    &= \frac{1}{p} \sum_{i=1}^p \frac{1}{h_{n,i}} k\left(\frac{x-\lambda_{n,i}}{h_{n,i}} \right),\\
    \mathcal{H}_{\tilde{f}_n}(x)
    &=\frac{1}{p} \sum_{i=1}^p \frac{1}{h_{n,i}}\mathcal{H}_k\left(\frac{x-\lambda_{n,i}}{h_{n,i}} \right),
\end{align}
where $h_{n,i} \equiv \lambda_{n,i} h_n$ for some suitable choice of bandwidth $h_n$.
In principle, many possibilities exist for the choice of the kernel function $k$.
In practice, it is advantageous to take a kernel with finite support and an analytically calculable Hilbert transform.
Reference~\cite{Ledoit:2018} advocates choosing the Wigner semicircle distribution,
\begin{align}
    k(x)             &= \frac{\sqrt{[4-x^2]^+}}{2\pi}, \\
    \mathcal{H}_k(x) &= \frac{\sgn(x)\sqrt{[4-x^2]^+} - x}{2\pi},
\end{align}
where $[x]^+ \equiv \max\{0, x\}$ for any $x\in \mathbb{R}$.
With this choice, the kernel estimators $\tilde{f}_n$ and $\mathcal{H}_{\tilde{f}_n}$ take the following form when evaluated at the eigenvalues:
\begin{align}
    \tilde{f}_n(\lambda_{n,i})
    &=\frac{1}{p} \sum_{i=1}^p \frac{\sqrt{4\lambda_{n,j}^2 h_n^2 - (\lambda_{n,i} - \lambda_{n,j})^2}}{2\pi\lambda_{n,j}^2 h_n^2}, \\
    \mathcal{H}_{\tilde{f}_n}(\lambda_{n,i})
    &=\frac{1}{p} \sum_{i=1}^p 
    \frac
        {\sgn(\lambda_{n,i} - \lambda_{n,j})
        \sqrt{[(\lambda_{n,i} - \lambda_{n,j})^2 - 4\lambda_{n,j}^2 h_n^2]^+} - \lambda_{n,i} + \lambda_{n,j}}
        {2\pi \lambda_{n,j}^2 h_n^2}.
\end{align}
Likewise, the sample estimator for the oracle function becomes 
\begin{equation}
    \tilde{d}_{n,i}
        = \frac{\lambda_{n,i}}{\left[\pi c \lambda_{n,i} \tilde{f}_n(\lambda_{n,i})\right]^2 
        + \left[1 - c + \pi c\lambda_{n,i} \mathcal{H}_{\tilde{f}_n}(\lambda_{n,i})\right]^2}.
\end{equation}
Some freedom exists in the choice of bandwidth.
For reasons of statistical convergence, {\it {\it i.e.}}, so that $\tilde{f}_n(x) \to f(x)$ and $\mathcal{H}_{\tilde{f}_n}(x) \to \mathcal{H}_f(x)$ uniformly in probability, Ref.~\cite{Ledoit:2018} argues that the bandwidth should vanish for large $n$ ($\lim_{n\to\infty} h_n = 0$) but not decrease too quickly ($\lim_{n\to\infty} n h_n^{5/2} = 0$).
We follow their recommendation of choosing $h_n \equiv n^{-0.35}$.

After shrinkage is applied, the new ``eigenvalues'' $\tilde{d}_{n,i}$ computed from $\lambda_{n,i}$ are not guaranteed to maintain their ascending order.
For this reason, the penultimate step is to restore ascending order by applying the pool adjacent violators (PAV) algorithm~\cite{Ayer:1955,Ledoit:2018}.
Finally, the shrinkage estimator for the sample covariance matrix is given by
\begin{align}
    \hat{\bm{d}}_n &\equiv \PAV(\tilde{\bm{d}}_{n}), \\
    \hat{S}_n &\equiv U_n \diag(\hat{\bm{d}}_n) U_n^T.
\end{align}
As above, the shrinkage estimator is then inverted, and $\hat{S}_n^{-1}$ is used in our fits.

The PAV algorithm is as follows.
Given an input set of data $\bm{d}$, the algorithm iteratively updates the values, locally pooling adjacent values which violate $d_i \ge d_{i+1}$, and replacing them with their average.
The process is repeated until the monotonicity condition is satisfied everywhere, yielding $\PAV(\bm{d})$.

Included in the supplementary material (in \texttt{shrink.py}) is a python implementation of the nonlinear shrinkage algorithm.

\subsection{Numerical examples of shrinkage}

\begin{figure}
    \centering
    \includegraphics[width=1.0\textwidth]{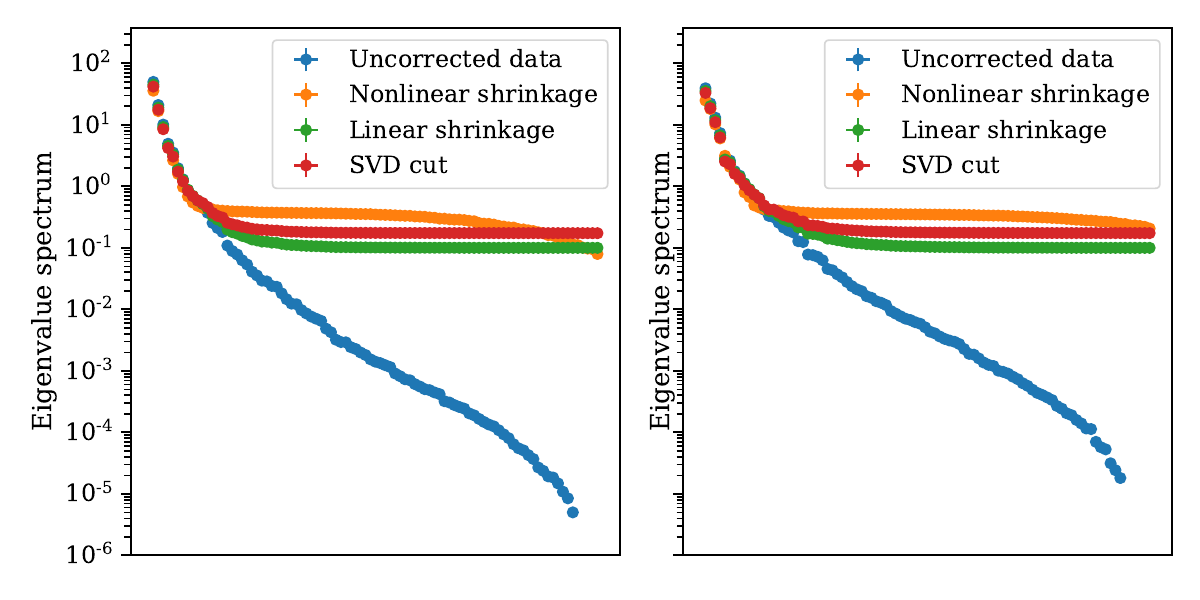}
    \caption{Comparison of eigenvalue spectra resulting
    before and after shrinkage or an SVD cut, for
    the correlation matrices for
    $C_{D_s}^{P}(t)$ (left) and $C_K^P(t, \bm{0})$ (right) on the physical-mass $a\approx 0.06 \fm$ ensemble.
    Linear shrinkage was applied with $\lambda=0.1$.
    An SVD cut of $10^{-3}$ was chosen to have an effect on the spectrum similar to shrinkage.}
    \label{fig:shrinkage_spectra}
\end{figure}

In this section, we present representative examples of correlation matrices appearing in our analysis.
For concreteness, we consider the correlation matrices for the two-point functions $C_{D_s}^{P}(t)$ and $C_K^P(t, \bm{0})$
(cf.\ \cref{eq:c2_heavy_P,eq:c2_light_P}) associated with the $D_s$ and $K$ mesons on the physical-mass $a\approx 0.06\fm$ ensemble.
The eigenvalue spectra associated with the correlation matrices are shown in \cref{fig:shrinkage_spectra},
for raw data, nonlinear shrinkage, linear shrinkage with $\lambda=0.1$, and an SVD cut of $10^{-3}$.
\footnote{
Some freedom exists in the implementation of an SVD cut.
One possibilty is setting to zero all eigenvalues below some threshold.
Instead, the method used for comparison in this appendix compares all the eigenvalues to the largest eigenvalue, $\lambda_{\rm max}$.
All eigenvalues below the threshold  \texttt{svdcut}$\times \lambda_{\rm max}$ are replaced by this value.
Theoretical and practical aspects of this convenction for SVD cuts are described in Ref.~\cite{Dowdall:2019bea}.
}
The raw spectra, shown in blue, span a range of roughly eight orders of magnitude. (In fact, not displayed are the last few eigenvalues, which are consistent with zero at double precision).
For the given parameter choices, linear shrinkage and the SVD cut give similar results.
With nonlinear shrinkage, the shape of the small-eigenvalue region of the spectrum retains some of its original curvature.
In the case of the kaon (left in \cref{fig:shrinkage_spectra}), the small eigenvalues from nonlinear shrinkage vary by approximately an order of magnitude over the region where they are roughly constant for linear shrinkage and SVD cut.

These methods all alter the covariance between pairs of data.
\Cref{fig:shrinkage_kaon_correlation_matrices,fig:shrinkage_ds_correlation_matrices} show heat maps for the corresponding correlation matrices.
As with the eigenvalue spectra in \cref{fig:shrinkage_spectra}, the results for linear shrinkage and SVD cut are qualitatively similar.
Compared with the other methods, nonlinear shrinkage tends to smooth the far off-diagonal correlation coefficients.
All three correction methods suppress the near-diagonal correlations which are nearly unity in the raw data.

Reference~\cite{Dowdall:2019bea} has argued that applying an SVD cut is a statistically conservative analysis choice, amounting to adding uncertainty to the data.
However, care must be given when interpreting the $\chi^2/{\rm DOF}$ when SVD cuts have been applied, since such cuts can result in artificially low values for the $\chi^2/{\rm DOF}$.
As described in Ref.~\cite{Dowdall:2019bea}, the standard diagnostic for this potential problem is to rerun fits with additional noise in the means, checking for the stability of posterior values and for the $\chi^2/{\rm DOF}$ to increase slightly but (at least for good fits) to remain of order unity.
Our analysis has carried out this check, with good stability observed.

\begin{figure}
    \centering
    \includegraphics[width=1.0\textwidth]{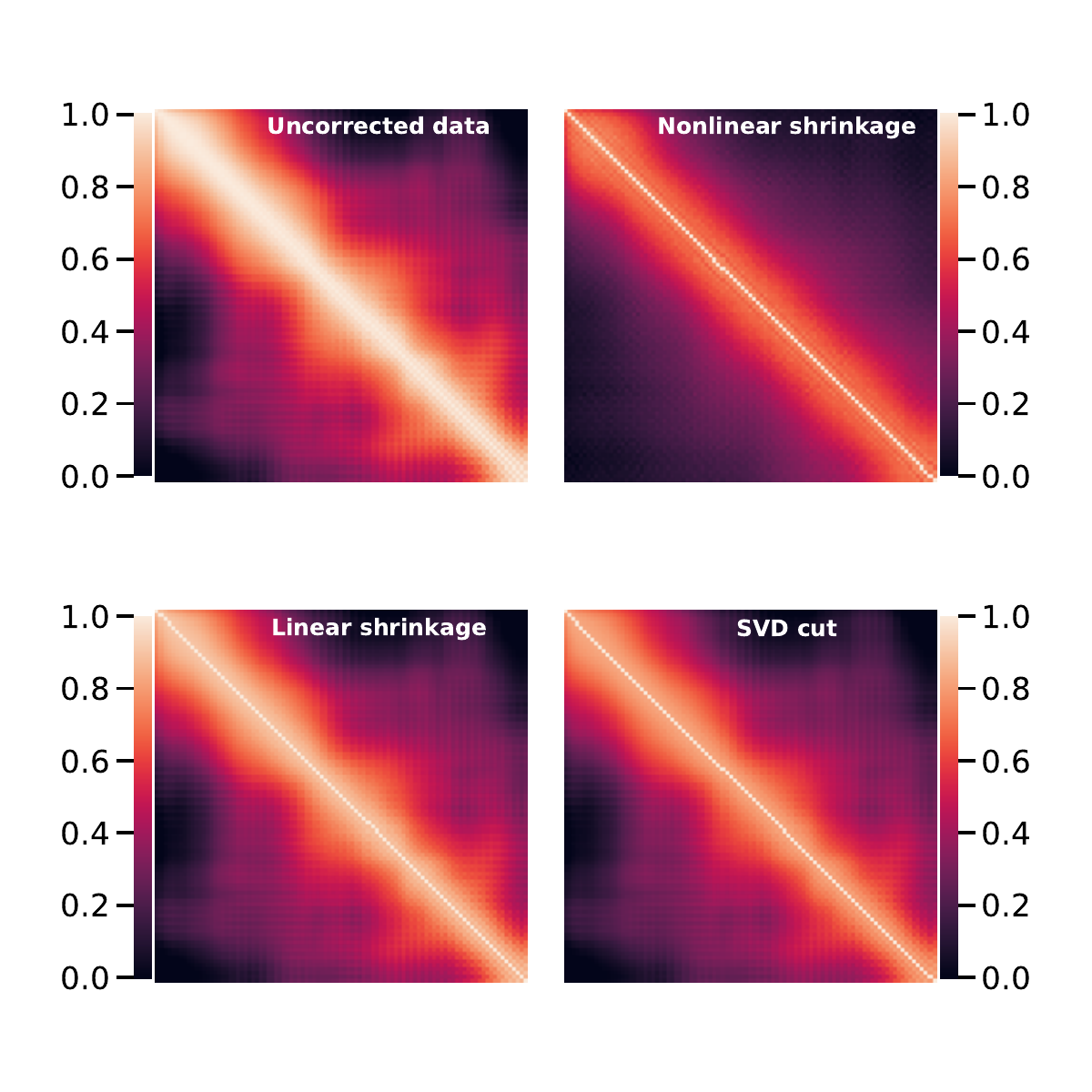}
    \caption{Comparison of correlation matrices resulting from different correction techniques applied to the zero-momentum $K$ two-point function $C_{K}^{P}(t, \bm{0})$ 
    on the physical-mass $a\approx 0.06\fm$ ensemble.
    The associated eigenvalue spectra are shown in \cref{fig:shrinkage_spectra}.
    \label{fig:shrinkage_kaon_correlation_matrices}
    }
\end{figure}

\begin{figure}
    \centering
    \includegraphics[width=1.0\textwidth]{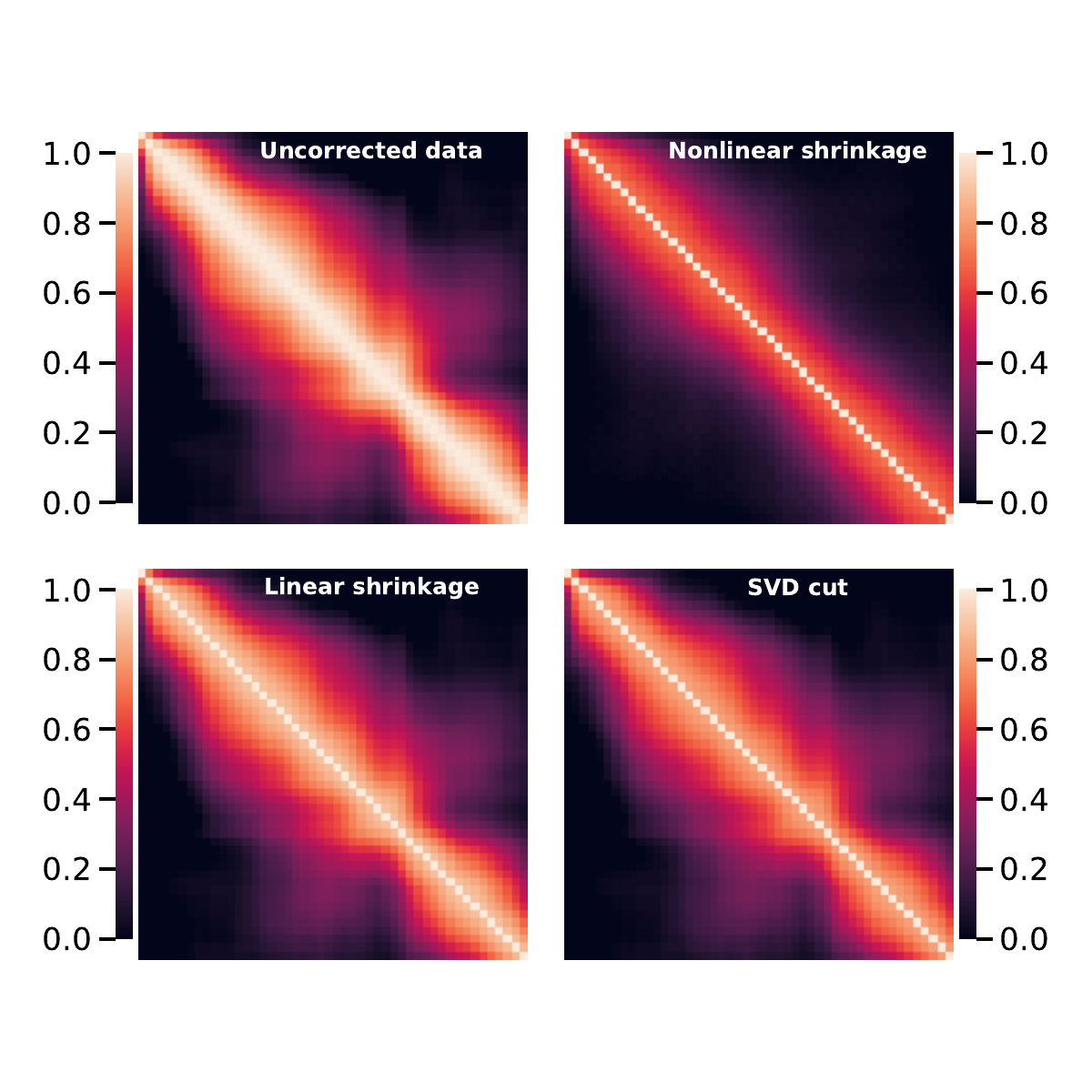}
    \caption{Comparison of correlation matrices resulting from different correction techniques applied to the $D_s$ two-point function $C_{D_s}^{P}(t)$ 
    on the physical-mass $a\approx 0.06\fm$ ensemble.
    The associated eigenvalue spectra are shown in \cref{fig:shrinkage_spectra}.
    \label{fig:shrinkage_ds_correlation_matrices}
    }
\end{figure}

\clearpage

\section{Fits: additional details and figures}
\label{app:extras}

This appendix gives additional details concerning the correlator, chiral-continuum, and $z$ expansion fits described in \cref{sec:correlator_analysis,ssec:chiral_ctm_fits,ssec:z-expansion}, and compiles figures that illustrate the robustness of our chiral-continuum analysis for $\DK$ and $\DsK$ decays. Analogous figures are included in the main text for $\Dpi$.

\subsection{Correlator fits}
\label{app:correlator_fits}

As introduced in \cref{sec:correlator_analysis}, the correlator fits must satisfy checks related to the ratios \cref{eq:ratio_s,eq:ratio_v4,eq:ratio_vi}.
\Cref{fig:d2pi_rbar} shows tests based on the ratio $R_0^{\Dpi}$ for the physical-mass $0.12\fm$ ensemble with the charm-quark mass approximately tuned to its physical value.
Similar figures are shown for the other decays and form factors in 
\cref{fig:d2pi_rbar_vs_tsnk,fig:d2k_rbar_vs_tsnk,fig:ds2k_rbar_vs_tsnk,fig:d2pi_rbar_vs_momentum,fig:d2k_rbar_vs_momentum,fig:ds2k_rbar_vs_momentum}.
The first test concerns the approach of the ratios $R_{0,\parallel,\perp}(t, T, \bm{p})$ to the asymptotic plateau region.
This behavior is examined in the top row of \cref{fig:d2pi_rbar} 
(and in \cref{fig:d2pi_rbar_vs_tsnk,fig:d2k_rbar_vs_tsnk,fig:ds2k_rbar_vs_tsnk})
by considering the ratios at fixed momentum as the source-sink separation is increased.
As $T$ increases, the data tend to flatten out as the ratio approaches the asymptotic limit.
In the right-hand panes, the data show the highest point,
\begin{align}
    \max_t R_{0,\parallel,\perp}(t, T, \bm{p}=2\pi(1,0,0)/N_sa) \label{eq:plateau},
\end{align}
as a convenient proxy for the value of ``plateau.''
As $T$ is increased, theses points gradually approach the form factor's fit posterior value, indicated by the horizontal band in both the left and right panes.
It bears emphasizing that the value of the form factor itself emerges from a fit to the spectral decomposition, \cref{eq:2pt_spectral_decomp_final,eq:2pt_spectral_decomp_initial,eq:3pt_spectral_decomp}, and therefore explicitly includes excited-state effects.

The third visual test checks the momentum dependence and is shown in the bottom row of \cref{fig:d2pi_rbar} (and in \cref{fig:d2pi_rbar_vs_momentum,fig:d2k_rbar_vs_momentum,fig:ds2k_rbar_vs_momentum}).
The left panel shows the ratio $R_0^{\Dpi}$, with each color corresponding to a different momentum.
The horizontal lines with matching colors show the central values of the posteriors for $f_0^{\Dpi}(\bm{p}^2)$.
For visual clarity, data are only shown for fixed $T_{\rm max}$, but all available source-sink separations $T$ were included in the fits.
Moving from top to bottom, the form factors fall monotonically with momentum, and the effects of excited states tend to decrease.
The bottom right panel shows the corresponding posterior values for $f_0^{\Dpi}(\bm{p}^2)$, which exhibit smooth dependence on the momentum.

\begin{figure}
    \centering
    \includegraphics[width=0.9\textwidth]{Figures/Rbar/d2pi_S_rbar_vs_tsnk.pdf}
    \includegraphics[width=0.9\textwidth]{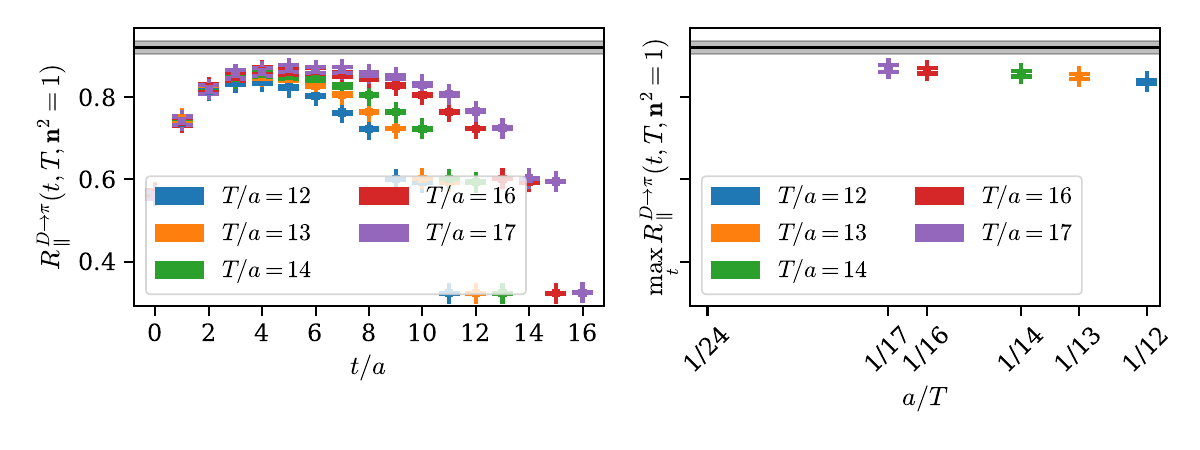}
    \includegraphics[width=0.9\textwidth]{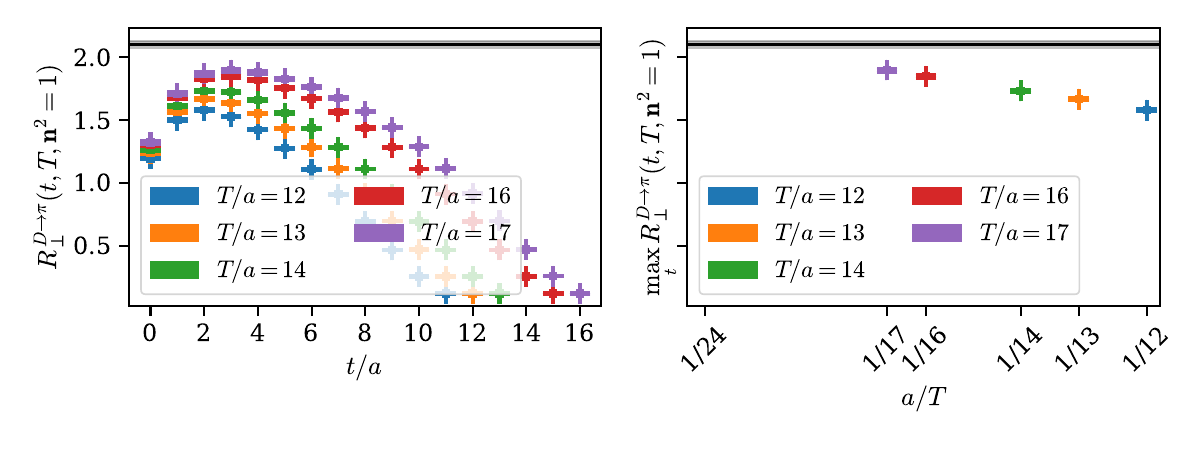}
    \caption{
    Comparing the ratios $R_{0,\parallel,\perp}^{D\to\pi}$,
    \cref{eq:ratio_s,eq:ratio_v4,eq:ratio_vi},
    with the form factor's fit posterior result at fixed momentum $\bm{p}=(1,0,0)$ on the physical-mass $0.12\fm$ ensemble.
    Left: The data are the ratios $R_{0,\parallel,\perp}^{D\to\pi}(t, T, \hat{p}^2=1)$, with each color corresponding to a different source-sink separation $T$.
    Right: The approach to the asymptotic plateau.
    Each point corresponds to the maximum point in the curves on the left, $\max_t R_{0,\parallel,\perp}^{D\to\pi}(t, T, \hat{p}^2=1)$.
    As the source-sink separation is increased, the data gradually approaches the form factor's posterior value given by the band.
    \label{fig:d2pi_rbar_vs_tsnk}
    }
\end{figure}

\begin{figure}
    \centering
    \includegraphics[width=0.9\textwidth]{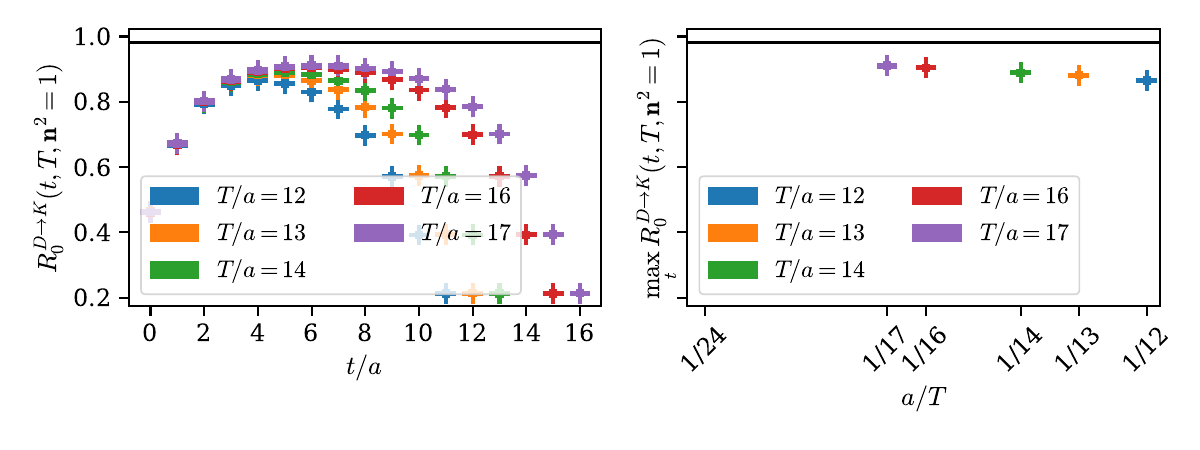}
    \includegraphics[width=0.9\textwidth]{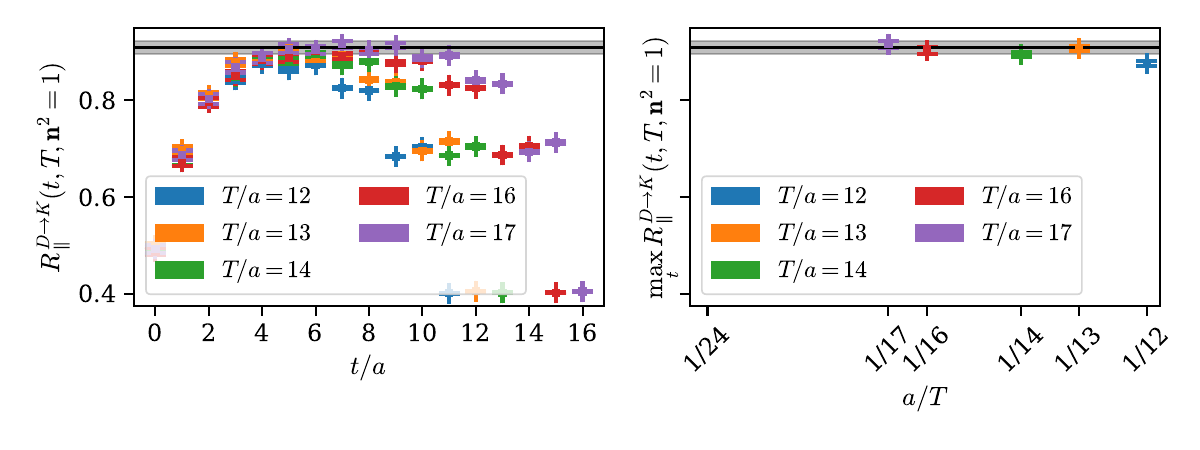}
    \includegraphics[width=0.9\textwidth]{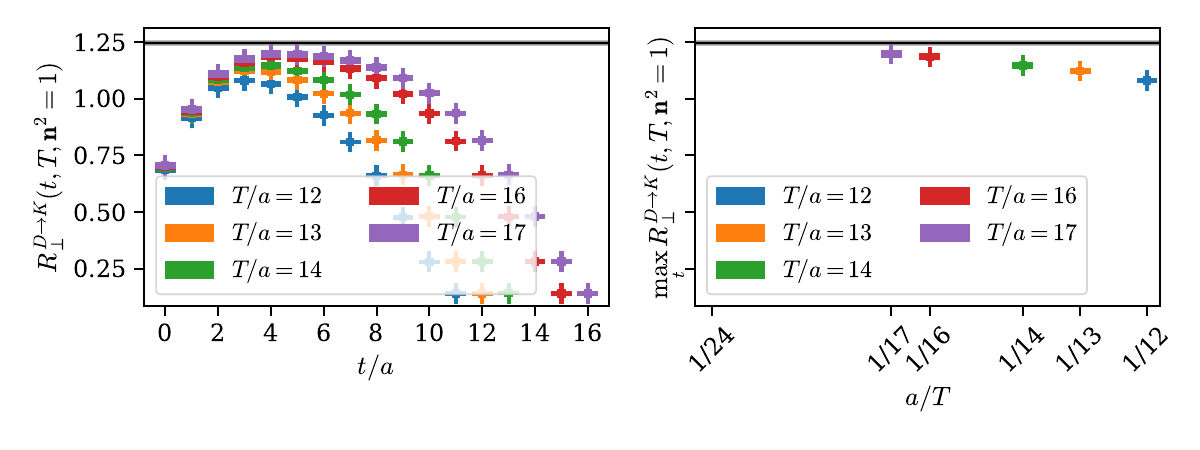}
    \caption{
    Comparing the ratios $R_{0,\parallel,\perp}^{D\to K}$,
    \cref{eq:ratio_s,eq:ratio_v4,eq:ratio_vi},
    with the form factor's fit posterior result at fixed momentum $\bm{p}=(1,0,0)$ on the physical-mass $0.12\fm$ ensemble.
    See the caption of \cref{fig:d2pi_rbar_vs_tsnk} for a detailed explanation.
    \label{fig:d2k_rbar_vs_tsnk}
    }
\end{figure}

\begin{figure}
    \centering
    \includegraphics[width=0.9\textwidth]{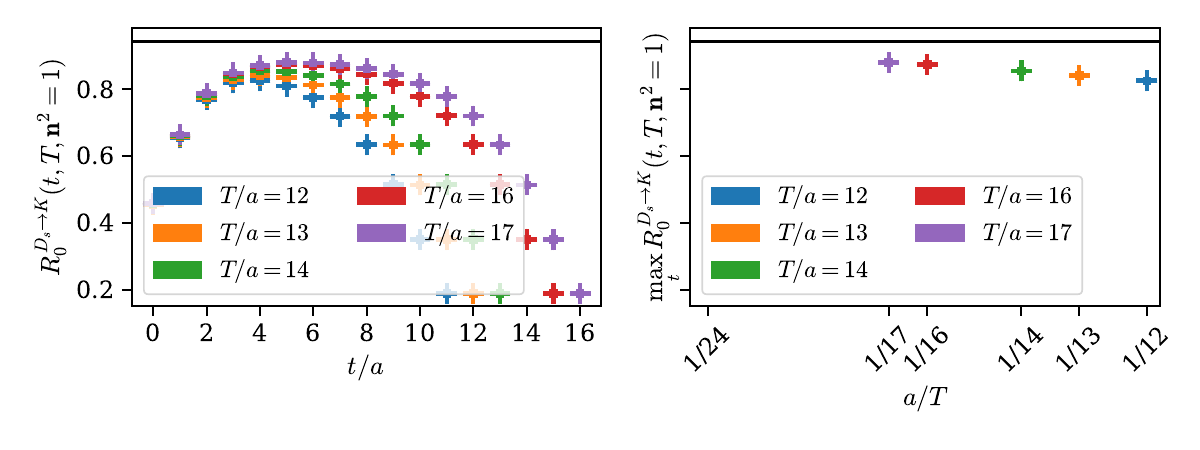}
    \includegraphics[width=0.9\textwidth]{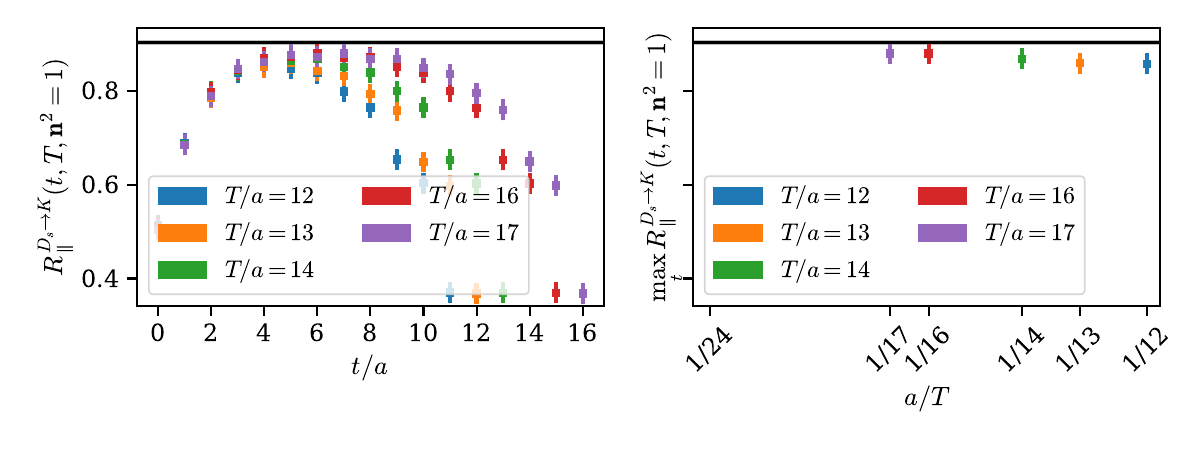}
    \includegraphics[width=0.9\textwidth]{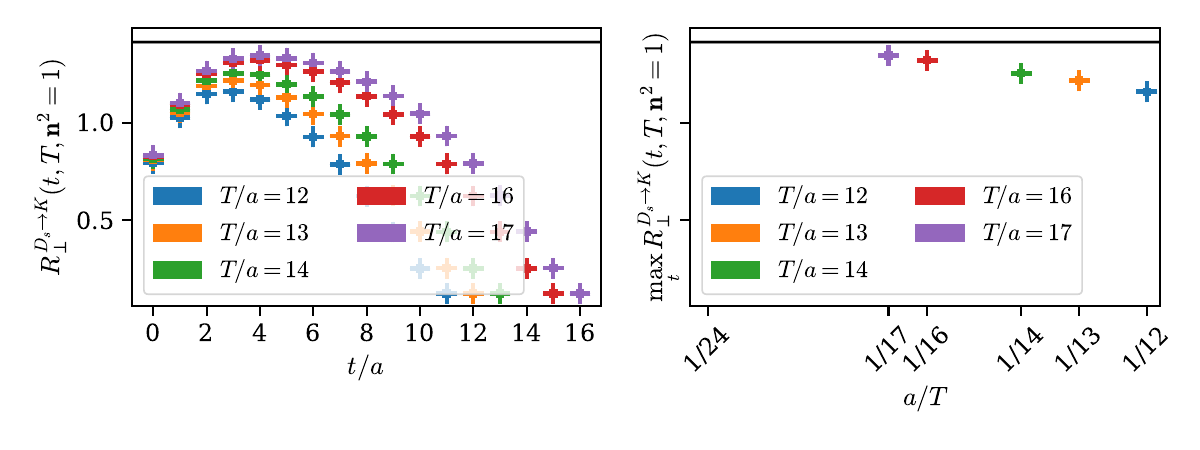}
    \caption{
    Comparing the ratios $R_{0,\parallel,\perp}^{D_s\to K}$,
    \cref{eq:ratio_s,eq:ratio_v4,eq:ratio_vi},
    with the form factor's fit posterior result at fixed momentum $\bm{p}=(1,0,0)$ on the physical-mass $0.12\fm$ ensemble.
    See the caption of \cref{fig:d2pi_rbar_vs_tsnk} for a detailed explanation.
    \label{fig:ds2k_rbar_vs_tsnk}
    }
\end{figure}

\begin{figure}
    \centering
    \includegraphics[width=0.9\textwidth]{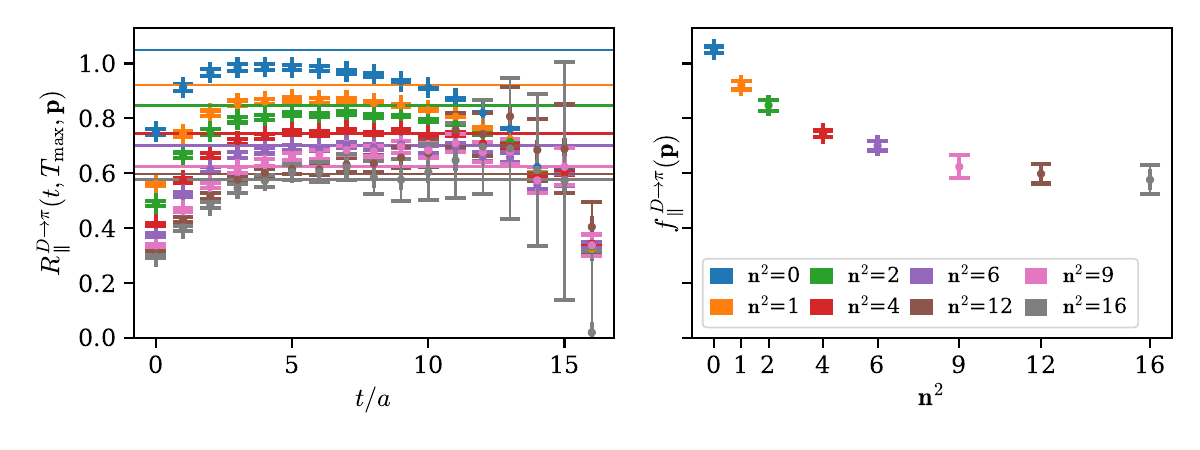}
    \includegraphics[width=0.9\textwidth]{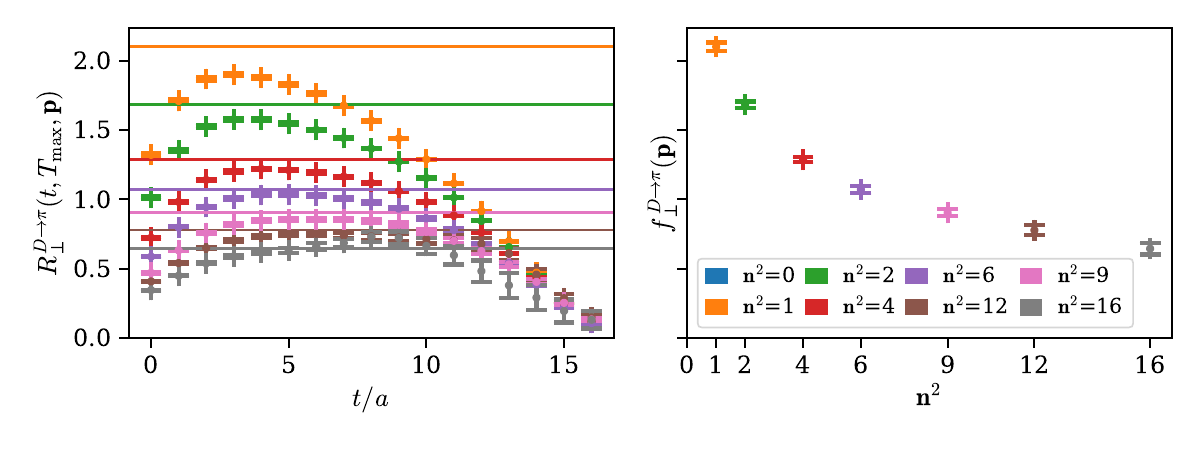}
    \caption{Comparing the ratios $R_{0,\parallel,\perp}^{D\to\pi}$,
    \cref{eq:ratio_s,eq:ratio_v4,eq:ratio_vi},
    with fit results for the form factors coming from the spectral decomposition on the physical-mass $0.12\fm$ ensemble.
    Left: The data are the ratios $R_{0,\parallel,\perp}^{D\to\pi}(t, T_{\rm max}, \bm{p})$, with each color corresponding to a different momentum.
    In each case, only the largest source-sink separation $T_{\rm max}$ is displayed.
    Horizontal lines denote the central values form factor's fit posterior values, coming from fits including all source-sink separations $T$.
    Right: The momentum dependence of the form factor's fit posterior values.
    \label{fig:d2pi_rbar_vs_momentum}
    }
\end{figure}

\begin{figure}
    \centering
    \includegraphics[width=1.0\textwidth]{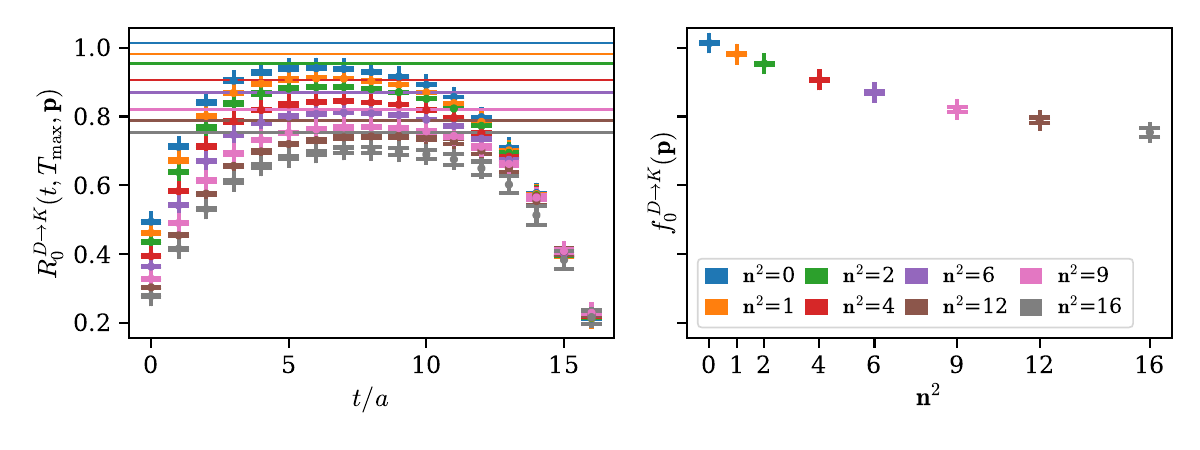}
    \includegraphics[width=1.0\textwidth]{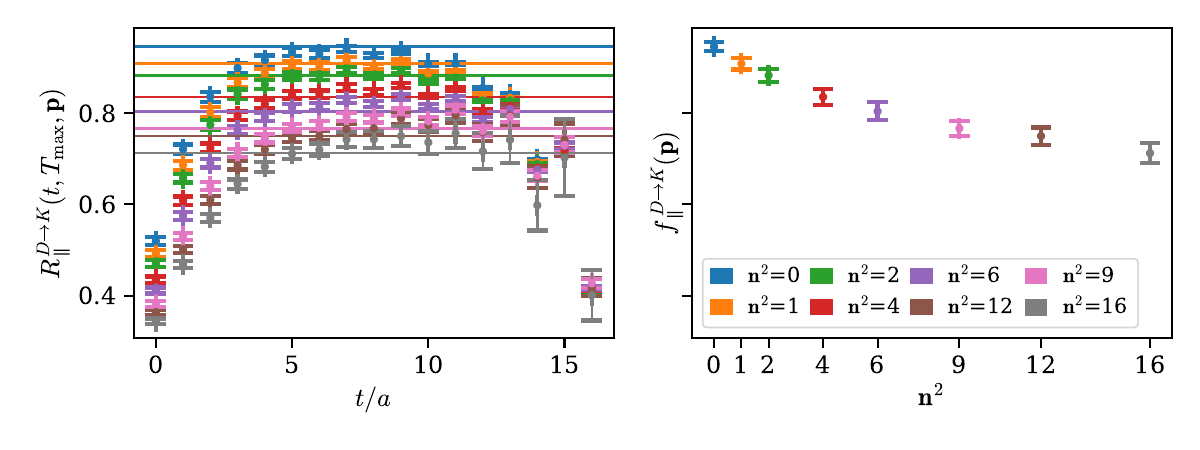}
    \includegraphics[width=1.0\textwidth]{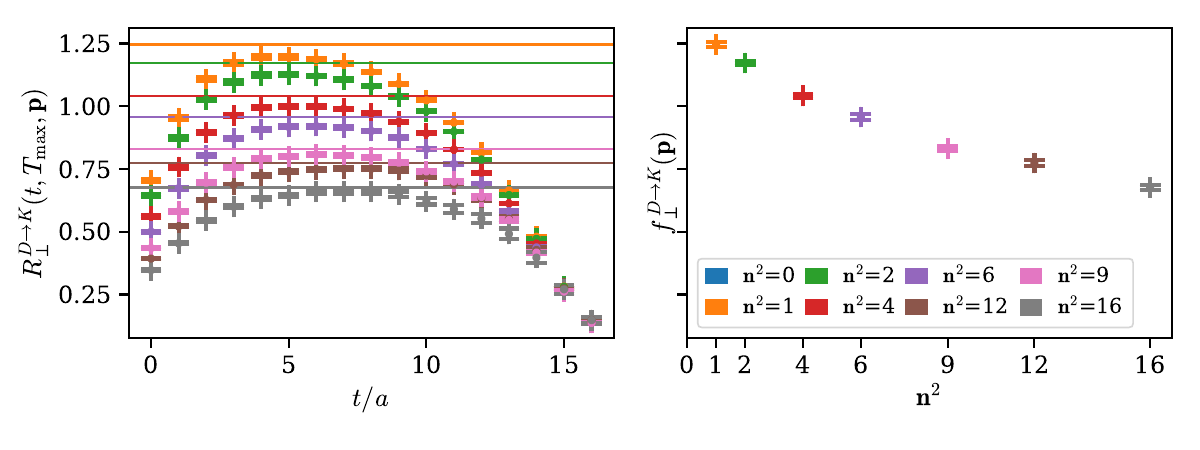}
    \caption{Comparing the ratios $R_{0,\parallel,\perp}^{D\to K}$,
    \cref{eq:ratio_s,eq:ratio_v4,eq:ratio_vi},
    with fit results for the form factors coming from the spectral decomposition on the physical-mass $0.12\fm$ ensemble.
    See the caption of \cref{fig:d2pi_rbar_vs_momentum} for a detailed explanation.
    \label{fig:d2k_rbar_vs_momentum}
    }
\end{figure}

\begin{figure}
    \centering
    \includegraphics[width=1.0\textwidth]{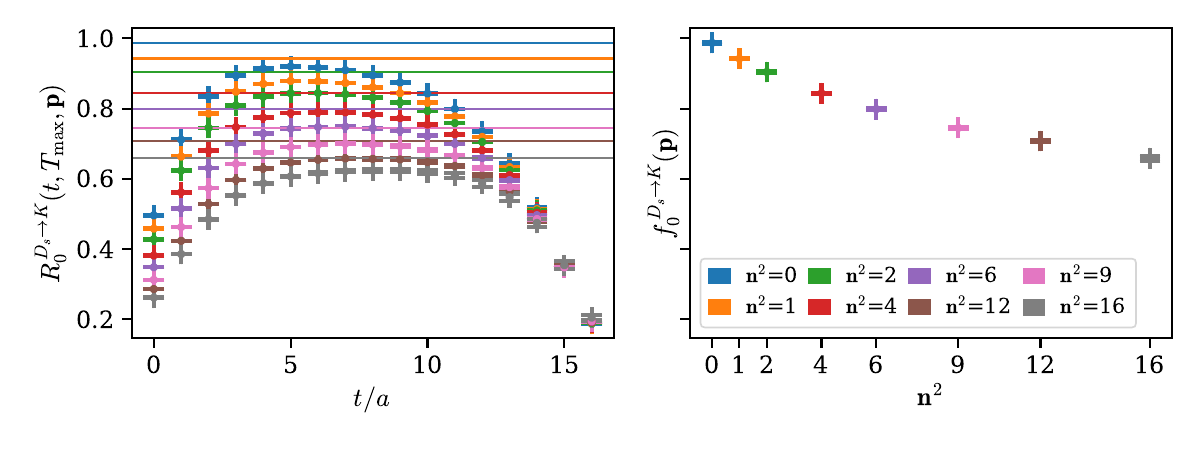}
    \includegraphics[width=1.0\textwidth]{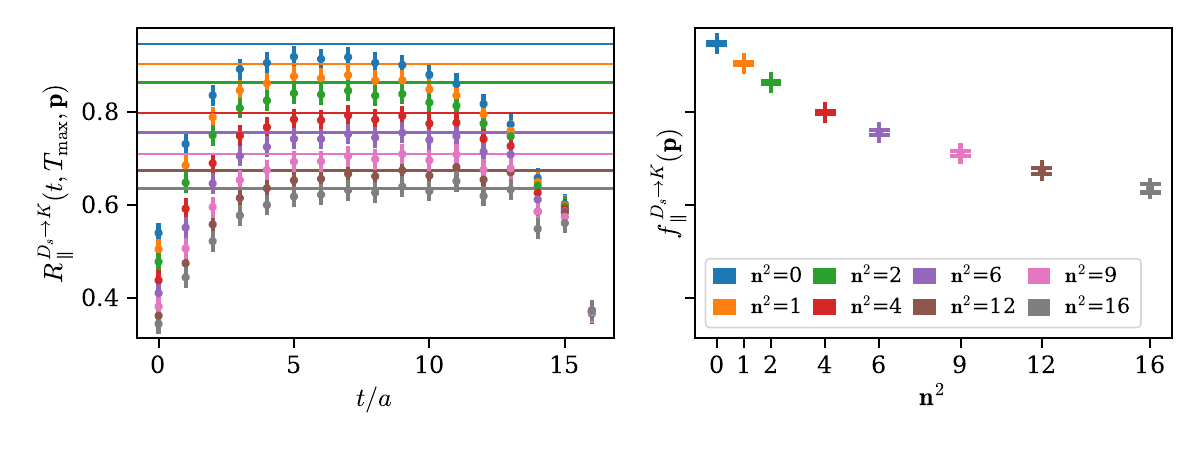}
    \includegraphics[width=1.0\textwidth]{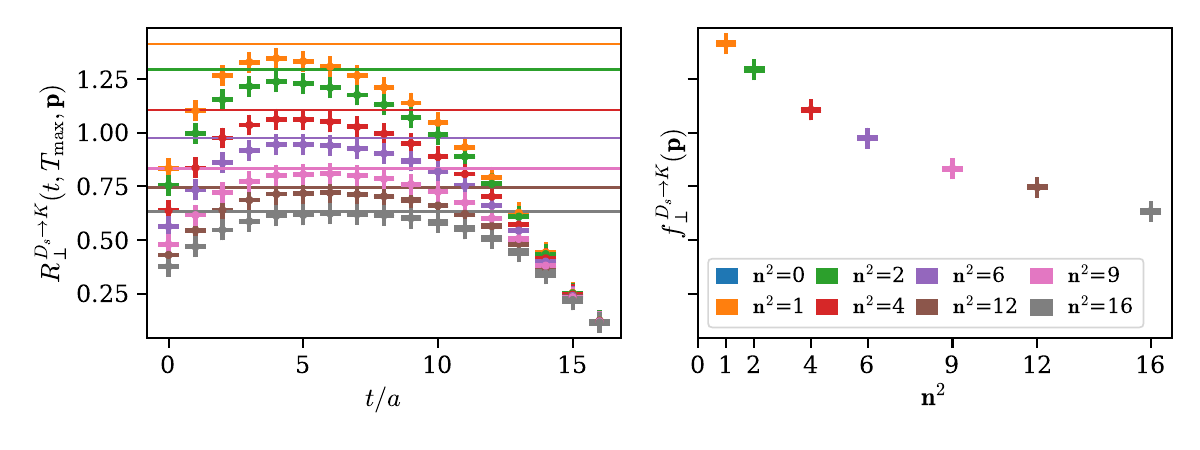}
    \caption{Comparing the ratios $R_{0,\parallel,\perp}^{D_s \to K}$,
    \cref{eq:ratio_s,eq:ratio_v4,eq:ratio_vi},
    with fit results for the form factors coming from the spectral decomposition on the physical-mass $0.12\fm$ ensemble.
    See the caption of \cref{fig:d2pi_rbar_vs_momentum} for a detailed explanation.
    \label{fig:ds2k_rbar_vs_momentum}
    }
\end{figure}

\clearpage

\subsection{Chiral-continuum fits: Results for \texorpdfstring{$\DK$}{D→K} and \texorpdfstring{$\DsK$}{Ds→K}}
\label{app:chipt_fits}


\begin{figure}[!htb]
    \centering
    \includegraphics[width=1.0\textwidth]{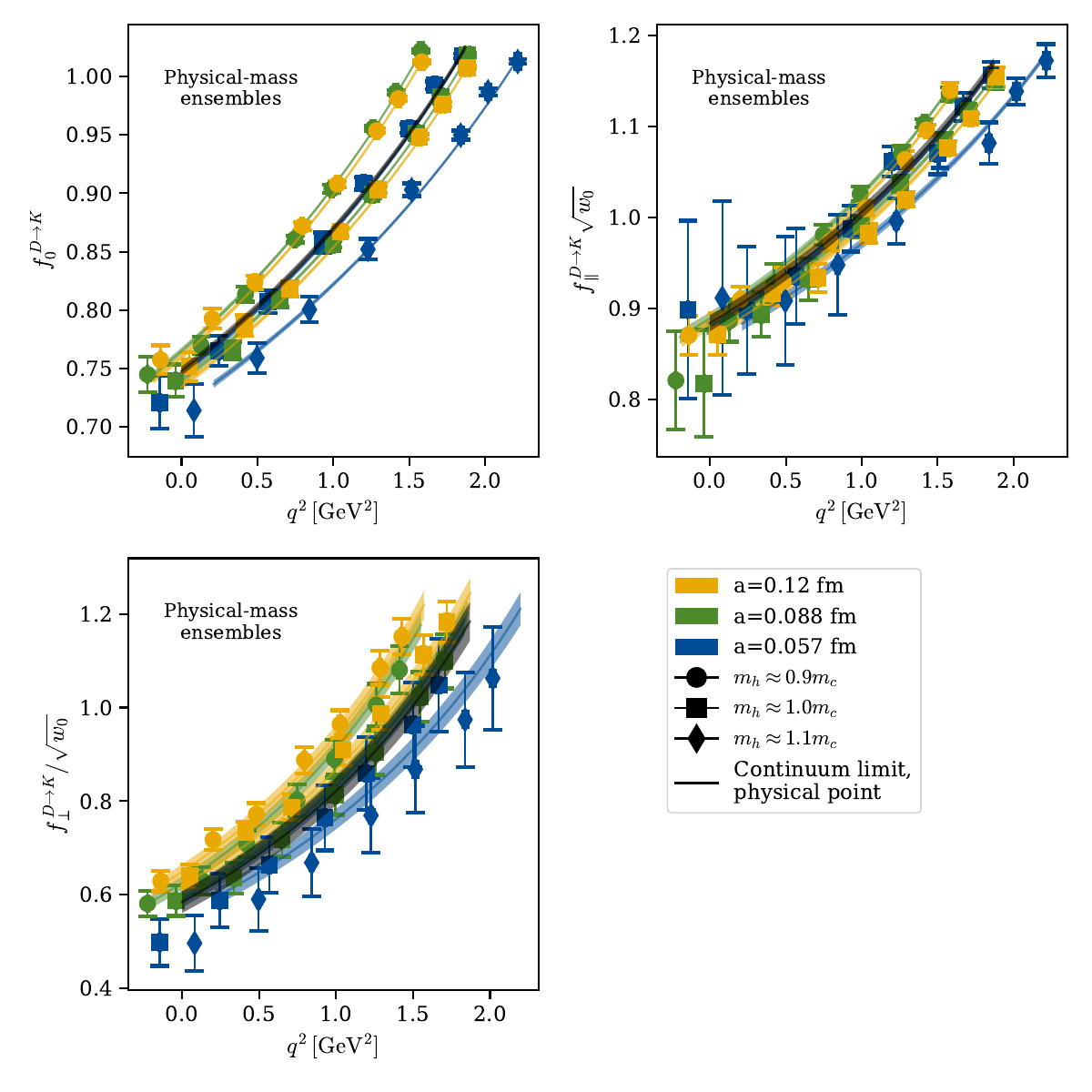}
    \caption{
    The result of the chiral-continuum fit for the $\DK$ form factors constructed using \cref{eq:f_perp,eq:f_parallel,eq:f_0} in units of the gradient-flow scale $w_0$.
    For visual clarity, only the physical-mass ensembles with heavy valence masses $m_h/m_c\in\{0.9, 1.0, 1.1\}$ are shown, although all ensembles in \cref{table:ensembles} were included in the fit.
    Points with $m_h/m_c\approx 1.1$ were only simulated on the $a\approx 0.06\fm$ ensemble.
    \label{fig:d2k_data_with_fit}}
\end{figure}


\begin{figure}[!htb]
    \centering
    \includegraphics[width=1.0\textwidth]{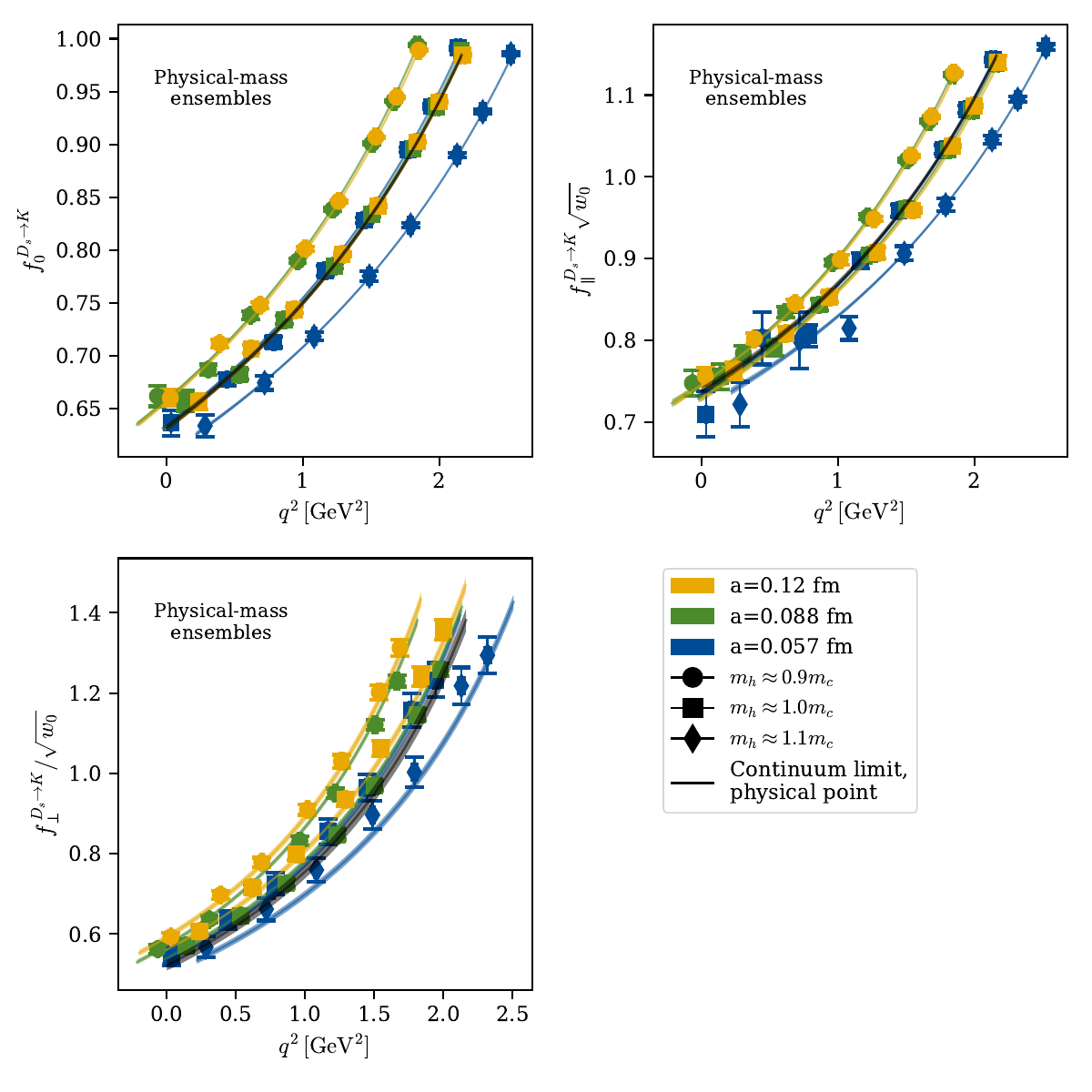}
    \caption{
    The result of the chiral-continuum fit for the $\DsK$ form factors  constructed using \cref{eq:f_perp,eq:f_parallel,eq:f_0} in units of the gradient-flow scale $w_0$.
    For visual clarity, only the physical-mass ensembles with heavy valence masses $m_h/m_c\in\{0.9, 1.0, 1.1\}$ are shown, although all ensembles in \cref{table:ensembles} were included in the fit.
    Points with $m_h/m_c\approx 1.1$ were only simulated on the $a\approx 0.06\fm$ ensemble.
    \label{fig:ds2k_data_with_fit}}
\end{figure}

\clearpage

\subsection{Chiral-continuum fits: Stability plots for \texorpdfstring{$\DK$}{D→K} and \texorpdfstring{$\DsK$}{Ds→K}}


\begin{figure}[!htb]
    \centering
    \includegraphics[width=0.75\textwidth]{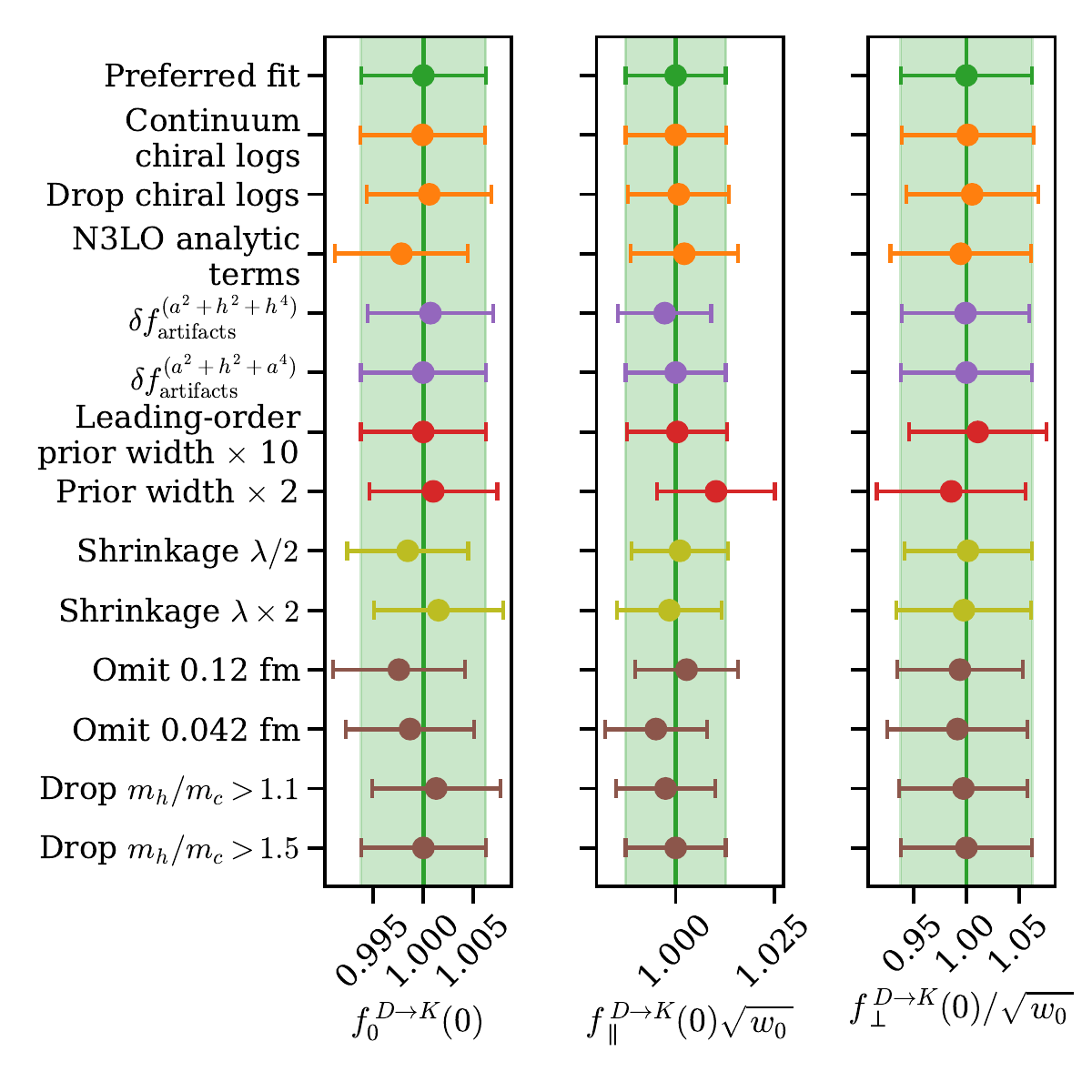}
\caption{
    Stability of the $\DK$ form factors $f_{\perp/\parallel/0}$ at $q^2=0$
    under variations to the EFT model, the model for discretization effects, 
    to the choice of data included in the fit, and other analysis choices as described in the main body. 
    The central values have been normalized by the central value of preferred fit in green. 
    All variations are statistically consistent with the preferred fit, highlighted by the green band in each panel.
    The statistical significance of the fits is indicated by the marker size, with larger points denoting better fits. 
    \label{fig:d2k_stability}}
\end{figure}


\begin{figure}[!htb]
    \centering
    \includegraphics[width=0.75\textwidth]{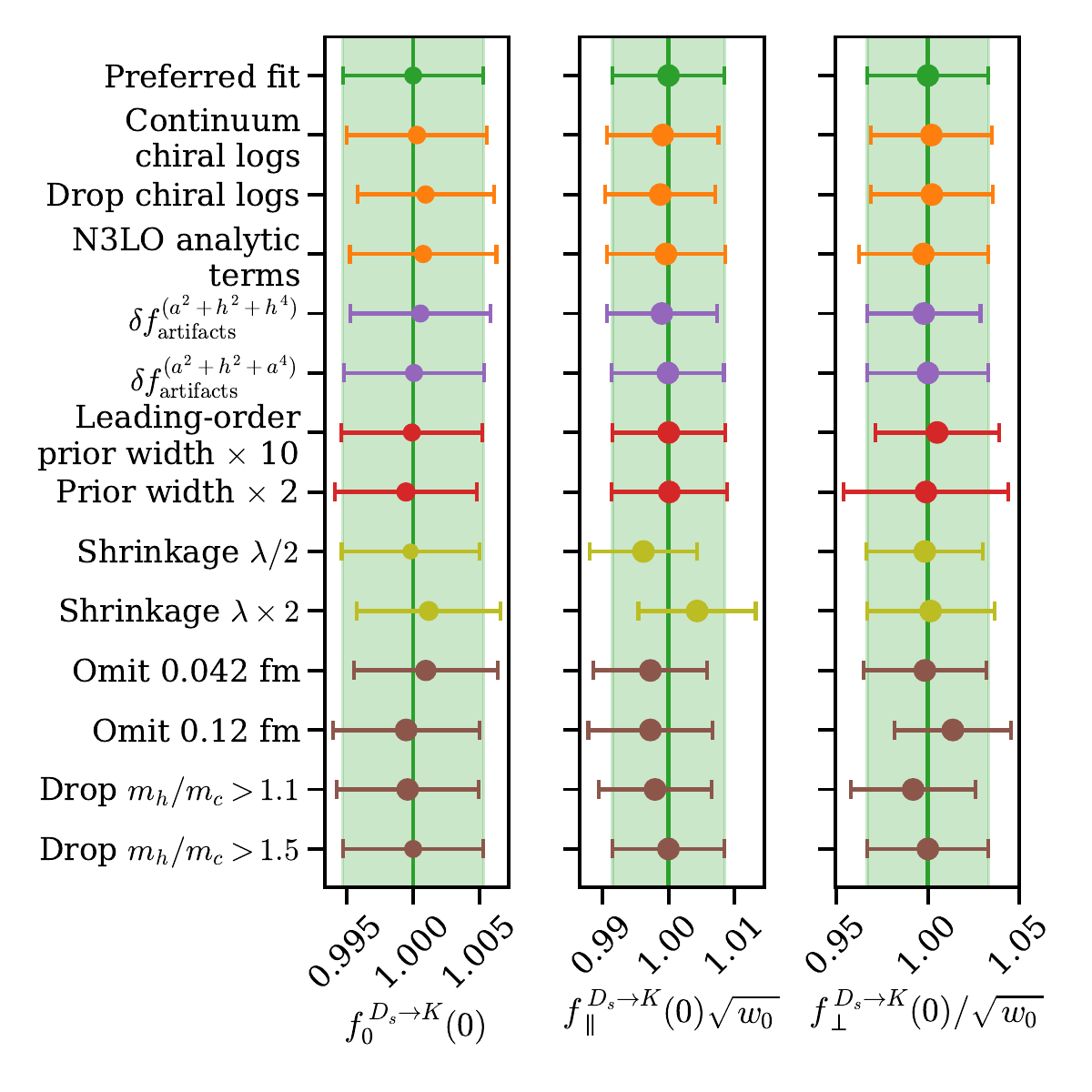}
\caption{
    Stability of the $\DsK$ form factors $f_{\perp/\parallel/0}$ at $q^2=0$
    under variations to the EFT model, the model for discretization effects, 
    to the choice of data included in the fit, and other analysis choices as described in the main body. 
    The central values have been normalized by the central value of preferred fit in green.
    All variations are statistically consistent with the preferred fit, highlighted by the green band in each panel.
    The statistical significance of the fits is indicated by the marker size, with larger points denoting better fits.
    \label{fig:ds2k_stability}}
\end{figure}

\clearpage
\subsection{Chiral continuum fits: Error breakdowns for \texorpdfstring{$\DK$}{D→K} and \texorpdfstring{$\DsK$}{Ds→K}}


\begin{figure}[!htb]
    \centering
    \includegraphics[width=1.0\textwidth]{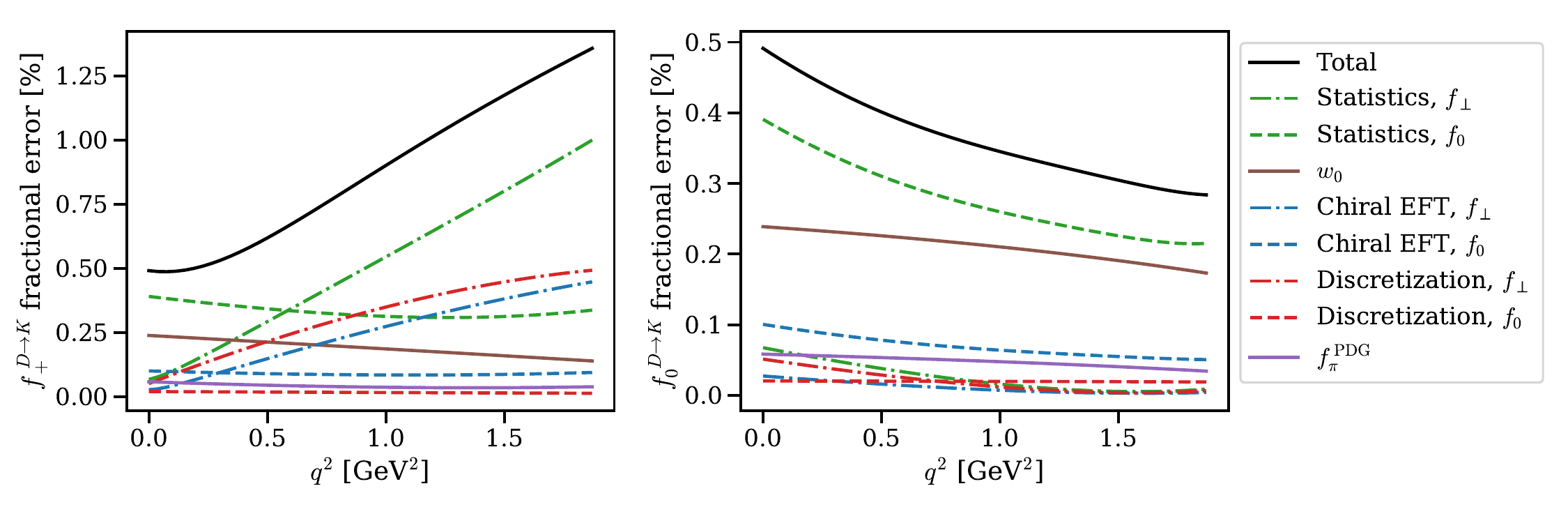}
\caption{
    Final error budget for the form factors $f^{\DK}_+$ and $f^{\DK}_0$ after the fit to the $z$~expansion.
    Contributions less than $0.01\%$ are not shown.
    \label{fig:d2k_final_error_budget}
    }
\end{figure}


\begin{figure}[!htb]
    \centering
    \includegraphics[width=1.0\textwidth]{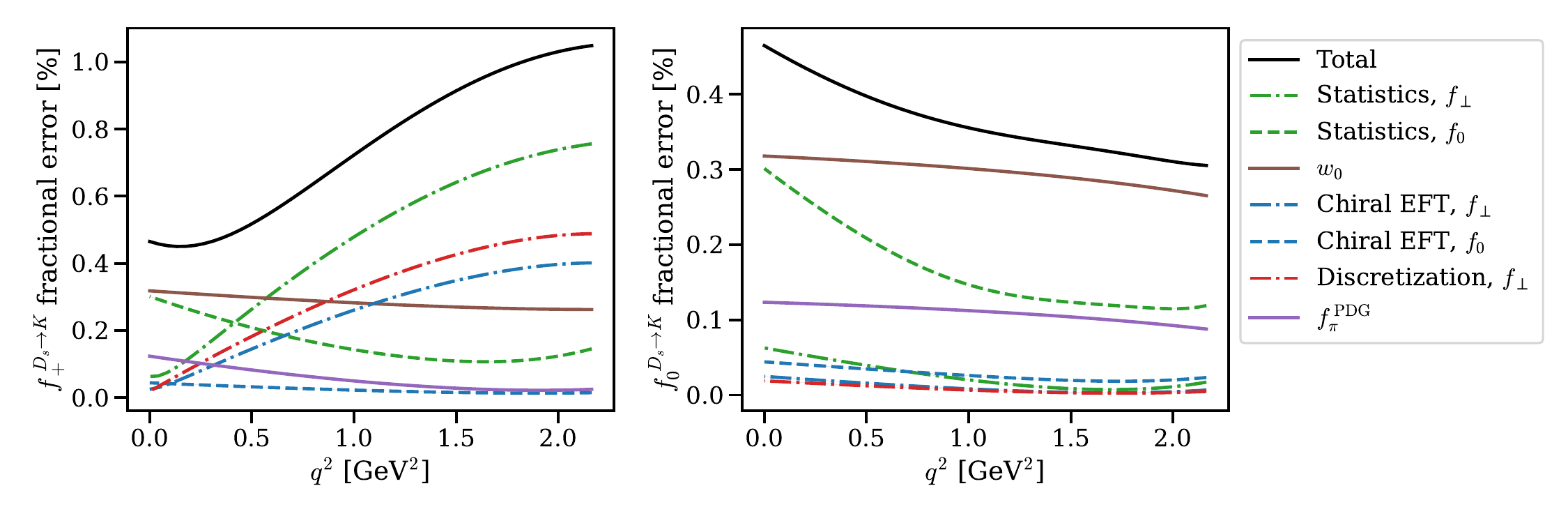}
\caption{
    Final error budget for the form factors $f^{\DsK}_+$ and $f^{\DsK}_0$ after the fit to the $z$~expansion.
    Contributions less than $0.01\%$ are not shown.
    \label{fig:ds2k_final_error_budget}
    }
\end{figure}

\clearpage

\subsection{\texorpdfstring{$z$}{z}-expansion fits: Joint fits to lattice-QCD form factors and experimental data \label{sec:joint_fit_results}}

\Cref{table:d2pi_zexpansion_comparison,table:d2k_zexpansion_comparison,table:ds2k_zexpansion_comparison} compare the results of the $z$-expansion fits for the decays $\Dpi$, $\DK$, and $\DsK$.
The fits enforce the kinematic identity $f_+(0)=f_0(0)$ by imposing $a_0=b_0$ [cf.\ \cref{eq:z_f0,eq:z_fplus}].
For the scalar form factor, the higher parameters $b_1$, $b_2$, and $b_3$ are unconstrained by the fits including experimental data.
In the joint fit, the lattice QCD form factors include a systematic from SIB, as described in \cref{ssec:sib_qed}.
No uncertainty from QED is included in the fit, since this is applied directly to $\Vcx$ as a final $1\%$ uncertainty.

\begin{table}[!htb]
\caption{
    Comparison of $z$-expansion fit results for the decay $\Dpi$.
    \label{table:d2pi_zexpansion_comparison}
}
    \begin{tabular}{c|c|c|c}
$\Dpi$ & LQCD only & Joint LQCD and Expt & Expt only \\
\hline\hline
$a_0 \equiv b_0$ & $0.6300(51)$ & $0.6306(47)$ & $0.1426(17)$\\
$a_1$ & $-0.610(99)$ & $-0.574(83)$ & $-0.157(45)$\\
$a_2$ & $-0.20(30)$ & $-0.009(393)$ & $-0.15(32)$\\
$a_3$ & $0.30(19)$ & $0.32(94)$ & $0.12(94)$\\
$b_1$ & $0.330(51)$ & $0.379(52)$ & $-$\\
$b_2$ & $-0.31(25)$ & $0.22(36)$ & $-$\\
$b_3$ & $-1.90(39)$ & $-0.54(84)$ & $-$\\
\hline\hline
\end{tabular}

\end{table}

\begin{table}[!htb]
\caption{
    Comparison of $z$-expansion fit results for the decay $\DK$.
    \label{table:d2k_zexpansion_comparison}
}
    \begin{tabular}{c|c|c|c}
$\DK$ & LQCD only & Joint LQCD and Expt & Expt only \\
\hline\hline
$a_0 \equiv b_0$ & $0.7452(31)$ & $0.7450(31)$ & $0.7246(26)$\\
$a_1$ & $-0.948(97)$ & $-1.036(73)$ & $-1.049(89)$\\
$a_2$ & $0.14(40)$ & $0.18(73)$ & $0.10(92)$\\
$a_3$ & $0.07(12)$ & $-0.03(1.00)$ & $-0.03(1.00)$\\
$b_1$ & $0.776(62)$ & $0.772(66)$ & $-$\\
$b_2$ & $0.14(34)$ & $0.08(56)$ & $-$\\
$b_3$ & $0.03(13)$ & $-0.02(99)$ & $-$\\
\hline\hline
\end{tabular}

\end{table}

\begin{table}[!htb]
\caption{
    Comparison of $z$-expansion fit results for the decay $\DsK$.
    \label{table:ds2k_zexpansion_comparison}
}
    \begin{tabular}{c|c|c|c}
$\DsK$ & LQCD only & Joint LQCD and Expt & Expt only \\
\hline\hline
$a_0 \equiv b_0$ & $0.6307(20)$ & $0.6306(20)$ & $0.164(18)$\\
$a_1$ & $-0.562(65)$ & $-0.557(72)$ & $-0.14(29)$\\
$a_2$ & $-0.19(20)$ & $-0.20(42)$ & $-0.03(98)$\\
$a_3$ & $0.33(29)$ & $0.04(98)$ & $0.008(1.000)$\\
$b_1$ & $0.347(27)$ & $0.346(35)$ & $-$\\
$b_2$ & $0.44(18)$ & $0.45(30)$ & $-$\\
$b_3$ & $-0.21(43)$ & $-0.11(96)$ & $-$\\
\hline\hline
\end{tabular}

\end{table}

\clearpage

\bibliography{refs}

\begin{thebibliography}{160}%
\makeatletter
\providecommand \@ifxundefined [1]{%
 \@ifx{#1\undefined}
}%
\providecommand \@ifnum [1]{%
 \ifnum #1\expandafter \@firstoftwo
 \else \expandafter \@secondoftwo
 \fi
}%
\providecommand \@ifx [1]{%
 \ifx #1\expandafter \@firstoftwo
 \else \expandafter \@secondoftwo
 \fi
}%
\providecommand \natexlab [1]{#1}%
\providecommand \enquote  [1]{``#1''}%
\providecommand \bibnamefont  [1]{#1}%
\providecommand \bibfnamefont [1]{#1}%
\providecommand \citenamefont [1]{#1}%
\providecommand \href@noop [0]{\@secondoftwo}%
\providecommand \href [0]{\begingroup \@sanitize@url \@href}%
\providecommand \@href[1]{\@@startlink{#1}\@@href}%
\providecommand \@@href[1]{\endgroup#1\@@endlink}%
\providecommand \@sanitize@url [0]{\catcode `\\12\catcode `\$12\catcode
  `\&12\catcode `\#12\catcode `\^12\catcode `\_12\catcode `\%12\relax}%
\providecommand \@@startlink[1]{}%
\providecommand \@@endlink[0]{}%
\providecommand \url  [0]{\begingroup\@sanitize@url \@url }%
\providecommand \@url [1]{\endgroup\@href {#1}{\urlprefix }}%
\providecommand \urlprefix  [0]{URL }%
\providecommand \Eprint [0]{\href }%
\providecommand \doibase [0]{https://doi.org/}%
\providecommand \selectlanguage [0]{\@gobble}%
\providecommand \bibinfo  [0]{\@secondoftwo}%
\providecommand \bibfield  [0]{\@secondoftwo}%
\providecommand \translation [1]{[#1]}%
\providecommand \BibitemOpen [0]{}%
\providecommand \bibitemStop [0]{}%
\providecommand \bibitemNoStop [0]{.\EOS\space}%
\providecommand \EOS [0]{\spacefactor3000\relax}%
\providecommand \BibitemShut  [1]{\csname bibitem#1\endcsname}%
\let\auto@bib@innerbib\@empty
\bibitem [{\citenamefont {Bjorken}\ and\ \citenamefont
  {Glashow}(1964)}]{Bjorken:1964gz}%
  \BibitemOpen
  \bibfield  {author} {\bibinfo {author} {\bibfnamefont {J.~D.}\ \bibnamefont
  {Bjorken}}\ and\ \bibinfo {author} {\bibfnamefont {S.~L.}\ \bibnamefont
  {Glashow}},\ }\bibfield  {title} {\bibinfo {title} {Elementary particles and
  {SU(4)}},\ }\href {https://doi.org/10.1016/0031-9163(64)90433-0} {\bibfield
  {journal} {\bibinfo  {journal} {Phys. Lett.}\ }\textbf {\bibinfo {volume}
  {11}},\ \bibinfo {pages} {255} (\bibinfo {year} {1964})}\BibitemShut
  {NoStop}%
\bibitem [{\citenamefont {Glashow}\ \emph {et~al.}(1970)\citenamefont
  {Glashow}, \citenamefont {Iliopoulos},\ and\ \citenamefont
  {Maiani}}]{Glashow:1970gm}%
  \BibitemOpen
  \bibfield  {author} {\bibinfo {author} {\bibfnamefont {S.~L.}\ \bibnamefont
  {Glashow}}, \bibinfo {author} {\bibfnamefont {J.}~\bibnamefont
  {Iliopoulos}},\ and\ \bibinfo {author} {\bibfnamefont {L.}~\bibnamefont
  {Maiani}},\ }\bibfield  {title} {\bibinfo {title} {Weak interactions with
  lepton-hadron symmetry},\ }\href {https://doi.org/10.1103/PhysRevD.2.1285}
  {\bibfield  {journal} {\bibinfo  {journal} {Phys. Rev. D}\ }\textbf {\bibinfo
  {volume} {2}},\ \bibinfo {pages} {1285} (\bibinfo {year} {1970})}\BibitemShut
  {NoStop}%
\bibitem [{\citenamefont {Gaillard}\ \emph {et~al.}(1975)\citenamefont
  {Gaillard}, \citenamefont {Lee},\ and\ \citenamefont
  {Rosner}}]{Gaillard:1974mw}%
  \BibitemOpen
  \bibfield  {author} {\bibinfo {author} {\bibfnamefont {M.~K.}\ \bibnamefont
  {Gaillard}}, \bibinfo {author} {\bibfnamefont {B.~W.}\ \bibnamefont {Lee}},\
  and\ \bibinfo {author} {\bibfnamefont {J.~L.}\ \bibnamefont {Rosner}},\
  }\bibfield  {title} {\bibinfo {title} {Search for charm},\ }\href
  {https://doi.org/10.1103/RevModPhys.47.277} {\bibfield  {journal} {\bibinfo
  {journal} {Rev. Mod. Phys.}\ }\textbf {\bibinfo {volume} {47}},\ \bibinfo
  {pages} {277} (\bibinfo {year} {1975})}\BibitemShut {NoStop}%
\bibitem [{\citenamefont {Aubert}\ \emph {et~al.}(1974)\citenamefont {Aubert}
  \emph {et~al.}}]{E598:1974sol}%
  \BibitemOpen
  \bibfield  {author} {\bibinfo {author} {\bibfnamefont {J.~J.}\ \bibnamefont
  {Aubert}} \emph {et~al.} (\bibinfo {collaboration} {E598}),\ }\bibfield
  {title} {\bibinfo {title} {Experimental observation of a heavy particle
  {$J$}},\ }\href {https://doi.org/10.1103/PhysRevLett.33.1404} {\bibfield
  {journal} {\bibinfo  {journal} {Phys. Rev. Lett.}\ }\textbf {\bibinfo
  {volume} {33}},\ \bibinfo {pages} {1404} (\bibinfo {year}
  {1974})}\BibitemShut {NoStop}%
\bibitem [{\citenamefont {Augustin}\ \emph {et~al.}(1974)\citenamefont
  {Augustin} \emph {et~al.}}]{SLAC-SP-017:1974ind}%
  \BibitemOpen
  \bibfield  {author} {\bibinfo {author} {\bibfnamefont {J.~E.}\ \bibnamefont
  {Augustin}} \emph {et~al.} (\bibinfo {collaboration} {SLAC-SP-017}),\
  }\bibfield  {title} {\bibinfo {title} {Discovery of a narrow resonance in
  $e^+ e^-$ annihilation},\ }\href
  {https://doi.org/10.1103/PhysRevLett.33.1406} {\bibfield  {journal} {\bibinfo
   {journal} {Phys. Rev. Lett.}\ }\textbf {\bibinfo {volume} {33}},\ \bibinfo
  {pages} {1406} (\bibinfo {year} {1974})}\BibitemShut {NoStop}%
\bibitem [{\citenamefont {Appelquist}\ and\ \citenamefont
  {Politzer}(1975)}]{Appelquist:1974zd}%
  \BibitemOpen
  \bibfield  {author} {\bibinfo {author} {\bibfnamefont {T.}~\bibnamefont
  {Appelquist}}\ and\ \bibinfo {author} {\bibfnamefont {H.~D.}\ \bibnamefont
  {Politzer}},\ }\bibfield  {title} {\bibinfo {title} {Orthocharmonium and
  $e^+e^-$ annihilation},\ }\href {https://doi.org/10.1103/PhysRevLett.34.43}
  {\bibfield  {journal} {\bibinfo  {journal} {Phys. Rev. Lett.}\ }\textbf
  {\bibinfo {volume} {34}},\ \bibinfo {pages} {43} (\bibinfo {year}
  {1975})}\BibitemShut {NoStop}%
\bibitem [{\citenamefont {De~Rujula}\ and\ \citenamefont
  {Glashow}(1975)}]{DeRujula:1974rkb}%
  \BibitemOpen
  \bibfield  {author} {\bibinfo {author} {\bibfnamefont {A.}~\bibnamefont
  {De~Rujula}}\ and\ \bibinfo {author} {\bibfnamefont {S.~L.}\ \bibnamefont
  {Glashow}},\ }\bibfield  {title} {\bibinfo {title} {Is bound charm found?},\
  }\href {https://doi.org/10.1103/PhysRevLett.34.46} {\bibfield  {journal}
  {\bibinfo  {journal} {Phys. Rev. Lett.}\ }\textbf {\bibinfo {volume} {34}},\
  \bibinfo {pages} {46} (\bibinfo {year} {1975})}\BibitemShut {NoStop}%
\bibitem [{\citenamefont {Appelquist}\ \emph {et~al.}(1975)\citenamefont
  {Appelquist}, \citenamefont {De~Rujula}, \citenamefont {Politzer},\ and\
  \citenamefont {Glashow}}]{Appelquist:1974yr}%
  \BibitemOpen
  \bibfield  {author} {\bibinfo {author} {\bibfnamefont {T.}~\bibnamefont
  {Appelquist}}, \bibinfo {author} {\bibfnamefont {A.}~\bibnamefont
  {De~Rujula}}, \bibinfo {author} {\bibfnamefont {H.~D.}\ \bibnamefont
  {Politzer}},\ and\ \bibinfo {author} {\bibfnamefont {S.~L.}\ \bibnamefont
  {Glashow}},\ }\bibfield  {title} {\bibinfo {title} {Charmonium
  spectroscopy},\ }\href {https://doi.org/10.1103/PhysRevLett.34.365}
  {\bibfield  {journal} {\bibinfo  {journal} {Phys. Rev. Lett.}\ }\textbf
  {\bibinfo {volume} {34}},\ \bibinfo {pages} {365} (\bibinfo {year}
  {1975})}\BibitemShut {NoStop}%
\bibitem [{\citenamefont {Eichten}\ \emph {et~al.}(1975)\citenamefont
  {Eichten}, \citenamefont {Gottfried}, \citenamefont {Kinoshita},
  \citenamefont {Kogut}, \citenamefont {Lane},\ and\ \citenamefont
  {Yan}}]{Eichten:1974af}%
  \BibitemOpen
  \bibfield  {author} {\bibinfo {author} {\bibfnamefont {E.}~\bibnamefont
  {Eichten}}, \bibinfo {author} {\bibfnamefont {K.}~\bibnamefont {Gottfried}},
  \bibinfo {author} {\bibfnamefont {T.}~\bibnamefont {Kinoshita}}, \bibinfo
  {author} {\bibfnamefont {J.~B.}\ \bibnamefont {Kogut}}, \bibinfo {author}
  {\bibfnamefont {K.~D.}\ \bibnamefont {Lane}},\ and\ \bibinfo {author}
  {\bibfnamefont {T.-M.}\ \bibnamefont {Yan}},\ }\bibfield  {title} {\bibinfo
  {title} {Spectrum of charmed quark-antiquark bound states},\ }\href
  {https://doi.org/10.1103/PhysRevLett.34.369} {\bibfield  {journal} {\bibinfo
  {journal} {Phys. Rev. Lett.}\ }\textbf {\bibinfo {volume} {34}},\ \bibinfo
  {pages} {369} (\bibinfo {year} {1975})},\ \bibinfo {note} {[Erratum:
  Phys.Rev.Lett. 36, 1276 (1976)]}\BibitemShut {NoStop}%
\bibitem [{\citenamefont {Albrecht}\ \emph {et~al.}(1987)\citenamefont
  {Albrecht} \emph {et~al.}}]{ARGUS:1987xtv}%
  \BibitemOpen
  \bibfield  {author} {\bibinfo {author} {\bibfnamefont {H.}~\bibnamefont
  {Albrecht}} \emph {et~al.} (\bibinfo {collaboration} {ARGUS}),\ }\bibfield
  {title} {\bibinfo {title} {Observation of {$B^0$-$\bar{B}^0$} mixing},\
  }\href {https://doi.org/10.1016/0370-2693(87)91177-4} {\bibfield  {journal}
  {\bibinfo  {journal} {Phys. Lett. B}\ }\textbf {\bibinfo {volume} {192}},\
  \bibinfo {pages} {245} (\bibinfo {year} {1987})}\BibitemShut {NoStop}%
\bibitem [{\citenamefont {Marciano}(1989)}]{Marciano:1989xd}%
  \BibitemOpen
  \bibfield  {author} {\bibinfo {author} {\bibfnamefont {W.~J.}\ \bibnamefont
  {Marciano}},\ }\bibfield  {title} {\bibinfo {title} {Heavy top quark mass
  predictions},\ }\href {https://doi.org/10.1103/PhysRevLett.62.2793}
  {\bibfield  {journal} {\bibinfo  {journal} {Phys. Rev. Lett.}\ }\textbf
  {\bibinfo {volume} {62}},\ \bibinfo {pages} {2793} (\bibinfo {year}
  {1989})}\BibitemShut {NoStop}%
\bibitem [{\citenamefont {Abe}\ \emph {et~al.}(1994)\citenamefont {Abe} \emph
  {et~al.}}]{CDF:1994vkk}%
  \BibitemOpen
  \bibfield  {author} {\bibinfo {author} {\bibfnamefont {F.}~\bibnamefont
  {Abe}} \emph {et~al.} (\bibinfo {collaboration} {CDF}),\ }\bibfield  {title}
  {\bibinfo {title} {{Evidence for top quark production in $\bar{p}p$
  collisions at $\sqrt{s} = 1.8$~TeV}},\ }\href
  {https://doi.org/10.1103/PhysRevD.50.2966} {\bibfield  {journal} {\bibinfo
  {journal} {Phys. Rev. D}\ }\textbf {\bibinfo {volume} {50}},\ \bibinfo
  {pages} {2966} (\bibinfo {year} {1994})}\BibitemShut {NoStop}%
\bibitem [{\citenamefont {Abe}\ \emph {et~al.}(1995)\citenamefont {Abe} \emph
  {et~al.}}]{CDF:1995wbb}%
  \BibitemOpen
  \bibfield  {author} {\bibinfo {author} {\bibfnamefont {F.}~\bibnamefont
  {Abe}} \emph {et~al.} (\bibinfo {collaboration} {CDF}),\ }\bibfield  {title}
  {\bibinfo {title} {{Observation of top quark production in $\bar{p}p$
  collisions}},\ }\href {https://doi.org/10.1103/PhysRevLett.74.2626}
  {\bibfield  {journal} {\bibinfo  {journal} {Phys. Rev. Lett.}\ }\textbf
  {\bibinfo {volume} {74}},\ \bibinfo {pages} {2626} (\bibinfo {year}
  {1995})},\ \Eprint {https://arxiv.org/abs/hep-ex/9503002}
  {arXiv:hep-ex/9503002} \BibitemShut {NoStop}%
\bibitem [{\citenamefont {Abachi}\ \emph {et~al.}(1995)\citenamefont {Abachi}
  \emph {et~al.}}]{D0:1995jca}%
  \BibitemOpen
  \bibfield  {author} {\bibinfo {author} {\bibfnamefont {S.}~\bibnamefont
  {Abachi}} \emph {et~al.} (\bibinfo {collaboration} {D0}),\ }\bibfield
  {title} {\bibinfo {title} {{Observation of the top quark}},\ }\href
  {https://doi.org/10.1103/PhysRevLett.74.2632} {\bibfield  {journal} {\bibinfo
   {journal} {Phys. Rev. Lett.}\ }\textbf {\bibinfo {volume} {74}},\ \bibinfo
  {pages} {2632} (\bibinfo {year} {1995})},\ \Eprint
  {https://arxiv.org/abs/hep-ex/9503003} {arXiv:hep-ex/9503003} \BibitemShut
  {NoStop}%
\bibitem [{\citenamefont {Artuso}\ \emph {et~al.}(2022)\citenamefont {Artuso},
  \citenamefont {Isidori},\ and\ \citenamefont {Stone}}]{Artuso:2022ijh}%
  \BibitemOpen
  \bibfield  {author} {\bibinfo {author} {\bibfnamefont {M.}~\bibnamefont
  {Artuso}}, \bibinfo {author} {\bibfnamefont {G.}~\bibnamefont {Isidori}},\
  and\ \bibinfo {author} {\bibfnamefont {S.}~\bibnamefont {Stone}},\ }\href
  {https://doi.org/10.1142/12696} {\emph {\bibinfo {title} {{New Physics in $b$
  Decays}}}}\ (\bibinfo  {publisher} {World Scientific},\ \bibinfo {address}
  {Singapore},\ \bibinfo {year} {2022})\BibitemShut {NoStop}%
\bibitem [{\citenamefont {Blanke}(2022)}]{Blanke:2022deg}%
  \BibitemOpen
  \bibfield  {author} {\bibinfo {author} {\bibfnamefont {M.}~\bibnamefont
  {Blanke}},\ }\bibfield  {title} {\bibinfo {title} {Theory perspective on
  heavy flavour physics},\ }in\ \href@noop {} {\emph {\bibinfo {booktitle}
  {{10\textsuperscript{th} Large Hadron Collider Physics Conference}}}}\
  (\bibinfo {year} {2022})\ \Eprint {https://arxiv.org/abs/2207.07354}
  {arXiv:2207.07354 [hep-ph]} \BibitemShut {NoStop}%
\bibitem [{\citenamefont {Hardy}\ and\ \citenamefont
  {Towner}(2020)}]{Hardy:2020qwl}%
  \BibitemOpen
  \bibfield  {author} {\bibinfo {author} {\bibfnamefont {J.~C.}\ \bibnamefont
  {Hardy}}\ and\ \bibinfo {author} {\bibfnamefont {I.~S.}\ \bibnamefont
  {Towner}},\ }\bibfield  {title} {\bibinfo {title} {{Superallowed $0^+ \to
  0^+$ nuclear $\beta$ decays: 2020 critical survey, with implications for
  $V_{ud}$ and CKM unitarity}},\ }\href
  {https://doi.org/10.1103/PhysRevC.102.045501} {\bibfield  {journal} {\bibinfo
   {journal} {Phys. Rev. C}\ }\textbf {\bibinfo {volume} {102}},\ \bibinfo
  {pages} {045501} (\bibinfo {year} {2020})}\BibitemShut {NoStop}%
\bibitem [{\citenamefont {Seng}\ \emph {et~al.}(2018)\citenamefont {Seng},
  \citenamefont {Gorchtein}, \citenamefont {Patel},\ and\ \citenamefont
  {Ramsey-Musolf}}]{Seng:2018yzq}%
  \BibitemOpen
  \bibfield  {author} {\bibinfo {author} {\bibfnamefont {C.-Y.}\ \bibnamefont
  {Seng}}, \bibinfo {author} {\bibfnamefont {M.}~\bibnamefont {Gorchtein}},
  \bibinfo {author} {\bibfnamefont {H.~H.}\ \bibnamefont {Patel}},\ and\
  \bibinfo {author} {\bibfnamefont {M.~J.}\ \bibnamefont {Ramsey-Musolf}},\
  }\bibfield  {title} {\bibinfo {title} {Reduced hadronic uncertainty in the
  determination of {$V_{ud}$}},\ }\href
  {https://doi.org/10.1103/PhysRevLett.121.241804} {\bibfield  {journal}
  {\bibinfo  {journal} {Phys. Rev. Lett.}\ }\textbf {\bibinfo {volume} {121}},\
  \bibinfo {pages} {241804} (\bibinfo {year} {2018})},\ \Eprint
  {https://arxiv.org/abs/1807.10197} {arXiv:1807.10197 [hep-ph]} \BibitemShut
  {NoStop}%
\bibitem [{\citenamefont {Seng}\ \emph {et~al.}(2019)\citenamefont {Seng},
  \citenamefont {Gorchtein},\ and\ \citenamefont
  {Ramsey-Musolf}}]{Seng:2018qru}%
  \BibitemOpen
  \bibfield  {author} {\bibinfo {author} {\bibfnamefont {C.~Y.}\ \bibnamefont
  {Seng}}, \bibinfo {author} {\bibfnamefont {M.}~\bibnamefont {Gorchtein}},\
  and\ \bibinfo {author} {\bibfnamefont {M.~J.}\ \bibnamefont
  {Ramsey-Musolf}},\ }\bibfield  {title} {\bibinfo {title} {{Dispersive
  evaluation of the inner radiative correction in neutron and nuclear $\beta$
  decay}},\ }\href {https://doi.org/10.1103/PhysRevD.100.013001} {\bibfield
  {journal} {\bibinfo  {journal} {Phys. Rev. D}\ }\textbf {\bibinfo {volume}
  {100}},\ \bibinfo {pages} {013001} (\bibinfo {year} {2019})},\ \Eprint
  {https://arxiv.org/abs/1812.03352} {arXiv:1812.03352 [nucl-th]} \BibitemShut
  {NoStop}%
\bibitem [{\citenamefont {Czarnecki}\ \emph {et~al.}(2019)\citenamefont
  {Czarnecki}, \citenamefont {Marciano},\ and\ \citenamefont
  {Sirlin}}]{Czarnecki:2019mwq}%
  \BibitemOpen
  \bibfield  {author} {\bibinfo {author} {\bibfnamefont {A.}~\bibnamefont
  {Czarnecki}}, \bibinfo {author} {\bibfnamefont {W.~J.}\ \bibnamefont
  {Marciano}},\ and\ \bibinfo {author} {\bibfnamefont {A.}~\bibnamefont
  {Sirlin}},\ }\bibfield  {title} {\bibinfo {title} {Radiative corrections to
  neutron and nuclear beta decays revisited},\ }\href
  {https://doi.org/10.1103/PhysRevD.100.073008} {\bibfield  {journal} {\bibinfo
   {journal} {Phys. Rev. D}\ }\textbf {\bibinfo {volume} {100}},\ \bibinfo
  {pages} {073008} (\bibinfo {year} {2019})},\ \Eprint
  {https://arxiv.org/abs/1907.06737} {arXiv:1907.06737 [hep-ph]} \BibitemShut
  {NoStop}%
\bibitem [{\citenamefont {Seng}\ \emph {et~al.}(2020)\citenamefont {Seng},
  \citenamefont {Feng}, \citenamefont {Gorchtein},\ and\ \citenamefont
  {Jin}}]{Seng:2020wjq}%
  \BibitemOpen
  \bibfield  {author} {\bibinfo {author} {\bibfnamefont {C.-Y.}\ \bibnamefont
  {Seng}}, \bibinfo {author} {\bibfnamefont {X.}~\bibnamefont {Feng}}, \bibinfo
  {author} {\bibfnamefont {M.}~\bibnamefont {Gorchtein}},\ and\ \bibinfo
  {author} {\bibfnamefont {L.-C.}\ \bibnamefont {Jin}},\ }\bibfield  {title}
  {\bibinfo {title} {{Joint lattice-QCD--dispersion-theory analysis confirms
  the quark-mixing top-row unitarity deficit}},\ }\href
  {https://doi.org/10.1103/PhysRevD.101.111301} {\bibfield  {journal} {\bibinfo
   {journal} {Phys. Rev. D}\ }\textbf {\bibinfo {volume} {101}},\ \bibinfo
  {pages} {111301} (\bibinfo {year} {2020})},\ \Eprint
  {https://arxiv.org/abs/2003.11264} {arXiv:2003.11264 [hep-ph]} \BibitemShut
  {NoStop}%
\bibitem [{\citenamefont {Shiells}\ \emph {et~al.}(2021)\citenamefont
  {Shiells}, \citenamefont {Blunden},\ and\ \citenamefont
  {Melnitchouk}}]{Shiells:2020fqp}%
  \BibitemOpen
  \bibfield  {author} {\bibinfo {author} {\bibfnamefont {K.}~\bibnamefont
  {Shiells}}, \bibinfo {author} {\bibfnamefont {P.~G.}\ \bibnamefont
  {Blunden}},\ and\ \bibinfo {author} {\bibfnamefont {W.}~\bibnamefont
  {Melnitchouk}},\ }\bibfield  {title} {\bibinfo {title} {{Electroweak axial
  structure functions and improved extraction of the $V_{ud}$ CKM matrix
  element}},\ }\href {https://doi.org/10.1103/PhysRevD.104.033003} {\bibfield
  {journal} {\bibinfo  {journal} {Phys. Rev. D}\ }\textbf {\bibinfo {volume}
  {104}},\ \bibinfo {pages} {033003} (\bibinfo {year} {2021})},\ \Eprint
  {https://arxiv.org/abs/2012.01580} {arXiv:2012.01580 [hep-ph]} \BibitemShut
  {NoStop}%
\bibitem [{\citenamefont {Gorchtein}(2019)}]{Gorchtein:2018fxl}%
  \BibitemOpen
  \bibfield  {author} {\bibinfo {author} {\bibfnamefont {M.}~\bibnamefont
  {Gorchtein}},\ }\bibfield  {title} {\bibinfo {title} {{$\gamma W$} box inside
  out: Nuclear polarizabilities distort the beta decay spectrum},\ }\href
  {https://doi.org/10.1103/PhysRevLett.123.042503} {\bibfield  {journal}
  {\bibinfo  {journal} {Phys. Rev. Lett.}\ }\textbf {\bibinfo {volume} {123}},\
  \bibinfo {pages} {042503} (\bibinfo {year} {2019})},\ \Eprint
  {https://arxiv.org/abs/1812.04229} {arXiv:1812.04229 [nucl-th]} \BibitemShut
  {NoStop}%
\bibitem [{\citenamefont {Marciano}(2004)}]{Marciano:2004uf}%
  \BibitemOpen
  \bibfield  {author} {\bibinfo {author} {\bibfnamefont {W.~J.}\ \bibnamefont
  {Marciano}},\ }\bibfield  {title} {\bibinfo {title} {{Precise determination
  of $|V_{us}|$ from lattice calculations of pseudoscalar decay constants}},\
  }\href {https://doi.org/10.1103/PhysRevLett.93.231803} {\bibfield  {journal}
  {\bibinfo  {journal} {Phys. Rev. Lett.}\ }\textbf {\bibinfo {volume} {93}},\
  \bibinfo {pages} {231803} (\bibinfo {year} {2004})},\ \Eprint
  {https://arxiv.org/abs/hep-ph/0402299} {arXiv:hep-ph/0402299} \BibitemShut
  {NoStop}%
\bibitem [{\citenamefont {Workman}\ \emph {et~al.}(2022)\citenamefont {Workman}
  \emph {et~al.}}]{Workman:2022ynf}%
  \BibitemOpen
  \bibfield  {author} {\bibinfo {author} {\bibfnamefont {R.~L.}\ \bibnamefont
  {Workman}} \emph {et~al.} (\bibinfo {collaboration} {Particle Data Group}),\
  }\bibfield  {title} {\bibinfo {title} {Review of particle physics},\ }\href
  {https://doi.org/10.1093/ptep/ptac097} {\bibfield  {journal} {\bibinfo
  {journal} {PTEP}\ }\textbf {\bibinfo {volume} {2022}},\ \bibinfo {pages}
  {083C01} (\bibinfo {year} {2022})}\BibitemShut {NoStop}%
\bibitem [{\citenamefont {Carrasco}\ \emph {et~al.}(2016)\citenamefont
  {Carrasco}, \citenamefont {Lami}, \citenamefont {Lubicz}, \citenamefont
  {Riggio}, \citenamefont {Simula},\ and\ \citenamefont
  {Tarantino}}]{Carrasco:2016kpy}%
  \BibitemOpen
  \bibfield  {author} {\bibinfo {author} {\bibfnamefont {N.}~\bibnamefont
  {Carrasco}}, \bibinfo {author} {\bibfnamefont {P.}~\bibnamefont {Lami}},
  \bibinfo {author} {\bibfnamefont {V.}~\bibnamefont {Lubicz}}, \bibinfo
  {author} {\bibfnamefont {L.}~\bibnamefont {Riggio}}, \bibinfo {author}
  {\bibfnamefont {S.}~\bibnamefont {Simula}},\ and\ \bibinfo {author}
  {\bibfnamefont {C.}~\bibnamefont {Tarantino}},\ }\bibfield  {title} {\bibinfo
  {title} {{$K \to \pi$ semileptonic form factors with $N_f=2+1+1$ twisted mass
  fermions}},\ }\href {https://doi.org/10.1103/PhysRevD.93.114512} {\bibfield
  {journal} {\bibinfo  {journal} {Phys. Rev. D}\ }\textbf {\bibinfo {volume}
  {93}},\ \bibinfo {pages} {114512} (\bibinfo {year} {2016})},\ \Eprint
  {https://arxiv.org/abs/1602.04113} {arXiv:1602.04113 [hep-lat]} \BibitemShut
  {NoStop}%
\bibitem [{\citenamefont {Bazavov}\ \emph
  {et~al.}(2019{\natexlab{a}})\citenamefont {Bazavov} \emph
  {et~al.}}]{FermilabLattice:2018zqv}%
  \BibitemOpen
  \bibfield  {author} {\bibinfo {author} {\bibfnamefont {A.}~\bibnamefont
  {Bazavov}} \emph {et~al.} (\bibinfo {collaboration} {Fermilab Lattice,
  MILC}),\ }\bibfield  {title} {\bibinfo {title} {{$|V_{us}|$ from $K_{\ell 3}$
  decay and four-flavor lattice QCD}},\ }\href
  {https://doi.org/10.1103/PhysRevD.99.114509} {\bibfield  {journal} {\bibinfo
  {journal} {Phys. Rev. D}\ }\textbf {\bibinfo {volume} {99}},\ \bibinfo
  {pages} {114509} (\bibinfo {year} {2019}{\natexlab{a}})},\ \Eprint
  {https://arxiv.org/abs/1809.02827} {arXiv:1809.02827 [hep-lat]} \BibitemShut
  {NoStop}%
\bibitem [{\citenamefont {Boyle}\ \emph {et~al.}(2015)\citenamefont {Boyle}
  \emph {et~al.}}]{RBCUKQCD:2015joy}%
  \BibitemOpen
  \bibfield  {author} {\bibinfo {author} {\bibfnamefont {P.~A.}\ \bibnamefont
  {Boyle}} \emph {et~al.} (\bibinfo {collaboration} {RBC, UKQCD}),\ }\bibfield
  {title} {\bibinfo {title} {{The kaon semileptonic form factor in $N_{f} = 2 +
  1$ domain wall lattice QCD with physical light quark masses}},\ }\href
  {https://doi.org/10.1007/JHEP06(2015)164} {\bibfield  {journal} {\bibinfo
  {journal} {JHEP}\ }\textbf {\bibinfo {volume} {06}},\ \bibinfo {pages}
  {164}},\ \Eprint {https://arxiv.org/abs/1504.01692} {arXiv:1504.01692
  [hep-lat]} \BibitemShut {NoStop}%
\bibitem [{\citenamefont {Ishikawa}\ \emph {et~al.}(2022)\citenamefont
  {Ishikawa}, \citenamefont {Ishizuka}, \citenamefont {Kuramashi},
  \citenamefont {Namekawa}, \citenamefont {Taniguchi}, \citenamefont {Ukita},
  \citenamefont {Yamazaki},\ and\ \citenamefont {Yoshi\'e}}]{Ishikawa:2022otj}%
  \BibitemOpen
  \bibfield  {author} {\bibinfo {author} {\bibfnamefont {K.-i.}\ \bibnamefont
  {Ishikawa}}, \bibinfo {author} {\bibfnamefont {N.}~\bibnamefont {Ishizuka}},
  \bibinfo {author} {\bibfnamefont {Y.}~\bibnamefont {Kuramashi}}, \bibinfo
  {author} {\bibfnamefont {Y.}~\bibnamefont {Namekawa}}, \bibinfo {author}
  {\bibfnamefont {Y.}~\bibnamefont {Taniguchi}}, \bibinfo {author}
  {\bibfnamefont {N.}~\bibnamefont {Ukita}}, \bibinfo {author} {\bibfnamefont
  {T.}~\bibnamefont {Yamazaki}},\ and\ \bibinfo {author} {\bibfnamefont
  {T.}~\bibnamefont {Yoshi\'e}} (\bibinfo {collaboration} {PACS}),\ }\bibfield
  {title} {\bibinfo {title} {{$K_{\ell3}$ form factors at the physical point:
  Toward the continuum limit}},\ }\href
  {https://doi.org/10.1103/PhysRevD.106.094501} {\bibfield  {journal} {\bibinfo
   {journal} {Phys. Rev. D}\ }\textbf {\bibinfo {volume} {106}},\ \bibinfo
  {pages} {094501} (\bibinfo {year} {2022})},\ \Eprint
  {https://arxiv.org/abs/2206.08654} {arXiv:2206.08654 [hep-lat]} \BibitemShut
  {NoStop}%
\bibitem [{\citenamefont {Dowdall}\ \emph {et~al.}(2013)\citenamefont
  {Dowdall}, \citenamefont {Davies}, \citenamefont {Lepage},\ and\
  \citenamefont {McNeile}}]{Dowdall:2013rya}%
  \BibitemOpen
  \bibfield  {author} {\bibinfo {author} {\bibfnamefont {R.~J.}\ \bibnamefont
  {Dowdall}}, \bibinfo {author} {\bibfnamefont {C.~T.~H.}\ \bibnamefont
  {Davies}}, \bibinfo {author} {\bibfnamefont {G.~P.}\ \bibnamefont {Lepage}},\
  and\ \bibinfo {author} {\bibfnamefont {C.}~\bibnamefont {McNeile}},\
  }\bibfield  {title} {\bibinfo {title} {{$V_{us}$ from $\pi$ and $K$ decay
  constants in full lattice QCD with physical $u$, $d$, $s$ and $c$ quarks}},\
  }\href {https://doi.org/10.1103/PhysRevD.88.074504} {\bibfield  {journal}
  {\bibinfo  {journal} {Phys. Rev. D}\ }\textbf {\bibinfo {volume} {88}},\
  \bibinfo {pages} {074504} (\bibinfo {year} {2013})},\ \Eprint
  {https://arxiv.org/abs/1303.1670} {arXiv:1303.1670 [hep-lat]} \BibitemShut
  {NoStop}%
\bibitem [{\citenamefont {Carrasco}\ \emph
  {et~al.}(2015{\natexlab{a}})\citenamefont {Carrasco} \emph
  {et~al.}}]{Carrasco:2014poa}%
  \BibitemOpen
  \bibfield  {author} {\bibinfo {author} {\bibfnamefont {N.}~\bibnamefont
  {Carrasco}} \emph {et~al.},\ }\bibfield  {title} {\bibinfo {title} {{Leptonic
  decay constants $f_{K}$, $f_{D}$, and $f_{{D}_{s}}$ with $N_{f} = 2+1+1$
  twisted-mass lattice QCD}},\ }\href
  {https://doi.org/10.1103/PhysRevD.91.054507} {\bibfield  {journal} {\bibinfo
  {journal} {Phys. Rev. D}\ }\textbf {\bibinfo {volume} {91}},\ \bibinfo
  {pages} {054507} (\bibinfo {year} {2015}{\natexlab{a}})},\ \Eprint
  {https://arxiv.org/abs/1411.7908} {arXiv:1411.7908 [hep-lat]} \BibitemShut
  {NoStop}%
\bibitem [{\citenamefont {Bazavov}\ \emph {et~al.}(2018)\citenamefont {Bazavov}
  \emph {et~al.}}]{Bazavov:2017lyh}%
  \BibitemOpen
  \bibfield  {author} {\bibinfo {author} {\bibfnamefont {A.}~\bibnamefont
  {Bazavov}} \emph {et~al.},\ }\bibfield  {title} {\bibinfo {title} {{$B$- and
  $D$-meson leptonic decay constants from four-flavor lattice QCD}},\ }\href
  {https://doi.org/10.1103/PhysRevD.98.074512} {\bibfield  {journal} {\bibinfo
  {journal} {Phys. Rev. D}\ }\textbf {\bibinfo {volume} {98}},\ \bibinfo
  {pages} {074512} (\bibinfo {year} {2018})},\ \Eprint
  {https://arxiv.org/abs/1712.09262} {arXiv:1712.09262 [hep-lat]} \BibitemShut
  {NoStop}%
\bibitem [{\citenamefont {Miller}\ \emph {et~al.}(2020)\citenamefont {Miller}
  \emph {et~al.}}]{Miller:2020xhy}%
  \BibitemOpen
  \bibfield  {author} {\bibinfo {author} {\bibfnamefont {N.}~\bibnamefont
  {Miller}} \emph {et~al.},\ }\bibfield  {title} {\bibinfo {title} {{$F_K /
  F_\pi$ from M\"obius domain-wall fermions solved on gradient-flowed HISQ
  ensembles}},\ }\href {https://doi.org/10.1103/PhysRevD.102.034507} {\bibfield
   {journal} {\bibinfo  {journal} {Phys. Rev. D}\ }\textbf {\bibinfo {volume}
  {102}},\ \bibinfo {pages} {034507} (\bibinfo {year} {2020})},\ \Eprint
  {https://arxiv.org/abs/2005.04795} {arXiv:2005.04795 [hep-lat]} \BibitemShut
  {NoStop}%
\bibitem [{\citenamefont {D\"urr}\ \emph {et~al.}(2017)\citenamefont {D\"urr}
  \emph {et~al.}}]{Durr:2016ulb}%
  \BibitemOpen
  \bibfield  {author} {\bibinfo {author} {\bibfnamefont {S.}~\bibnamefont
  {D\"urr}} \emph {et~al.},\ }\bibfield  {title} {\bibinfo {title} {{Leptonic
  decay-constant ratio $f_K/f_\pi$ from lattice QCD using 2+1 clover-improved
  fermion flavors with 2-HEX smearing}},\ }\href
  {https://doi.org/10.1103/PhysRevD.95.054513} {\bibfield  {journal} {\bibinfo
  {journal} {Phys. Rev. D}\ }\textbf {\bibinfo {volume} {95}},\ \bibinfo
  {pages} {054513} (\bibinfo {year} {2017})},\ \Eprint
  {https://arxiv.org/abs/1601.05998} {arXiv:1601.05998 [hep-lat]} \BibitemShut
  {NoStop}%
\bibitem [{\citenamefont {Bornyakov}\ \emph {et~al.}(2017)\citenamefont
  {Bornyakov}, \citenamefont {Horsley}, \citenamefont {Nakamura}, \citenamefont
  {Perlt}, \citenamefont {Pleiter}, \citenamefont {Rakow}, \citenamefont
  {Schierholz}, \citenamefont {Schiller}, \citenamefont {St\"uben},\ and\
  \citenamefont {Zanotti}}]{QCDSF-UKQCD:2016rau}%
  \BibitemOpen
  \bibfield  {author} {\bibinfo {author} {\bibfnamefont {V.~G.}\ \bibnamefont
  {Bornyakov}}, \bibinfo {author} {\bibfnamefont {R.}~\bibnamefont {Horsley}},
  \bibinfo {author} {\bibfnamefont {Y.}~\bibnamefont {Nakamura}}, \bibinfo
  {author} {\bibfnamefont {H.}~\bibnamefont {Perlt}}, \bibinfo {author}
  {\bibfnamefont {D.}~\bibnamefont {Pleiter}}, \bibinfo {author} {\bibfnamefont
  {P.~E.~L.}\ \bibnamefont {Rakow}}, \bibinfo {author} {\bibfnamefont
  {G.}~\bibnamefont {Schierholz}}, \bibinfo {author} {\bibfnamefont
  {A.}~\bibnamefont {Schiller}}, \bibinfo {author} {\bibfnamefont
  {H.}~\bibnamefont {St\"uben}},\ and\ \bibinfo {author} {\bibfnamefont
  {J.~M.}\ \bibnamefont {Zanotti}} (\bibinfo {collaboration} {QCDSF, UKQCD}),\
  }\bibfield  {title} {\bibinfo {title} {{Flavour breaking effects in the
  pseudoscalar meson decay constants}},\ }\href
  {https://doi.org/10.1016/j.physletb.2017.02.018} {\bibfield  {journal}
  {\bibinfo  {journal} {Phys. Lett. B}\ }\textbf {\bibinfo {volume} {767}},\
  \bibinfo {pages} {366} (\bibinfo {year} {2017})},\ \Eprint
  {https://arxiv.org/abs/1612.04798} {arXiv:1612.04798 [hep-lat]} \BibitemShut
  {NoStop}%
\bibitem [{\citenamefont {Aoki}\ \emph {et~al.}(2022)\citenamefont {Aoki} \emph
  {et~al.}}]{Aoki:2021kgd}%
  \BibitemOpen
  \bibfield  {author} {\bibinfo {author} {\bibfnamefont {Y.}~\bibnamefont
  {Aoki}} \emph {et~al.} (\bibinfo {collaboration} {Flavour Lattice Averaging
  Group}),\ }\bibfield  {title} {\bibinfo {title} {{FLAG Review 2021}},\ }\href
  {https://doi.org/10.1140/epjc/s10052-022-10536-1} {\bibfield  {journal}
  {\bibinfo  {journal} {Eur. Phys. J. C}\ }\textbf {\bibinfo {volume} {82}},\
  \bibinfo {pages} {869} (\bibinfo {year} {2022})},\ \Eprint
  {https://arxiv.org/abs/2111.09849} {arXiv:2111.09849 [hep-lat]} \BibitemShut
  {NoStop}%
\bibitem [{\citenamefont {Cirigliano}\ \emph {et~al.}(2022)\citenamefont
  {Cirigliano}, \citenamefont {Crivellin}, \citenamefont {Hoferichter},\ and\
  \citenamefont {Moulson}}]{Cirigliano:2022yyo}%
  \BibitemOpen
  \bibfield  {author} {\bibinfo {author} {\bibfnamefont {V.}~\bibnamefont
  {Cirigliano}}, \bibinfo {author} {\bibfnamefont {A.}~\bibnamefont
  {Crivellin}}, \bibinfo {author} {\bibfnamefont {M.}~\bibnamefont
  {Hoferichter}},\ and\ \bibinfo {author} {\bibfnamefont {M.}~\bibnamefont
  {Moulson}},\ }\bibfield  {title} {\bibinfo {title} {{Scrutinizing CKM
  unitarity with a new measurement of the $K_{\mu 3}/K_{\mu 2}$ branching
  fraction}},\ }\href@noop {} {\  (\bibinfo {year} {2022})},\ \Eprint
  {https://arxiv.org/abs/2208.11707} {arXiv:2208.11707 [hep-ph]} \BibitemShut
  {NoStop}%
\bibitem [{\citenamefont {Moulson}(2017)}]{Moulson:2017ive}%
  \BibitemOpen
  \bibfield  {author} {\bibinfo {author} {\bibfnamefont {M.}~\bibnamefont
  {Moulson}},\ }\bibfield  {title} {\bibinfo {title} {{Experimental
  determination of $V_{us}$ from kaon decays}},\ }\href
  {https://doi.org/10.22323/1.291.0033} {\bibfield  {journal} {\bibinfo
  {journal} {PoS}\ }\textbf {\bibinfo {volume} {CKM2016}},\ \bibinfo {pages}
  {033} (\bibinfo {year} {2017})},\ \Eprint {https://arxiv.org/abs/1704.04104}
  {arXiv:1704.04104 [hep-ex]} \BibitemShut {NoStop}%
\bibitem [{\citenamefont {Giusti}\ \emph {et~al.}(2018)\citenamefont {Giusti},
  \citenamefont {Lubicz}, \citenamefont {Martinelli}, \citenamefont
  {Sachrajda}, \citenamefont {Sanfilippo}, \citenamefont {Simula},
  \citenamefont {Tantalo},\ and\ \citenamefont {Tarantino}}]{Giusti:2017dwk}%
  \BibitemOpen
  \bibfield  {author} {\bibinfo {author} {\bibfnamefont {D.}~\bibnamefont
  {Giusti}}, \bibinfo {author} {\bibfnamefont {V.}~\bibnamefont {Lubicz}},
  \bibinfo {author} {\bibfnamefont {G.}~\bibnamefont {Martinelli}}, \bibinfo
  {author} {\bibfnamefont {C.~T.}\ \bibnamefont {Sachrajda}}, \bibinfo {author}
  {\bibfnamefont {F.}~\bibnamefont {Sanfilippo}}, \bibinfo {author}
  {\bibfnamefont {S.}~\bibnamefont {Simula}}, \bibinfo {author} {\bibfnamefont
  {N.}~\bibnamefont {Tantalo}},\ and\ \bibinfo {author} {\bibfnamefont
  {C.}~\bibnamefont {Tarantino}},\ }\bibfield  {title} {\bibinfo {title}
  {{First lattice calculation of the QED corrections to leptonic decay
  rates}},\ }\href {https://doi.org/10.1103/PhysRevLett.120.072001} {\bibfield
  {journal} {\bibinfo  {journal} {Phys. Rev. Lett.}\ }\textbf {\bibinfo
  {volume} {120}},\ \bibinfo {pages} {072001} (\bibinfo {year} {2018})},\
  \Eprint {https://arxiv.org/abs/1711.06537} {arXiv:1711.06537 [hep-lat]}
  \BibitemShut {NoStop}%
\bibitem [{\citenamefont {Di~Carlo}\ \emph {et~al.}(2019)\citenamefont
  {Di~Carlo}, \citenamefont {Giusti}, \citenamefont {Lubicz}, \citenamefont
  {Martinelli}, \citenamefont {Sachrajda}, \citenamefont {Sanfilippo},
  \citenamefont {Simula},\ and\ \citenamefont {Tantalo}}]{DiCarlo:2019thl}%
  \BibitemOpen
  \bibfield  {author} {\bibinfo {author} {\bibfnamefont {M.}~\bibnamefont
  {Di~Carlo}}, \bibinfo {author} {\bibfnamefont {D.}~\bibnamefont {Giusti}},
  \bibinfo {author} {\bibfnamefont {V.}~\bibnamefont {Lubicz}}, \bibinfo
  {author} {\bibfnamefont {G.}~\bibnamefont {Martinelli}}, \bibinfo {author}
  {\bibfnamefont {C.~T.}\ \bibnamefont {Sachrajda}}, \bibinfo {author}
  {\bibfnamefont {F.}~\bibnamefont {Sanfilippo}}, \bibinfo {author}
  {\bibfnamefont {S.}~\bibnamefont {Simula}},\ and\ \bibinfo {author}
  {\bibfnamefont {N.}~\bibnamefont {Tantalo}},\ }\bibfield  {title} {\bibinfo
  {title} {{Light-meson leptonic decay rates in lattice QCD+QED}},\ }\href
  {https://doi.org/10.1103/PhysRevD.100.034514} {\bibfield  {journal} {\bibinfo
   {journal} {Phys. Rev. D}\ }\textbf {\bibinfo {volume} {100}},\ \bibinfo
  {pages} {034514} (\bibinfo {year} {2019})},\ \Eprint
  {https://arxiv.org/abs/1904.08731} {arXiv:1904.08731 [hep-lat]} \BibitemShut
  {NoStop}%
\bibitem [{\citenamefont {Seng}\ \emph
  {et~al.}(2021{\natexlab{a}})\citenamefont {Seng}, \citenamefont {Galviz},
  \citenamefont {Gorchtein},\ and\ \citenamefont {Mei\ss{}ner}}]{Seng:2021boy}%
  \BibitemOpen
  \bibfield  {author} {\bibinfo {author} {\bibfnamefont {C.-Y.}\ \bibnamefont
  {Seng}}, \bibinfo {author} {\bibfnamefont {D.}~\bibnamefont {Galviz}},
  \bibinfo {author} {\bibfnamefont {M.}~\bibnamefont {Gorchtein}},\ and\
  \bibinfo {author} {\bibfnamefont {U.-G.}\ \bibnamefont {Mei\ss{}ner}},\
  }\bibfield  {title} {\bibinfo {title} {{High-precision determination of the
  $K_{e3}$ radiative corrections}},\ }\href
  {https://doi.org/10.1016/j.physletb.2021.136522} {\bibfield  {journal}
  {\bibinfo  {journal} {Phys. Lett. B}\ }\textbf {\bibinfo {volume} {820}},\
  \bibinfo {pages} {136522} (\bibinfo {year} {2021}{\natexlab{a}})},\ \Eprint
  {https://arxiv.org/abs/2103.00975} {arXiv:2103.00975 [hep-ph]} \BibitemShut
  {NoStop}%
\bibitem [{\citenamefont {Seng}\ \emph
  {et~al.}(2021{\natexlab{b}})\citenamefont {Seng}, \citenamefont {Galviz},
  \citenamefont {Gorchtein},\ and\ \citenamefont {Mei\ss{}ner}}]{Seng:2021wcf}%
  \BibitemOpen
  \bibfield  {author} {\bibinfo {author} {\bibfnamefont {C.-Y.}\ \bibnamefont
  {Seng}}, \bibinfo {author} {\bibfnamefont {D.}~\bibnamefont {Galviz}},
  \bibinfo {author} {\bibfnamefont {M.}~\bibnamefont {Gorchtein}},\ and\
  \bibinfo {author} {\bibfnamefont {U.-G.}\ \bibnamefont {Mei\ss{}ner}},\
  }\bibfield  {title} {\bibinfo {title} {{Improved $K_{e3}$ radiative
  corrections sharpen the $K_{\mu 2}$-$K_{l3}$ discrepancy}},\ }\href
  {https://doi.org/10.1007/JHEP11(2021)172} {\bibfield  {journal} {\bibinfo
  {journal} {JHEP}\ }\textbf {\bibinfo {volume} {11}},\ \bibinfo {pages}
  {172}},\ \Eprint {https://arxiv.org/abs/2103.04843} {arXiv:2103.04843
  [hep-ph]} \BibitemShut {NoStop}%
\bibitem [{\citenamefont {Seng}\ \emph {et~al.}(2022)\citenamefont {Seng},
  \citenamefont {Galviz}, \citenamefont {Gorchtein},\ and\ \citenamefont
  {Mei\ss{}ner}}]{Seng:2022wcw}%
  \BibitemOpen
  \bibfield  {author} {\bibinfo {author} {\bibfnamefont {C.-Y.}\ \bibnamefont
  {Seng}}, \bibinfo {author} {\bibfnamefont {D.}~\bibnamefont {Galviz}},
  \bibinfo {author} {\bibfnamefont {M.}~\bibnamefont {Gorchtein}},\ and\
  \bibinfo {author} {\bibfnamefont {U.-G.}\ \bibnamefont {Mei\ss{}ner}},\
  }\bibfield  {title} {\bibinfo {title} {{Complete theory of radiative
  corrections to K$_{\ell3}$ decays and the $V_{us}$ update}},\ }\href
  {https://doi.org/10.1007/JHEP07(2022)071} {\bibfield  {journal} {\bibinfo
  {journal} {JHEP}\ }\textbf {\bibinfo {volume} {07}},\ \bibinfo {pages}
  {071}},\ \Eprint {https://arxiv.org/abs/2203.05217} {arXiv:2203.05217
  [hep-ph]} \BibitemShut {NoStop}%
\bibitem [{\citenamefont {Amhis}\ \emph {et~al.}(2022)\citenamefont {Amhis}
  \emph {et~al.}}]{HFLAV:2022pwe}%
  \BibitemOpen
  \bibfield  {author} {\bibinfo {author} {\bibfnamefont {Y.}~\bibnamefont
  {Amhis}} \emph {et~al.} (\bibinfo {collaboration} {HFLAV}),\ }\bibfield
  {title} {\bibinfo {title} {{Averages of $b$-hadron, $c$-hadron, and
  $\tau$-lepton properties as of 2021}},\ }\href@noop {} {\  (\bibinfo {year}
  {2022})},\ \Eprint {https://arxiv.org/abs/2206.07501} {arXiv:2206.07501
  [hep-ex]} \BibitemShut {NoStop}%
\bibitem [{\citenamefont {Aubin}\ \emph {et~al.}(2005)\citenamefont {Aubin}
  \emph {et~al.}}]{FermilabLattice:2004ncd}%
  \BibitemOpen
  \bibfield  {author} {\bibinfo {author} {\bibfnamefont {C.}~\bibnamefont
  {Aubin}} \emph {et~al.} (\bibinfo {collaboration} {Fermilab Lattice, MILC,
  HPQCD}),\ }\bibfield  {title} {\bibinfo {title} {{Semileptonic decays of $D$
  mesons in three-flavor lattice QCD}},\ }\href
  {https://doi.org/10.1103/PhysRevLett.94.011601} {\bibfield  {journal}
  {\bibinfo  {journal} {Phys. Rev. Lett.}\ }\textbf {\bibinfo {volume} {94}},\
  \bibinfo {pages} {011601} (\bibinfo {year} {2005})},\ \Eprint
  {https://arxiv.org/abs/hep-ph/0408306} {arXiv:hep-ph/0408306} \BibitemShut
  {NoStop}%
\bibitem [{\citenamefont {Becirevic}\ \emph {et~al.}(2007)\citenamefont
  {Becirevic}, \citenamefont {Haas},\ and\ \citenamefont
  {Mescia}}]{Becirevic:2007cr}%
  \BibitemOpen
  \bibfield  {author} {\bibinfo {author} {\bibfnamefont {D.}~\bibnamefont
  {Becirevic}}, \bibinfo {author} {\bibfnamefont {B.}~\bibnamefont {Haas}},\
  and\ \bibinfo {author} {\bibfnamefont {F.}~\bibnamefont {Mescia}},\
  }\bibfield  {title} {\bibinfo {title} {{Semileptonic $D$ decays and lattice
  QCD}},\ }\href {https://doi.org/10.22323/1.042.0355} {\bibfield  {journal}
  {\bibinfo  {journal} {PoS}\ }\textbf {\bibinfo {volume} {LATTICE2007}},\
  \bibinfo {pages} {355} (\bibinfo {year} {2007})},\ \Eprint
  {https://arxiv.org/abs/0710.1741} {arXiv:0710.1741 [hep-lat]} \BibitemShut
  {NoStop}%
\bibitem [{\citenamefont {Di~Vita}\ \emph {et~al.}(2010)\citenamefont
  {Di~Vita}, \citenamefont {Haas}, \citenamefont {Lubicz}, \citenamefont
  {Mescia}, \citenamefont {Simula},\ and\ \citenamefont
  {Tarantino}}]{DiVita:2010mlb}%
  \BibitemOpen
  \bibfield  {author} {\bibinfo {author} {\bibfnamefont {S.}~\bibnamefont
  {Di~Vita}}, \bibinfo {author} {\bibfnamefont {B.}~\bibnamefont {Haas}},
  \bibinfo {author} {\bibfnamefont {V.}~\bibnamefont {Lubicz}}, \bibinfo
  {author} {\bibfnamefont {F.}~\bibnamefont {Mescia}}, \bibinfo {author}
  {\bibfnamefont {S.}~\bibnamefont {Simula}},\ and\ \bibinfo {author}
  {\bibfnamefont {C.}~\bibnamefont {Tarantino}} (\bibinfo {collaboration}
  {ETM}),\ }\bibfield  {title} {\bibinfo {title} {{Form factors of the $D \to
  \pi$ and $D \to K$ semileptonic decays with $N_f = 2$ twisted mass lattice
  QCD}},\ }\href@noop {} {\bibfield  {journal} {\bibinfo  {journal} {PoS}\
  }\textbf {\bibinfo {volume} {LATTICE2010}},\ \bibinfo {pages} {301} (\bibinfo
  {year} {2010})},\ \Eprint {https://arxiv.org/abs/1104.0869} {arXiv:1104.0869
  [hep-lat]} \BibitemShut {NoStop}%
\bibitem [{\citenamefont {Na}\ \emph {et~al.}(2010)\citenamefont {Na},
  \citenamefont {Davies}, \citenamefont {Follana}, \citenamefont {Lepage},\
  and\ \citenamefont {Shigemitsu}}]{Na:2010uf}%
  \BibitemOpen
  \bibfield  {author} {\bibinfo {author} {\bibfnamefont {H.}~\bibnamefont
  {Na}}, \bibinfo {author} {\bibfnamefont {C.~T.~H.}\ \bibnamefont {Davies}},
  \bibinfo {author} {\bibfnamefont {E.}~\bibnamefont {Follana}}, \bibinfo
  {author} {\bibfnamefont {G.~P.}\ \bibnamefont {Lepage}},\ and\ \bibinfo
  {author} {\bibfnamefont {J.}~\bibnamefont {Shigemitsu}},\ }\bibfield  {title}
  {\bibinfo {title} {{The $D \rightarrow K, l \nu$ semileptonic decay scalar
  form factor and $|V_{cs}|$ from lattice QCD}},\ }\href
  {https://doi.org/10.1103/PhysRevD.82.114506} {\bibfield  {journal} {\bibinfo
  {journal} {Phys. Rev. D}\ }\textbf {\bibinfo {volume} {82}},\ \bibinfo
  {pages} {114506} (\bibinfo {year} {2010})},\ \Eprint
  {https://arxiv.org/abs/1008.4562} {arXiv:1008.4562 [hep-lat]} \BibitemShut
  {NoStop}%
\bibitem [{\citenamefont {Na}\ \emph {et~al.}(2011)\citenamefont {Na},
  \citenamefont {Davies}, \citenamefont {Follana}, \citenamefont {Koponen},
  \citenamefont {Lepage},\ and\ \citenamefont {Shigemitsu}}]{Na:2011mc}%
  \BibitemOpen
  \bibfield  {author} {\bibinfo {author} {\bibfnamefont {H.}~\bibnamefont
  {Na}}, \bibinfo {author} {\bibfnamefont {C.~T.~H.}\ \bibnamefont {Davies}},
  \bibinfo {author} {\bibfnamefont {E.}~\bibnamefont {Follana}}, \bibinfo
  {author} {\bibfnamefont {J.}~\bibnamefont {Koponen}}, \bibinfo {author}
  {\bibfnamefont {G.~P.}\ \bibnamefont {Lepage}},\ and\ \bibinfo {author}
  {\bibfnamefont {J.}~\bibnamefont {Shigemitsu}},\ }\bibfield  {title}
  {\bibinfo {title} {{$D \rightarrow \pi, l \nu$ semileptonic decays,
  $|V_{cd}|$ and 2$^{\rm nd}$ row unitarity from lattice QCD}},\ }\href
  {https://doi.org/10.1103/PhysRevD.84.114505} {\bibfield  {journal} {\bibinfo
  {journal} {Phys. Rev. D}\ }\textbf {\bibinfo {volume} {84}},\ \bibinfo
  {pages} {114505} (\bibinfo {year} {2011})},\ \Eprint
  {https://arxiv.org/abs/1109.1501} {arXiv:1109.1501 [hep-lat]} \BibitemShut
  {NoStop}%
\bibitem [{\citenamefont {Koponen}\ \emph {et~al.}(2011)\citenamefont
  {Koponen}, \citenamefont {Davies}, \citenamefont {Donald}, \citenamefont
  {Follana}, \citenamefont {Lepage}, \citenamefont {Na},\ and\ \citenamefont
  {Shigemitsu}}]{Koponen:2011ev}%
  \BibitemOpen
  \bibfield  {author} {\bibinfo {author} {\bibfnamefont {J.}~\bibnamefont
  {Koponen}}, \bibinfo {author} {\bibfnamefont {C.~T.~H.}\ \bibnamefont
  {Davies}}, \bibinfo {author} {\bibfnamefont {G.}~\bibnamefont {Donald}},
  \bibinfo {author} {\bibfnamefont {E.}~\bibnamefont {Follana}}, \bibinfo
  {author} {\bibfnamefont {G.~P.}\ \bibnamefont {Lepage}}, \bibinfo {author}
  {\bibfnamefont {H.}~\bibnamefont {Na}},\ and\ \bibinfo {author}
  {\bibfnamefont {J.}~\bibnamefont {Shigemitsu}} (\bibinfo {collaboration}
  {HPQCD}),\ }\bibfield  {title} {\bibinfo {title} {{The $D\to K$ and $D\to\pi$
  semileptonic decay form factors from lattice QCD}},\ }\href
  {https://doi.org/10.22323/1.139.0286} {\bibfield  {journal} {\bibinfo
  {journal} {PoS}\ }\textbf {\bibinfo {volume} {LATTICE2011}},\ \bibinfo
  {pages} {286} (\bibinfo {year} {2011})},\ \Eprint
  {https://arxiv.org/abs/1111.0225} {arXiv:1111.0225 [hep-lat]} \BibitemShut
  {NoStop}%
\bibitem [{\citenamefont {Bailey}\ \emph {et~al.}(2012)\citenamefont {Bailey},
  \citenamefont {Du}, \citenamefont {El-Khadra}, \citenamefont {Gottlieb},
  \citenamefont {Jain}, \citenamefont {Kronfeld}, \citenamefont {Van~de
  Water},\ and\ \citenamefont {Zhou}}]{Bailey:2012sa}%
  \BibitemOpen
  \bibfield  {author} {\bibinfo {author} {\bibfnamefont {J.~A.}\ \bibnamefont
  {Bailey}}, \bibinfo {author} {\bibfnamefont {D.}~\bibnamefont {Du}}, \bibinfo
  {author} {\bibfnamefont {A.~X.}\ \bibnamefont {El-Khadra}}, \bibinfo {author}
  {\bibfnamefont {S.}~\bibnamefont {Gottlieb}}, \bibinfo {author}
  {\bibfnamefont {R.~D.}\ \bibnamefont {Jain}}, \bibinfo {author}
  {\bibfnamefont {A.~S.}\ \bibnamefont {Kronfeld}}, \bibinfo {author}
  {\bibfnamefont {R.~S.}\ \bibnamefont {Van~de Water}},\ and\ \bibinfo {author}
  {\bibfnamefont {R.}~\bibnamefont {Zhou}} (\bibinfo {collaboration} {Fermilab
  Lattice, MILC}),\ }\bibfield  {title} {\bibinfo {title} {{Charm semileptonic
  decays and $|V_{cs(d)}|$ from heavy clover quarks and 2+1 flavor asqtad
  staggered ensembles}},\ }\href {https://doi.org/10.22323/1.164.0272}
  {\bibfield  {journal} {\bibinfo  {journal} {PoS}\ }\textbf {\bibinfo {volume}
  {LATTICE2012}},\ \bibinfo {pages} {272} (\bibinfo {year} {2012})},\ \Eprint
  {https://arxiv.org/abs/1211.4964} {arXiv:1211.4964 [hep-lat]} \BibitemShut
  {NoStop}%
\bibitem [{\citenamefont {Koponen}\ \emph {et~al.}(2012)\citenamefont
  {Koponen}, \citenamefont {Davies},\ and\ \citenamefont
  {Donald}}]{Koponen:2012di}%
  \BibitemOpen
  \bibfield  {author} {\bibinfo {author} {\bibfnamefont {J.}~\bibnamefont
  {Koponen}}, \bibinfo {author} {\bibfnamefont {C.~T.~H.}\ \bibnamefont
  {Davies}},\ and\ \bibinfo {author} {\bibfnamefont {G.}~\bibnamefont {Donald}}
  (\bibinfo {collaboration} {HPQCD}),\ }\bibfield  {title} {\bibinfo {title}
  {{$D\to K$ and $D\to\pi$ semileptonic form factors from lattice QCD}},\ }in\
  \href@noop {} {\emph {\bibinfo {booktitle} {{5th International Workshop on
  Charm Physics}}}}\ (\bibinfo {year} {2012})\ \Eprint
  {https://arxiv.org/abs/1208.6242} {arXiv:1208.6242 [hep-lat]} \BibitemShut
  {NoStop}%
\bibitem [{\citenamefont {Koponen}\ \emph {et~al.}(2013)\citenamefont
  {Koponen}, \citenamefont {Davies}, \citenamefont {Donald}, \citenamefont
  {Follana}, \citenamefont {Lepage}, \citenamefont {Na},\ and\ \citenamefont
  {Shigemitsu}}]{Koponen:2013tua}%
  \BibitemOpen
  \bibfield  {author} {\bibinfo {author} {\bibfnamefont {J.}~\bibnamefont
  {Koponen}}, \bibinfo {author} {\bibfnamefont {C.~T.~H.}\ \bibnamefont
  {Davies}}, \bibinfo {author} {\bibfnamefont {G.~C.}\ \bibnamefont {Donald}},
  \bibinfo {author} {\bibfnamefont {E.}~\bibnamefont {Follana}}, \bibinfo
  {author} {\bibfnamefont {G.~P.}\ \bibnamefont {Lepage}}, \bibinfo {author}
  {\bibfnamefont {H.}~\bibnamefont {Na}},\ and\ \bibinfo {author}
  {\bibfnamefont {J.}~\bibnamefont {Shigemitsu}},\ }\bibfield  {title}
  {\bibinfo {title} {{The shape of the $D \to K$ semileptonic form factor from
  full lattice QCD and $V_{cs}$}},\ }\href@noop {} {\  (\bibinfo {year}
  {2013})},\ \Eprint {https://arxiv.org/abs/1305.1462} {arXiv:1305.1462
  [hep-lat]} \BibitemShut {NoStop}%
\bibitem [{\citenamefont {Primer}\ \emph {et~al.}(2016)\citenamefont {Primer},
  \citenamefont {Bernard}, \citenamefont {DeTar}, \citenamefont {El-Khadra},
  \citenamefont {G\'amiz}, \citenamefont {Komijani}, \citenamefont {Kronfeld},
  \citenamefont {Simone}, \citenamefont {Toussaint},\ and\ \citenamefont
  {Van~de Water}}]{LATTICE-FERMILAB:2015wnj}%
  \BibitemOpen
  \bibfield  {author} {\bibinfo {author} {\bibfnamefont {T.}~\bibnamefont
  {Primer}}, \bibinfo {author} {\bibfnamefont {C.}~\bibnamefont {Bernard}},
  \bibinfo {author} {\bibfnamefont {C.}~\bibnamefont {DeTar}}, \bibinfo
  {author} {\bibfnamefont {A.}~\bibnamefont {El-Khadra}}, \bibinfo {author}
  {\bibfnamefont {E.}~\bibnamefont {G\'amiz}}, \bibinfo {author} {\bibfnamefont
  {J.}~\bibnamefont {Komijani}}, \bibinfo {author} {\bibfnamefont
  {A.}~\bibnamefont {Kronfeld}}, \bibinfo {author} {\bibfnamefont
  {J.}~\bibnamefont {Simone}}, \bibinfo {author} {\bibfnamefont
  {D.}~\bibnamefont {Toussaint}},\ and\ \bibinfo {author} {\bibfnamefont
  {R.~S.}\ \bibnamefont {Van~de Water}} (\bibinfo {collaboration} {Fermilab
  Lattice, MILC}),\ }\bibfield  {title} {\bibinfo {title} {{$D$}-meson
  semileptonic form factors at zero momentum transfer in (2+1+1)-flavor lattice
  {QCD}},\ }\href@noop {} {\bibfield  {journal} {\bibinfo  {journal} {PoS}\
  }\textbf {\bibinfo {volume} {LATTICE2015}},\ \bibinfo {pages} {338} (\bibinfo
  {year} {2016})},\ \Eprint {https://arxiv.org/abs/1511.04000}
  {arXiv:1511.04000 [hep-lat]} \BibitemShut {NoStop}%
\bibitem [{\citenamefont {Primer}\ \emph {et~al.}(2017)\citenamefont {Primer}
  \emph {et~al.}}]{FermilabLattice:2017eea}%
  \BibitemOpen
  \bibfield  {author} {\bibinfo {author} {\bibfnamefont {T.}~\bibnamefont
  {Primer}} \emph {et~al.} (\bibinfo {collaboration} {Fermilab Lattice,
  MILC}),\ }\bibfield  {title} {\bibinfo {title} {{$D$ meson semileptonic form
  factors with HISQ valence and sea quarks}},\ }\href
  {https://doi.org/10.22323/1.256.0305} {\bibfield  {journal} {\bibinfo
  {journal} {PoS}\ }\textbf {\bibinfo {volume} {LATTICE2016}},\ \bibinfo
  {pages} {305} (\bibinfo {year} {2017})}\BibitemShut {NoStop}%
\bibitem [{\citenamefont {Kaneko}\ \emph {et~al.}(2018)\citenamefont {Kaneko},
  \citenamefont {Colquhoun}, \citenamefont {Fukaya},\ and\ \citenamefont
  {Hashimoto}}]{Kaneko:2017xgg}%
  \BibitemOpen
  \bibfield  {author} {\bibinfo {author} {\bibfnamefont {T.}~\bibnamefont
  {Kaneko}}, \bibinfo {author} {\bibfnamefont {B.}~\bibnamefont {Colquhoun}},
  \bibinfo {author} {\bibfnamefont {H.}~\bibnamefont {Fukaya}},\ and\ \bibinfo
  {author} {\bibfnamefont {S.}~\bibnamefont {Hashimoto}} (\bibinfo
  {collaboration} {JLQCD}),\ }\bibfield  {title} {\bibinfo {title} {{$D$} meson
  semileptonic form factors in {$N_f=3$ QCD} with {Möbius} domain-wall
  quarks},\ }\href {https://doi.org/10.1051/epjconf/201817513007} {\bibfield
  {journal} {\bibinfo  {journal} {EPJ Web Conf.}\ }\textbf {\bibinfo {volume}
  {175}},\ \bibinfo {pages} {13007} (\bibinfo {year} {2018})},\ \Eprint
  {https://arxiv.org/abs/1711.11235} {arXiv:1711.11235 [hep-lat]} \BibitemShut
  {NoStop}%
\bibitem [{\citenamefont {Lubicz}\ \emph {et~al.}(2017)\citenamefont {Lubicz},
  \citenamefont {Riggio}, \citenamefont {Salerno}, \citenamefont {Simula},\
  and\ \citenamefont {Tarantino}}]{Lubicz:2017syv}%
  \BibitemOpen
  \bibfield  {author} {\bibinfo {author} {\bibfnamefont {V.}~\bibnamefont
  {Lubicz}}, \bibinfo {author} {\bibfnamefont {L.}~\bibnamefont {Riggio}},
  \bibinfo {author} {\bibfnamefont {G.}~\bibnamefont {Salerno}}, \bibinfo
  {author} {\bibfnamefont {S.}~\bibnamefont {Simula}},\ and\ \bibinfo {author}
  {\bibfnamefont {C.}~\bibnamefont {Tarantino}} (\bibinfo {collaboration}
  {ETM}),\ }\bibfield  {title} {\bibinfo {title} {{Scalar and vector form
  factors of $D \to \pi(K) \ell \nu$ decays with $N_f=2+1+1$ twisted
  fermions}},\ }\href {https://doi.org/10.1103/PhysRevD.96.054514} {\bibfield
  {journal} {\bibinfo  {journal} {Phys. Rev. D}\ }\textbf {\bibinfo {volume}
  {96}},\ \bibinfo {pages} {054514} (\bibinfo {year} {2017})},\ \bibinfo {note}
  {[Errata: \href{https://doi.org/10.1103/PhysRevD.99.099902}{Phys. Rev. D
  \textbf{99}, 099902 (2019)} and
  \href{https://doi.org/10.1103/PhysRevD.100.079901}{Phys. Rev. D \textbf{100},
  079901 (2019)}]},\ \Eprint {https://arxiv.org/abs/1706.03017}
  {arXiv:1706.03017 [hep-lat]} \BibitemShut {NoStop}%
\bibitem [{\citenamefont {Lubicz}\ \emph {et~al.}(2018)\citenamefont {Lubicz},
  \citenamefont {Riggio}, \citenamefont {Salerno}, \citenamefont {Simula},\
  and\ \citenamefont {Tarantino}}]{Lubicz:2018rfs}%
  \BibitemOpen
  \bibfield  {author} {\bibinfo {author} {\bibfnamefont {V.}~\bibnamefont
  {Lubicz}}, \bibinfo {author} {\bibfnamefont {L.}~\bibnamefont {Riggio}},
  \bibinfo {author} {\bibfnamefont {G.}~\bibnamefont {Salerno}}, \bibinfo
  {author} {\bibfnamefont {S.}~\bibnamefont {Simula}},\ and\ \bibinfo {author}
  {\bibfnamefont {C.}~\bibnamefont {Tarantino}} (\bibinfo {collaboration}
  {ETM}),\ }\bibfield  {title} {\bibinfo {title} {{Tensor form factor of $D \to
  \pi(K) \ell \nu$ and $D \to \pi(K) \ell \ell$ decays with $N_f=2+1+1$
  twisted-mass fermions}},\ }\href {https://doi.org/10.1103/PhysRevD.98.014516}
  {\bibfield  {journal} {\bibinfo  {journal} {Phys. Rev. D}\ }\textbf {\bibinfo
  {volume} {98}},\ \bibinfo {pages} {014516} (\bibinfo {year} {2018})},\
  \Eprint {https://arxiv.org/abs/1803.04807} {arXiv:1803.04807 [hep-lat]}
  \BibitemShut {NoStop}%
\bibitem [{\citenamefont {Bazavov}\ \emph {et~al.}(2014)\citenamefont {Bazavov}
  \emph {et~al.}}]{FermilabLattice:2014tsy}%
  \BibitemOpen
  \bibfield  {author} {\bibinfo {author} {\bibfnamefont {A.}~\bibnamefont
  {Bazavov}} \emph {et~al.} (\bibinfo {collaboration} {Fermilab Lattice,
  MILC}),\ }\bibfield  {title} {\bibinfo {title} {Charmed and light
  pseudoscalar meson decay constants from four-flavor lattice {QCD} with
  physical light quarks},\ }\href {https://doi.org/10.1103/PhysRevD.90.074509}
  {\bibfield  {journal} {\bibinfo  {journal} {Phys. Rev. D}\ }\textbf {\bibinfo
  {volume} {90}},\ \bibinfo {pages} {074509} (\bibinfo {year} {2014})},\
  \Eprint {https://arxiv.org/abs/1407.3772} {arXiv:1407.3772 [hep-lat]}
  \BibitemShut {NoStop}%
\bibitem [{\citenamefont {Li}\ \emph {et~al.}(2019)\citenamefont {Li} \emph
  {et~al.}}]{FermilabLattice:2019ycs}%
  \BibitemOpen
  \bibfield  {author} {\bibinfo {author} {\bibfnamefont {R.}~\bibnamefont {Li}}
  \emph {et~al.} (\bibinfo {collaboration} {Fermilab Lattice, MILC}),\
  }\bibfield  {title} {\bibinfo {title} {{$D$} meson semileptonic decay form
  factors at $q^2 = 0$},\ }\href {https://doi.org/10.22323/1.334.0269}
  {\bibfield  {journal} {\bibinfo  {journal} {PoS}\ }\textbf {\bibinfo {volume}
  {LATTICE2018}},\ \bibinfo {pages} {269} (\bibinfo {year} {2019})},\ \Eprint
  {https://arxiv.org/abs/1901.08989} {arXiv:1901.08989 [hep-lat]} \BibitemShut
  {NoStop}%
\bibitem [{\citenamefont {Jay}\ \emph {et~al.}(2022)\citenamefont {Jay},
  \citenamefont {Lytle}, \citenamefont {DeTar}, \citenamefont {El-Khadra},
  \citenamefont {Gamiz}, \citenamefont {Gelzer}, \citenamefont {Gottlieb},
  \citenamefont {Kronfeld}, \citenamefont {Simone},\ and\ \citenamefont
  {Vaquero}}]{FermilabLattice:2021bxu}%
  \BibitemOpen
  \bibfield  {author} {\bibinfo {author} {\bibfnamefont {W.~I.}\ \bibnamefont
  {Jay}}, \bibinfo {author} {\bibfnamefont {A.}~\bibnamefont {Lytle}}, \bibinfo
  {author} {\bibfnamefont {C.}~\bibnamefont {DeTar}}, \bibinfo {author}
  {\bibfnamefont {A.~X.}\ \bibnamefont {El-Khadra}}, \bibinfo {author}
  {\bibfnamefont {E.}~\bibnamefont {Gamiz}}, \bibinfo {author} {\bibfnamefont
  {Z.}~\bibnamefont {Gelzer}}, \bibinfo {author} {\bibfnamefont
  {S.}~\bibnamefont {Gottlieb}}, \bibinfo {author} {\bibfnamefont
  {A.}~\bibnamefont {Kronfeld}}, \bibinfo {author} {\bibfnamefont
  {J.}~\bibnamefont {Simone}},\ and\ \bibinfo {author} {\bibfnamefont
  {A.}~\bibnamefont {Vaquero}} (\bibinfo {collaboration} {Fermilab Lattice,
  MILC}),\ }\bibfield  {title} {\bibinfo {title} {{$B$- and $D$-meson
  semileptonic decays with highly improved staggered quarks}},\ }\href
  {https://doi.org/10.22323/1.396.0109} {\bibfield  {journal} {\bibinfo
  {journal} {PoS}\ }\textbf {\bibinfo {volume} {LATTICE2021}},\ \bibinfo
  {pages} {109} (\bibinfo {year} {2022})},\ \Eprint
  {https://arxiv.org/abs/2111.05184} {arXiv:2111.05184 [hep-lat]} \BibitemShut
  {NoStop}%
\bibitem [{\citenamefont {Chakraborty}\ \emph {et~al.}(2021)\citenamefont
  {Chakraborty}, \citenamefont {Parrott}, \citenamefont {Bouchard},
  \citenamefont {Davies}, \citenamefont {Koponen},\ and\ \citenamefont
  {Lepage}}]{Chakraborty:2021qav}%
  \BibitemOpen
  \bibfield  {author} {\bibinfo {author} {\bibfnamefont {B.}~\bibnamefont
  {Chakraborty}}, \bibinfo {author} {\bibfnamefont {W.~G.}\ \bibnamefont
  {Parrott}}, \bibinfo {author} {\bibfnamefont {C.}~\bibnamefont {Bouchard}},
  \bibinfo {author} {\bibfnamefont {C.~T.~H.}\ \bibnamefont {Davies}}, \bibinfo
  {author} {\bibfnamefont {J.}~\bibnamefont {Koponen}},\ and\ \bibinfo {author}
  {\bibfnamefont {G.~P.}\ \bibnamefont {Lepage}} (\bibinfo {collaboration}
  {HPQCD}),\ }\bibfield  {title} {\bibinfo {title} {{Improved $V_{cs}$
  determination using precise lattice QCD form factors for $D\to K\ell\nu$}},\
  }\href {https://doi.org/10.1103/PhysRevD.104.034505} {\bibfield  {journal}
  {\bibinfo  {journal} {Phys. Rev. D}\ }\textbf {\bibinfo {volume} {104}},\
  \bibinfo {pages} {034505} (\bibinfo {year} {2021})},\ \Eprint
  {https://arxiv.org/abs/2104.09883} {arXiv:2104.09883 [hep-lat]} \BibitemShut
  {NoStop}%
\bibitem [{\citenamefont {Parrott}\ \emph {et~al.}(2022)\citenamefont
  {Parrott}, \citenamefont {Bouchard},\ and\ \citenamefont
  {Davies}}]{Parrott:2022rgu}%
  \BibitemOpen
  \bibfield  {author} {\bibinfo {author} {\bibfnamefont {W.~G.}\ \bibnamefont
  {Parrott}}, \bibinfo {author} {\bibfnamefont {C.}~\bibnamefont {Bouchard}},\
  and\ \bibinfo {author} {\bibfnamefont {C.~T.~H.}\ \bibnamefont {Davies}}
  (\bibinfo {collaboration} {HPQCD}),\ }\bibfield  {title} {\bibinfo {title}
  {{$B\to K$ and $D\to K$ form factors from fully relativistic lattice QCD}},\
  }\href@noop {} {\  (\bibinfo {year} {2022})},\ \Eprint
  {https://arxiv.org/abs/2207.12468} {arXiv:2207.12468 [hep-lat]} \BibitemShut
  {NoStop}%
\bibitem [{\citenamefont {Sirlin}(1982)}]{Sirlin:1981ie}%
  \BibitemOpen
  \bibfield  {author} {\bibinfo {author} {\bibfnamefont {A.}~\bibnamefont
  {Sirlin}},\ }\bibfield  {title} {\bibinfo {title} {Large {$m_W$, $m_Z$}
  behavior of the {$O(\alpha)$} corrections to semileptonic processes mediated
  by {$W$}},\ }\href {https://doi.org/10.1016/0550-3213(82)90303-0} {\bibfield
  {journal} {\bibinfo  {journal} {Nucl. Phys. B}\ }\textbf {\bibinfo {volume}
  {196}},\ \bibinfo {pages} {83} (\bibinfo {year} {1982})}\BibitemShut
  {NoStop}%
\bibitem [{\citenamefont {Follana}\ \emph {et~al.}(2007)\citenamefont
  {Follana}, \citenamefont {Mason}, \citenamefont {Davies}, \citenamefont
  {Hornbostel}, \citenamefont {Lepage}, \citenamefont {Shigemitsu},
  \citenamefont {Trottier},\ and\ \citenamefont {Wong}}]{Follana:2006rc}%
  \BibitemOpen
  \bibfield  {author} {\bibinfo {author} {\bibfnamefont {E.}~\bibnamefont
  {Follana}}, \bibinfo {author} {\bibfnamefont {Q.}~\bibnamefont {Mason}},
  \bibinfo {author} {\bibfnamefont {C.}~\bibnamefont {Davies}}, \bibinfo
  {author} {\bibfnamefont {K.}~\bibnamefont {Hornbostel}}, \bibinfo {author}
  {\bibfnamefont {G.~P.}\ \bibnamefont {Lepage}}, \bibinfo {author}
  {\bibfnamefont {J.}~\bibnamefont {Shigemitsu}}, \bibinfo {author}
  {\bibfnamefont {H.}~\bibnamefont {Trottier}},\ and\ \bibinfo {author}
  {\bibfnamefont {K.}~\bibnamefont {Wong}} (\bibinfo {collaboration} {HPQCD}),\
  }\bibfield  {title} {\bibinfo {title} {{Highly improved staggered quarks on
  the lattice, with applications to charm physics}},\ }\href
  {https://doi.org/10.1103/PhysRevD.75.054502} {\bibfield  {journal} {\bibinfo
  {journal} {Phys. Rev. D}\ }\textbf {\bibinfo {volume} {75}},\ \bibinfo
  {pages} {054502} (\bibinfo {year} {2007})},\ \Eprint
  {https://arxiv.org/abs/hep-lat/0610092} {arXiv:hep-lat/0610092} \BibitemShut
  {NoStop}%
\bibitem [{\citenamefont {Karsten}\ and\ \citenamefont
  {Smit}(1981)}]{Karsten:1980wd}%
  \BibitemOpen
  \bibfield  {author} {\bibinfo {author} {\bibfnamefont {L.~H.}\ \bibnamefont
  {Karsten}}\ and\ \bibinfo {author} {\bibfnamefont {J.}~\bibnamefont {Smit}},\
  }\bibfield  {title} {\bibinfo {title} {Lattice fermions: Species doubling,
  chiral invariance, and the triangle anomaly},\ }\href
  {https://doi.org/10.1016/0550-3213(81)90549-6} {\bibfield  {journal}
  {\bibinfo  {journal} {Nucl. Phys. B}\ }\textbf {\bibinfo {volume} {183}},\
  \bibinfo {pages} {103} (\bibinfo {year} {1981})}\BibitemShut {NoStop}%
\bibitem [{\citenamefont {Smit}\ and\ \citenamefont
  {Vink}(1988)}]{Smit:1987zh}%
  \BibitemOpen
  \bibfield  {author} {\bibinfo {author} {\bibfnamefont {J.}~\bibnamefont
  {Smit}}\ and\ \bibinfo {author} {\bibfnamefont {J.~C.}\ \bibnamefont
  {Vink}},\ }\bibfield  {title} {\bibinfo {title} {Renormalized
  {Ward-Takahashi} relations and topological susceptibility with staggered
  fermions},\ }\href {https://doi.org/10.1016/0550-3213(88)90354-9} {\bibfield
  {journal} {\bibinfo  {journal} {Nucl. Phys. B}\ }\textbf {\bibinfo {volume}
  {298}},\ \bibinfo {pages} {557} (\bibinfo {year} {1988})}\BibitemShut
  {NoStop}%
\bibitem [{\citenamefont {Bazavov}\ \emph {et~al.}(2010)\citenamefont {Bazavov}
  \emph {et~al.}}]{MILC:2010pul}%
  \BibitemOpen
  \bibfield  {author} {\bibinfo {author} {\bibfnamefont {A.}~\bibnamefont
  {Bazavov}} \emph {et~al.} (\bibinfo {collaboration} {MILC}),\ }\bibfield
  {title} {\bibinfo {title} {{Scaling studies of QCD with the dynamical HISQ
  action}},\ }\href {https://doi.org/10.1103/PhysRevD.82.074501} {\bibfield
  {journal} {\bibinfo  {journal} {Phys. Rev. D}\ }\textbf {\bibinfo {volume}
  {82}},\ \bibinfo {pages} {074501} (\bibinfo {year} {2010})},\ \Eprint
  {https://arxiv.org/abs/1004.0342} {arXiv:1004.0342 [hep-lat]} \BibitemShut
  {NoStop}%
\bibitem [{\citenamefont {Bazavov}\ \emph {et~al.}(2013)\citenamefont {Bazavov}
  \emph {et~al.}}]{MILC:2012znn}%
  \BibitemOpen
  \bibfield  {author} {\bibinfo {author} {\bibfnamefont {A.}~\bibnamefont
  {Bazavov}} \emph {et~al.} (\bibinfo {collaboration} {MILC}),\ }\bibfield
  {title} {\bibinfo {title} {Lattice {QCD} ensembles with four flavors of
  highly improved staggered quarks},\ }\href
  {https://doi.org/10.1103/PhysRevD.87.054505} {\bibfield  {journal} {\bibinfo
  {journal} {Phys. Rev. D}\ }\textbf {\bibinfo {volume} {87}},\ \bibinfo
  {pages} {054505} (\bibinfo {year} {2013})},\ \Eprint
  {https://arxiv.org/abs/1212.4768} {arXiv:1212.4768 [hep-lat]} \BibitemShut
  {NoStop}%
\bibitem [{\citenamefont {Bazavov}\ \emph
  {et~al.}(2016{\natexlab{a}})\citenamefont {Bazavov} \emph
  {et~al.}}]{MILC:2015tqx}%
  \BibitemOpen
  \bibfield  {author} {\bibinfo {author} {\bibfnamefont {A.}~\bibnamefont
  {Bazavov}} \emph {et~al.} (\bibinfo {collaboration} {MILC}),\ }\bibfield
  {title} {\bibinfo {title} {{Gradient flow and scale setting on MILC HISQ
  ensembles}},\ }\href {https://doi.org/10.1103/PhysRevD.93.094510} {\bibfield
  {journal} {\bibinfo  {journal} {Phys. Rev. D}\ }\textbf {\bibinfo {volume}
  {93}},\ \bibinfo {pages} {094510} (\bibinfo {year} {2016}{\natexlab{a}})},\
  \Eprint {https://arxiv.org/abs/1503.02769} {arXiv:1503.02769 [hep-lat]}
  \BibitemShut {NoStop}%
\bibitem [{\citenamefont {Brown}(2018)}]{Brown:2018jtv}%
  \BibitemOpen
  \bibfield  {author} {\bibinfo {author} {\bibfnamefont {N.~J.}\ \bibnamefont
  {Brown}},\ }\emph {\bibinfo {title} {Lattice Scales from Gradient Flow and
  Chiral Analysis on the {MILC Collaboration's HISQ} Ensembles}},\ \href
  {https://doi.org/10.7936/K7S181ZQ} {Ph.D. thesis},\ \bibinfo  {school}
  {Washington U., St. Louis} (\bibinfo {year} {2018})\BibitemShut {NoStop}%
\bibitem [{\citenamefont {Albanese}\ \emph {et~al.}(1987)\citenamefont
  {Albanese} \emph {et~al.}}]{APE:1987ehd}%
  \BibitemOpen
  \bibfield  {author} {\bibinfo {author} {\bibfnamefont {M.}~\bibnamefont
  {Albanese}} \emph {et~al.} (\bibinfo {collaboration} {APE}),\ }\bibfield
  {title} {\bibinfo {title} {Glueball masses and string tension in lattice
  {QCD}},\ }\href {https://doi.org/10.1016/0370-2693(87)91160-9} {\bibfield
  {journal} {\bibinfo  {journal} {Phys. Lett. B}\ }\textbf {\bibinfo {volume}
  {192}},\ \bibinfo {pages} {163} (\bibinfo {year} {1987})}\BibitemShut
  {NoStop}%
\bibitem [{\citenamefont {Aubin}\ \emph {et~al.}(2004)\citenamefont {Aubin},
  \citenamefont {Bernard}, \citenamefont {DeTar}, \citenamefont {Osborn},
  \citenamefont {Gottlieb}, \citenamefont {Gregory}, \citenamefont {Toussaint},
  \citenamefont {Heller}, \citenamefont {Hetrick},\ and\ \citenamefont
  {Sugar}}]{MILC:2004qnl}%
  \BibitemOpen
  \bibfield  {author} {\bibinfo {author} {\bibfnamefont {C.}~\bibnamefont
  {Aubin}}, \bibinfo {author} {\bibfnamefont {C.}~\bibnamefont {Bernard}},
  \bibinfo {author} {\bibfnamefont {C.~E.}\ \bibnamefont {DeTar}}, \bibinfo
  {author} {\bibfnamefont {J.}~\bibnamefont {Osborn}}, \bibinfo {author}
  {\bibfnamefont {S.}~\bibnamefont {Gottlieb}}, \bibinfo {author}
  {\bibfnamefont {E.~B.}\ \bibnamefont {Gregory}}, \bibinfo {author}
  {\bibfnamefont {D.}~\bibnamefont {Toussaint}}, \bibinfo {author}
  {\bibfnamefont {U.~M.}\ \bibnamefont {Heller}}, \bibinfo {author}
  {\bibfnamefont {J.~E.}\ \bibnamefont {Hetrick}},\ and\ \bibinfo {author}
  {\bibfnamefont {R.}~\bibnamefont {Sugar}} (\bibinfo {collaboration} {MILC}),\
  }\bibfield  {title} {\bibinfo {title} {{Light pseudoscalar decay constants,
  quark masses, and low energy constants from three-flavor lattice QCD}},\
  }\href {https://doi.org/10.1103/PhysRevD.70.114501} {\bibfield  {journal}
  {\bibinfo  {journal} {Phys. Rev. D}\ }\textbf {\bibinfo {volume} {70}},\
  \bibinfo {pages} {114501} (\bibinfo {year} {2004})},\ \Eprint
  {https://arxiv.org/abs/hep-lat/0407028} {arXiv:hep-lat/0407028} \BibitemShut
  {NoStop}%
\bibitem [{\citenamefont {Bali}\ \emph {et~al.}(2010)\citenamefont {Bali},
  \citenamefont {Collins},\ and\ \citenamefont {Schäfer}}]{Bali:2009hu}%
  \BibitemOpen
  \bibfield  {author} {\bibinfo {author} {\bibfnamefont {G.~S.}\ \bibnamefont
  {Bali}}, \bibinfo {author} {\bibfnamefont {S.}~\bibnamefont {Collins}},\ and\
  \bibinfo {author} {\bibfnamefont {A.}~\bibnamefont {Schäfer}},\ }\bibfield
  {title} {\bibinfo {title} {Effective noise reduction techniques for
  disconnected loops in lattice {QCD}},\ }\href
  {https://doi.org/10.1016/j.cpc.2010.05.008} {\bibfield  {journal} {\bibinfo
  {journal} {Comput. Phys. Commun.}\ }\textbf {\bibinfo {volume} {181}},\
  \bibinfo {pages} {1570} (\bibinfo {year} {2010})},\ \Eprint
  {https://arxiv.org/abs/0910.3970} {arXiv:0910.3970 [hep-lat]} \BibitemShut
  {NoStop}%
\bibitem [{\citenamefont {Alexandrou}\ \emph {et~al.}(2012)\citenamefont
  {Alexandrou}, \citenamefont {Hadjiyiannakou}, \citenamefont {Koutsou},
  \citenamefont {O'Cais},\ and\ \citenamefont
  {Strelchenko}}]{Alexandrou:2012zz}%
  \BibitemOpen
  \bibfield  {author} {\bibinfo {author} {\bibfnamefont {C.}~\bibnamefont
  {Alexandrou}}, \bibinfo {author} {\bibfnamefont {K.}~\bibnamefont
  {Hadjiyiannakou}}, \bibinfo {author} {\bibfnamefont {G.}~\bibnamefont
  {Koutsou}}, \bibinfo {author} {\bibfnamefont {A.}~\bibnamefont {O'Cais}},\
  and\ \bibinfo {author} {\bibfnamefont {A.}~\bibnamefont {Strelchenko}},\
  }\bibfield  {title} {\bibinfo {title} {{Evaluation of fermion loops applied
  to the calculation of the $\eta'$ mass and the nucleon scalar and
  electromagnetic form factors}},\ }\href
  {https://doi.org/10.1016/j.cpc.2012.01.023} {\bibfield  {journal} {\bibinfo
  {journal} {Comput. Phys. Commun.}\ }\textbf {\bibinfo {volume} {183}},\
  \bibinfo {pages} {1215} (\bibinfo {year} {2012})},\ \Eprint
  {https://arxiv.org/abs/1108.2473} {arXiv:1108.2473 [hep-lat]} \BibitemShut
  {NoStop}%
\bibitem [{\citenamefont {Lepage}\ \emph {et~al.}(2002)\citenamefont {Lepage},
  \citenamefont {Clark}, \citenamefont {Davies}, \citenamefont {Hornbostel},
  \citenamefont {Mackenzie}, \citenamefont {Morningstar},\ and\ \citenamefont
  {Trottier}}]{Lepage:2001ym}%
  \BibitemOpen
  \bibfield  {author} {\bibinfo {author} {\bibfnamefont {G.~P.}\ \bibnamefont
  {Lepage}}, \bibinfo {author} {\bibfnamefont {B.}~\bibnamefont {Clark}},
  \bibinfo {author} {\bibfnamefont {C.~T.~H.}\ \bibnamefont {Davies}}, \bibinfo
  {author} {\bibfnamefont {K.}~\bibnamefont {Hornbostel}}, \bibinfo {author}
  {\bibfnamefont {P.~B.}\ \bibnamefont {Mackenzie}}, \bibinfo {author}
  {\bibfnamefont {C.}~\bibnamefont {Morningstar}},\ and\ \bibinfo {author}
  {\bibfnamefont {H.}~\bibnamefont {Trottier}},\ }\bibfield  {title} {\bibinfo
  {title} {{Constrained curve fitting}},\ }\href
  {https://doi.org/10.1016/S0920-5632(01)01638-3} {\bibfield  {journal}
  {\bibinfo  {journal} {Nucl. Phys. B Proc. Suppl.}\ }\textbf {\bibinfo
  {volume} {106}},\ \bibinfo {pages} {12} (\bibinfo {year} {2002})},\ \Eprint
  {https://arxiv.org/abs/hep-lat/0110175} {arXiv:hep-lat/0110175} \BibitemShut
  {NoStop}%
\bibitem [{\citenamefont {Morningstar}(2002)}]{Morningstar:2001je}%
  \BibitemOpen
  \bibfield  {author} {\bibinfo {author} {\bibfnamefont {C.}~\bibnamefont
  {Morningstar}},\ }\bibfield  {title} {\bibinfo {title} {{Bayesian curve
  fitting for lattice gauge theorists}},\ }\href
  {https://doi.org/10.1016/S0920-5632(02)01413-5} {\bibfield  {journal}
  {\bibinfo  {journal} {Nucl. Phys. B Proc. Suppl.}\ }\textbf {\bibinfo
  {volume} {109A}},\ \bibinfo {pages} {185} (\bibinfo {year} {2002})},\ \Eprint
  {https://arxiv.org/abs/hep-lat/0112023} {arXiv:hep-lat/0112023} \BibitemShut
  {NoStop}%
\bibitem [{\citenamefont {Jay}\ and\ \citenamefont {Neil}(2021)}]{Jay:2020jkz}%
  \BibitemOpen
  \bibfield  {author} {\bibinfo {author} {\bibfnamefont {W.~I.}\ \bibnamefont
  {Jay}}\ and\ \bibinfo {author} {\bibfnamefont {E.~T.}\ \bibnamefont {Neil}},\
  }\bibfield  {title} {\bibinfo {title} {{Bayesian model averaging for analysis
  of lattice field theory results}},\ }\href
  {https://doi.org/10.1103/PhysRevD.103.114502} {\bibfield  {journal} {\bibinfo
   {journal} {Phys. Rev. D}\ }\textbf {\bibinfo {volume} {103}},\ \bibinfo
  {pages} {114502} (\bibinfo {year} {2021})},\ \Eprint
  {https://arxiv.org/abs/2008.01069} {arXiv:2008.01069 [stat.ME]} \BibitemShut
  {NoStop}%
\bibitem [{\citenamefont {Ledoit}\ and\ \citenamefont
  {Wolf}(2018)}]{Ledoit:2018}%
  \BibitemOpen
  \bibfield  {author} {\bibinfo {author} {\bibfnamefont {O.}~\bibnamefont
  {Ledoit}}\ and\ \bibinfo {author} {\bibfnamefont {M.}~\bibnamefont {Wolf}},\
  }\bibfield  {title} {\bibinfo {title} {Direct nonlinear shrinkage estimation
  of large-dimensional covariance matrices},\ }\bibfield  {journal} {\bibinfo
  {journal} {Working Paper Series, Department of Economics, University of
  Zurich}\ }\href {https://doi.org/http://dx.doi.org/10.5167/uzh-139880}
  {http://dx.doi.org/10.5167/uzh-139880} (\bibinfo {year} {2018})\BibitemShut
  {NoStop}%
\bibitem [{\citenamefont {Toussaint}\ and\ \citenamefont
  {Freeman}(2008)}]{Toussaint:2008ke}%
  \BibitemOpen
  \bibfield  {author} {\bibinfo {author} {\bibfnamefont {D.}~\bibnamefont
  {Toussaint}}\ and\ \bibinfo {author} {\bibfnamefont {W.}~\bibnamefont
  {Freeman}},\ }\bibfield  {title} {\bibinfo {title} {{Sample size effects in
  multivariate fitting of correlated data}},\ }\href@noop {} {\  (\bibinfo
  {year} {2008})},\ \Eprint {https://arxiv.org/abs/0808.2211} {arXiv:0808.2211
  [hep-lat]} \BibitemShut {NoStop}%
\bibitem [{\citenamefont {Bazavov}\ \emph
  {et~al.}(2016{\natexlab{b}})\citenamefont {Bazavov} \emph
  {et~al.}}]{FermilabLattice:2016ipl}%
  \BibitemOpen
  \bibfield  {author} {\bibinfo {author} {\bibfnamefont {A.}~\bibnamefont
  {Bazavov}} \emph {et~al.} (\bibinfo {collaboration} {Fermilab Lattice,
  MILC}),\ }\bibfield  {title} {\bibinfo {title} {{$B^0_{(s)}$-mixing matrix
  elements from lattice QCD for the Standard Model and beyond}},\ }\href
  {https://doi.org/10.1103/PhysRevD.93.113016} {\bibfield  {journal} {\bibinfo
  {journal} {Phys. Rev. D}\ }\textbf {\bibinfo {volume} {93}},\ \bibinfo
  {pages} {113016} (\bibinfo {year} {2016}{\natexlab{b}})},\ \Eprint
  {https://arxiv.org/abs/1602.03560} {arXiv:1602.03560 [hep-lat]} \BibitemShut
  {NoStop}%
\bibitem [{\citenamefont {Aubin}\ and\ \citenamefont
  {Bernard}(2006)}]{Aubin:2005aq}%
  \BibitemOpen
  \bibfield  {author} {\bibinfo {author} {\bibfnamefont {C.}~\bibnamefont
  {Aubin}}\ and\ \bibinfo {author} {\bibfnamefont {C.}~\bibnamefont
  {Bernard}},\ }\bibfield  {title} {\bibinfo {title} {{Staggered chiral
  perturbation theory for heavy-light mesons}},\ }\href
  {https://doi.org/10.1103/PhysRevD.73.014515} {\bibfield  {journal} {\bibinfo
  {journal} {Phys. Rev. D}\ }\textbf {\bibinfo {volume} {73}},\ \bibinfo
  {pages} {014515} (\bibinfo {year} {2006})},\ \Eprint
  {https://arxiv.org/abs/hep-lat/0510088} {arXiv:hep-lat/0510088} \BibitemShut
  {NoStop}%
\bibitem [{\citenamefont {Aubin}\ and\ \citenamefont
  {Bernard}(2007)}]{Aubin:2007mc}%
  \BibitemOpen
  \bibfield  {author} {\bibinfo {author} {\bibfnamefont {C.}~\bibnamefont
  {Aubin}}\ and\ \bibinfo {author} {\bibfnamefont {C.}~\bibnamefont
  {Bernard}},\ }\bibfield  {title} {\bibinfo {title} {{Heavy-light semileptonic
  decays in staggered chiral perturbation theory}},\ }\href
  {https://doi.org/10.1103/PhysRevD.76.014002} {\bibfield  {journal} {\bibinfo
  {journal} {Phys. Rev. D}\ }\textbf {\bibinfo {volume} {76}},\ \bibinfo
  {pages} {014002} (\bibinfo {year} {2007})},\ \Eprint
  {https://arxiv.org/abs/0704.0795} {arXiv:0704.0795 [hep-lat]} \BibitemShut
  {NoStop}%
\bibitem [{\citenamefont {Bailey}\ \emph {et~al.}(2016)\citenamefont {Bailey}
  \emph {et~al.}}]{Bailey:2015dka}%
  \BibitemOpen
  \bibfield  {author} {\bibinfo {author} {\bibfnamefont {J.~A.}\ \bibnamefont
  {Bailey}} \emph {et~al.} (\bibinfo {collaboration} {Fermilab Lattice,
  MILC}),\ }\bibfield  {title} {\bibinfo {title} {{$B\to Kl^+l^-$} decay form
  factors from three-flavor lattice {QCD}},\ }\href
  {https://doi.org/10.1103/PhysRevD.93.025026} {\bibfield  {journal} {\bibinfo
  {journal} {Phys. Rev. D}\ }\textbf {\bibinfo {volume} {93}},\ \bibinfo
  {pages} {025026} (\bibinfo {year} {2016})},\ \Eprint
  {https://arxiv.org/abs/1509.06235} {arXiv:1509.06235 [hep-lat]} \BibitemShut
  {NoStop}%
\bibitem [{\citenamefont {Bazavov}\ \emph
  {et~al.}(2019{\natexlab{b}})\citenamefont {Bazavov} \emph
  {et~al.}}]{FermilabLattice:2019ikx}%
  \BibitemOpen
  \bibfield  {author} {\bibinfo {author} {\bibfnamefont {A.}~\bibnamefont
  {Bazavov}} \emph {et~al.} (\bibinfo {collaboration} {Fermilab Lattice,
  MILC}),\ }\bibfield  {title} {\bibinfo {title} {{$B_s\to K\ell\nu$ decay from
  lattice QCD}},\ }\href {https://doi.org/10.1103/PhysRevD.100.034501}
  {\bibfield  {journal} {\bibinfo  {journal} {Phys. Rev. D}\ }\textbf {\bibinfo
  {volume} {100}},\ \bibinfo {pages} {034501} (\bibinfo {year}
  {2019}{\natexlab{b}})},\ \Eprint {https://arxiv.org/abs/1901.02561}
  {arXiv:1901.02561 [hep-lat]} \BibitemShut {NoStop}%
\bibitem [{\citenamefont {Flynn}\ and\ \citenamefont
  {Sachrajda}(2009)}]{Flynn:2008tg}%
  \BibitemOpen
  \bibfield  {author} {\bibinfo {author} {\bibfnamefont {J.~M.}\ \bibnamefont
  {Flynn}}\ and\ \bibinfo {author} {\bibfnamefont {C.~T.}\ \bibnamefont
  {Sachrajda}} (\bibinfo {collaboration} {RBC, UKQCD}),\ }\bibfield  {title}
  {\bibinfo {title} {{SU(2) chiral perturbation theory for $K_{l3}$ decay
  amplitudes}},\ }\href {https://doi.org/10.1016/j.nuclphysb.2008.12.001}
  {\bibfield  {journal} {\bibinfo  {journal} {Nucl. Phys. B}\ }\textbf
  {\bibinfo {volume} {812}},\ \bibinfo {pages} {64} (\bibinfo {year} {2009})},\
  \Eprint {https://arxiv.org/abs/0809.1229} {arXiv:0809.1229 [hep-ph]}
  \BibitemShut {NoStop}%
\bibitem [{\citenamefont {Bijnens}\ and\ \citenamefont
  {Jemos}(2010)}]{Bijnens:2010ws}%
  \BibitemOpen
  \bibfield  {author} {\bibinfo {author} {\bibfnamefont {J.}~\bibnamefont
  {Bijnens}}\ and\ \bibinfo {author} {\bibfnamefont {I.}~\bibnamefont
  {Jemos}},\ }\bibfield  {title} {\bibinfo {title} {Hard pion chiral
  perturbation theory for {$B\to\pi$ and $D\to\pi$} form factors},\ }\href
  {https://doi.org/10.1016/j.nuclphysb.2010.06.021} {\bibfield  {journal}
  {\bibinfo  {journal} {Nucl. Phys. B}\ }\textbf {\bibinfo {volume} {840}},\
  \bibinfo {pages} {54} (\bibinfo {year} {2010})},\ \bibinfo {note} {[Erratum:
  \href{https://doi.org/10.1016/j.nuclphysb.2010.10.024}{Nucl. Phys. B
  \textbf{844}, 182 (2011)}]},\ \Eprint {https://arxiv.org/abs/1006.1197}
  {arXiv:1006.1197 [hep-ph]} \BibitemShut {NoStop}%
\bibitem [{\citenamefont {Bijnens}\ and\ \citenamefont
  {Jemos}(2011)}]{Bijnens:2010jg}%
  \BibitemOpen
  \bibfield  {author} {\bibinfo {author} {\bibfnamefont {J.}~\bibnamefont
  {Bijnens}}\ and\ \bibinfo {author} {\bibfnamefont {I.}~\bibnamefont
  {Jemos}},\ }\bibfield  {title} {\bibinfo {title} {Vector form factors in hard
  pion chiral perturbation theory},\ }\href
  {https://doi.org/10.1016/j.nuclphysb.2010.12.012} {\bibfield  {journal}
  {\bibinfo  {journal} {Nucl. Phys. B}\ }\textbf {\bibinfo {volume} {846}},\
  \bibinfo {pages} {145} (\bibinfo {year} {2011})},\ \Eprint
  {https://arxiv.org/abs/1011.6531} {arXiv:1011.6531 [hep-ph]} \BibitemShut
  {NoStop}%
\bibitem [{\citenamefont {Burdman}\ \emph {et~al.}(1994)\citenamefont
  {Burdman}, \citenamefont {Ligeti}, \citenamefont {Neubert},\ and\
  \citenamefont {Nir}}]{Burdman:1993es}%
  \BibitemOpen
  \bibfield  {author} {\bibinfo {author} {\bibfnamefont {G.}~\bibnamefont
  {Burdman}}, \bibinfo {author} {\bibfnamefont {Z.}~\bibnamefont {Ligeti}},
  \bibinfo {author} {\bibfnamefont {M.}~\bibnamefont {Neubert}},\ and\ \bibinfo
  {author} {\bibfnamefont {Y.}~\bibnamefont {Nir}},\ }\bibfield  {title}
  {\bibinfo {title} {Decay {$B\to\pi\ell\nu$} in heavy quark effective
  theory},\ }\href {https://doi.org/10.1103/PhysRevD.49.2331} {\bibfield
  {journal} {\bibinfo  {journal} {Phys. Rev. D}\ }\textbf {\bibinfo {volume}
  {49}},\ \bibinfo {pages} {2331} (\bibinfo {year} {1994})},\ \Eprint
  {https://arxiv.org/abs/hep-ph/9309272} {arXiv:hep-ph/9309272} \BibitemShut
  {NoStop}%
\bibitem [{\citenamefont {Bernard}\ and\ \citenamefont
  {Komijani}(2013)}]{Bernard:2013qwa}%
  \BibitemOpen
  \bibfield  {author} {\bibinfo {author} {\bibfnamefont {C.}~\bibnamefont
  {Bernard}}\ and\ \bibinfo {author} {\bibfnamefont {J.}~\bibnamefont
  {Komijani}},\ }\bibfield  {title} {\bibinfo {title} {Chiral perturbation
  theory for all-staggered heavy-light mesons},\ }\href
  {https://doi.org/10.1103/PhysRevD.88.094017} {\bibfield  {journal} {\bibinfo
  {journal} {Phys. Rev. D}\ }\textbf {\bibinfo {volume} {88}},\ \bibinfo
  {pages} {094017} (\bibinfo {year} {2013})},\ \Eprint
  {https://arxiv.org/abs/1309.4533} {arXiv:1309.4533 [hep-lat]} \BibitemShut
  {NoStop}%
\bibitem [{\citenamefont {Becirevic}\ \emph
  {et~al.}(2003{\natexlab{a}})\citenamefont {Becirevic}, \citenamefont
  {Prelovsek},\ and\ \citenamefont {Zupan}}]{Becirevic:2002sc}%
  \BibitemOpen
  \bibfield  {author} {\bibinfo {author} {\bibfnamefont {D.}~\bibnamefont
  {Becirevic}}, \bibinfo {author} {\bibfnamefont {S.}~\bibnamefont
  {Prelovsek}},\ and\ \bibinfo {author} {\bibfnamefont {J.}~\bibnamefont
  {Zupan}},\ }\bibfield  {title} {\bibinfo {title} {{$B\to\pi$ and $B\to K$}
  transitions in standard and quenched chiral perturbation theory},\ }\href
  {https://doi.org/10.1103/PhysRevD.67.054010} {\bibfield  {journal} {\bibinfo
  {journal} {Phys. Rev. D}\ }\textbf {\bibinfo {volume} {67}},\ \bibinfo
  {pages} {054010} (\bibinfo {year} {2003}{\natexlab{a}})},\ \Eprint
  {https://arxiv.org/abs/hep-lat/0210048} {arXiv:hep-lat/0210048} \BibitemShut
  {NoStop}%
\bibitem [{\citenamefont {Becirevic}\ \emph
  {et~al.}(2003{\natexlab{b}})\citenamefont {Becirevic}, \citenamefont
  {Prelovsek},\ and\ \citenamefont {Zupan}}]{Becirevic:2003ad}%
  \BibitemOpen
  \bibfield  {author} {\bibinfo {author} {\bibfnamefont {D.}~\bibnamefont
  {Becirevic}}, \bibinfo {author} {\bibfnamefont {S.}~\bibnamefont
  {Prelovsek}},\ and\ \bibinfo {author} {\bibfnamefont {J.}~\bibnamefont
  {Zupan}},\ }\bibfield  {title} {\bibinfo {title} {{$B\to\pi$ and $B\to K$}
  transitions transitions in partially quenched chiral perturbation theory},\
  }\href {https://doi.org/10.1103/PhysRevD.68.074003} {\bibfield  {journal}
  {\bibinfo  {journal} {Phys. Rev. D}\ }\textbf {\bibinfo {volume} {68}},\
  \bibinfo {pages} {074003} (\bibinfo {year} {2003}{\natexlab{b}})},\ \Eprint
  {https://arxiv.org/abs/hep-lat/0305001} {arXiv:hep-lat/0305001} \BibitemShut
  {NoStop}%
\bibitem [{\citenamefont {Chakraborty}\ \emph {et~al.}(2015)\citenamefont
  {Chakraborty}, \citenamefont {Davies}, \citenamefont {Galloway},
  \citenamefont {Knecht}, \citenamefont {Koponen}, \citenamefont {Donald},
  \citenamefont {Dowdall}, \citenamefont {Lepage},\ and\ \citenamefont
  {McNeile}}]{Chakraborty:2014aca}%
  \BibitemOpen
  \bibfield  {author} {\bibinfo {author} {\bibfnamefont {B.}~\bibnamefont
  {Chakraborty}}, \bibinfo {author} {\bibfnamefont {C.~T.~H.}\ \bibnamefont
  {Davies}}, \bibinfo {author} {\bibfnamefont {B.}~\bibnamefont {Galloway}},
  \bibinfo {author} {\bibfnamefont {P.}~\bibnamefont {Knecht}}, \bibinfo
  {author} {\bibfnamefont {J.}~\bibnamefont {Koponen}}, \bibinfo {author}
  {\bibfnamefont {G.~C.}\ \bibnamefont {Donald}}, \bibinfo {author}
  {\bibfnamefont {R.~J.}\ \bibnamefont {Dowdall}}, \bibinfo {author}
  {\bibfnamefont {G.~P.}\ \bibnamefont {Lepage}},\ and\ \bibinfo {author}
  {\bibfnamefont {C.}~\bibnamefont {McNeile}},\ }\bibfield  {title} {\bibinfo
  {title} {{High-precision quark masses and QCD coupling from $n_f=4$ lattice
  QCD}},\ }\href {https://doi.org/10.1103/PhysRevD.91.054508} {\bibfield
  {journal} {\bibinfo  {journal} {Phys. Rev. D}\ }\textbf {\bibinfo {volume}
  {91}},\ \bibinfo {pages} {054508} (\bibinfo {year} {2015})},\ \Eprint
  {https://arxiv.org/abs/1408.4169} {arXiv:1408.4169 [hep-lat]} \BibitemShut
  {NoStop}%
\bibitem [{\citenamefont {Komijani}(2018)}]{Komijani:2018}%
  \BibitemOpen
  \bibfield  {author} {\bibinfo {author} {\bibfnamefont {J.}~\bibnamefont
  {Komijani}},\ }\href@noop {} {}\bibinfo {howpublished} {private
  communication} (\bibinfo {year} {2018})\BibitemShut {NoStop}%
\bibitem [{\citenamefont {Anastassov}\ \emph {et~al.}(2002)\citenamefont
  {Anastassov} \emph {et~al.}}]{CLEO:2001sxb}%
  \BibitemOpen
  \bibfield  {author} {\bibinfo {author} {\bibfnamefont {A.}~\bibnamefont
  {Anastassov}} \emph {et~al.} (\bibinfo {collaboration} {CLEO}),\ }\bibfield
  {title} {\bibinfo {title} {First measurement of {$\Gamma(D^{*+})$} and
  precision measurement of {$m_{D^{*+}} - m_{D^0}$}},\ }\href
  {https://doi.org/10.1103/PhysRevD.65.032003} {\bibfield  {journal} {\bibinfo
  {journal} {Phys. Rev. D}\ }\textbf {\bibinfo {volume} {65}},\ \bibinfo
  {pages} {032003} (\bibinfo {year} {2002})},\ \Eprint
  {https://arxiv.org/abs/hep-ex/0108043} {arXiv:hep-ex/0108043} \BibitemShut
  {NoStop}%
\bibitem [{\citenamefont {Lees}\ \emph
  {et~al.}(2013{\natexlab{a}})\citenamefont {Lees} \emph
  {et~al.}}]{BaBar:2013zgp}%
  \BibitemOpen
  \bibfield  {author} {\bibinfo {author} {\bibfnamefont {J.~P.}\ \bibnamefont
  {Lees}} \emph {et~al.} (\bibinfo {collaboration} {BaBar}),\ }\bibfield
  {title} {\bibinfo {title} {{Measurement of the $D^*(2010)^+$ natural line
  width and the $D^*(2010)^+ - D^0$ mass difference}},\ }\href
  {https://doi.org/10.1103/PhysRevD.88.052003} {\bibfield  {journal} {\bibinfo
  {journal} {Phys. Rev. D}\ }\textbf {\bibinfo {volume} {88}},\ \bibinfo
  {pages} {052003} (\bibinfo {year} {2013}{\natexlab{a}})},\ \bibinfo {note}
  {[Erratum: \href{http://doi.org/10.1103/PhysRevD.88.079902}{Phys. Rev. D
  \textbf{88}, 079902 (2013)}]},\ \Eprint {https://arxiv.org/abs/1304.5009}
  {arXiv:1304.5009 [hep-ex]} \BibitemShut {NoStop}%
\bibitem [{\citenamefont {Lees}\ \emph
  {et~al.}(2013{\natexlab{b}})\citenamefont {Lees} \emph
  {et~al.}}]{BaBar:2013thi}%
  \BibitemOpen
  \bibfield  {author} {\bibinfo {author} {\bibfnamefont {J.~P.}\ \bibnamefont
  {Lees}} \emph {et~al.} (\bibinfo {collaboration} {BaBar}),\ }\bibfield
  {title} {\bibinfo {title} {{Measurement of the $D^*(2010)^+$ meson width and
  the $D^*(2010)^+$-$D^0$ mass difference}},\ }\href
  {https://doi.org/10.1103/PhysRevLett.111.111801} {\bibfield  {journal}
  {\bibinfo  {journal} {Phys. Rev. Lett.}\ }\textbf {\bibinfo {volume} {111}},\
  \bibinfo {pages} {111801} (\bibinfo {year} {2013}{\natexlab{b}})},\ \Eprint
  {https://arxiv.org/abs/1304.5657} {arXiv:1304.5657 [hep-ex]} \BibitemShut
  {NoStop}%
\bibitem [{\citenamefont {Detmold}\ \emph
  {et~al.}(2012{\natexlab{a}})\citenamefont {Detmold}, \citenamefont {Lin},\
  and\ \citenamefont {Meinel}}]{Detmold:2011bp}%
  \BibitemOpen
  \bibfield  {author} {\bibinfo {author} {\bibfnamefont {W.}~\bibnamefont
  {Detmold}}, \bibinfo {author} {\bibfnamefont {C.~J.~D.}\ \bibnamefont
  {Lin}},\ and\ \bibinfo {author} {\bibfnamefont {S.}~\bibnamefont {Meinel}},\
  }\bibfield  {title} {\bibinfo {title} {{Axial couplings and strong decay
  widths of heavy hadrons}},\ }\href
  {https://doi.org/10.1103/PhysRevLett.108.172003} {\bibfield  {journal}
  {\bibinfo  {journal} {Phys. Rev. Lett.}\ }\textbf {\bibinfo {volume} {108}},\
  \bibinfo {pages} {172003} (\bibinfo {year} {2012}{\natexlab{a}})},\ \Eprint
  {https://arxiv.org/abs/1109.2480} {arXiv:1109.2480 [hep-lat]} \BibitemShut
  {NoStop}%
\bibitem [{\citenamefont {Detmold}\ \emph
  {et~al.}(2012{\natexlab{b}})\citenamefont {Detmold}, \citenamefont {Lin},\
  and\ \citenamefont {Meinel}}]{Detmold:2012ge}%
  \BibitemOpen
  \bibfield  {author} {\bibinfo {author} {\bibfnamefont {W.}~\bibnamefont
  {Detmold}}, \bibinfo {author} {\bibfnamefont {C.~J.~D.}\ \bibnamefont
  {Lin}},\ and\ \bibinfo {author} {\bibfnamefont {S.}~\bibnamefont {Meinel}},\
  }\bibfield  {title} {\bibinfo {title} {{Calculation of the heavy-hadron axial
  couplings $g_1$, $g_2$ and $g_3$ using lattice QCD}},\ }\href
  {https://doi.org/10.1103/PhysRevD.85.114508} {\bibfield  {journal} {\bibinfo
  {journal} {Phys. Rev. D}\ }\textbf {\bibinfo {volume} {85}},\ \bibinfo
  {pages} {114508} (\bibinfo {year} {2012}{\natexlab{b}})},\ \Eprint
  {https://arxiv.org/abs/1203.3378} {arXiv:1203.3378 [hep-lat]} \BibitemShut
  {NoStop}%
\bibitem [{\citenamefont {Can}\ \emph {et~al.}(2013)\citenamefont {Can},
  \citenamefont {Erkol}, \citenamefont {Oka}, \citenamefont {Ozpineci},\ and\
  \citenamefont {Takahashi}}]{Can:2012tx}%
  \BibitemOpen
  \bibfield  {author} {\bibinfo {author} {\bibfnamefont {K.~U.}\ \bibnamefont
  {Can}}, \bibinfo {author} {\bibfnamefont {G.}~\bibnamefont {Erkol}}, \bibinfo
  {author} {\bibfnamefont {M.}~\bibnamefont {Oka}}, \bibinfo {author}
  {\bibfnamefont {A.}~\bibnamefont {Ozpineci}},\ and\ \bibinfo {author}
  {\bibfnamefont {T.~T.}\ \bibnamefont {Takahashi}},\ }\bibfield  {title}
  {\bibinfo {title} {{Vector and axial-vector couplings of $D$ and $D^*$ mesons
  in 2+1 flavor lattice QCD}},\ }\href
  {https://doi.org/10.1016/j.physletb.2012.12.050} {\bibfield  {journal}
  {\bibinfo  {journal} {Phys. Lett. B}\ }\textbf {\bibinfo {volume} {719}},\
  \bibinfo {pages} {103} (\bibinfo {year} {2013})},\ \Eprint
  {https://arxiv.org/abs/1210.0869} {arXiv:1210.0869 [hep-lat]} \BibitemShut
  {NoStop}%
\bibitem [{\citenamefont {Becirevic}\ and\ \citenamefont
  {Sanfilippo}(2013)}]{Becirevic:2012pf}%
  \BibitemOpen
  \bibfield  {author} {\bibinfo {author} {\bibfnamefont {D.}~\bibnamefont
  {Becirevic}}\ and\ \bibinfo {author} {\bibfnamefont {F.}~\bibnamefont
  {Sanfilippo}},\ }\bibfield  {title} {\bibinfo {title} {{Theoretical estimate
  of the $D^* \to D\pi$ decay rate}},\ }\href
  {https://doi.org/10.1016/j.physletb.2013.03.004} {\bibfield  {journal}
  {\bibinfo  {journal} {Phys. Lett. B}\ }\textbf {\bibinfo {volume} {721}},\
  \bibinfo {pages} {94} (\bibinfo {year} {2013})},\ \Eprint
  {https://arxiv.org/abs/1210.5410} {arXiv:1210.5410 [hep-lat]} \BibitemShut
  {NoStop}%
\bibitem [{\citenamefont {Flynn}\ \emph {et~al.}(2016)\citenamefont {Flynn},
  \citenamefont {Fritzsch}, \citenamefont {Kawanai}, \citenamefont {Lehner},
  \citenamefont {Samways}, \citenamefont {Sachrajda}, \citenamefont {Van~de
  Water},\ and\ \citenamefont {Witzel}}]{Flynn:2015xna}%
  \BibitemOpen
  \bibfield  {author} {\bibinfo {author} {\bibfnamefont {J.~M.}\ \bibnamefont
  {Flynn}}, \bibinfo {author} {\bibfnamefont {P.}~\bibnamefont {Fritzsch}},
  \bibinfo {author} {\bibfnamefont {T.}~\bibnamefont {Kawanai}}, \bibinfo
  {author} {\bibfnamefont {C.}~\bibnamefont {Lehner}}, \bibinfo {author}
  {\bibfnamefont {B.}~\bibnamefont {Samways}}, \bibinfo {author} {\bibfnamefont
  {C.~T.}\ \bibnamefont {Sachrajda}}, \bibinfo {author} {\bibfnamefont {R.~S.}\
  \bibnamefont {Van~de Water}},\ and\ \bibinfo {author} {\bibfnamefont
  {O.}~\bibnamefont {Witzel}} (\bibinfo {collaboration} {RBC, UKQCD}),\
  }\bibfield  {title} {\bibinfo {title} {{The $B^*B\pi$} coupling using
  relativistic heavy quarks},\ }\href
  {https://doi.org/10.1103/PhysRevD.93.014510} {\bibfield  {journal} {\bibinfo
  {journal} {Phys. Rev. D}\ }\textbf {\bibinfo {volume} {93}},\ \bibinfo
  {pages} {014510} (\bibinfo {year} {2016})},\ \Eprint
  {https://arxiv.org/abs/1506.06413} {arXiv:1506.06413 [hep-lat]} \BibitemShut
  {NoStop}%
\bibitem [{\citenamefont {Bernardoni}\ \emph {et~al.}(2015)\citenamefont
  {Bernardoni}, \citenamefont {Bulava}, \citenamefont {Donnellan},\ and\
  \citenamefont {Sommer}}]{Bernardoni:2014kla}%
  \BibitemOpen
  \bibfield  {author} {\bibinfo {author} {\bibfnamefont {F.}~\bibnamefont
  {Bernardoni}}, \bibinfo {author} {\bibfnamefont {J.}~\bibnamefont {Bulava}},
  \bibinfo {author} {\bibfnamefont {M.}~\bibnamefont {Donnellan}},\ and\
  \bibinfo {author} {\bibfnamefont {R.}~\bibnamefont {Sommer}} (\bibinfo
  {collaboration} {ALPHA}),\ }\bibfield  {title} {\bibinfo {title} {{Precision
  lattice QCD computation of the $B^*B\pi$ coupling}},\ }\href
  {https://doi.org/10.1016/j.physletb.2014.11.051} {\bibfield  {journal}
  {\bibinfo  {journal} {Phys. Lett. B}\ }\textbf {\bibinfo {volume} {740}},\
  \bibinfo {pages} {278} (\bibinfo {year} {2015})},\ \Eprint
  {https://arxiv.org/abs/1404.6951} {arXiv:1404.6951 [hep-lat]} \BibitemShut
  {NoStop}%
\bibitem [{\citenamefont {Bailey}\ \emph {et~al.}(2015)\citenamefont {Bailey}
  \emph {et~al.}}]{FermilabLattice:2015mwy}%
  \BibitemOpen
  \bibfield  {author} {\bibinfo {author} {\bibfnamefont {J.~A.}\ \bibnamefont
  {Bailey}} \emph {et~al.} (\bibinfo {collaboration} {Fermilab Lattice,
  MILC}),\ }\bibfield  {title} {\bibinfo {title} {{$|V_{ub}|$ from
  $B\to\pi\ell\nu$ decays and (2+1)-flavor lattice QCD}},\ }\href
  {https://doi.org/10.1103/PhysRevD.92.014024} {\bibfield  {journal} {\bibinfo
  {journal} {Phys. Rev. D}\ }\textbf {\bibinfo {volume} {92}},\ \bibinfo
  {pages} {014024} (\bibinfo {year} {2015})},\ \Eprint
  {https://arxiv.org/abs/1503.07839} {arXiv:1503.07839 [hep-lat]} \BibitemShut
  {NoStop}%
\bibitem [{\citenamefont {Dowdall}\ \emph {et~al.}(2019)\citenamefont
  {Dowdall}, \citenamefont {Davies}, \citenamefont {Horgan}, \citenamefont
  {Lepage}, \citenamefont {Monahan}, \citenamefont {Shigemitsu},\ and\
  \citenamefont {Wingate}}]{Dowdall:2019bea}%
  \BibitemOpen
  \bibfield  {author} {\bibinfo {author} {\bibfnamefont {R.~J.}\ \bibnamefont
  {Dowdall}}, \bibinfo {author} {\bibfnamefont {C.~T.~H.}\ \bibnamefont
  {Davies}}, \bibinfo {author} {\bibfnamefont {R.~R.}\ \bibnamefont {Horgan}},
  \bibinfo {author} {\bibfnamefont {G.~P.}\ \bibnamefont {Lepage}}, \bibinfo
  {author} {\bibfnamefont {C.~J.}\ \bibnamefont {Monahan}}, \bibinfo {author}
  {\bibfnamefont {J.}~\bibnamefont {Shigemitsu}},\ and\ \bibinfo {author}
  {\bibfnamefont {M.}~\bibnamefont {Wingate}},\ }\bibfield  {title} {\bibinfo
  {title} {{Neutral $B$-meson mixing from full lattice QCD at the physical
  point}},\ }\href {https://doi.org/10.1103/PhysRevD.100.094508} {\bibfield
  {journal} {\bibinfo  {journal} {Phys. Rev. D}\ }\textbf {\bibinfo {volume}
  {100}},\ \bibinfo {pages} {094508} (\bibinfo {year} {2019})},\ \Eprint
  {https://arxiv.org/abs/1907.01025} {arXiv:1907.01025 [hep-lat]} \BibitemShut
  {NoStop}%
\bibitem [{\citenamefont {Zyla}\ \emph {et~al.}(2020)\citenamefont {Zyla} \emph
  {et~al.}}]{ParticleDataGroup:2020ssz}%
  \BibitemOpen
  \bibfield  {author} {\bibinfo {author} {\bibfnamefont {P.~A.}\ \bibnamefont
  {Zyla}} \emph {et~al.} (\bibinfo {collaboration} {Particle Data Group}),\
  }\bibfield  {title} {\bibinfo {title} {Review of particle physics},\ }\href
  {https://doi.org/10.1093/ptep/ptaa104} {\bibfield  {journal} {\bibinfo
  {journal} {PTEP}\ }\textbf {\bibinfo {volume} {2020}},\ \bibinfo {pages}
  {083C01} (\bibinfo {year} {2020})}\BibitemShut {NoStop}%
\bibitem [{\citenamefont {Boyd}\ \emph {et~al.}(1995)\citenamefont {Boyd},
  \citenamefont {Grinstein},\ and\ \citenamefont {Lebed}}]{Boyd:1994tt}%
  \BibitemOpen
  \bibfield  {author} {\bibinfo {author} {\bibfnamefont {C.~G.}\ \bibnamefont
  {Boyd}}, \bibinfo {author} {\bibfnamefont {B.}~\bibnamefont {Grinstein}},\
  and\ \bibinfo {author} {\bibfnamefont {R.~F.}\ \bibnamefont {Lebed}},\
  }\bibfield  {title} {\bibinfo {title} {{Constraints on form-factors for
  exclusive semileptonic heavy to light meson decays}},\ }\href
  {https://doi.org/10.1103/PhysRevLett.74.4603} {\bibfield  {journal} {\bibinfo
   {journal} {Phys. Rev. Lett.}\ }\textbf {\bibinfo {volume} {74}},\ \bibinfo
  {pages} {4603} (\bibinfo {year} {1995})},\ \Eprint
  {https://arxiv.org/abs/hep-ph/9412324} {arXiv:hep-ph/9412324} \BibitemShut
  {NoStop}%
\bibitem [{\citenamefont {Bourrely}\ \emph {et~al.}(2009)\citenamefont
  {Bourrely}, \citenamefont {Caprini},\ and\ \citenamefont
  {Lellouch}}]{Bourrely:2008za}%
  \BibitemOpen
  \bibfield  {author} {\bibinfo {author} {\bibfnamefont {C.}~\bibnamefont
  {Bourrely}}, \bibinfo {author} {\bibfnamefont {I.}~\bibnamefont {Caprini}},\
  and\ \bibinfo {author} {\bibfnamefont {L.}~\bibnamefont {Lellouch}},\
  }\bibfield  {title} {\bibinfo {title} {{Model-independent description of
  $B\to\pi l\nu$ decays and a determination of $|V_{ub}|$}},\ }\href
  {https://doi.org/10.1103/PhysRevD.79.013008} {\bibfield  {journal} {\bibinfo
  {journal} {Phys. Rev. D}\ }\textbf {\bibinfo {volume} {79}},\ \bibinfo
  {pages} {013008} (\bibinfo {year} {2009})},\ \bibinfo {note} {[Erratum:
  \href{https://doi.org/10.1103/PhysRevD.82.099902}{Phys. Rev. D \textbf{82},
  099902 (2010)}]},\ \Eprint {https://arxiv.org/abs/0807.2722} {arXiv:0807.2722
  [hep-ph]} \BibitemShut {NoStop}%
\bibitem [{\citenamefont {Ablikim}\ \emph
  {et~al.}(2019{\natexlab{a}})\citenamefont {Ablikim} \emph
  {et~al.}}]{BESIII:2018xre}%
  \BibitemOpen
  \bibfield  {author} {\bibinfo {author} {\bibfnamefont {M.}~\bibnamefont
  {Ablikim}} \emph {et~al.} (\bibinfo {collaboration} {BES~III}),\ }\bibfield
  {title} {\bibinfo {title} {First measurement of the form factors in
  {$D^+_{s}\rightarrow K^0 e^+\nu_e$ and $D^+_{s}\rightarrow K^{*0} e^+\nu_e$}
  decays},\ }\href {https://doi.org/10.1103/PhysRevLett.122.061801} {\bibfield
  {journal} {\bibinfo  {journal} {Phys. Rev. Lett.}\ }\textbf {\bibinfo
  {volume} {122}},\ \bibinfo {pages} {061801} (\bibinfo {year}
  {2019}{\natexlab{a}})},\ \Eprint {https://arxiv.org/abs/1811.02911}
  {arXiv:1811.02911 [hep-ex]} \BibitemShut {NoStop}%
\bibitem [{\citenamefont {Parrott}\ \emph {et~al.}(2021)\citenamefont
  {Parrott}, \citenamefont {Bouchard}, \citenamefont {Davies},\ and\
  \citenamefont {Hatton}}]{Parrott:2020vbe}%
  \BibitemOpen
  \bibfield  {author} {\bibinfo {author} {\bibfnamefont {W.~G.}\ \bibnamefont
  {Parrott}}, \bibinfo {author} {\bibfnamefont {C.}~\bibnamefont {Bouchard}},
  \bibinfo {author} {\bibfnamefont {C.~T.~H.}\ \bibnamefont {Davies}},\ and\
  \bibinfo {author} {\bibfnamefont {D.}~\bibnamefont {Hatton}} (\bibinfo
  {collaboration} {HPQCD}),\ }\bibfield  {title} {\bibinfo {title} {{Toward
  accurate form factors for $B$-to-light meson decay from lattice QCD}},\
  }\href {https://doi.org/10.1103/PhysRevD.103.094506} {\bibfield  {journal}
  {\bibinfo  {journal} {Phys. Rev. D}\ }\textbf {\bibinfo {volume} {103}},\
  \bibinfo {pages} {094506} (\bibinfo {year} {2021})},\ \Eprint
  {https://arxiv.org/abs/2010.07980} {arXiv:2010.07980 [hep-lat]} \BibitemShut
  {NoStop}%
\bibitem [{\citenamefont {Lepage}\ \emph {et~al.}(2022)\citenamefont {Lepage},
  \citenamefont {Gohlke},\ and\ \citenamefont {Hackett}}]{gvar:2022}%
  \BibitemOpen
  \bibfield  {author} {\bibinfo {author} {\bibfnamefont {P.}~\bibnamefont
  {Lepage}}, \bibinfo {author} {\bibfnamefont {C.}~\bibnamefont {Gohlke}},\
  and\ \bibinfo {author} {\bibfnamefont {D.}~\bibnamefont {Hackett}},\ }\href
  {https://doi.org/10.5281/zenodo.6350630} {\bibinfo {title}
  {\href{https://doi.org/10.5281/zenodo.6350630}{gplepage/gvar: gvar}}}
  (\bibinfo {year} {2022})\BibitemShut {NoStop}%
\bibitem [{\citenamefont {Bouchard}\ \emph {et~al.}(2014)\citenamefont
  {Bouchard}, \citenamefont {Lepage}, \citenamefont {Monahan}, \citenamefont
  {Na},\ and\ \citenamefont {Shigemitsu}}]{Bouchard:2014ypa}%
  \BibitemOpen
  \bibfield  {author} {\bibinfo {author} {\bibfnamefont {C.~M.}\ \bibnamefont
  {Bouchard}}, \bibinfo {author} {\bibfnamefont {G.~P.}\ \bibnamefont
  {Lepage}}, \bibinfo {author} {\bibfnamefont {C.}~\bibnamefont {Monahan}},
  \bibinfo {author} {\bibfnamefont {H.}~\bibnamefont {Na}},\ and\ \bibinfo
  {author} {\bibfnamefont {J.}~\bibnamefont {Shigemitsu}},\ }\bibfield  {title}
  {\bibinfo {title} {{$B_s \to K \ell \nu$ form factors from lattice QCD}},\
  }\href {https://doi.org/10.1103/PhysRevD.90.054506} {\bibfield  {journal}
  {\bibinfo  {journal} {Phys. Rev. D}\ }\textbf {\bibinfo {volume} {90}},\
  \bibinfo {pages} {054506} (\bibinfo {year} {2014})},\ \Eprint
  {https://arxiv.org/abs/1406.2279} {arXiv:1406.2279 [hep-lat]} \BibitemShut
  {NoStop}%
\bibitem [{\citenamefont {Arndt}\ and\ \citenamefont
  {Lin}(2004)}]{Arndt:2004bg}%
  \BibitemOpen
  \bibfield  {author} {\bibinfo {author} {\bibfnamefont {D.}~\bibnamefont
  {Arndt}}\ and\ \bibinfo {author} {\bibfnamefont {C.~J.~D.}\ \bibnamefont
  {Lin}},\ }\bibfield  {title} {\bibinfo {title} {{Heavy meson chiral
  perturbation theory in finite volume}},\ }\href
  {https://doi.org/10.1103/PhysRevD.70.014503} {\bibfield  {journal} {\bibinfo
  {journal} {Phys. Rev. D}\ }\textbf {\bibinfo {volume} {70}},\ \bibinfo
  {pages} {014503} (\bibinfo {year} {2004})},\ \Eprint
  {https://arxiv.org/abs/hep-lat/0403012} {arXiv:hep-lat/0403012} \BibitemShut
  {NoStop}%
\bibitem [{\citenamefont {Laiho}\ and\ \citenamefont {Van~de
  Water}(2006)}]{Laiho:2005ue}%
  \BibitemOpen
  \bibfield  {author} {\bibinfo {author} {\bibfnamefont {J.}~\bibnamefont
  {Laiho}}\ and\ \bibinfo {author} {\bibfnamefont {R.~S.}\ \bibnamefont {Van~de
  Water}},\ }\bibfield  {title} {\bibinfo {title} {{$B \to D^* \ell \nu$ and $B
  \to D \ell \nu$ form factors in staggered chiral perturbation theory.}},\
  }\href {https://doi.org/10.1103/PhysRevD.73.054501} {\bibfield  {journal}
  {\bibinfo  {journal} {Phys. Rev. D}\ }\textbf {\bibinfo {volume} {73}},\
  \bibinfo {pages} {054501} (\bibinfo {year} {2006})},\ \Eprint
  {https://arxiv.org/abs/hep-lat/0512007} {arXiv:hep-lat/0512007} \BibitemShut
  {NoStop}%
\bibitem [{\citenamefont {Brower}\ \emph {et~al.}(2003)\citenamefont {Brower},
  \citenamefont {Chandrasekharan}, \citenamefont {Negele},\ and\ \citenamefont
  {Wiese}}]{Brower:2003yx}%
  \BibitemOpen
  \bibfield  {author} {\bibinfo {author} {\bibfnamefont {R.}~\bibnamefont
  {Brower}}, \bibinfo {author} {\bibfnamefont {S.}~\bibnamefont
  {Chandrasekharan}}, \bibinfo {author} {\bibfnamefont {J.~W.}\ \bibnamefont
  {Negele}},\ and\ \bibinfo {author} {\bibfnamefont {U.~J.}\ \bibnamefont
  {Wiese}},\ }\bibfield  {title} {\bibinfo {title} {{QCD at fixed topology}},\
  }\href {https://doi.org/10.1016/S0370-2693(03)00369-1} {\bibfield  {journal}
  {\bibinfo  {journal} {Phys. Lett. B}\ }\textbf {\bibinfo {volume} {560}},\
  \bibinfo {pages} {64} (\bibinfo {year} {2003})},\ \Eprint
  {https://arxiv.org/abs/hep-lat/0302005} {arXiv:hep-lat/0302005} \BibitemShut
  {NoStop}%
\bibitem [{\citenamefont {Bernard}\ and\ \citenamefont
  {Toussaint}(2018)}]{Bernard:2017npd}%
  \BibitemOpen
  \bibfield  {author} {\bibinfo {author} {\bibfnamefont {C.}~\bibnamefont
  {Bernard}}\ and\ \bibinfo {author} {\bibfnamefont {D.}~\bibnamefont
  {Toussaint}} (\bibinfo {collaboration} {MILC}),\ }\bibfield  {title}
  {\bibinfo {title} {{Effects of nonequilibrated topological charge
  distributions on pseudoscalar meson masses and decay constants}},\ }\href
  {https://doi.org/10.1103/PhysRevD.97.074502} {\bibfield  {journal} {\bibinfo
  {journal} {Phys. Rev. D}\ }\textbf {\bibinfo {volume} {97}},\ \bibinfo
  {pages} {074502} (\bibinfo {year} {2018})},\ \Eprint
  {https://arxiv.org/abs/1707.05430} {arXiv:1707.05430 [hep-lat]} \BibitemShut
  {NoStop}%
\bibitem [{\citenamefont {Billeter}\ \emph {et~al.}(2004)\citenamefont
  {Billeter}, \citenamefont {DeTar},\ and\ \citenamefont
  {Osborn}}]{Billeter:2004wx}%
  \BibitemOpen
  \bibfield  {author} {\bibinfo {author} {\bibfnamefont {B.}~\bibnamefont
  {Billeter}}, \bibinfo {author} {\bibfnamefont {C.~E.}\ \bibnamefont
  {DeTar}},\ and\ \bibinfo {author} {\bibfnamefont {J.}~\bibnamefont
  {Osborn}},\ }\bibfield  {title} {\bibinfo {title} {{Topological
  susceptibility in staggered fermion chiral perturbation theory}},\ }\href
  {https://doi.org/10.1103/PhysRevD.70.077502} {\bibfield  {journal} {\bibinfo
  {journal} {Phys. Rev. D}\ }\textbf {\bibinfo {volume} {70}},\ \bibinfo
  {pages} {077502} (\bibinfo {year} {2004})},\ \Eprint
  {https://arxiv.org/abs/hep-lat/0406032} {arXiv:hep-lat/0406032} \BibitemShut
  {NoStop}%
\bibitem [{\citenamefont {Leutwyler}\ and\ \citenamefont
  {Smilga}(1992)}]{Leutwyler:1992yt}%
  \BibitemOpen
  \bibfield  {author} {\bibinfo {author} {\bibfnamefont {H.}~\bibnamefont
  {Leutwyler}}\ and\ \bibinfo {author} {\bibfnamefont {A.~V.}\ \bibnamefont
  {Smilga}},\ }\bibfield  {title} {\bibinfo {title} {{Spectrum of Dirac
  operator and role of winding number in QCD}},\ }\href
  {https://doi.org/10.1103/PhysRevD.46.5607} {\bibfield  {journal} {\bibinfo
  {journal} {Phys. Rev. D}\ }\textbf {\bibinfo {volume} {46}},\ \bibinfo
  {pages} {5607} (\bibinfo {year} {1992})}\BibitemShut {NoStop}%
\bibitem [{\citenamefont {Link}\ \emph {et~al.}(2005)\citenamefont {Link} \emph
  {et~al.}}]{FOCUS:2004meh}%
  \BibitemOpen
  \bibfield  {author} {\bibinfo {author} {\bibfnamefont {J.~M.}\ \bibnamefont
  {Link}} \emph {et~al.} (\bibinfo {collaboration} {FOCUS}),\ }\bibfield
  {title} {\bibinfo {title} {{Measurements of the $q^2$ dependence of the
  $D^0\to K^-\mu^+\nu$ and $D^0\to\pi^-\mu^+\nu$ form factors}},\ }\href
  {https://doi.org/10.1016/j.physletb.2004.12.036} {\bibfield  {journal}
  {\bibinfo  {journal} {Phys. Lett. B}\ }\textbf {\bibinfo {volume} {607}},\
  \bibinfo {pages} {233} (\bibinfo {year} {2005})},\ \Eprint
  {https://arxiv.org/abs/hep-ex/0410037} {arXiv:hep-ex/0410037} \BibitemShut
  {NoStop}%
\bibitem [{\citenamefont {Widhalm}\ \emph {et~al.}(2006)\citenamefont {Widhalm}
  \emph {et~al.}}]{Belle:2006idb}%
  \BibitemOpen
  \bibfield  {author} {\bibinfo {author} {\bibfnamefont {L.}~\bibnamefont
  {Widhalm}} \emph {et~al.} (\bibinfo {collaboration} {Belle}),\ }\bibfield
  {title} {\bibinfo {title} {Measurement of {$D^0\to\pi l\nu (Kl\nu)$} form
  factors and absolute branching fractions},\ }\href
  {https://doi.org/10.1103/PhysRevLett.97.061804} {\bibfield  {journal}
  {\bibinfo  {journal} {Phys. Rev. Lett.}\ }\textbf {\bibinfo {volume} {97}},\
  \bibinfo {pages} {061804} (\bibinfo {year} {2006})},\ \Eprint
  {https://arxiv.org/abs/hep-ex/0604049} {arXiv:hep-ex/0604049} \BibitemShut
  {NoStop}%
\bibitem [{\citenamefont {Aubert}\ \emph {et~al.}(2007)\citenamefont {Aubert}
  \emph {et~al.}}]{BaBar:2007zgf}%
  \BibitemOpen
  \bibfield  {author} {\bibinfo {author} {\bibfnamefont {B.}~\bibnamefont
  {Aubert}} \emph {et~al.} (\bibinfo {collaboration} {BaBar}),\ }\bibfield
  {title} {\bibinfo {title} {Measurement of the hadronic form-factor in
  {$D^0\to K^-e^+\nu_e$}},\ }\href {https://doi.org/10.1103/PhysRevD.76.052005}
  {\bibfield  {journal} {\bibinfo  {journal} {Phys. Rev. D}\ }\textbf {\bibinfo
  {volume} {76}},\ \bibinfo {pages} {052005} (\bibinfo {year} {2007})},\
  \Eprint {https://arxiv.org/abs/0704.0020} {arXiv:0704.0020 [hep-ex]}
  \BibitemShut {NoStop}%
\bibitem [{\citenamefont {Lees}\ \emph {et~al.}(2015)\citenamefont {Lees} \emph
  {et~al.}}]{BaBar:2014xzf}%
  \BibitemOpen
  \bibfield  {author} {\bibinfo {author} {\bibfnamefont {J.~P.}\ \bibnamefont
  {Lees}} \emph {et~al.} (\bibinfo {collaboration} {BaBar}),\ }\bibfield
  {title} {\bibinfo {title} {{Measurement of the $D^0 \to \pi^- e^+ \nu_e$
  differential decay branching fraction as a function of $q^2$ and study of
  form factor parameterizations}},\ }\href
  {https://doi.org/10.1103/PhysRevD.91.052022} {\bibfield  {journal} {\bibinfo
  {journal} {Phys. Rev. D}\ }\textbf {\bibinfo {volume} {91}},\ \bibinfo
  {pages} {052022} (\bibinfo {year} {2015})},\ \Eprint
  {https://arxiv.org/abs/1412.5502} {arXiv:1412.5502 [hep-ex]} \BibitemShut
  {NoStop}%
\bibitem [{\citenamefont {Besson}\ \emph {et~al.}(2009)\citenamefont {Besson}
  \emph {et~al.}}]{CLEO:2009svp}%
  \BibitemOpen
  \bibfield  {author} {\bibinfo {author} {\bibfnamefont {D.}~\bibnamefont
  {Besson}} \emph {et~al.} (\bibinfo {collaboration} {CLEO}),\ }\bibfield
  {title} {\bibinfo {title} {{Improved measurements of $D$ meson semileptonic
  decays to $\pi$ and $K$ mesons}},\ }\href
  {https://doi.org/10.1103/PhysRevD.80.032005} {\bibfield  {journal} {\bibinfo
  {journal} {Phys. Rev. D}\ }\textbf {\bibinfo {volume} {80}},\ \bibinfo
  {pages} {032005} (\bibinfo {year} {2009})},\ \Eprint
  {https://arxiv.org/abs/0906.2983} {arXiv:0906.2983 [hep-ex]} \BibitemShut
  {NoStop}%
\bibitem [{\citenamefont {Ablikim}\ \emph {et~al.}(2015)\citenamefont {Ablikim}
  \emph {et~al.}}]{BESIII:2015tql}%
  \BibitemOpen
  \bibfield  {author} {\bibinfo {author} {\bibfnamefont {M.}~\bibnamefont
  {Ablikim}} \emph {et~al.} (\bibinfo {collaboration} {BES~III}),\ }\bibfield
  {title} {\bibinfo {title} {Study of dynamics of {$D^0 \to K^- e^+ \nu_{e}$
  and $D^0\to\pi^- e^+ \nu_{e}$} decays},\ }\href
  {https://doi.org/10.1103/PhysRevD.92.072012} {\bibfield  {journal} {\bibinfo
  {journal} {Phys. Rev. D}\ }\textbf {\bibinfo {volume} {92}},\ \bibinfo
  {pages} {072012} (\bibinfo {year} {2015})},\ \Eprint
  {https://arxiv.org/abs/1508.07560} {arXiv:1508.07560 [hep-ex]} \BibitemShut
  {NoStop}%
\bibitem [{\citenamefont {Ablikim}\ \emph {et~al.}(2016)\citenamefont {Ablikim}
  \emph {et~al.}}]{BESIII:2016gbw}%
  \BibitemOpen
  \bibfield  {author} {\bibinfo {author} {\bibfnamefont {M.}~\bibnamefont
  {Ablikim}} \emph {et~al.} (\bibinfo {collaboration} {BES~III}),\ }\bibfield
  {title} {\bibinfo {title} {{Improved measurement of the absolute branching
  fraction of $D^{+}\rightarrow \bar{K}^0 \mu ^{+}\nu _{\mu }$}},\ }\href
  {https://doi.org/10.1140/epjc/s10052-016-4198-2} {\bibfield  {journal}
  {\bibinfo  {journal} {Eur. Phys. J. C}\ }\textbf {\bibinfo {volume} {76}},\
  \bibinfo {pages} {369} (\bibinfo {year} {2016})},\ \Eprint
  {https://arxiv.org/abs/1605.00068} {arXiv:1605.00068 [hep-ex]} \BibitemShut
  {NoStop}%
\bibitem [{\citenamefont {Ablikim}\ \emph {et~al.}(2017)\citenamefont {Ablikim}
  \emph {et~al.}}]{BESIII:2017ylw}%
  \BibitemOpen
  \bibfield  {author} {\bibinfo {author} {\bibfnamefont {M.}~\bibnamefont
  {Ablikim}} \emph {et~al.} (\bibinfo {collaboration} {BES~III}),\ }\bibfield
  {title} {\bibinfo {title} {{Analysis of $D^+\to\bar K^0e^+\nu_e$ and
  $D^+\to\pi^0e^+\nu_e$ semileptonic decays}},\ }\href
  {https://doi.org/10.1103/PhysRevD.96.012002} {\bibfield  {journal} {\bibinfo
  {journal} {Phys. Rev. D}\ }\textbf {\bibinfo {volume} {96}},\ \bibinfo
  {pages} {012002} (\bibinfo {year} {2017})},\ \Eprint
  {https://arxiv.org/abs/1703.09084} {arXiv:1703.09084 [hep-ex]} \BibitemShut
  {NoStop}%
\bibitem [{\citenamefont {Ablikim}\ \emph {et~al.}(2018)\citenamefont {Ablikim}
  \emph {et~al.}}]{BESIII:2018nzb}%
  \BibitemOpen
  \bibfield  {author} {\bibinfo {author} {\bibfnamefont {M.}~\bibnamefont
  {Ablikim}} \emph {et~al.} (\bibinfo {collaboration} {BES~III}),\ }\bibfield
  {title} {\bibinfo {title} {{Measurement of the branching fraction for the
  semi-leptonic decay $D^{0(+)}\to \pi^{-(0)}\mu^+\nu_\mu$ and test of lepton
  universality}},\ }\href {https://doi.org/10.1103/PhysRevLett.121.171803}
  {\bibfield  {journal} {\bibinfo  {journal} {Phys. Rev. Lett.}\ }\textbf
  {\bibinfo {volume} {121}},\ \bibinfo {pages} {171803} (\bibinfo {year}
  {2018})},\ \Eprint {https://arxiv.org/abs/1802.05492} {arXiv:1802.05492
  [hep-ex]} \BibitemShut {NoStop}%
\bibitem [{\citenamefont {Ablikim}\ \emph
  {et~al.}(2019{\natexlab{b}})\citenamefont {Ablikim} \emph
  {et~al.}}]{BESIII:2018ccy}%
  \BibitemOpen
  \bibfield  {author} {\bibinfo {author} {\bibfnamefont {M.}~\bibnamefont
  {Ablikim}} \emph {et~al.} (\bibinfo {collaboration} {BES~III}),\ }\bibfield
  {title} {\bibinfo {title} {{Study of the $D^0\to K^-\mu^+\nu_\mu$ dynamics
  and test of lepton flavor universality with $D^0\to K^-\ell^+\nu_\ell$
  decays}},\ }\href {https://doi.org/10.1103/PhysRevLett.122.011804} {\bibfield
   {journal} {\bibinfo  {journal} {Phys. Rev. Lett.}\ }\textbf {\bibinfo
  {volume} {122}},\ \bibinfo {pages} {011804} (\bibinfo {year}
  {2019}{\natexlab{b}})},\ \Eprint {https://arxiv.org/abs/1810.03127}
  {arXiv:1810.03127 [hep-ex]} \BibitemShut {NoStop}%
\bibitem [{\citenamefont {Rong}\ \emph {et~al.}(2015)\citenamefont {Rong},
  \citenamefont {Fang}, \citenamefont {Ma},\ and\ \citenamefont
  {Zhao}}]{Rong:2014hea}%
  \BibitemOpen
  \bibfield  {author} {\bibinfo {author} {\bibfnamefont {G.}~\bibnamefont
  {Rong}}, \bibinfo {author} {\bibfnamefont {Y.}~\bibnamefont {Fang}}, \bibinfo
  {author} {\bibfnamefont {H.~L.}\ \bibnamefont {Ma}},\ and\ \bibinfo {author}
  {\bibfnamefont {J.~Y.}\ \bibnamefont {Zhao}},\ }\bibfield  {title} {\bibinfo
  {title} {Determination of $f_+^\pi(0)$ or extraction of {$|V_{cd}|$} from
  semileptonic {$D$} decays},\ }\href
  {https://doi.org/10.1016/j.physletb.2015.02.049} {\bibfield  {journal}
  {\bibinfo  {journal} {Phys. Lett. B}\ }\textbf {\bibinfo {volume} {743}},\
  \bibinfo {pages} {315} (\bibinfo {year} {2015})},\ \Eprint
  {https://arxiv.org/abs/1410.3232} {arXiv:1410.3232 [hep-ex]} \BibitemShut
  {NoStop}%
\bibitem [{\citenamefont {Fang}\ \emph {et~al.}(2015)\citenamefont {Fang},
  \citenamefont {Rong}, \citenamefont {Ma},\ and\ \citenamefont
  {Zhao}}]{Fang:2014sqa}%
  \BibitemOpen
  \bibfield  {author} {\bibinfo {author} {\bibfnamefont {Y.}~\bibnamefont
  {Fang}}, \bibinfo {author} {\bibfnamefont {G.}~\bibnamefont {Rong}}, \bibinfo
  {author} {\bibfnamefont {H.~L.}\ \bibnamefont {Ma}},\ and\ \bibinfo {author}
  {\bibfnamefont {J.~Y.}\ \bibnamefont {Zhao}},\ }\bibfield  {title} {\bibinfo
  {title} {{Determination of $f_+^K(0)$ and extraction of $|V_{cs}|$ from
  semileptonic $D$ decays}},\ }\href
  {https://doi.org/10.1140/epjc/s10052-014-3226-3} {\bibfield  {journal}
  {\bibinfo  {journal} {Eur. Phys. J. C}\ }\textbf {\bibinfo {volume} {75}},\
  \bibinfo {pages} {10} (\bibinfo {year} {2015})},\ \Eprint
  {https://arxiv.org/abs/1409.8049} {arXiv:1409.8049 [hep-ex]} \BibitemShut
  {NoStop}%
\bibitem [{\citenamefont {Golonka}\ and\ \citenamefont {{W\c
  as}}(2006)}]{Golonka:2005pn}%
  \BibitemOpen
  \bibfield  {author} {\bibinfo {author} {\bibfnamefont {P.}~\bibnamefont
  {Golonka}}\ and\ \bibinfo {author} {\bibfnamefont {Z.}~\bibnamefont {{W\c
  as}}},\ }\bibfield  {title} {\bibinfo {title} {{PHOTOS Monte Carlo: A
  precision tool for QED corrections in $Z$ and $W$ decays}},\ }\href
  {https://doi.org/10.1140/epjc/s2005-02396-4} {\bibfield  {journal} {\bibinfo
  {journal} {Eur. Phys. J. C}\ }\textbf {\bibinfo {volume} {45}},\ \bibinfo
  {pages} {97} (\bibinfo {year} {2006})},\ \Eprint
  {https://arxiv.org/abs/hep-ph/0506026} {arXiv:hep-ph/0506026} \BibitemShut
  {NoStop}%
\bibitem [{\citenamefont {Barberio}\ and\ \citenamefont {{W\c
  as}}(1994)}]{Barberio:1993qi}%
  \BibitemOpen
  \bibfield  {author} {\bibinfo {author} {\bibfnamefont {E.}~\bibnamefont
  {Barberio}}\ and\ \bibinfo {author} {\bibfnamefont {Z.}~\bibnamefont {{W\c
  as}}},\ }\bibfield  {title} {\bibinfo {title} {{PHOTOS: A universal Monte
  Carlo for QED radiative corrections: Version 2.0}},\ }\href
  {https://doi.org/10.1016/0010-4655(94)90074-4} {\bibfield  {journal}
  {\bibinfo  {journal} {Comput. Phys. Commun.}\ }\textbf {\bibinfo {volume}
  {79}},\ \bibinfo {pages} {291} (\bibinfo {year} {1994})}\BibitemShut
  {NoStop}%
\bibitem [{\citenamefont {Cirigliano}\ \emph {et~al.}(2008)\citenamefont
  {Cirigliano}, \citenamefont {Giannotti},\ and\ \citenamefont
  {Neufeld}}]{Cirigliano:2008wn}%
  \BibitemOpen
  \bibfield  {author} {\bibinfo {author} {\bibfnamefont {V.}~\bibnamefont
  {Cirigliano}}, \bibinfo {author} {\bibfnamefont {M.}~\bibnamefont
  {Giannotti}},\ and\ \bibinfo {author} {\bibfnamefont {H.}~\bibnamefont
  {Neufeld}},\ }\bibfield  {title} {\bibinfo {title} {{Electromagnetic effects
  in $K_{l3}$ decays}},\ }\href {https://doi.org/10.1088/1126-6708/2008/11/006}
  {\bibfield  {journal} {\bibinfo  {journal} {JHEP}\ }\textbf {\bibinfo
  {volume} {11}},\ \bibinfo {pages} {006}},\ \Eprint
  {https://arxiv.org/abs/0807.4507} {arXiv:0807.4507 [hep-ph]} \BibitemShut
  {NoStop}%
\bibitem [{\citenamefont {Cirigliano}\ \emph {et~al.}(2012)\citenamefont
  {Cirigliano}, \citenamefont {Ecker}, \citenamefont {Neufeld}, \citenamefont
  {Pich},\ and\ \citenamefont {Portoles}}]{Cirigliano:2011ny}%
  \BibitemOpen
  \bibfield  {author} {\bibinfo {author} {\bibfnamefont {V.}~\bibnamefont
  {Cirigliano}}, \bibinfo {author} {\bibfnamefont {G.}~\bibnamefont {Ecker}},
  \bibinfo {author} {\bibfnamefont {H.}~\bibnamefont {Neufeld}}, \bibinfo
  {author} {\bibfnamefont {A.}~\bibnamefont {Pich}},\ and\ \bibinfo {author}
  {\bibfnamefont {J.}~\bibnamefont {Portoles}},\ }\bibfield  {title} {\bibinfo
  {title} {Kaon decays in the standard model},\ }\href
  {https://doi.org/10.1103/RevModPhys.84.399} {\bibfield  {journal} {\bibinfo
  {journal} {Rev. Mod. Phys.}\ }\textbf {\bibinfo {volume} {84}},\ \bibinfo
  {pages} {399} (\bibinfo {year} {2012})},\ \Eprint
  {https://arxiv.org/abs/1107.6001} {arXiv:1107.6001 [hep-ph]} \BibitemShut
  {NoStop}%
\bibitem [{\citenamefont {Ginsberg}(1968)}]{Ginsberg:1968pz}%
  \BibitemOpen
  \bibfield  {author} {\bibinfo {author} {\bibfnamefont {E.~S.}\ \bibnamefont
  {Ginsberg}},\ }\bibfield  {title} {\bibinfo {title} {Radiative corrections to
  {$K_{e3}^0$} decays and the {$\Delta I=1/2$} rule},\ }\href
  {https://doi.org/10.1103/PhysRev.171.1675} {\bibfield  {journal} {\bibinfo
  {journal} {Phys. Rev.}\ }\textbf {\bibinfo {volume} {171}},\ \bibinfo {pages}
  {1675} (\bibinfo {year} {1968})},\ \bibinfo {note} {[Erratum:
  \href{https://doi.org/10.1103/PhysRev.174.2169}{Phys.\ Rev.\ \textbf{174},
  2169 (1968)}]}\BibitemShut {NoStop}%
\bibitem [{\citenamefont {Atwood}\ and\ \citenamefont
  {Marciano}(1990)}]{Atwood:1989em}%
  \BibitemOpen
  \bibfield  {author} {\bibinfo {author} {\bibfnamefont {D.}~\bibnamefont
  {Atwood}}\ and\ \bibinfo {author} {\bibfnamefont {W.~J.}\ \bibnamefont
  {Marciano}},\ }\bibfield  {title} {\bibinfo {title} {Radiative corrections
  and semileptonic {$B$} decays},\ }\href
  {https://doi.org/10.1103/PhysRevD.41.1736} {\bibfield  {journal} {\bibinfo
  {journal} {Phys. Rev. D}\ }\textbf {\bibinfo {volume} {41}},\ \bibinfo
  {pages} {1736} (\bibinfo {year} {1990})}\BibitemShut {NoStop}%
\bibitem [{\citenamefont {Cal\'\i{}}\ \emph {et~al.}(2019)\citenamefont
  {Cal\'\i{}}, \citenamefont {Klaver}, \citenamefont {Rotondo},\ and\
  \citenamefont {Sciascia}}]{Cali:2019nwp}%
  \BibitemOpen
  \bibfield  {author} {\bibinfo {author} {\bibfnamefont {S.}~\bibnamefont
  {Cal\'\i{}}}, \bibinfo {author} {\bibfnamefont {S.}~\bibnamefont {Klaver}},
  \bibinfo {author} {\bibfnamefont {M.}~\bibnamefont {Rotondo}},\ and\ \bibinfo
  {author} {\bibfnamefont {B.}~\bibnamefont {Sciascia}},\ }\bibfield  {title}
  {\bibinfo {title} {{Impacts of radiative corrections on measurements of
  lepton flavour universality in $B \to D \ell \nu_{\ell}$ decays}},\ }\href
  {https://doi.org/10.1140/epjc/s10052-019-7254-x} {\bibfield  {journal}
  {\bibinfo  {journal} {Eur. Phys. J. C}\ }\textbf {\bibinfo {volume} {79}},\
  \bibinfo {pages} {744} (\bibinfo {year} {2019})},\ \Eprint
  {https://arxiv.org/abs/1905.02702} {arXiv:1905.02702 [hep-ph]} \BibitemShut
  {NoStop}%
\bibitem [{\citenamefont {de~Boer}\ \emph {et~al.}(2018)\citenamefont
  {de~Boer}, \citenamefont {Kitahara},\ and\ \citenamefont
  {Nisandzic}}]{deBoer:2018ipi}%
  \BibitemOpen
  \bibfield  {author} {\bibinfo {author} {\bibfnamefont {S.}~\bibnamefont
  {de~Boer}}, \bibinfo {author} {\bibfnamefont {T.}~\bibnamefont {Kitahara}},\
  and\ \bibinfo {author} {\bibfnamefont {I.}~\bibnamefont {Nisandzic}},\
  }\bibfield  {title} {\bibinfo {title} {Soft-photon corrections to {$\bar{B}
  \to D \tau^{-} \bar{\nu}_{\tau}$} relative to {$\bar{B} \to D \mu^{-}
  \bar{\nu}_{\mu}$}},\ }\href {https://doi.org/10.1103/PhysRevLett.120.261804}
  {\bibfield  {journal} {\bibinfo  {journal} {Phys. Rev. Lett.}\ }\textbf
  {\bibinfo {volume} {120}},\ \bibinfo {pages} {261804} (\bibinfo {year}
  {2018})},\ \Eprint {https://arxiv.org/abs/1803.05881} {arXiv:1803.05881
  [hep-ph]} \BibitemShut {NoStop}%
\bibitem [{\citenamefont {Riggio}\ \emph {et~al.}(2018)\citenamefont {Riggio},
  \citenamefont {Salerno},\ and\ \citenamefont {Simula}}]{Riggio:2017zwh}%
  \BibitemOpen
  \bibfield  {author} {\bibinfo {author} {\bibfnamefont {L.}~\bibnamefont
  {Riggio}}, \bibinfo {author} {\bibfnamefont {G.}~\bibnamefont {Salerno}},\
  and\ \bibinfo {author} {\bibfnamefont {S.}~\bibnamefont {Simula}},\
  }\bibfield  {title} {\bibinfo {title} {{Extraction of $|V_{cd}|$ and
  $|V_{cs}|$ from experimental decay rates using lattice QCD $D \to \pi(K) \ell
  \nu$ form factors}},\ }\href {https://doi.org/10.1140/epjc/s10052-018-5943-5}
  {\bibfield  {journal} {\bibinfo  {journal} {Eur. Phys. J. C}\ }\textbf
  {\bibinfo {volume} {78}},\ \bibinfo {pages} {501} (\bibinfo {year} {2018})},\
  \Eprint {https://arxiv.org/abs/1706.03657} {arXiv:1706.03657 [hep-lat]}
  \BibitemShut {NoStop}%
\bibitem [{\citenamefont {Davies}\ \emph {et~al.}(2010)\citenamefont {Davies},
  \citenamefont {McNeile}, \citenamefont {Follana}, \citenamefont {Lepage},
  \citenamefont {Na},\ and\ \citenamefont {Shigemitsu}}]{Davies:2010ip}%
  \BibitemOpen
  \bibfield  {author} {\bibinfo {author} {\bibfnamefont {C.~T.~H.}\
  \bibnamefont {Davies}}, \bibinfo {author} {\bibfnamefont {C.}~\bibnamefont
  {McNeile}}, \bibinfo {author} {\bibfnamefont {E.}~\bibnamefont {Follana}},
  \bibinfo {author} {\bibfnamefont {G.~P.}\ \bibnamefont {Lepage}}, \bibinfo
  {author} {\bibfnamefont {H.}~\bibnamefont {Na}},\ and\ \bibinfo {author}
  {\bibfnamefont {J.}~\bibnamefont {Shigemitsu}},\ }\bibfield  {title}
  {\bibinfo {title} {{Update: Precision $D_s$ decay constant from full lattice
  QCD using very fine lattices}},\ }\href
  {https://doi.org/10.1103/PhysRevD.82.114504} {\bibfield  {journal} {\bibinfo
  {journal} {Phys. Rev. D}\ }\textbf {\bibinfo {volume} {82}},\ \bibinfo
  {pages} {114504} (\bibinfo {year} {2010})},\ \Eprint
  {https://arxiv.org/abs/1008.4018} {arXiv:1008.4018 [hep-lat]} \BibitemShut
  {NoStop}%
\bibitem [{\citenamefont {Bazavov}\ \emph {et~al.}(2012)\citenamefont {Bazavov}
  \emph {et~al.}}]{FermilabLattice:2011njy}%
  \BibitemOpen
  \bibfield  {author} {\bibinfo {author} {\bibfnamefont {A.}~\bibnamefont
  {Bazavov}} \emph {et~al.} (\bibinfo {collaboration} {Fermilab Lattice,
  MILC}),\ }\bibfield  {title} {\bibinfo {title} {{$B$- and $D$-meson decay
  constants from three-flavor lattice QCD}},\ }\href
  {https://doi.org/10.1103/PhysRevD.85.114506} {\bibfield  {journal} {\bibinfo
  {journal} {Phys. Rev. D}\ }\textbf {\bibinfo {volume} {85}},\ \bibinfo
  {pages} {114506} (\bibinfo {year} {2012})},\ \Eprint
  {https://arxiv.org/abs/1112.3051} {arXiv:1112.3051 [hep-lat]} \BibitemShut
  {NoStop}%
\bibitem [{\citenamefont {Na}\ \emph {et~al.}(2012)\citenamefont {Na},
  \citenamefont {Davies}, \citenamefont {Follana}, \citenamefont {Lepage},\
  and\ \citenamefont {Shigemitsu}}]{Na:2012iu}%
  \BibitemOpen
  \bibfield  {author} {\bibinfo {author} {\bibfnamefont {H.}~\bibnamefont
  {Na}}, \bibinfo {author} {\bibfnamefont {C.~T.~H.}\ \bibnamefont {Davies}},
  \bibinfo {author} {\bibfnamefont {E.}~\bibnamefont {Follana}}, \bibinfo
  {author} {\bibfnamefont {G.~P.}\ \bibnamefont {Lepage}},\ and\ \bibinfo
  {author} {\bibfnamefont {J.}~\bibnamefont {Shigemitsu}},\ }\bibfield  {title}
  {\bibinfo {title} {{$|V_{cd}|$ from $D$-meson leptonic decays}},\ }\href
  {https://doi.org/10.1103/PhysRevD.86.054510} {\bibfield  {journal} {\bibinfo
  {journal} {Phys. Rev. D}\ }\textbf {\bibinfo {volume} {86}},\ \bibinfo
  {pages} {054510} (\bibinfo {year} {2012})},\ \Eprint
  {https://arxiv.org/abs/1206.4936} {arXiv:1206.4936 [hep-lat]} \BibitemShut
  {NoStop}%
\bibitem [{\citenamefont {Boyle}\ \emph {et~al.}(2017)\citenamefont {Boyle},
  \citenamefont {Del~Debbio}, \citenamefont {J\"uttner}, \citenamefont
  {Khamseh}, \citenamefont {Sanfilippo},\ and\ \citenamefont
  {Tsang}}]{Boyle:2017jwu}%
  \BibitemOpen
  \bibfield  {author} {\bibinfo {author} {\bibfnamefont {P.~A.}\ \bibnamefont
  {Boyle}}, \bibinfo {author} {\bibfnamefont {L.}~\bibnamefont {Del~Debbio}},
  \bibinfo {author} {\bibfnamefont {A.}~\bibnamefont {J\"uttner}}, \bibinfo
  {author} {\bibfnamefont {A.}~\bibnamefont {Khamseh}}, \bibinfo {author}
  {\bibfnamefont {F.}~\bibnamefont {Sanfilippo}},\ and\ \bibinfo {author}
  {\bibfnamefont {J.~T.}\ \bibnamefont {Tsang}},\ }\bibfield  {title} {\bibinfo
  {title} {{The decay constants $f_D$ and $f_{D_{s}}$ in the continuum limit of
  $N_f=2+1$ domain wall lattice QCD}},\ }\href
  {https://doi.org/10.1007/JHEP12(2017)008} {\bibfield  {journal} {\bibinfo
  {journal} {JHEP}\ }\textbf {\bibinfo {volume} {12}},\ \bibinfo {pages}
  {008}},\ \Eprint {https://arxiv.org/abs/1701.02644} {arXiv:1701.02644
  [hep-lat]} \BibitemShut {NoStop}%
\bibitem [{\citenamefont {Charles}\ \emph {et~al.}(2005)\citenamefont
  {Charles}, \citenamefont {Höcker}, \citenamefont {Lacker}, \citenamefont
  {Laplace}, \citenamefont {Le~Diberder}, \citenamefont {Malcles},
  \citenamefont {Ocariz}, \citenamefont {Pivk},\ and\ \citenamefont
  {Roos}}]{Charles:2004jd}%
  \BibitemOpen
  \bibfield  {author} {\bibinfo {author} {\bibfnamefont {J.}~\bibnamefont
  {Charles}}, \bibinfo {author} {\bibfnamefont {A.}~\bibnamefont {Höcker}},
  \bibinfo {author} {\bibfnamefont {H.}~\bibnamefont {Lacker}}, \bibinfo
  {author} {\bibfnamefont {S.}~\bibnamefont {Laplace}}, \bibinfo {author}
  {\bibfnamefont {F.~R.}\ \bibnamefont {Le~Diberder}}, \bibinfo {author}
  {\bibfnamefont {J.}~\bibnamefont {Malcles}}, \bibinfo {author} {\bibfnamefont
  {J.}~\bibnamefont {Ocariz}}, \bibinfo {author} {\bibfnamefont
  {M.}~\bibnamefont {Pivk}},\ and\ \bibinfo {author} {\bibfnamefont
  {L.}~\bibnamefont {Roos}} (\bibinfo {collaboration} {CKMfitter Group}),\
  }\bibfield  {title} {\bibinfo {title} {{$CP$ violation and the CKM matrix:
  Assessing the impact of the asymmetric $B$ factories}},\ }\href
  {https://doi.org/10.1140/epjc/s2005-02169-1} {\bibfield  {journal} {\bibinfo
  {journal} {Eur. Phys. J. C}\ }\textbf {\bibinfo {volume} {41}},\ \bibinfo
  {pages} {1} (\bibinfo {year} {2005})},\ \bibinfo {note} {updated results and
  plots available at:
  \href{http://ckmfitter.in2p3.fr/}{http://ckmfitter.in2p3.fr/}},\ \Eprint
  {https://arxiv.org/abs/hep-ph/0406184} {arXiv:hep-ph/0406184} \BibitemShut
  {NoStop}%
\bibitem [{\citenamefont {Bona}\ \emph {et~al.}(2022)\citenamefont {Bona} \emph
  {et~al.}}]{UTfit:2022hsi}%
  \BibitemOpen
  \bibfield  {author} {\bibinfo {author} {\bibfnamefont {M.}~\bibnamefont
  {Bona}} \emph {et~al.} (\bibinfo {collaboration} {UTfit}),\ }\bibfield
  {title} {\bibinfo {title} {New {UTfit} analysis of the unitarity triangle in
  the {Cabibbo-Kobayashi-Maskawa} scheme},\ }\href@noop {} {\  (\bibinfo {year}
  {2022})},\ \Eprint {https://arxiv.org/abs/2212.03894} {arXiv:2212.03894
  [hep-ph]} \BibitemShut {NoStop}%
\bibitem [{\citenamefont {Buras}\ \emph {et~al.}(1994)\citenamefont {Buras},
  \citenamefont {Lautenbacher},\ and\ \citenamefont
  {Ostermaier}}]{Buras:1994ec}%
  \BibitemOpen
  \bibfield  {author} {\bibinfo {author} {\bibfnamefont {A.~J.}\ \bibnamefont
  {Buras}}, \bibinfo {author} {\bibfnamefont {M.~E.}\ \bibnamefont
  {Lautenbacher}},\ and\ \bibinfo {author} {\bibfnamefont {G.}~\bibnamefont
  {Ostermaier}},\ }\bibfield  {title} {\bibinfo {title} {{Waiting for the top
  quark mass, $K^+\to\pi^+\nu\bar{\nu}$, $B_s^0$-$\bar{B}_s^0$ mixing, and $CP$
  asymmetries in $B$ decays}},\ }\href
  {https://doi.org/10.1103/PhysRevD.50.3433} {\bibfield  {journal} {\bibinfo
  {journal} {Phys. Rev. D}\ }\textbf {\bibinfo {volume} {50}},\ \bibinfo
  {pages} {3433} (\bibinfo {year} {1994})},\ \Eprint
  {https://arxiv.org/abs/hep-ph/9403384} {arXiv:hep-ph/9403384} \BibitemShut
  {NoStop}%
\bibitem [{\citenamefont {Carrasco}\ \emph
  {et~al.}(2015{\natexlab{b}})\citenamefont {Carrasco}, \citenamefont {Lubicz},
  \citenamefont {Martinelli}, \citenamefont {Sachrajda}, \citenamefont
  {Tantalo}, \citenamefont {Tarantino},\ and\ \citenamefont
  {Testa}}]{Carrasco:2015xwa}%
  \BibitemOpen
  \bibfield  {author} {\bibinfo {author} {\bibfnamefont {N.}~\bibnamefont
  {Carrasco}}, \bibinfo {author} {\bibfnamefont {V.}~\bibnamefont {Lubicz}},
  \bibinfo {author} {\bibfnamefont {G.}~\bibnamefont {Martinelli}}, \bibinfo
  {author} {\bibfnamefont {C.~T.}\ \bibnamefont {Sachrajda}}, \bibinfo {author}
  {\bibfnamefont {N.}~\bibnamefont {Tantalo}}, \bibinfo {author} {\bibfnamefont
  {C.}~\bibnamefont {Tarantino}},\ and\ \bibinfo {author} {\bibfnamefont
  {M.}~\bibnamefont {Testa}},\ }\bibfield  {title} {\bibinfo {title} {{QED}
  corrections to hadronic processes in lattice {QCD}},\ }\href
  {https://doi.org/10.1103/PhysRevD.91.074506} {\bibfield  {journal} {\bibinfo
  {journal} {Phys. Rev. D}\ }\textbf {\bibinfo {volume} {91}},\ \bibinfo
  {pages} {074506} (\bibinfo {year} {2015}{\natexlab{b}})},\ \Eprint
  {https://arxiv.org/abs/1502.00257} {arXiv:1502.00257 [hep-lat]} \BibitemShut
  {NoStop}%
\bibitem [{\citenamefont {Desiderio}\ \emph {et~al.}(2021)\citenamefont
  {Desiderio} \emph {et~al.}}]{Desiderio:2020oej}%
  \BibitemOpen
  \bibfield  {author} {\bibinfo {author} {\bibfnamefont {A.}~\bibnamefont
  {Desiderio}} \emph {et~al.},\ }\bibfield  {title} {\bibinfo {title} {{First
  lattice calculation of radiative leptonic decay rates of pseudoscalar
  mesons}},\ }\href {https://doi.org/10.1103/PhysRevD.103.014502} {\bibfield
  {journal} {\bibinfo  {journal} {Phys. Rev. D}\ }\textbf {\bibinfo {volume}
  {103}},\ \bibinfo {pages} {014502} (\bibinfo {year} {2021})},\ \Eprint
  {https://arxiv.org/abs/2006.05358} {arXiv:2006.05358 [hep-lat]} \BibitemShut
  {NoStop}%
\bibitem [{\citenamefont {Frezzotti}\ \emph
  {et~al.}(2021{\natexlab{a}})\citenamefont {Frezzotti}, \citenamefont
  {Garofalo}, \citenamefont {Lubicz}, \citenamefont {Martinelli}, \citenamefont
  {Sachrajda}, \citenamefont {Sanfilippo}, \citenamefont {Simula},\ and\
  \citenamefont {Tantalo}}]{Frezzotti:2020bfa}%
  \BibitemOpen
  \bibfield  {author} {\bibinfo {author} {\bibfnamefont {R.}~\bibnamefont
  {Frezzotti}}, \bibinfo {author} {\bibfnamefont {M.}~\bibnamefont {Garofalo}},
  \bibinfo {author} {\bibfnamefont {V.}~\bibnamefont {Lubicz}}, \bibinfo
  {author} {\bibfnamefont {G.}~\bibnamefont {Martinelli}}, \bibinfo {author}
  {\bibfnamefont {C.~T.}\ \bibnamefont {Sachrajda}}, \bibinfo {author}
  {\bibfnamefont {F.}~\bibnamefont {Sanfilippo}}, \bibinfo {author}
  {\bibfnamefont {S.}~\bibnamefont {Simula}},\ and\ \bibinfo {author}
  {\bibfnamefont {N.}~\bibnamefont {Tantalo}},\ }\bibfield  {title} {\bibinfo
  {title} {{Comparison of lattice QCD+QED predictions for radiative leptonic
  decays of light mesons with experimental data}},\ }\href
  {https://doi.org/10.1103/PhysRevD.103.053005} {\bibfield  {journal} {\bibinfo
   {journal} {Phys. Rev. D}\ }\textbf {\bibinfo {volume} {103}},\ \bibinfo
  {pages} {053005} (\bibinfo {year} {2021}{\natexlab{a}})},\ \Eprint
  {https://arxiv.org/abs/2012.02120} {arXiv:2012.02120 [hep-ph]} \BibitemShut
  {NoStop}%
\bibitem [{\citenamefont {Frezzotti}\ \emph
  {et~al.}(2021{\natexlab{b}})\citenamefont {Frezzotti}, \citenamefont
  {Gagliardi}, \citenamefont {Lubicz}, \citenamefont {Sanfilippo},\ and\
  \citenamefont {Simula}}]{Frezzotti:2021slr}%
  \BibitemOpen
  \bibfield  {author} {\bibinfo {author} {\bibfnamefont {R.}~\bibnamefont
  {Frezzotti}}, \bibinfo {author} {\bibfnamefont {G.}~\bibnamefont
  {Gagliardi}}, \bibinfo {author} {\bibfnamefont {V.}~\bibnamefont {Lubicz}},
  \bibinfo {author} {\bibfnamefont {F.}~\bibnamefont {Sanfilippo}},\ and\
  \bibinfo {author} {\bibfnamefont {S.}~\bibnamefont {Simula}},\ }\bibfield
  {title} {\bibinfo {title} {{Rotated twisted-mass: a convenient regularization
  scheme for isospin breaking QCD and QED lattice calculations}},\ }\href
  {https://doi.org/10.1140/epja/s10050-021-00579-5} {\bibfield  {journal}
  {\bibinfo  {journal} {Eur. Phys. J. A}\ }\textbf {\bibinfo {volume} {57}},\
  \bibinfo {pages} {282} (\bibinfo {year} {2021}{\natexlab{b}})},\ \Eprint
  {https://arxiv.org/abs/2106.07107} {arXiv:2106.07107 [hep-lat]} \BibitemShut
  {NoStop}%
\bibitem [{\citenamefont {Gagliardi}\ \emph {et~al.}(2022)\citenamefont
  {Gagliardi}, \citenamefont {Sanfilippo}, \citenamefont {Simula},
  \citenamefont {Lubicz}, \citenamefont {Mazzetti}, \citenamefont {Martinelli},
  \citenamefont {Sachrajda},\ and\ \citenamefont
  {Tantalo}}]{Gagliardi:2022szw}%
  \BibitemOpen
  \bibfield  {author} {\bibinfo {author} {\bibfnamefont {G.}~\bibnamefont
  {Gagliardi}}, \bibinfo {author} {\bibfnamefont {F.}~\bibnamefont
  {Sanfilippo}}, \bibinfo {author} {\bibfnamefont {S.}~\bibnamefont {Simula}},
  \bibinfo {author} {\bibfnamefont {V.}~\bibnamefont {Lubicz}}, \bibinfo
  {author} {\bibfnamefont {F.}~\bibnamefont {Mazzetti}}, \bibinfo {author}
  {\bibfnamefont {G.}~\bibnamefont {Martinelli}}, \bibinfo {author}
  {\bibfnamefont {C.~T.}\ \bibnamefont {Sachrajda}},\ and\ \bibinfo {author}
  {\bibfnamefont {N.}~\bibnamefont {Tantalo}},\ }\bibfield  {title} {\bibinfo
  {title} {{Virtual photon emission in leptonic decays of charged pseudoscalar
  mesons}},\ }\href {https://doi.org/10.1103/PhysRevD.105.114507} {\bibfield
  {journal} {\bibinfo  {journal} {Phys. Rev. D}\ }\textbf {\bibinfo {volume}
  {105}},\ \bibinfo {pages} {114507} (\bibinfo {year} {2022})},\ \Eprint
  {https://arxiv.org/abs/2202.03833} {arXiv:2202.03833 [hep-lat]} \BibitemShut
  {NoStop}%
\bibitem [{\citenamefont {Stanzione}\ \emph {et~al.}(2020)\citenamefont
  {Stanzione}, \citenamefont {West}, \citenamefont {Evans}, \citenamefont
  {Minard}, \citenamefont {Ghattas},\ and\ \citenamefont {Panda}}]{Frontera}%
  \BibitemOpen
  \bibfield  {author} {\bibinfo {author} {\bibfnamefont {D.}~\bibnamefont
  {Stanzione}}, \bibinfo {author} {\bibfnamefont {J.}~\bibnamefont {West}},
  \bibinfo {author} {\bibfnamefont {R.~T.}\ \bibnamefont {Evans}}, \bibinfo
  {author} {\bibfnamefont {T.}~\bibnamefont {Minard}}, \bibinfo {author}
  {\bibfnamefont {O.}~\bibnamefont {Ghattas}},\ and\ \bibinfo {author}
  {\bibfnamefont {D.~K.}\ \bibnamefont {Panda}},\ }\bibfield  {title} {\bibinfo
  {title} {Frontera: The evolution of leadership computing at the national
  science foundation},\ }in\ \href
  {https://doi.org/doi:10.1145/3311790.3396656} {\emph {\bibinfo {booktitle}
  {Practice and Experience in Advanced Research Computing 2020}}}\ (\bibinfo
  {year} {July 26–30, 2020})\BibitemShut {NoStop}%
\bibitem [{\citenamefont {Towns}\ \emph {et~al.}(2014)\citenamefont {Towns},
  \citenamefont {Cockerill}, \citenamefont {Dahan}, \citenamefont {Foster},
  \citenamefont {Gaither}, \citenamefont {Grimshaw}, \citenamefont {Hazlewood},
  \citenamefont {Lathrop}, \citenamefont {Lifka}, \citenamefont {Peterson},
  \citenamefont {Roskies}, \citenamefont {Scott},\ and\ \citenamefont
  {N.}}]{XSEDE}%
  \BibitemOpen
  \bibfield  {author} {\bibinfo {author} {\bibfnamefont {J.}~\bibnamefont
  {Towns}}, \bibinfo {author} {\bibfnamefont {T.}~\bibnamefont {Cockerill}},
  \bibinfo {author} {\bibfnamefont {M.}~\bibnamefont {Dahan}}, \bibinfo
  {author} {\bibfnamefont {I.}~\bibnamefont {Foster}}, \bibinfo {author}
  {\bibfnamefont {K.}~\bibnamefont {Gaither}}, \bibinfo {author} {\bibfnamefont
  {A.}~\bibnamefont {Grimshaw}}, \bibinfo {author} {\bibfnamefont
  {V.}~\bibnamefont {Hazlewood}}, \bibinfo {author} {\bibfnamefont
  {S.}~\bibnamefont {Lathrop}}, \bibinfo {author} {\bibfnamefont
  {D.}~\bibnamefont {Lifka}}, \bibinfo {author} {\bibfnamefont {G.~D.}\
  \bibnamefont {Peterson}}, \bibinfo {author} {\bibfnamefont {R.}~\bibnamefont
  {Roskies}}, \bibinfo {author} {\bibfnamefont {J.~R.}\ \bibnamefont {Scott}},\
  and\ \bibinfo {author} {\bibfnamefont {W.-D.}\ \bibnamefont {N.}},\
  }\bibfield  {title} {\bibinfo {title} {{XSEDE}: Accelerating scientific
  discovery},\ }\href {https://doi.org/doi:10.1109/MCSE.2014.80} {\bibfield
  {journal} {\bibinfo  {journal} {Comput. Sci. Eng.}\ }\textbf {\bibinfo
  {volume} {16}},\ \bibinfo {pages} {62} (\bibinfo {year} {2014})}\BibitemShut
  {NoStop}%
\bibitem [{\citenamefont {Bailey}\ \emph {et~al.}(2009)\citenamefont {Bailey}
  \emph {et~al.}}]{Bailey:2008wp}%
  \BibitemOpen
  \bibfield  {author} {\bibinfo {author} {\bibfnamefont {J.~A.}\ \bibnamefont
  {Bailey}} \emph {et~al.},\ }\bibfield  {title} {\bibinfo {title} {{The $B \to
  \pi \ell \nu$ semileptonic form factor from three-flavor lattice QCD: A
  model-independent determination of $|V_{ub}|$}},\ }\href
  {https://doi.org/10.1103/PhysRevD.79.054507} {\bibfield  {journal} {\bibinfo
  {journal} {Phys. Rev. D}\ }\textbf {\bibinfo {volume} {79}},\ \bibinfo
  {pages} {054507} (\bibinfo {year} {2009})},\ \Eprint
  {https://arxiv.org/abs/0811.3640} {arXiv:0811.3640 [hep-lat]} \BibitemShut
  {NoStop}%
\bibitem [{\citenamefont {Monahan}\ \emph {et~al.}(2013)\citenamefont
  {Monahan}, \citenamefont {Shigemitsu},\ and\ \citenamefont
  {Horgan}}]{Monahan:2012dq}%
  \BibitemOpen
  \bibfield  {author} {\bibinfo {author} {\bibfnamefont {C.}~\bibnamefont
  {Monahan}}, \bibinfo {author} {\bibfnamefont {J.}~\bibnamefont
  {Shigemitsu}},\ and\ \bibinfo {author} {\bibfnamefont {R.}~\bibnamefont
  {Horgan}},\ }\bibfield  {title} {\bibinfo {title} {{Matching lattice and
  continuum axial-vector and vector currents with nonrelativistic QCD and
  highly improved staggered quarks}},\ }\href
  {https://doi.org/10.1103/PhysRevD.87.034017} {\bibfield  {journal} {\bibinfo
  {journal} {Phys. Rev. D}\ }\textbf {\bibinfo {volume} {87}},\ \bibinfo
  {pages} {034017} (\bibinfo {year} {2013})},\ \Eprint
  {https://arxiv.org/abs/1211.6966} {arXiv:1211.6966 [hep-lat]} \BibitemShut
  {NoStop}%
\bibitem [{\citenamefont {El-Khadra}\ \emph {et~al.}(1997)\citenamefont
  {El-Khadra}, \citenamefont {Kronfeld},\ and\ \citenamefont
  {Mackenzie}}]{El-Khadra:1996wdx}%
  \BibitemOpen
  \bibfield  {author} {\bibinfo {author} {\bibfnamefont {A.~X.}\ \bibnamefont
  {El-Khadra}}, \bibinfo {author} {\bibfnamefont {A.~S.}\ \bibnamefont
  {Kronfeld}},\ and\ \bibinfo {author} {\bibfnamefont {P.~B.}\ \bibnamefont
  {Mackenzie}},\ }\bibfield  {title} {\bibinfo {title} {{Massive fermions in
  lattice gauge theory}},\ }\href {https://doi.org/10.1103/PhysRevD.55.3933}
  {\bibfield  {journal} {\bibinfo  {journal} {Phys. Rev. D}\ }\textbf {\bibinfo
  {volume} {55}},\ \bibinfo {pages} {3933} (\bibinfo {year} {1997})},\ \Eprint
  {https://arxiv.org/abs/hep-lat/9604004} {arXiv:hep-lat/9604004} \BibitemShut
  {NoStop}%
\bibitem [{\citenamefont {El-Khadra}\ \emph {et~al.}(1998)\citenamefont
  {El-Khadra}, \citenamefont {Kronfeld}, \citenamefont {Mackenzie},
  \citenamefont {Ryan},\ and\ \citenamefont {Simone}}]{El-Khadra:1997kgw}%
  \BibitemOpen
  \bibfield  {author} {\bibinfo {author} {\bibfnamefont {A.~X.}\ \bibnamefont
  {El-Khadra}}, \bibinfo {author} {\bibfnamefont {A.~S.}\ \bibnamefont
  {Kronfeld}}, \bibinfo {author} {\bibfnamefont {P.~B.}\ \bibnamefont
  {Mackenzie}}, \bibinfo {author} {\bibfnamefont {S.~M.}\ \bibnamefont
  {Ryan}},\ and\ \bibinfo {author} {\bibfnamefont {J.~N.}\ \bibnamefont
  {Simone}},\ }\bibfield  {title} {\bibinfo {title} {{$B$ and $D$ meson decay
  constants in lattice QCD}},\ }\href
  {https://doi.org/10.1103/PhysRevD.58.014506} {\bibfield  {journal} {\bibinfo
  {journal} {Phys. Rev. D}\ }\textbf {\bibinfo {volume} {58}},\ \bibinfo
  {pages} {014506} (\bibinfo {year} {1998})},\ \Eprint
  {https://arxiv.org/abs/hep-ph/9711426} {arXiv:hep-ph/9711426} \BibitemShut
  {NoStop}%
\bibitem [{\citenamefont {Ledoit}\ and\ \citenamefont
  {Wolf}(2004)}]{Ledoit:2004}%
  \BibitemOpen
  \bibfield  {author} {\bibinfo {author} {\bibfnamefont {O.}~\bibnamefont
  {Ledoit}}\ and\ \bibinfo {author} {\bibfnamefont {M.}~\bibnamefont {Wolf}},\
  }\bibfield  {title} {\bibinfo {title} {A well-conditioned estimator for
  large-dimensional covariance matrices},\ }\href
  {https://doi.org/https://doi.org/10.1016/S0047-259X(03)00096-4} {\bibfield
  {journal} {\bibinfo  {journal} {J. Multivariate Anal.}\ }\textbf {\bibinfo
  {volume} {88}},\ \bibinfo {pages} {365} (\bibinfo {year} {2004})}\BibitemShut
  {NoStop}%
\bibitem [{\citenamefont {Rinaldi}\ \emph {et~al.}(2019)\citenamefont
  {Rinaldi}, \citenamefont {Syritsyn}, \citenamefont {Wagman}, \citenamefont
  {Buchoff}, \citenamefont {Schroeder},\ and\ \citenamefont
  {Wasem}}]{Rinaldi:2019thf}%
  \BibitemOpen
  \bibfield  {author} {\bibinfo {author} {\bibfnamefont {E.}~\bibnamefont
  {Rinaldi}}, \bibinfo {author} {\bibfnamefont {S.}~\bibnamefont {Syritsyn}},
  \bibinfo {author} {\bibfnamefont {M.~L.}\ \bibnamefont {Wagman}}, \bibinfo
  {author} {\bibfnamefont {M.~I.}\ \bibnamefont {Buchoff}}, \bibinfo {author}
  {\bibfnamefont {C.}~\bibnamefont {Schroeder}},\ and\ \bibinfo {author}
  {\bibfnamefont {J.}~\bibnamefont {Wasem}},\ }\bibfield  {title} {\bibinfo
  {title} {{Lattice QCD determination of neutron-antineutron matrix elements
  with physical quark masses}},\ }\href
  {https://doi.org/10.1103/PhysRevD.99.074510} {\bibfield  {journal} {\bibinfo
  {journal} {Phys. Rev. D}\ }\textbf {\bibinfo {volume} {99}},\ \bibinfo
  {pages} {074510} (\bibinfo {year} {2019})},\ \Eprint
  {https://arxiv.org/abs/1901.07519} {arXiv:1901.07519 [hep-lat]} \BibitemShut
  {NoStop}%
\bibitem [{\citenamefont {Ayer}\ \emph {et~al.}(1955)\citenamefont {Ayer},
  \citenamefont {Brunk}, \citenamefont {Ewing}, \citenamefont {Reid},\ and\
  \citenamefont {Silverman}}]{Ayer:1955}%
  \BibitemOpen
  \bibfield  {author} {\bibinfo {author} {\bibfnamefont {M.}~\bibnamefont
  {Ayer}}, \bibinfo {author} {\bibfnamefont {H.~D.}\ \bibnamefont {Brunk}},
  \bibinfo {author} {\bibfnamefont {G.~M.}\ \bibnamefont {Ewing}}, \bibinfo
  {author} {\bibfnamefont {W.~T.}\ \bibnamefont {Reid}},\ and\ \bibinfo
  {author} {\bibfnamefont {E.}~\bibnamefont {Silverman}},\ }\bibfield  {title}
  {\bibinfo {title} {An empirical distribution function for sampling with
  incomplete information},\ }\href {https://doi.org/10.1214/aoms/1177728423}
  {\bibfield  {journal} {\bibinfo  {journal} {Ann. Math. Statist.}\ }\textbf
  {\bibinfo {volume} {26}},\ \bibinfo {pages} {641 } (\bibinfo {year}
  {1955})}\BibitemShut {NoStop}%
\end{thebibliography}%


%

\end{document}